\documentclass[a4paper,11pt]{article}

\usepackage[T1]{fontenc}

\usepackage[numbers,sort&compress]{natbib}
\usepackage{latexsym}
\usepackage{longtable}
\usepackage{epsfig}
\usepackage{graphicx,bbm,psfrag} 
\usepackage{amssymb,amsmath,amsthm,booktabs,mathtools}
\usepackage{bm}
\usepackage{color}
\usepackage{dutchcal}
\usepackage{hyperref}
\hypersetup{
    colorlinks,
    citecolor=red,
    linkcolor=blue,
}

\numberwithin{equation}{section}

\newtheorem{theorem}{Theorem}[section]
\newtheorem{rema}{Remark}[section]

\newtheorem{defi}[theorem]{Definition}
\newtheorem{lemma}[theorem]{Lemma}
\newtheorem{corol}[theorem]{Corollary}
\newtheorem{assum}[theorem]{Assumption}
\newtheorem{algo}[theorem]{Algorithm}

\def\doi#1{\\ DOI: \href{https://www.doi.org/#1}{#1}}

\newcommand{\bc}{\begin{center}}
\newcommand{\ec}{\end{center}}
\def\ba#1{\begin{array}{#1}\displaystyle}
\newcommand{\ea}{\end{array}}
\newcommand{\z}{\\[0.2cm] \displaystyle}
\newcommand{\beq}{\begin{equation}}
\newcommand{\eeq}{\end{equation}}
\newcommand{\beqa}{\begin{eqnarray}}
\newcommand{\eeqa}{\end{eqnarray}}
\newcommand{\no}{\nonumber}
\newcommand{\n}{\nonumber\\}
\newcommand{\bi}{\begin{itemize}}
\newcommand{\ei}{\end{itemize}}

\def\ssect#1{\subsection{#1}}

\def\lt#1{\left#1}
\def\rt#1{\right#1}
\def\t#1{\tilde{#1}}
\def\h#1{\hat{#1}}
\def\b#1{\bar{#1}}
\def\frc#1#2{\frac{#1}{#2}}

\def\bs#1{\boldsymbol{#1}}

\DeclareMathOperator{\argmax}{{\rm argmax}}

\DeclareMathOperator{\argmin}{{\rm argmin}}
\newcommand{\p}{\partial}
\newcommand{\dbra}{[\![}
\newcommand{\dket}{]\!]}
\newcommand{\Z}{{\mathbb{Z}}}
\newcommand{\N}{{\mathbb{N}}}
\newcommand{\R}{{\mathbb{R}}}

\newcommand{\Rpu}{\mathbb{R}_{+\!\!\uparrow}}
\newcommand{\Rne}{\mathbb{R}_{\neq}}

\newcommand{\ep}{\epsilon}
\newcommand{\varep}{\varepsilon}

\newcommand{\re}{{\rm e}}
\newcommand{\dd}{{\rm d}}
\newcommand{\1}{{\bf 1}}
\newcommand{\0}{{\bf 0}}
\DeclareMathOperator{\sgn}{sgn}
\DeclareMathOperator{\diag}{{\rm diag}}
\DeclareMathOperator{\Mat}{{\rm Mat}}

\DeclareMathOperator{\dist}{{\rm dist}}

\DeclareMathOperator{\core}{{\rm core}}
\DeclareMathOperator{\reg}{{\rm reg}}
\DeclareMathOperator{\neigh}{{\rm neigh}}

\DeclareMathOperator{\wlim}{{\rm w--lim}}
\def\wlimu#1{\underset{#1}{\wlim}\,}

\newcommand{\halmos}{\rule{1ex}{1.4ex}}
\newcommand{\eproof}{\hspace*{\fill}\mbox{$\halmos$}}

\begin{document}

\begin{center}
{\Large {\sc Where solitons are in a KdV soliton gas}}

\vspace{1cm}

{\large Benjamin Doyon
}

\vspace{0.2cm}
Department of Mathematics, King’s College London, Strand, London WC2R 2LS, U.K.
\ec

The Korteweg-De Vries (KdV) equation is a paradigmatic model of integrable classical fields, admitting solitonic solutions where localised profiles interact by elastic, factorised scattering. When solitons interact, their shapes are modified, and individual solitons cannot be identified anymore. This is a key problem in soliton gas theory where one needs a density of solitons per unit distance: can we locate asymptotic objects -- quasiparticles -- even at finite times in the macroscopic limit? Quasiparticles' positions should satisfy locality properties, the field in any mesoscopic region being determined by quasiparticles lying there. We solve this problem for the KdV equation. We define solitons' positions as functions of their spectral and impact parameters via an expansion bijection generalising that of the hard rod gas. We show that these ``KdV rods'' solve the semiclassical Bethe equations, implementing the collision rate ansatz. We show that they satisfy a repulsion condition generalising the rods' lengths to the KdV two-body scattering shifts and closely related to the condensate condition. This gives a full and explicit characterisation of multi-soliton fields in terms of solitons' positions instead of impact parameters. We construct fluid-cell projections, that project out solitons lying outside mesoscopic regions, and show that the result is a field composed solely of the remaining solitons, unchanged in this region and supported there up to exponentially decaying corrections. From this we show that in the weak limit, conserved densities are given by the empirical density for these solitons' positions. Thus we have associated a kinetic theory of quasiparticles to the KdV equation.
The results hold under natural assumptions on spectral and impact parameters, and no randomness  is required. Our proof is based on novel solitonic tau functions. This is a crucial step towards rigorously deriving the kinetic equation of the KdV soliton gas, first proposed by Gennady El in 2003. We believe the methods are generalisable to other solitonic models.


\medskip

\tableofcontents

\section{Introduction}

{\em KdV solitons---} The Korteweg-de-Vries (KdV) equation
\beq\label{kdv}
	\p_t u + 6 u\p_x u + \p_x^3 u = 0
\eeq
is a paradigmatic model of integrable PDEs \cite{miura1968korteweg,gardner1974korteweg,zakharov1991integrability}. One of its main properties is that it admits solutions composed of many solitons. A soliton is a field configuration with exponentially localised energy, that keeps its shape under the dynamics \eqref{kdv} and propagates with constant velocity, like a particle. Solitonic models such as KdV have received a lot of attention recently in the context of {\em soliton gases}, where one sees random fields $u$ as distributions of solitons, with properties akin to a gas or fluid. This was first proposed by Zakharov \cite{zakharov1971kinetic}, and Gennady El and co-authors \cite{EL2003374,el2005kinetic,el2011kinetic,pavlov2012generalized} first conjectured the soliton gas kinetic equation governing its large-scale dynamics, see also \cite{gurevich2000statistical}. This has been extended to generalised hydrodynamics (GHD), the hydrodynamic theory for classical and quantum many-body integrable models \cite{castro2016emergent,bertini2016transport}, see \cite{el2021soliton,suret2024soliton} for reviews on soliton gas, and \cite{doyon2020lecture,bastianello2022introduction,spohn2024hydrodynamic,doyon2025generalized} for reviews on GHD.

In the KdV equation, there is a 1-soliton solution for  every {\em spectral parameter} $\chi>0$, with velocity $v=4\chi^2$, and {\em impact parameter} $y\in\R$, the position of its maximum at time $t=0$:
\beq\label{1sol}
	u(x,t) = u_{\chi,y+vt}(x):= \frc{2\chi^2}{\cosh^2(\chi (x-y-vt))},\quad v := 4\chi^2,\,\chi>0,\;y\in\R.
\eeq
A $N$-soliton solution has a more complicated form \cite{gardner1974korteweg}, but it is explicit, and its scattering structure is surprisingly simple. It tends to a sum of $N$ separate one-soliton fields \eqref{1sol} as $t\to\pm\infty$, with a unique set of spectral parameters $\chi_i$'s in both asymptotic time regions (elastic scattering), and with corresponding one-soliton impact parameters $y=x_i^\pm$ in \eqref{1sol}. Further, the Wigner scattering shifts are sums of two-soliton shifts, as if the complex multi-soliton scattering had occurred as well-separated two-solitons scattering events (factorised scattering),
\beq\label{factscat}
	x_i^+ - x_i ^- = -\sum_{j=1\atop j\neq i}^N \sgn(v_i-v_j)\varphi_{ij},\quad \varphi_{ij} = \,\frc1{\chi_i} \log\Big|\frc{\chi_i-\chi_j}{\chi_i+\chi_j}\Big|,\quad v_i = 4\chi_i^2.
\eeq
For convenience we define our ``impact parameters'' as $y_i = (x_i^++x_i^-)/2$; there is one distinct $N$-soliton field $u_{\bs\chi,\bs y}$ for every $\bs\chi\in\Rpu^N=\{\bs\chi\in\R^N:0<\chi_1<\chi_2<\ldots<\chi_N\}$ and $\bs y\in\R^N$, and this is the full set.
Time evolution of the $N$-soliton field is rather simple. Spectral parameters are invariant, and impact parameters evolve linearly:
\beq\label{uevolintro}
	u(x,t) = u_{\bs \chi,\bs y+\bs v t}(x)
\eeq
solves \eqref{kdv}.

%
%
\medskip
\noindent {\em Quasiparticles---} The linearisation of time evolution, Eq.~\eqref{uevolintro}, could na\"ively be seen as a transformation to action-angle variables, known to exist from integrability \cite{babelon2003introduction}. However, in many-body scattering systems, linearisation of time evolution is always achieved by going to asymptotic coordinates, even in non-integrable models \cite{derezinski2013scattering} (see e.g.~\cite[Thm 3]{hubacher1989classical})! Elastic, factorised scattering is truly a consequence of integrability, but it is of more conceptual than practical use.

One of the most important aspect of integrability, both for conceptual and practical use, is yet another notion: the existence of {\em stable quasiparticles}. Within the KdV context, this informally says that {\em in a large family of multi-soliton configurations, there is a meaningful notion of individual solitons' positions $x_i^{(t)}$ at all times $t$, up to a ``small'' imprecision $\Delta x \ll N$ as $N\to\infty$}. This is why it makes sense to talk about a gas of solitons: it is a gas of quasiparticles.

But what is a ``meaningful'' notion of solitons' positions? In a multi-soliton field, there is of course no universal notion of a soliton's position other then at asymptotically large times, or more generally at well-separated impact parameters -- solitons' shapes are modified due to interactions when their distances are comparable to their sizes, and individual solitons are not manifestly identifiable anymore in terms of where the energy is localised, or where the field takes its maximum values. 

The most basic property of meaningful solitons' positions is that they should describe weak limits of conserved densities. Let $\mathtt P_k[u](x)$ be the $k$'th conserved density admitted by the KdV equation \cite{miura1968korteweg}. This is a polynomial in $u,\,\p_xu,\,\p_x^2 u,\ldots$. With an appropriate normalisation, it takes total values $\int_{-\infty}^\infty \dd x\,\mathtt P_k[u_{\bs\chi,\bs y}](x)=\sum_{i=1}^N \chi_i^{2k+1}$ ($k=0,1,2,\ldots$) on $N$-soliton fields. Then it should be that for every Schwartz function $f:\R\to\R$, and every sequence with adequate ``finite-density'' constraints\footnote{Usually, one takes an appropriate family of measures, parameterised by $N$, on $N$-soliton configurations, such as local generalised Gibbs ensembles \cite{doyon2020lecture}, and either the left-hand side of \eqref{euler} holds a.s.~under this measure, or holds on average. But in general it is not necessary to consider measures.} $\N\ni N\,\mapsto \bs\chi\in\Rpu^N,\,\bs y\in\R^N$:
\beq\label{euler}
	\lim_{N\to\infty}\Bigg[\int \dd x\, f(x)  \mathtt P_k[u_{\bs\chi,\bs y}](Nx)
	-
	\frc1N\sum_i f(x_i/N)\chi_i^{2k+1}\Bigg]=0.
\eeq
This is non-trivial: we must identify functions of spectral and impact parameters, $x_i=x_i(\bs\chi,\bs y)$, such that solitons with positions $x_i\approx x$ are those that give information about the field around $x$ as per \eqref{euler}. For application to the Euler scaling limit, it should be that with \eqref{uevolintro}, and defining trajectories as $x_i^{(t)}=x_i(\bs\chi,\bs y+\bs v t)$, Eq.~\eqref{euler} holds for a proportion of times that tends to 1, in macroscopic intervals $t\in[0,NT]$. Then the vector of quasiparticle positions $\bs x^{(Nt)}/N$ give local information about the field $u_{\bs\chi,\bs y}(Nx,Nt)$. More generally, there should be a  relation between weak limit $\rho(\chi,x,t) = \wlim_{N\to\infty} \rho_{\rm emp}(\chi,Nx,Nt)$ of the empirical density
\beq\label{empirical}
	\rho_{\rm emp}(\chi,x,t) = \sum_i \delta(x-x_i^{(t)})\delta(\chi-\chi_i)
\eeq
and the field $u_{\bs\chi,\bs y}(Nx,Nt)$ and its derivatives. The empirical density of quasiparticles is a central object in soliton gas theory and GHD; beyond the hydrodynamic equation, it also describes hydrodynamic fluctuations \cite{doyon2023emergence,doyon2023ballistic} as well as the hydrodynamic theory beyond the Euler scale \cite{gopalakrishnan2024non,yoshimura2025anomalous,mcculloch2025ballistic,hydrodynamic2025Yoshimura,doyon2026hydrodynamic}.

%
%

There is up to now no proposal for $x_i(\bs\chi,\bs y)$ that satisfy such properties. There are proposals for constructing $\rho(\chi,x,t)$ itself in different ways. For instance, one may use a ``time-of-flight'' thought experiment \cite{doyon2020lecture}. This is experimentally realised in the context of GHD in cold atomic gases \cite{malvania2021generalized}. In this procedure, at some time $Nt$, we put to 0 the field outside a mesoscopic interval $I$ centred on the position $Nx$, of length $1\ll L\ll N$, and determine the spectral parameters of solitons lying on $I$ by letting the resulting field evolve, using the soliton resolution conjecture \cite{jendrej2025resolution}. The resulting spectral density divided by $L$ gives the {\em density of states} $\rho(\chi,x,t)$: the quantity $\dd x\dd\chi\, \rho(\chi,x)$ is interpreted as the number of solitons with spectral parameters within $[\chi,\chi+\dd\chi]$ and macroscopic positions within $[x,x+\dd x]$ at macroscopic time $t$. See Fig.~\ref{figtof}. This indirect procedure to define $\rho(\chi,x,t)$ is very general, and does not even rely on integrability. But it does not guarantee the existence of meaningful solitons' trajectories.
\begin{figure}
\bc\includegraphics[width=12cm]{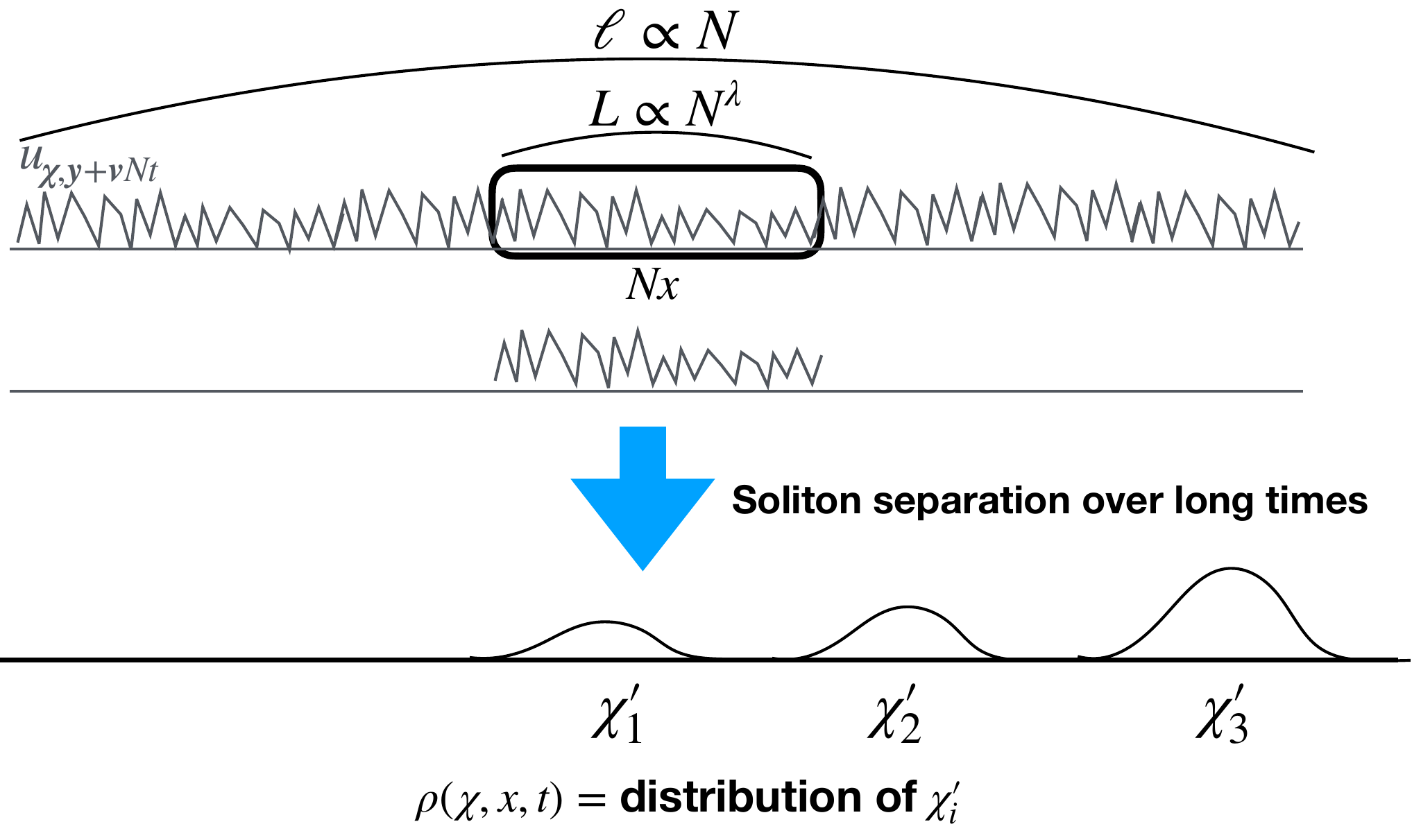}\ec
\caption{A representation of the time-of-flight thought experiment.}
\label{figtof}
\end{figure}
Another way, from soliton gas theory, is to define $\rho(\chi,x,t)$ via a Witham modulation procedure from multi-phase, finite-gap solutions  (see recent results on how $N$-soliton solutions arise from finite-gap solutions \cite{jenkins2025approximation}), and in GHD it is defined indirectly as a tool to represent the infinite set of conserved densities \cite{castro2016emergent,bertini2016transport}. The full connection with \eqref{empirical} is still missing.

\medskip
\noindent {\em The quasiparticle criterium and fluid-cell projections---} Going beyond \eqref{euler}, we argue that a good notion of solitons' positions, inspired by the time-of-flight thought experiment, is what we refer to as the {\em quasiparticle criterium}, expressed in terms of {\em fluid-cell projections}:
\beq\label{loose}
\fbox{\begin{minipage}[t]{15cm}
There exists for every closed interval $I$ a fluid-cell projection, $\mathcal F_I:u_{\bs\chi,\bs y} \mapsto u_{\bs\chi',\bs y'}$, that projects out all solitons with  $x_i\not\in I$: the projected field $u_{\bs\chi',\bs y'}$ contains only solitons lying in $I$, $\{\chi_i'\} = \{\chi_i:x_i\in I\}$. In the large-$N$ limit, for $I$ of mesoscopic length, the multi-soliton field is unchanged on $I$, $u_{\bs\chi,\bs y}|_I = u_{\bs\chi',\bs y'}|_I$, and is itself supported on $I$, $u_{\bs\chi',\bs y'}|_{\R\setminus I}=0$. 
\end{minipage}}
\eeq
%

From the quasiparticle criterium, one can show that the time-of-flight $\rho(\chi,x,t)$ is the weak limit of the empirical density \eqref{empirical}, and \eqref{euler}  follows. Importantly, fluid-cell maps give an extension of \eqref{euler} to {\em local observables that are not conserved densities}, as one can show that the fluid-cell average of local observables is equal to the average for the projected field; fluid-cell averaging is a crucial notion in hydrodynamics (see e.g.~\cite{doyon2020lecture,doyon2025nonlinear,doyon2026hydrodynamic}).

The lack of a concrete, precise notion of quasiparticles -- of solitons' positions -- has been up to now a major stumbling block in making sense of, and deriving rigorously, the hydrodynamic theory of soliton gases. Establishing the quasiparticle criterium for a large family of spectral and impact parameters is a crucial step in proving the emergent hydrodynamics and thermodynamics of soliton gases.

\medskip
\noindent {\em Results---} Because of the linear evolution of impact parameters, Eq.~\eqref{uevolintro}, it is sufficient to concentrate on $t=0$. In this paper, we propose {\em exact quasiparticle positions $x_i(\bs\chi,\bs y)$'s}, and show, under certain natural assumptions on spectral and impact parameters, that {\em they describe the weak limit of conserved densities}, Eq.~\eqref{euler}, and more generally that {\em they fulfil the quasiparticle criterium}, with rigorous expressions of the statements in Eq.~\eqref{loose}. Further, we show that these quasiparticle positions {\em satisfy a discrete version of the geometric solution to the GHD equation \cite{doyon2018geometric,hubner2024new,hubner2024existence}, referred to as the semi-classical Bethe equations (SBE) \cite{doyon2023ab,doyon2024new,bonnemain2025hamiltonian}}, which is an implementation of the collision rate ansatz from soliton gas theory \cite{el2005kinetic,doyon2018soliton,doyon2019generalized,congy2021soliton,bonnemain2022generalized,bonnemain2025soliton}, and from which the GHD equation can be proven. We leave for future works an analysis of how our assumptions on impact parameters remain valid in time, and a proof of the GHD equation from the SBE.

We define quasiparticle positions $\bs x(\bs\chi,\bs y)\in\R^N_{\neq} = \{\bs x\in\R^N:x_i\neq x_j\,\forall\,i\neq j\}$ via an explicit, finite {\em expansion algorithm}, which expands in space the freely-propagating coordinates $y_i$ in order to account for the solitonic interactions (Definition \ref{defiqp}, Algorithm \ref{leftanchor}, see also Fig.~\ref{figexpansion}). This generalises the expansion map used to describe the dynamics of hard rods in terms of that of point particles \cite{spohn2012large}. Because of this, the result is what we refer to as the {\em KdV rods}. In fact, this algorithm works for any two-body scattering shift $\varphi_{ij}<0$. We show that the KdV rods satisfy the SBE (Theorems \ref{theosbe}, \ref{theodesc}):
\beq\label{bethe}
	y_i = x_i + \frc12 \sum_{j=1\atop j\neq i}^N \sgn(x_i-x_j)
	\varphi_{ij},\quad i=1,\ldots,N,
\eeq
as well as a minimal-distance condition which generalises that of the hard rods:
\beq\label{separationintro}
	x_i-x_j \geq -\sum_{k:x_k\in[x_j,x_i)} \varphi_{ik}\quad \forall\;i,j:x_i>x_j.
\eeq
Under a simple condition of right-continuity in time of $\bs x^{(t)}$ (see Eq.~\eqref{bethet}), we show that the SBE \eqref{bethe} give a bijection between $\bs y\in\R^N$, and the set of $\bs x\in\R^N_{\neq}$ satisfying \eqref{separationintro} (Theorem \ref{theodesc}). Thus, we have a complete characterisation of the set of multi-soliton fields in terms of spectral parameters and solitons' positions, $u_{\bs\chi,\bs y} = u_{\bs \chi;\bs x(\bs\chi,\bs y)}$:
\beqa
	\lefteqn{\mbox{set of all multi-soliton fields with at least one soliton}}&&\n &=& 
	\{u_{\bs\chi;\bs x}:\,N\in\N,\,\bs\chi\in\Rpu^N,\, \bs x\in\R^N_{\neq}\mbox{ with \eqref{separationintro} and right-continuity in time}\}.
\eeqa

We show (Theorem \ref{theodesc}, Definition \ref{defimapfc}) that the associated fluid-cell projection is given by, for every closed interval $I$,
\beq\label{fcintro}
	\mathcal F_I : u_{\bs\chi,\bs y} \mapsto u_{\bs\chi',\bs y'}\quad
	\mbox{or equivalently}\quad
	\mathcal F_I : u_{\bs\chi;\bs x} \mapsto u_{\bs\chi';\bs x'}
\eeq
where $\bs\chi' = (\chi_i)_{i:x_i\in I}$, $\bs x' = \bs x(\bs\chi',\bs y') = (x_i)_{i:x_i\in I}$ and
\beq\label{fcintro2}
	\bs y' = \Big(y_i + \frc12\sum_{j\,:\,x_j\,\text{to the right of}\,I} \varphi_{ij}
	-
	\frc12\sum_{j\,:\,x_j\,\text{to the left of}\,I} \varphi_{ij}\Big)_{i:x_i\in I}.
\eeq
Clearly by the second expression in \eqref{fcintro}, it has the projection property $\mathcal F_I^2 = \mathcal F_I$ on the space of multi-soliton fields.

Finally, we show (Corollaries \ref{corolfc}, \ref{corolvanishfc}, \ref{coroboundedfc} and Theorem \ref{theofc}), under certain conditions on the density of solitons and on the spectral parameters, that the quasiparticle criterium Eq.~\eqref{loose} is fulfilled for $\mathcal F_I$. Using fluid-cell projections, we further show Eq.~\eqref{euler} (Theorem \ref{theoweak}).

Thus, re-introducing time, we map, in the case of the multi-soliton solutions, the dynamics of the KdV field to that of quasiparticle positions, which satisfy at all times
\beq\label{bethet}
	y_i  +v_i t= x_i^{(t)} + \frc12 \sum_{j=1\atop j\neq i}^N \sgn(x_i^{(t)}-x_j^{(t)})
	\varphi_{ij},\quad i=1,\ldots,N
\eeq
and are predictive for local observables of the KdV field. In this way, we have associated to the KdV time evolution a {\em kinetic theory} of effective objects, the quasiparticles, whose empirical density reproduces in the weak Euler scaling limit all conserved densities, and more generally which give local observables via the fluid-cell projection. We believe that  \eqref{bethet}, \eqref{fcintro} and \eqref{fcintro2} are universal expressions in solitonic models, and not specific to KdV (with appropriate choices of $v_i$ and $\varphi_{ij}$), and so is Algorithm \ref{leftanchor}, at least for models with negative scattering shift $\varphi_{ij}<0$.

An important aspect of our work is that we establish {\em precise and general} conditions for the fluid-cell criterium to be satisfied. These are simple conditions on spectral parameters (in general Assumption \ref{asssp}, with simplified version \eqref{c1intro}), and some {\em density conditions} (Assumptions \ref{assde}, \ref{assimp}, \ref{asscont}, Lemmas \ref{lemconcentration} \ref{lemrealcontracted}, or the simplified version Eq.~\eqref{c4intro}). Our density condition informally says that there must exist $W\in[0,1)$ such that for all $x_*\in\R$ and all $i$ with $|x_i-x_*|$ large enough:
\beq\label{densityconditionintro}
	\sum_{j:x_j\in[x_*,x_i)}|\varphi_{ij}|
	\leq W|x_i-x_*|
\eeq
where $[x_*,x_i) := (x_i,x_*]$ if $x_i<x_*$.
We give the simple example of the ``dilute'' gas, whose positions $x_i$'s have a density that decreases as $N^{-\ep}$ for $\ep>0$ as small as desired. We note that our assumptions do not require an analysis of the full $N$-soliton field or the inverse scattering structures for the KdV equation -- all information lies within the solution to the SBE (the KdV rods).

The inequalities \eqref{separationintro} and the density condition \eqref{densityconditionintro} are very natural: in the weak limit, a simple heuristic argument shows that they correspond to the positivity of the ``density of space'' \cite{doyon2018geometric,doyon2020lecture}, the amount of space that remains for the solitons to move in. Thus, we believe our construction holds as near as desired to the soliton condensate \cite{el2021soliton,suret2024soliton}, where the density of space vanishes and the maximal soliton density is obtained.


Our techniques also allow us to obtain new fundamental results on the KdV $N$-soliton field: a bound on its support (Corollary \ref{corolextent}), and a bound on the supremum of the field and its derivatives (Corollary \ref{corolbound2}), as $N$ gets large. Bounds on the support can be obtained by Riemann-Hilbert techniques, this is done in \cite{girotti2021rigorous} for a particular family of $N$-soliton solutions as $N\to\infty$. Our bound is not as tight, however we believe it applies on a wider family of such limits (we leave an in-depth comparison for future works). Concerning bounds on the supremum of the field, we are only aware of results in \cite{lundina1985compactness,gesztesy1992limits}. Our bound is also weaker, however we obtain bounds on the derivatives as well, for which it seems fully rigorous results were still lacking.

Interestingly our proofs do not explicitly require the phenomenon of Anderson localisation nor any randomness, which were important principle in the recent works \cite{aggarwal2025effective,aggarwal2026asymptoticscatteringrelationtoda} in order to localise Lax eigenvectors (see below). Our result suggests that randomness is not required to have a notion of quasiparticles.

Our proof is instead directly based on new tau-function representations of the KdV multi-soliton solution, which allow us to easily recognise which solitons contribute to the field around any chosen observation point. We note that the specific form of the two-body scattering shift $\varphi(\chi,\chi')=\frc1{\chi} \log\big|\frc{\chi-\chi'}{\chi+\chi'}\big|$ plays a crucial role: the fact that $\log|\chi-\chi'|$ is the one-dimensional projection of the Green's function for the non-positive Laplacian operator $\nabla^2$ in two dimensions is important; this fact was also used in analysing the nonlinear dispersion relation of related soliton gases in \cite{kuijlaars2021minimal}. We believe the techniques based on tau functions developed here are powerful enough to be widely applicable, including, for instance, to higher-dimensional integrable field equations such as the KP equation \cite{ablowitz1981solitons}.

\medskip
\noindent {\em Literature remarks---} The belief in the existence of ``stable'' quasiparticles in quantum integrable models stems from the Bethe ansatz \cite{takahashi1999thermodynamics}. In classical integrability it is suggested by the early works from the inverse scattering method such as that of Babelon and Bernard \cite{babelon1993sine} on the sine-Gordon model, and in Zakharov's and El's works on soliton gases cited above. In general quasiparticles are a fundamental concept of GHD, see e.g.~\cite{doyon2020lecture}. In the box-ball system \cite{takahashi1990soliton}, a ``discretization'' of the KdV equation, there are explicit algorithms to locate solitons, in particular that of Ferrari, Nguyen,  Rolla and Wang \cite{ferrari2021soliton}, and although the equivalent of \eqref{euler} and \eqref{loose} have not been proven explicitly, it is relatively clear that they should hold from such algorithms (excursions that allow to locate solitons are not expected to go to macroscopic distances). Further, the recent work of Aggarwal \cite{aggarwal2025effective,aggarwal2026asymptoticscatteringrelationtoda} locates quasiparticles in terms of localised Lax eigenvectors; again \eqref{euler} and \eqref{loose} have not explicitly been proven but they are expected to follow, as there are locality results that can be translated into properties of fluid-cell projections, see e.g.~\cite[point 2, page 5; Cor 5.5, 5.6]{aggarwal2026asymptoticscatteringrelationtoda}. These two latter group of works lend credence to the existence of quasiparticles. As far as we are aware, however, Eqs.~\eqref{fcintro}, \eqref{fcintro2} were not explicitly proposed or proven previously. 

The SBE Eq.~\eqref{bethet} has appeared before. The kinetic equation of soliton gases (or GHD equation) strongly suggests that quasiparticles' trajectories are obtained by accumulating two-body scattering shifts  $\varphi_{ij}$ at soliton crossings. This is the intuitive picture referred to as the {\em collision rate ansatz} \cite{el2005kinetic,doyon2018soliton,doyon2019generalized,congy2021soliton,bonnemain2022generalized,bonnemain2025soliton}, a local version of the asymptotic factorised scattering \eqref{factscat}. Eq.~\eqref{bethet} is closely related to the soliton rate ansatz (although in these works, this equation was not written explicitly). Eq.~\eqref{bethet} arose in the form written here for the first time from an analysis of the Euler-scale limit of Bethe ansatz wave functions in the quantum Lieb-Liniger model  \cite{doyon2023ab}; it directly comes out from a semiclassical limit of the Bethe ansatz wave functions, thus its name. It was then proposed in \cite{doyon2024new} that solitons in classical integrable models should be assigned positions $x_i^{(t)}$'s that solve \eqref{bethet}, as this implements the collision rate ansatz (a larger family of such equations was considered, where the collision structure at finite distances is modified). The special case $\varphi_{ij}=-a$ (negative constant) appeared earlier as the ``contraction map'' used to solve the hard-rod model, see e.g.~\cite{spohn2012large}, and their generalisation to varying rod lengths $\varphi_{ij}=-a_j$ appeared in \cite{ferrari2022hard}. The flea-gas algorithm \cite{doyon2018soliton}, used to solve the GHD equations, implements \eqref{bethet} for general $\varphi_{ij}$. Under the condition $\varphi_{ij}\geq0$ (not satisfied in KdV), it was 
shown  \cite{doyon2026generalised}, and further developed in \cite{bonnemain2025hamiltonian}, that a version of \eqref{bethet} where the sign function is mollified (and a larger family of such equations) have a unique solution corresponding to trajectories of local integrable Hamiltonian systems with elastic and factorised scattering as \eqref{factscat}, arising from so-called $T\b T$-deformations. It is not too hard to show heuristically that the GHD equation arises for the weak limit of the empirical density associated to \eqref{bethet}, see e.g.~\cite{doyon2024new,doyon2023ab}, and rigorous proofs should also be obtainable using current techniques, see e.g.~\cite{croydon2021generalized}. It was in fact conjectured in \cite{doyon2024new} that \eqref{bethet} should give a good description of many-body integrable models at all orders of hydrodynamics, including their fluctuations, an idea developed numerically in \cite{urilyon2025simulating,kethepalli2026ballistic,urilyon2026anomalousdiffusionsuperdiffusionintegrable}.

Crucially, in the context of solitonic models, it is possible to show \cite{Nakabayashiprepa}, using the results of \cite{ferrari2021soliton}, that Eq.~\eqref{bethe} with for $\varphi_{ij} = -2\min\{i,j\}$ describes solitons' positions arising from the ``slot decomposition'' of the box-ball cellular automaton \cite{takahashi1990soliton}, and thereby to obtain a new proof  of the GHD equation in this model (proven in \cite{croydon2021generalized}). Further, recently, it was shown by Aggarwal \cite{aggarwal2025effective,aggarwal2026asymptoticscatteringrelationtoda} that, again with appropriate $\varphi_{ij}$, they give a good description of the centres of Lax eigenvectors in thermal states of the Toda chain, leading to a proof of the GHD prediction for Euler-scale correlation functions of conserved densities \cite{aggarwal2026fluctuationstodalattice}.

Explicit algorithms solving \eqref{bethe}, however, were known until now only for the hard rod gas and its varying-length generalisation \cite{ferrari2022hard}.

In fact, the definition of solitons' positions via Lax eigenvectors was proposed earlier in \cite{bulchandani2019kinetic,fache2025dissipation} (referred to as the Bethe-Lax correspondence), and further studied in the work \cite{bonnemain2026exact} on KdV.

\medskip
\noindent
{\em Organisation of the paper---} The paper is organised as follows. In Sec \ref{sectstatement} we make the general discussion above more precise, express  simplified versions of our main theorems, and  explain the example of the ultra-dilute soliton gases. In Sec.~\ref{secmain} we present the setup, and provide a set of intermediate definitions and lemmas, including new tau functions representation for the $N$-soliton solution (the in-out representation, and the centred representation) as well as our definition of quasiparticle positions. In Sec.~\ref{sectmainresult} we express the main results related to the quasiparticle criterium, leaving longer proofs to Sec.~\ref{secproofs}. In Sec.~\ref{sectdisc}, we discuss the physical meaning of the results,  how they should apply to dense soliton gases, and  how the kinetic equation of soliton gases should emerge. Finally, we conclude in Sec.~\ref{secconclu}.

\medskip
\noindent {\em Notation---} Throughout, we use $\Z_+ = \{0,1,2,\ldots\}$,  $\N=\{1,2,3,\ldots\}$, and for $N\in\N$ we use the notation $\dbra 1,N\dket = \{1,2,\ldots,N\}$. For any $z\in\R,\,I\subset\R$, we denote $\dist(z,I) = \inf_{x\in I}(|z-x|)$. We use the boldface font for vectors $\bs a=(a_1,a_2,\ldots,a_N)\in\R^N$, and we write $\1_N = (1,1,\ldots,1)\in\R^N$. We denote the set of real vectors in $\R^N$ with positive, ordered entries as
\beq\label{Rpudef}
	\Rpu^N = \{\bs \chi\in\R^N: 0<\chi_1<\chi_2<\ldots<\chi_N\},
\eeq
and those where coordinates are all different as
\beq
	\Rne^N = \{\bs x\in\R^N: x_i\neq x_j\;\forall\;i\neq j\}.
\eeq 
We will also use the $\varphi$-exclusion space,
\beq\label{exclusionnot}
	\R_\varphi^N:=\Big\{\bs x\in\Rne^N,\ 
	x_i-x_k\geq -\sum_{j:x_j\in[x_k,x_i)} \varphi_{ij} \ \forall\  i,k:
	x_k< x_i\Big\}.
\eeq

For any $s\subseteq \dbra 1,N\dket$ and $\bs a\in \R^N$, we denote
\beq\label{setext}
	\bs a_s = (a_i)_{i\in s}
\eeq
the restriction of $\bs a$ onto $s$, with the induced order. We take the convention that the sign function $\sgn$ takes value 1 at 0:
\beq\label{sgn}
	\sgn(a) := \lt\{
	\ba{ll}1 & (a\geq 0) \\ -1  & (a<0). \ea\rt.
\eeq
For every subset $s\subseteq\N$, it will also be convenient to define the ``indicator sign function''
\beq\label{sgns}
	\sgn_s(j)=\lt\{\ba{ll} 1 & (j\in s)\\ -1 & (\mbox{otherwise}).\ea\rt.
\eeq

By a slight abuse of notation, we write $[a,b)$ for the standard closed-open interval from $a$ to $b$ if $a<b$, and for the open-closed interval $(b,a]$ if $a>b$:
\beq\label{defab}
	[a,b):=(b,a]\mbox{ if }a>b,\quad \emptyset \mbox{ if }a=b.
\eeq

For any $s\subset\N$, $\argmin_{i\in s} (f(i))$ is the set of all arguments that minimise $f$, and similarly for $\argmax_{i\in s} (f(i))$:
\beq\label{argminmax}
	 \argmin_{i\in s} (f(i)) = \{j\in s: f(j) = \min_{i\in s}f(i)\},\quad
	 \argmax_{i\in s} (f(i)) = \{j\in s: f(j) = \max_{i\in s}f(i)\}.
\eeq
If there is a single minimum, then we will sometimes understand these functions as returning that single value, instead of the single-element set.

Finally, for any $Z\subseteq \R$, we denote the largest distance between nearest neighbours are
\beq\label{neigh}
	\neigh(Z) = \sup_{x\in Z}
	\Big(\dist(x,Z\setminus\{x\})\Big).
\eeq

\section{Statement of the problem and overview of the main results}\label{sectstatement}

The basic ideas behind the problem of solitons' positions for the KdV equation were expressed in the Introduction. We now give details of the statement of the problem, and express simplified versions of our main theorems.

\subsection{Background: $N$-soliton solutions and conserved densities}\label{ssectbackground}

A $N$-soliton solution $u(x,t) = u_{\bs\chi,\bs y}(x,t)$ to the KdV equation \eqref{kdv} on the plane $(x,t)\in\R^2$ with vanishing asymptotic conditions $\lim_{x\to\pm\infty} u_{\bs\chi,\bs y}(x,t)=0\,\forall\,t\in\R$, is a $C^\infty$ function $\R^2\to\R$ characterised by $2N$ real parameters: the {\em spectral parameters} $\chi_i,\,i=1,\ldots,N$, which we can order as $\bs \chi\in\Rpu^N$, and the {\em impact parameters} $y_i\in\R,\,i=1,\ldots N$. See \cite{miura1968korteweg,gardner1974korteweg,lax1983small,lundina1985compactness,gesztesy1992limits} for fundamental results on the $N$-soliton solution, and in particular \cite{gesztesy1992limits,girotti2021rigorous,jenkins2025approximation} for results about the $N\to\infty$ limit.

Time evolution is expressed as a linear shift of impact parameters:
\beq\label{timeevol}
	u_{\bs\chi,\bs y}(x,t) = u_{\bs \chi,\bs y + \bs v t}(x),\quad
	\bs v = (v_i)_i = (4\chi_i^2)_i
\eeq
where by a slight abuse of notation we denote $u_{\bs \chi,\bs y}(x,0) = u_{\bs \chi,\bs y}(x)$. A useful, explicit form of $u_{\bs \chi,\bs y}(x)$ is in terms of what we call the in-out representation (Lemma \ref{lemtaudecomp} -- this is a more general form of the minor expansion written for instance in \cite{lax1983small}):
\beq
	u_{\bs \chi,\bs y}(x)=2\p_x^2 \log \tau_{\bs\chi,\bs y}^s(x)
\eeq
where we may choose any $s\subseteq\dbra 1,N\dket$, and
\beq\label{tauoverview}
	\tau_{\bs \chi,\bs y}^s(x):=
	\sum_{r\subseteq \dbra 1,N\dket}\ 
	\exp \sum_{i\in r} 2\chi_i\sgn_s(i)(x-a_i^{s})\ 
	\exp \sum_{i,j,\in r,\,i\neq j} \sgn_s(i)\sgn_s(j)
	\log\Big|\frc{\chi_i-\chi_j}{\chi_i+\chi_j}\Big|
\eeq
with (see \eqref{factscat})
\beq\label{naiveimpacts}
	a_i^{s} := y_i + \frc12\sum_{j=1\atop j\neq i}^N \sgn_s(j)\varphi_{ij}.
\eeq
At asymptotically large times, the $N$-soliton field separates into single-soliton solutions \eqref{1sol}: for every $\bs\chi\in\Rpu^N,\,\bs y\in\R^N$ and $x,\,v\in\R$, we have
\beq\label{uasymp}
	\lim_{t\to\pm\infty} u_{\bs\chi,\bs y}(x+vt,t)  = \lim_{t\to\pm\infty}
	\sum_i u_{\chi_i,x_i^{\pm}+v_i t}(x+vt)=
	\sum_{i=1}^N \lt\{\ba{ll}
	u_{\chi_i,x_i^\pm}(x) & (v=v_i) \\
	0 & (\mbox{otherwise})
	\ea\rt\}
\eeq
where
\beq\label{outgoingincoming}
	x_i^\pm = y_i \pm \frc12\sum_{j\neq i}\sgn(v_j-v_i)\varphi_{ij}
\eeq
are the outgoing and incoming impact parameters, respectively. Their differences are the Wigner scattering shifts associated to this $N$-soliton collision event, and take the factorised form \eqref{factscat}, and
\beq\label{yaverage}
	 y_i = \frc{x_i^++x_i^-}{2}.
\eeq
More generally, whenever impact parameters are well separated, the $N$-soliton field separates into 1-soliton profiles: for every ``velocity vector'' $\bs w\in\R^N$, one has the decomposition
\beq\label{uasympgen}
	\lim_{t\to\infty} u_{\bs\chi,\bs y+\bs w t}(x+vt)  =
	\sum_{i=1}^N \lt\{\ba{ll}
	u_{\chi_i,\t x_i}(x) & (v=w_i) \\
	0 & (\mbox{otherwise})
	\ea\rt\}
\eeq
where
\beq\label{tx}
	\t x_i = y_i + \frc12\sum_{j\neq i}\sgn(w_j-w_i)\varphi_{ij}.
\eeq
These standard results are shown in Section \ref{ssecttau} for completeness.

Quantities that will play an important role below are the {\em right- and left-extremal positions}: for every $N\in\N$ and $\bs\chi\in\Rpu^N,\,\bs y\in\R^N$, set
\beq\label{Xextremeintro}
	X_i^\pm(\bs\chi,\bs y) := y_i \mp \frc12\sum_{j=1\atop j\neq i}^N\varphi_{ij}.
\eeq
As we will see, right-extremal positions are good solitons' positions for those the furthest to the right, and similarly for left-extremal positions. We note that $X_i^+(\bs\chi,\bs y)  = a_i^\emptyset$, $X_i^-(\bs\chi,\bs y)  = a_i^{\dbra 1,N\dket}$. These are the parameters appearing in the expansion \eqref{tauoverview} for these choices of sets $s$. These parameters are natural in the context of the inverse scattering theory: they are simply related to the ``norming constants'' $c_i$. For a standard choice of norming constants, see e.g.~\cite[Eq 3.9]{gardner1974korteweg}, we have
\beq\label{normingconst}
	c_i = \sqrt{2\chi_i}\re^{\chi_i a_i^\emptyset} 
	= \sqrt{2\chi_i}\re^{\chi_i X_i^+(\bs\chi,\bs y)}.
\eeq
The parameters $a_i^s$ for general $s$ give generalisations of norming constants. In our work, they will be seen as ``contracted coordinates'' of solitons with respect to general observation points.

The KdV equation is integrable. The existence of the multi-soliton solution, with elastic and factorised scattering, are fundamental properties of integrability. Another fundamental property is the presence of an infinite number of conserved quantities with local densities \cite{miura1968korteweg,gardner1974korteweg}. There exist an infinite number of polynomial functions $P_k$ of $u,\,\p_x u,\,\p_x^2 u,\ldots$, the conserved densities\footnote{In the GHD literature, these are often denoted $q(x)$; here we use the notation that is more habitual for KdV.},
\beq\label{conspoly}
	\mathtt P_k[u](x) = P_k\big(u(x),\p_xu(x),\p_x^2u(x),\ldots \big ),\quad k=0,1,2,\ldots
\eeq
such that, for any $u(x,t)$ satisfying \eqref{kdv}, their total integral on $\R$ is conserved in time:
\beq
	\p_t \int_{-\infty}^\infty \dd x\,\mathtt P_k[u(\cdot,t)](x) = 0.
\eeq
With an appropriate normalisation, on the $N$-soliton solution these integrals take the form
\beq\label{integralconserved}
	\int_{-\infty}^\infty \dd x\,\mathtt P_k[u_{\bs\chi,\bs y}](x)
	=\sum_{i=1}^N\chi_i^{2k+1}.
\eeq
In particular,
\beq
	P_0(u) = \frc{u}4,\quad P_1(u) = \frc{3u^2}{16},\quad
	P_2(u,\p_x u) = \frc{5}{64}\big((\p_x u)^2 -2  u^3\big).
\eeq

In the context of the $N$-soliton solution to the KdV equation, a soliton gas can naturally be seen as a probability distribution of spectral and impact parameters $\bs\chi,\,\bs y$, taken in the limit where $N$ is large. What this probability distribution may be is still under investigation \cite{suret2024soliton}, see also our Conclusion, Section \ref{secconclu}.

In this paper, we define a soliton gas as {\em a sequence of spectral parameters, and a sequence of impact parameters, with increasing number of solitons},
\beq\label{seqintro}
	\N\ni N\mapsto \bs\chi\in\Rpu^N,\ \bs y\in\R^N.
\eeq
One may simply adjoin $\chi_{N+1},\,y_{N+1}$ to $\bs\chi\in\Rpu^N,\,\bs y\in\R^N$ for every $N$, but this is not necessary, and our analysis (in particular of projections) in fact requires more general sequences. The interest is in understanding what happens as $N\to\infty$. Fundamental results about this limit, with bounded $\max_{i\in\dbra 1,N\dket}(\chi_i)$ and certain weak conditions on $\bs y_i$, were obtained in \cite{gesztesy1992limits}, where the simple adjoining procedure was taken, including the existence of the limit. Further, it has been established how certain $N$-soliton solutions reproduce,  as $N\to\infty$, a KdV solution interpolating between a cnoidal wave and an asymptoticlally vanishing field \cite{girotti2021rigorous}; and a recent analysis \cite{jenkins2025approximation} connects finite-gap solutions to certain $N\to\infty$ limits of $N$-soliton solutions. Here, we do not need $\max_{i\in\dbra 1,N\dket}(\chi_i)$ to be bounded in our general results, although it is convenient to assume this in the specialisation of our results below. We do not need the limit $N\to\infty$ of $u_{\bs\chi,\bs y}$  to converge in any sense.

\subsection{Heuristics: time-of-flight thought experiment}\label{ssecttof}

The discussion in this section provides a heuristic understanding of the fluid-cell projection. It is not fully mathematically rigorous, and is not essential for our results.

The main object of a KdV soliton gas is the {\em density of states}: a function $\rho(\chi,x,t)$ on $\R_+\times\R\times \R$, which, in some sense, describes the density of solitons in phase-space $\chi,x$ at time $t$. As far as we are aware, there is no rigorous definition of $\rho(\chi,x,t)$ yet in terms of $u_{\bs\chi,\bs y}(x,t)$ (our work provides the first proposal). Notwithstanding this lack of precision, the soliton-gas theory still predicts a universal kinetic equation \cite{suret2024soliton}. We are looking to understand {\em how to define} $\rho(\chi,x,t)$, and {\em under what conditions} on the sequences \eqref{seqintro} of spectral and impact parameters it is related $u_{\bs\chi,\bs y}$ as in \eqref{euler} and \eqref{loose}.

Here we emphasise that {\em there is no need to have a notion of solitons' positions in order to define a density of states $\rho(\chi,x,t)$}. Indeed, a natural definition of $\rho(\chi,x,t)$ in terms of $u_{\bs\chi,\bs y}(x,t)$ in the large-$N$ limit is from the physical picture of the ``time-of-flight thought experiment'', see e.g.~\cite{doyon2020lecture}. Although we will not directly implement this exact procedure, our construction amounts to a slight modification of it, so it is instructive to describe what it is in the present context.

Under appropriate conditions of finite density, as $N$ gets large, we expect $u_{\bs\chi,\bs y}(\cdot,Nt)$ to be essentially supported on an interval of length $\ell(t)\propto N$, for all $t\in\R$ (see \cite[Lem 2.1]{girotti2021rigorous} for a family of examples, and Theorem \ref{theoextentintro} for a result on the support). So we take the ``macroscopic length scale'' as being equal to $N$. Now take an interval of length $L>0$ and centre $x_*\in\R$,
\beq\label{Ioverview}
	I = [x_*-L/2,x_*+L/2],
\eeq
that is large but not too large, say
\beq\label{Loverview}
	L = N^\lambda,\quad 0<\lambda<1.
\eeq
$L$ is a ``mesoscopic length scale'' and $I$ is a ``fluid cell''. Then construct a na\"ive fluid-cell projected field at time $t$:
\beq
	u^{(I)}_{\bs\chi,\bs y}(x,t) := \lt\{\ba{ll}
	u_{\bs\chi,\bs y}(x,t) & (x\in I)\\
	0 & (\dist(x,I)>\Delta x)\\
	\mbox{some smooth interpolation} & (\mbox{otherwise})
	\ea\rt.
\eeq
for some $\Delta x$ small enough. Denoting by $\tau_T$ the evolution under \eqref{kdv} to time $T$, that is $u(x,T) = (\tau_Tu(\cdot,0))(x)$ for  $u(\cdot,0)\in C^\infty(\R)$, evaluate the set of spectral parameters that represent the fluid-cell projected field: the ones occurring asymptotically in the ``thought experiment'' of letting it evolve for a long time,
\beq\label{setchiprime}
	U=\{\chi'>0:\lim_{T\to\infty} (\tau_T u^{(I)}_{\bs\chi,\bs y}(\cdot,t))(x+4{\chi'}^2T) \neq 0\} = \{\chi_i'\}_{i=1}^M,\quad M=|U|.
\eeq
The soliton resolution conjecture says that generic initial conditions to the KdV equation \eqref{kdv} give, asymptotically, fields that have solitonic and radiative parts \cite{gardner1974korteweg}; the above extracts spectral parameters of the solitonic part, as we only look at positive velocities $4{\chi'}^2$ (while radiation modes have negative velocities) -- see \eqref{uasymp}. Then define the ``empirical density'' associated to this thought experiment:
\beq
	\t\rho_{\rm emp}(\chi,x_*,t) = \frc1L\sum_{i=1}^M
	\delta(\chi-\chi_i').
\eeq
The density of states is the following weak limit, at macroscopic positions and times: for every $\chi>0,\,x,\,t\in\R$,
\beq\label{rhowlim}
	\rho(\chi, x,t) = \wlimu{N\to\infty}\t\rho_{\rm emp}(\chi,N x,N t).
\eeq

That is, we cut-out a mesoscopic piece of the $N$-soliton field in a neighbourhood of the point $x_*$ of mesoscopic length $L$ at time $t$, and let this piece evolve in time until solitons separate. This ``extracts'' the solitons that were, so to speak, in that neighbourhood, and the density of states is the corresponding phase-space density at large scales. See Fig.~\ref{figtof}. This thought experiment is, in fact,  implemented in actual experiments on cold atoms \cite{malvania2021generalized}. One expects that, under physically natural probability distributions, the weak limit in \eqref{rhowlim} exists almost surely, and is an almost-sure function (that is, the result does not fluctuate); otherwise one may take an average over the distribution.

In order for any emergent dynamical equation for $\rho(\chi,x,t)$ to be predictive, one would also like to go in the other direction, and {\em extract information about $u_{\bs \chi,\bs y}(Nx,Nt)$ from $\rho(\chi,x,t)$}. Clearly, because the na\"ive fluid-cell projected field lies on $I$, we have, with the only approximation being that we take the solitonic part only,
\beq\label{Prho}
	\frc1L\int_{x_*-L/2-\Delta x}^{x_*+L/2+\Delta x}\dd x\,
	\mathtt P_k[u_{\bs\chi,\bs y}(\cdot,Nt)](x)
	=
	\frc1L\int_{-\infty}^\infty \dd x\,
	\mathtt P_k[u^{(I)}_{\bs\chi,\bs y}(\cdot,Nt)](x)
	\stackrel{\text{solitonic part}}=
	\frc1L\sum_{i=1}^M {\chi_i'}^{2k+1}.
\eeq
With $x_*= N x$, and taking the limit $N\to\infty$ where we can replace $\int_{[x_*-L/2-\Delta x , x_*+L/2+\Delta x]}\dd x$ by $\int_{I}\dd x$, this becomes the {\em fluid-cell average of $\mathtt P_k[u_{\bs\chi,\bs y}]$ at the macroscopic point $x,t$}. As in the weak Euler-scale limit we may divide $\R$ into mesoscopic intervals where the Schwartz function against which we integrate is approximately constant, this implies \eqref{euler}.
That is, $\rho(\chi,x,t)$ gives us information about the total conserved quantities present in mesoscopic intervals, or local information at the Euler scale.

The time-of-flight construction raises a few questions.

\bi
\item First, Eq.~\eqref{setchiprime} only extracts the solitonic component of $u^{(I)}_{\bs\chi,\bs y}(x,t)$. Yet, it is expected to have a radiative component as well. Is it indeed true that this radiative part does not contribute in evaluating averages, i.e.~terms neglected in Eq.~\eqref{Prho} are negligible?

\item Second, it would be natural for the extracted spectral parameters to form a {\em subset} of the original spectral parameters:
\beq
	\{\chi_i':i=1,\ldots,M\} \subseteq \{\chi_j:j=1,\ldots,N\}\mbox{ ?}
\eeq
That is, we are extracting solitons that were already in the $N$-soliton solution. But this is not immediate from the time-of-flight construction \eqref{setchiprime}.

\item Third, what about predictions for fluid-cell averages of other observables, say functions of $u$ and finitely-many of its derivatives, $O(u,\p_x u,\ldots)$:
\beq\label{Frho}
	\lim_{N\to\infty} \frc1L\int_I \dd x'\,O(u_{\bs\chi,\bs y}(x',N t),
	\p_{x'} u_{\bs\chi,\bs y}(x',N t),\ldots) 
	= \mbox{?}
\eeq
The knowledge of the left-hand side of \eqref{Prho} for all $k$ is not sufficient to determine \eqref{Frho} more generally.

\item Finally, for what values of the mesoscopic exponent $\lambda$, Eq.~\eqref{Loverview}, the above would work, and under what conditions on how spectral and impact parameters behave as $N$ gets large?
\ei

Even with the current extensive theory of the KdV equation, or more generally of integrable field equations and integrable interacting particles (where similar heuristic constructions can be done and questions asked), proving the statement \eqref{rhowlim} with this time-of-flight construction, proving \eqref{Prho}, and evaluating \eqref{Frho}, appear to be out of reach.

We now explain how to define fluid-cell projections for which we can prove good properties, in terms of quasiparticles.

\subsection{Quasiparticles and fluid-cell projections: the KdV rods}\label{ssectquestion}

We concentrate on field configurations $u$ lying in the space of solitonic fields,
\beq\label{defiS}
	\mathfrak S = \{u_{\emptyset,\emptyset}=0\}\bigcup
	\{u_{\bs\chi,\bs y}:N\in\N,\,\bs \chi\in\Rpu^N,\,
	\bs y\in\R^N\}\subset C^\infty(\R).
\eeq
The multi-soliton solution Eq.~\eqref{utau} gives an explicit map
\beqa
	\Omega\quad :
	\quad
	\mathsf S:=\{(\emptyset,\emptyset)\}\bigcup \coprod_{N\in\N}\Rpu^N\times \R^N&\to&\mathfrak S\n
	(\bs\chi,\bs y)&\mapsto& u_{\bs\chi,\bs y}.
	\label{Omega}
\eeqa
It is simple to see, from our discussion in Sec.~\ref{ssectbackground}, that $\Omega$ is a bijection: it is surjective by definition, and injective as any two different elements of $\mathsf S$ give fields $u_{\bs\chi,\bs y}$ that manifestly have different form at large times under the KdV evolution \eqref{kdv}. As mentioned, as time evolution is given by \eqref{timeevol}, we may set without loss of generality\footnote{In this respect, we note that from the time evolution \eqref{timeevol}, it may be tempting to see spectral parameters as ``action variables'' and impact parameters as non-compact versions of ``angle variables'' from classical integrability \cite{babelon2003introduction}. However, linearisation of time evolution by mapping to impact parameters is expected to hold by the general many-body scattering theory, even without integrability \cite{derezinski2013scattering}. In fact, by general properties of scattering processes, the scattering map $\Sigma:(\bs\chi,\bs x^-)\mapsto (\bs\chi,\bs x^+)$, which takes incoming to outgoing asymptotic data, always factorises in terms of M\o ller transforms $\Sigma = \Omega_+^{-1}\Omega_-$. Outgoing / incoming M\o ller transforms $\Omega_\pm$ connects freely-evolving asymptotic coordinates $(\bs\chi,\bs x^\pm)$ to interacting objects. Here interacting objects are the solitonic fields in $\mathfrak S$, i.e.~$\Omega_\pm:(\bs\chi,\bs x^\pm)\mapsto u_{\bs\chi,\bs x^\pm \mp \frc12\sum_{j\neq i}\Delta_{ij}}$ where $\Delta_{ij} = \sgn(v_i-v_j)\varphi_{ij}$. But, contrary to linearisation of the dynamics, elastic and factorised scattering \eqref{outgoingincoming} are properties of $\Sigma$ that are, indeed, fundamental to integrability. Using this, for the purpose of our discussion it is convenient to simply define ``the M\o ller transform'' as $\Omega$, Eq.~\eqref{Omega}.} $t=0$.

Quasiparticles may be introduced by requiring that the {\em M\o ller transform} \cite{derezinski2013scattering} factorises as
\beq\label{Moller}
	\Omega = \Upsilon \mathcal U
\eeq
with
\beqa
	\mathcal U
	\quad :
	\quad
	\mathsf S&\to&\mathsf Q\subseteq\mathsf S\n
	(\bs\chi,\bs y)&\mapsto&  (\bs\chi,\bs x)
	\label{mapU}
\eeqa
and $\Upsilon:\mathsf Q\to \mathfrak S$. Here $\bs x\in\R^N$ is the vector of quasiparticle positions, representing where individual solitons lie in space within the multi-soliton field. We write the vector of quasiparticle positions as a function\footnote{Here and below we make this slight abuse of notation whereby $\bs x(\bs\chi,\bs y)$ is the vector of quasiparticle positions while $\bs x$ is just a variable, here chosen to simplify the notation; that is, we do not use $\bs x$ as the map $\bs\chi,\bs y\mapsto \bs x(\bs\chi,\bs y)$.}
\beq\label{xUy}
 	\bs x = \bs x(\bs\chi,\bs y),
\eeq
that is $\mathcal U(\bs\chi,\bs y) = (\bs \chi,\bs x(\bs\chi,\bs y))$.

\begin{figure}\bc\includegraphics[width=7cm]{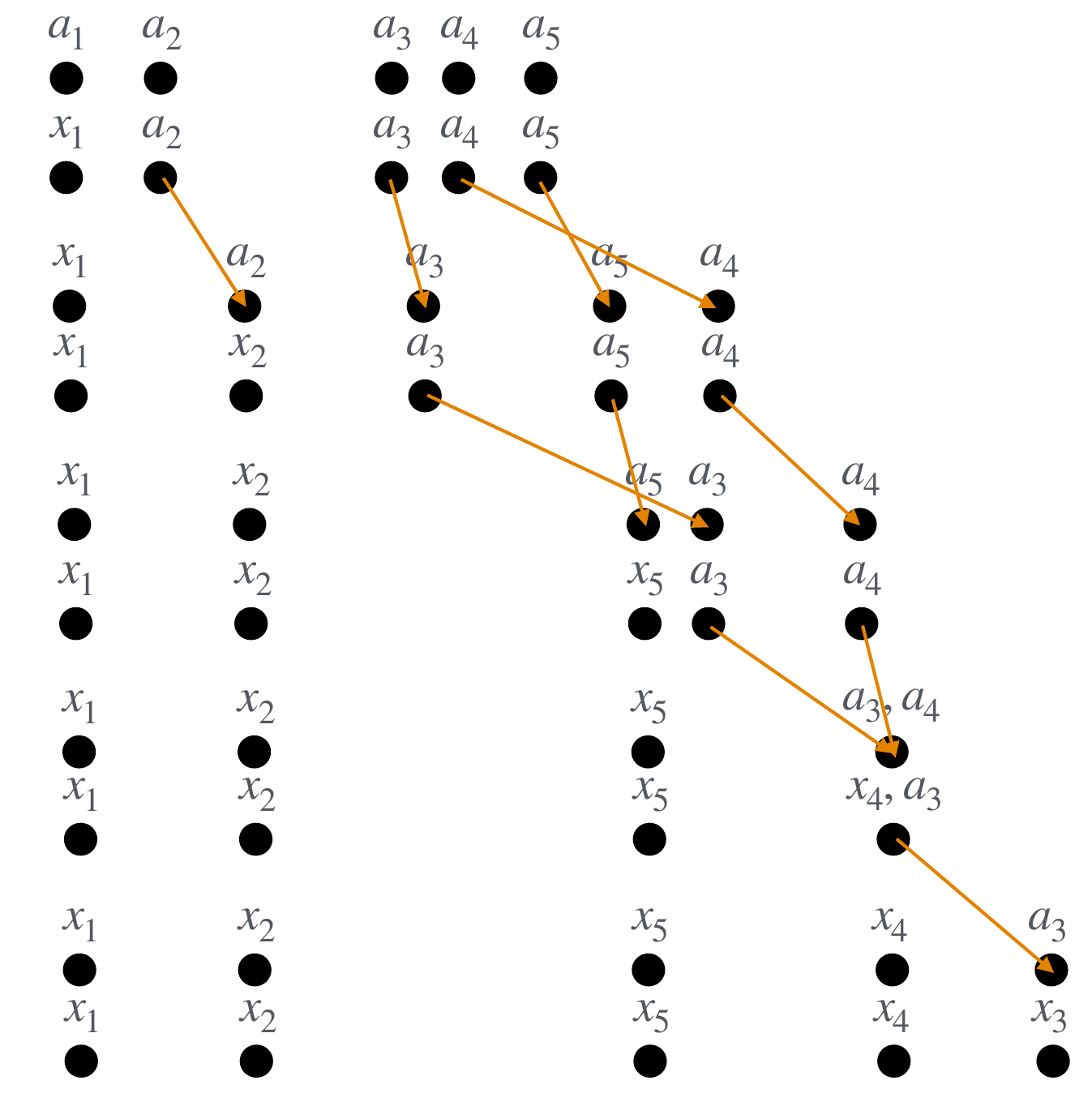}\ec
\caption{The expansion map defining the quasiparticle positions $\bs x = \bs x(\bs\chi,\bs y)$ from point particles with positions $\bs a$. Each line is either taking the left-most point as the left anchor of a new quasiparticle, or expanding this left anchor to a full KdV rod by pushing towards the right the point particles that have not become a rod yet. In this example, $v_4<v_3$.}
\label{figexpansion}
\end{figure}
We define $\bs x(\bs\chi,\bs y)$ via a simple algorithm, an {\em expansion map} generalising that mapping point particles to rods in hard rod gases \cite{spohn2012large,ferrari2022hard}, representing the ``left anchors'' of the quasiparticles:
\beq\label{qpreview}
	\bs x(\bs\chi,\bs y) = {\tt Leftanchor}(\bs\chi ,\bs y).
\eeq
The algorithm is as follows. Set $\bs a = \bs a^{\dbra 1,N\dket} = \bs X^-(\bs\chi,\bs y)$ (Eqs.~\eqref{naiveimpacts}, \eqref{Xextremeintro}), and start with $i_1 = \argmin_i(a_i)$ and set $x_{i_1}=a_{i_1}$. If $\argmin_i(a_i)$ has more than one element, take that which minimises $v_i$. This is the position of the left-most soliton, which we now mark as done. To get the next one, take away the interaction due to soliton $i_1$: replace $a_j\to a_j -\varphi_{ji_1}$ for all remaining solitons, $j\neq i_1$. As $\varphi_{ij}<0$, this moves them to the right. Repeat to get $i_2,\,i_3,\ldots,i_N$ and the associated soliton positions $x_{i_1}<x_{i_2}<\cdots<x_{i_N}$. See Fig.~\ref{figexpansion}. As $a_i = y_i+\frc12\sum_{j\neq i} \varphi_{ij}$, at every step we change one sign in the sum, that of the soliton we passed through. Hence the resulting positions have the form $x_i = y_i + \frc12\sum_{j\neq i}\sgn(x_j-x_i)\varphi_{ij}$, thus they satisfy the SBE \eqref{bethe}.

It turns out that the resulting quasiparticle positions also satisfy  the separation property \eqref{separationintro} (Theorem \ref{theosbe}, Theorem \ref{theodesc}), which in its more general form is
\beq\label{separationsimplified}
	x_i(\bs\chi,\bs y)-x_*\geq -\sum_{j:x_j(\bs\chi,\bs y)\in[x_*,x_i(\bs\chi,\bs y))} \varphi_{ij} \quad \forall\  x_*
	\leq x_i(\bs\chi,\bs y),\ i\in\dbra 1,N\dket.
\eeq

Using the choice above to define $\mathcal U$, Eq.~\eqref{mapU}, we get its image, the space of quasiparticle coordinates $\mathsf Q$. We show that this is characterised as follows. Recall that we define the $\varphi$-exclusion space as (Eq.~\eqref{exclusionnot})
\beq\label{exclusion}
	\R_\varphi^N:=\Big\{\bs x\in\Rne^N,\ 
	x_i-x_k\geq -\sum_{j:x_j\in[x_k,x_i)} \varphi_{ij} \ \forall\  i,k:
	x_k< x_i\Big\}.
\eeq
Then (Theorem \ref{theodesc})
\beq\label{Qexact}
	\mathsf Q = \{(\emptyset,\emptyset)\}\cup\Big\{(\bs\chi,\bs x):
	N\in\N,\,\bs \chi\in\Rpu^N,\,\bs x\in\Rne^N,\, \exists\,\ep>0\,|\,\bs x + \bs v t \in\R_\varphi^N
	\ \forall\ t\in[0,\ep]\Big\}.
\eeq
That is, this is the set of coordinates such that the separation property holds, and such that time evolution
\beq
	\bs x^{(t)}(\bs\chi,\bs y) = \bs x(\bs\chi,\bs y+\bs v t) 
\eeq
is right-continuous at all times. As $\mathcal U:\mathsf S\to \mathsf Q$ is a bijection, we then have a characterisation of all multi-soliton fields in terms of quasiparticle coordinates: defining
\beq
	u_{\bs\chi;\bs x} = u_{\mathcal U^{-1}(\bs\chi,\bs x)},
\eeq
we have
\beq
	\mathfrak S = \{u_{\bs\chi;\bs x}:(\bs\chi,\bs x)\in\mathsf Q\}.
\eeq
Explicitly, $u_{\bs\chi;\bs x}(x)$ is obtained using any known representation of the multi-soliton field, such as tau functions \eqref{tauoverview} or \eqref{tau}, with \eqref{naiveimpacts} (related to norming constants as in \eqref{normingconst}) written as
\beq
	a_i^{s} := x_i + \sum_{j=1\atop j\neq i}^N (\delta_{j\in s}\Theta(x_i-x_j) - \delta_{j\in\dbra 1,N\dket\setminus s}\Theta(x_j-x_i))\varphi_{ij}
\eeq
where $\Theta$ is Heaviside's step function. One may also use the centred form, Lemma \ref{lemtaumain} below.

We are looking to improve on the time-of-flight construction, and define, for every interval $I\subset \R$, a fluid-cell projection $\mathcal F_I:\mathfrak S \to \mathfrak S$ and associated projection on the space of spectral and impact parameters $\h{\mathcal P}_I:\mathsf S \to \mathsf S$:
\beq\label{mapintro}
	\mathcal F_I:u_{\bs\chi,\bs y} \mapsto u_{\bs\chi',\bs y'},\quad
	(\bs\chi',\bs y') = \h{\mathcal P}_I (\bs\chi,\bs y),\quad
	\mathcal F_I = \Omega\h{\mathcal P}_I\Omega^{-1}.
\eeq
Naturally $\mathcal F_I 0 = 0$, and, for every $N\in\N,\,(\bs\chi,\bs y)\in\Rpu^N\times \R^N$,
\beq
	(\bs\chi',\bs y')= (\emptyset,\emptyset) \quad\mbox{or}\quad
	(\bs\chi',\bs y')\in\Rpu^M\times \R^M,\quad 
	M\in\N,\,M\leq N.
\eeq
Note that because this maps to a multi-soliton field, without radiation mode, the time-of-flight thought experiment directly gives $\bs\chi'$, by \eqref{uasymp}.

We obtain the fluid-cell projection using the quasiparticles as follows. For any closed interval $I\subset \R$, the fluid-cell map $\mathcal F_I$ from \eqref{loose} is defined as
\beq\label{fcupsilon}
	\mathcal F_I = \Upsilon \mathcal P_I \Upsilon^{-1} = \Omega \mathcal U^{-1}\mathcal P_I\mathcal U\Omega^{-1} ,\quad
	\mathcal P_I(\bs\chi,\bs x) = ((\chi_i)_{i:x_i\in I},\,(x_i)_{i:x_i\in I}).
\eeq
We show (Theorem \ref{theodesc}) that the space $\mathsf Q$ is invariant under projection, 
\beq\label{Piinclusion}
	\mathcal P_I \mathsf Q\subseteq\mathsf Q
\eeq
for all closed intervals $I$. This is a crucial property, as it guarantees that \eqref{fcupsilon} makes sense (the product $\mathcal U^{-1} \mathcal P_I$ makes sense). It is then clear that
\beq
	\mathcal F_I^2 = \mathcal F_I.
\eeq
In terms of \eqref{mapintro},
\beq\label{hatPUP}
	\hat{\mathcal P}_I = \mathcal U^{-1}\mathcal P_I\mathcal U
\eeq
is then indeed a projection, $\h{\mathcal P}_I^2 = \h{\mathcal P}_I$.

Clearly, by this construction
\beq
	\mathcal F_I u_{\bs\chi;\bs x} = u_{\bs\chi';\bs x'}
\eeq
where
\beq\label{chiprimeintro}
	\bs\chi' = (\chi_i)_{i:x_i\in I},\quad\bs x' = (x_i)_{i:x_i\in I}.
\eeq
We show (Lemma \ref{lemextractsbe}, Definition \ref{defimapfc}) that the fluid-cell projection \eqref{fcupsilon} acts as \eqref{mapintro} with
\beq\label{yprimeintro}
	\bs y' = \Big(y_i + \frc12\sum_{j\,:\,x_j\,\text{to the right of}\,I} \varphi_{ij}
	-
	\frc12\sum_{j\,:\,x_j\,\text{to the left of}\,I} \varphi_{ij}\Big)_{i:x_i\in I}
\eeq
where $\bs x=\bs x(\bs\chi,\bs y)$. Here, impact parameters are shifted in accordance with their scattering with the solitons that have been taken away. Eqs.~\eqref{chiprimeintro} and \eqref{yprimeintro} define the map $\hat {\mathcal P}_I$ in \eqref{mapintro}. Eq.~\eqref{yprimeintro} is natural: when solitons are well separated, which happens when impact parameters  are well separated, Eq.~\eqref{uasympgen}, then solitons are easily identifiable, and we should have, for every $\bs w\in\R^N$ (see \eqref{tx} for $\t{\bs x}$):
\beq\label{conditionUintro}
	\lim_{t\to\infty} \bs x(\bs\chi,\bs y + \bs w t)-\bs w t=
	\t{\bs x}.
\eeq
The expression Eq.~\eqref{yprimeintro} is consistent with this.

For $\bs x(\bs\chi,\bs y)$ to be good solitons' positions, $\Upsilon:\mathsf Q \to\mathfrak S$ in Eq.~\eqref{Moller} must map solitons' positions $\bs x\in\mathsf Q$ to the KdV field $u$ with {\em appropriate locality properties}. A natural, informal expression of locality for $\Upsilon$ is that, in any neighbourhood of $x\in\R$, local observables $u(x),\,\p_x u(x),\ldots$ and functions of these, where $u = \Upsilon(\bs\chi,\bs x)$,  can be solely expressed in terms of solitons data $\chi_i,x_i$ for solitons whose positions $x_i$ lie in this neighbourhoods. More precisely, we need to verify the quasiparticle criterium \eqref{loose}:  for every $I$, $u_{\bs\chi,\bs y}|_I=\big(\mathcal F_Iu_{\bs\chi,\bs y}\big)|_I$ and $\big(\mathcal F_Iu_{\bs\chi,\bs y}\big)|_{\R\setminus I}=0$ in an appropriate sense as $N\to\infty$, and  $\Upsilon$ is the intertwiner relating $\mathcal F_I$ and the simple projection $\mathcal P_I$, 
\beq
	\mathcal F_I\Upsilon = \Upsilon \mathcal P_I.
\eeq
We review our results on this in Sec.~\ref{ssectovermain}.


It is important to note that, from \eqref{timeevol}, {\em the KdV dynamics is transferred to a dynamics for quasiparticle positions}: solitons' trajectories $t\mapsto \bs x^{(t)}$ are given by
\beq
	\bs x^{(t)} = \bs x(\bs \chi,\bs y+\bs vt).
\eeq
The analysis of this dynamics at the Euler scale is a separate problem, which we do not address here rigorously (we provide non-rigorous arguments). We emphasise that relating solitons' positions to local observables of the KdV field $u_{\bs\chi,\bs y}$ itself is an essential problem. It is non-trivial in soliton gases and other models such as the Toda model \cite{aggarwal2025effective,aggarwal2026asymptoticscatteringrelationtoda}, a major difference with respects to integrable models of interacting particles where the quasiparticles are simply related to the particles themselves, such as in the hard rod gas \cite{spohn2012large,doyon2017dynamics}. Having an emergent, large-scale dynamics for solitons seen as point particles is not sufficient: it is essential that this emergent dynamics describes well all large-scale aggregate quantities evaluated from the KdV field --- the quasiparticle criterium must be fulfilled.

\subsection{The quasiparticle criterium: how the KdV rods relate to the KdV field}\label{ssectovermain}

We now describe our main rigorous results on the quasiparticle criterium and related results, whose full expressions are Theorems \ref{theo}, \ref{theofc} and \ref{theoweak}, and Corollaries \ref{corolextent}, \ref{corolbound2}, \ref{corolfc}, \ref{corolvanishfc}, \ref{coroboundedfc}. We use here a simplified set of assumptions, that is nevertheless physically relevant; the theorems and corollaries quoted have more general versions, with weaker assumptions. All theorems below are proven in Section \ref{ssectproofoverviewtheorems} from our main theorems.

As mentioned, we take a soliton gas as a sequence of spectral and impact parameters, with increasing number of solitons, Eq.~\eqref{seqintro}. The interest is in understanding what happens as $N\to\infty$. We keep the dependence of $\chi_i$'s and $y_i$'s on $N$ implicit, in order to lighten the notation.

We say that {\em a soliton gas \eqref{seqintro} has regular spectral parameters} if the sequence $N\mapsto \bs\chi$ satisfies the following conditions: there exist $C>\chi_*>0$, and for every $\eta>0$ there exists $A=A(\eta)>0$, such that for every $N\in\N$,
\beq
	\label{c1intro}
	\chi_*\leq \chi_i\leq C,\quad
	|\chi_i-\chi_j|\geq \re^{-A N^{\eta}}\quad
	\forall\ i\neq j\in\dbra 1,N\dket.
\eeq
That is, there is a minimum and maximum spectral parameter, and the minimal separation between spectral parameters is at most decaying as a stretched exponential, for every stretch exponent $\eta$. Note that if we distribute impact parameters regularly between $\chi_*$ and $C$, then their minimal separation is of order $1/N$, hence the minimum-separation condition is rather weak. It is a simple matter to construct a soliton gas with regular spectral parameters. Our most general result allow for $\max_i(\chi_i)$ to increase with $N$ as a power law, see \eqref{c1}. We note that a soliton gas with regular spectral parameters has finite density (Lemma \ref{lemrealdens}),
\beq\label{fdintro}
	\sup_{x_*\in\R,\,r\geq 1}\frc{\big|\{i:|x_i(\bs\chi,\bs y)-x_*|\leq r\}\big|}{2r} \leq \b\rho:= \frc12\Big(1+\frc{C^2}{\chi_*}\Big).
\eeq

Regularity of spectral parameters is not enough in order to obtain our results on soliton gases: we will need requirements stating that the ``density'' of the soliton gas is not too large. But these requirements must take into account the interaction and soliton positions. 
Establishing such density requirements is one of the most non-trivial parts of our results. 

\medskip
\noindent {\em First result: support of the multi-soliton field (Corollary \ref{corolextent}).} Recall the right- and left-extremal positions, Eq.~\eqref{Xextremeintro}. These partially account for solitons' interactions. In fact, $X_i^-(\bs\chi,\bs y)$ is more accurate as a position for soliton $i$ the less there are solitons whose extremal left positions are to its left (it becomes a good position if there are none); and vice versa for $X_i^+(\bs\chi,\bs y)$. According to this interpretation, $\min_i(X_i^-(\bs\chi,\bs y))$ and $\max_i(X_i^+(\bs\chi,\bs y))$ are therefore accurate estimates of where the left-most  and right-most solitons lie. Define the {\em core} of the $N$-soliton field as the interval
\beq\label{coreintro}
	\core(\bs\chi,\bs y) = \Big[
	\min_{i\in\dbra 1,N\dket}(X_i^-(\bs\chi,\bs y)) , 
	\max_{i\in\dbra 1,N\dket}(X_i^+(\bs\chi,\bs y))\Big].
\eeq
In this theorem, the assumption is based on a simplification of Assumption \ref{assde}.
\begin{theorem}\label{theoextentintro}
Consider a soliton gas \eqref{seqintro} that has regular spectral parameters, Eq.~\eqref{c1intro}. Assume that for every $\eta>0$ there exists $D>0$ such that, for all $N\in\N$ and $d\geq N^\eta$,
\beq\label{c2intro}
	\frc{\Big| \{i:|X_i^-(\bs\chi,\bs y)-\min_j(X_j^-(\bs\chi,\bs y))|\leq d\}\Big|}{2d} \leq D,
	\quad
	\frc{\Big| \{i:|X_i^+(\bs\chi,\bs y)-\max_j(X_j^+(\bs\chi,\bs y))|\leq d\}\Big|}{2d} \leq D.
\eeq
Then for every $\alpha>0$, there exist constants $D_n>0,\,\kappa_n\in\R$ for $n=0,1,2,\ldots$, and $E>0$, such that
\beqa
	\lefteqn{\Big|\p_x^n u_{\bs \chi,\bs y}(x)\Big| \leq D_n
	N^{\kappa_n} \exp\Big(EN^{\alpha}-2\chi_*\dist(x,\core(\bs\chi,\bs y))\Big)} &&\n
	&& \qquad \quad \forall\ x\in\R\,:\, \dist(x,\core(\bs\chi,\bs y))>0,\ n\in\Z_+,\ N\in\N.
	\label{projvanishintro}
\eeqa
\end{theorem}
Recall that $\chi_*$ is the constant involved in \eqref{c1intro}. The assumption is that densities for the coordinates $X_i^\pm(\bs\chi,\bs y)$ around the boundaries of the core are bounded. The theorem says that the $N$-soliton field and all its derivatives are essentially supported on the core, being exponentially decaying outside it, up to a ``skin effect'' whose thickness grows more slowly than any power law in $N$. The constants $D_n,\,\kappa_n$ and $E$ depend only on the data of the assumptions: on $\alpha,\,D=D(\eta)$, and on $\chi_*,\,C,\,A(\eta)$ from \eqref{c1intro}; they do not depend on the choice of sequences \eqref{seqintro} satisfying these assumptions. It would be interesting to compare the above result with \cite[Lem 2.1]{girotti2021rigorous}, where a particular family of $N$-soliton solutions is shown to be supported  on $[-\t CN,0]$ for some $\t C>0$, up to exponentially decaying terms as $N\to\infty$.

For the next three results, we need the quasiparticle positions $x_i(\bs\chi,\bs y)$, defined in Sec.~\ref{ssectquestion}. Recall that they satisfy the minimal-distance property \eqref{separationsimplified}. In particular, this implies that $\sum_{j:x_j(\bs\chi,\bs y)\in[x_*,x_i(\bs\chi,\bs y))}|\varphi_{ij}|\leq |x_i(\bs\chi,\bs y)-x_*|$ for all $i$ such that $x_i\geq x_*$. Equality, on large scales, corresponds to the condition that the soliton gas be in its condensate regime \cite{el2021soliton,suret2024soliton}, see Sec.~\ref{ssectdiscmetric}. Therefore, for physically ``typical'' choices of spectral and impact parameters away from the condensate regime, one would expect strict inequality, and further, this inequality should also hold for $x_i\leq x_*$.
This justifies our next assumption (a simplification of Assumption \ref{asscont}): we say that {\em a soliton gas \eqref{seqintro} is away from the condensate regime} if for every $\eta>0$ there exists $W\in[0,1)$, such that for all $N\in\N$, $x_*\in\R$ and $i$ such that $|x_i(\bs\chi,\bs y)-x_*|\geq 2N^\eta/(1-W)$,
\beq\label{c4intro}
	\sum_{j:x_j(\bs\chi,\bs y)\in[x_*,x_i(\bs\chi,\bs y))}|\varphi_{ij}|
	\leq W|x_i(\bs\chi,\bs y)-x_*|
\eeq
(recall our convention \eqref{defab}).
Importantly, if this is the case, and if the soliton gas also has regular spectral parameters, Eq.~\eqref{c1intro}, then the conditions \eqref{c2intro} of Theorem \ref{theoextentintro} are satisfied (Lemmas \ref{lemconcentration}, \ref{lemrealcontracted}), with
\beq\label{DVintro}
	D = \frc{2\b\rho}{1-W}
\eeq 
(see Eq.~\eqref{fdintro}).

\medskip
\noindent{\em Second result: bound on the multi-soliton field (Corollary \ref{corolbound2}).} It has been established that as $N\to\infty$, the $N$-soliton field has bounded supremum for bounded spectral parameters \cite{lundina1985compactness,gesztesy1992limits}; see also \cite{girotti2021rigorous,jenkins2025approximation} for other $N\to\infty$ results. However, this is not quite enough for our purposes. Here, under certain density assumptions, we find a stronger result: we bound the supremum of the field and all its derivatives.

\begin{theorem}\label{theoboundintro}
Consider a soliton gas \eqref{seqintro} that has regular spectral parameters and that is away from the condensate regime, Eqs.~\eqref{c1intro} and \eqref{c4intro}. Then for every $\alpha>0$, there exist $T>0$ such that
\beq\label{bound2intro}
	\sup_{x\in \R}|\p_x^nu_{\bs \chi,\bs y}(x)|\leq 
	2(TN^{\alpha})^{(n+2)}\,(n+2)!\,\sum_{m_1,\ldots,m_{n+2}\geq 0\atop \sum_j jm_j= n+2}
	\frc{\Big(\sum_j m_j-1\Big)!}{\prod_{j=1}^{n+2} j!^{m_j}m_j!}\,
	\quad
	\forall \ n\in\Z_+,\ N\in\N.
\eeq
\end{theorem}
That is, the supremum of the KdV field on the line, and of all its derivatives, grows more slowly than any power law in $N$, with a coefficient whose dependence on $n$ is as given. Again, $T$ only depends on the parameters characterising the regularity conditions.

\medskip
\noindent{\em Third result: local form of the KdV field (Theorem \ref{theo})}. This is our main technical theorem, and we are not aware of any previous result of a similar type. It says how a KdV field ``looks like'' locally, within the interval of the fluid-cell projection. Here we consider a local region of space-time -- this is the only place where time appears explicitly, as for large-scale results, we need to make solitons' positions time-dependent instead, as in \eqref{bethe}. We set
\beq
	I = [I_-,I_+],\quad I_-<I_+\in\R
\eeq
and consider the fluid-cell projection \eqref{fcintro},
\beq
	\mathcal F_I u_{\bs\chi,\bs y} = u_{\bs\chi',\bs y'}.
\eeq
%
\begin{theorem}\label{theointro} Consider a soliton gas \eqref{seqintro} that has regular spectral parameters and that is away from the condensate regime, Eqs.~\eqref{c1intro} and \eqref{c4intro}. For every $\alpha>0$, there exist constants $D_{n,m}>0,\,\kappa_{n,m}\in\R$ for $n,m=0,1,2,\ldots$, and $E>0$, such that
\beq\begin{aligned}
	\sup_{x\in[I_-+VL/2,I_+-VL/2]\atop
	t\in[-C^{-3},C^{-3}]}\Big|\p_x^n\p_t^m \big(u_{\bs \chi,\bs y+\bs v t}(x) - u_{\bs \chi',\bs y'+\bs v' t}(x)\big)\Big| \leq D_{n,m}
	N^{\kappa_{n,m}} \exp\Big(E N^{\alpha}-\chi_*L\Big)
	&\\
	\forall\
	L\geq 2,\ 
	I_+-I_-\geq VL,\ 
	n,m\in\Z_+,\
	N\in\N.
	&
	\label{projintro}
	\end{aligned}
\eeq
\end{theorem}

Recall that $\chi_*,\,C,\,V$ are the constants involved in \eqref{c1intro}, \eqref{c4intro}. This is a partial projection theorem: taking $L = N^\gamma$ for any $\gamma>0$ as small as desired, the right-hand side of \eqref{projintro} vanishes as $N\to\infty$ for $\alpha$ small enough. Hence this tells us that only a ``small'' number of solitons, just those lying within $I$, may contribute to the shape of the $N$-soliton solution $u_{\bs\chi,\bs y}$ on a slight reduction of the interval $I$. The result is not quite strong enough in order to allow us to perform a fluid-cell average, as we need to control what happens on the boundaries of $I$; but this can be done using our support and boundedness theorems. Again, the constants $D_n,\,\kappa_n$ and $E$ depend only the constants characterising the regularity of spectral parameters and the finite density.

\medskip
\noindent {\em Fourth result: fluid-cell averaging (Theorem \ref{theofc}) and weak limit (Theorem \ref{theoweak}).} These are our main and final theorems, which implement rigorously the consequences of the quasiparticle criterium \eqref{loose}. Again, we are not aware of previous results of this type in the literature.

For the following theorem, we need not only a soliton gas, but also a sequence of fluid cells, intervals $I$ of lengths $L=N^\lambda$ centred at positions $x_*$ that may depend on $N$,
\beq
	I_- = x_*-N^\lambda/2,\quad I_+ = x_*+N^\lambda/2,\quad \lambda>0
\eeq
along with the associated sequence of fluid-cell projections \eqref{fcintro} with \eqref{fcintro2} for $\bs x = \bs x(\bs\chi,\bs y)$ as defined in \eqref{qpreview}. Recall that for the Euler scaling limit, we are interested in taking $x_*\propto N$ and $0<\lambda<1$, which is covered by the theorem.

\begin{theorem}\label{theofcintro} Consider a soliton gas \eqref{seqintro} which has regular spectral parameters and that is away from the condensate regime, Eqs.~\eqref{c1intro} and \eqref{c4intro}.
%
Let $n\in\Z_+$, and let $O:\R^{n+1}\to \R$ be either: a bounded Lipschitz function for the $L^1$ norm on $\R^{n+1}$ with $O(\bs 0) = 0$, or a polynomial without constant term. Consider the local observable
\beq
	\mathtt O[u] (x) := O(u(x),\p_x u(x),\ldots,\p_x^nu(x)),\quad x\in\R.
\eeq
Then
we have
\beq\label{resintro}
	\lim_{N\to\infty} \Bigg(\frc1{L}\int_{I} \dd x\,\mathtt O[u_{\bs\chi,\bs y}] (x)
	- \frc1{L}\int_{-\infty}^{\infty} \dd x\,\mathtt O[u_{\bs\chi',\bs y'}] (x)
	\Bigg)=0
\eeq
and in particular, for every local conserved density $\mathtt P_k[u](x)$ Eq.~\eqref{conspoly},
\beq\label{res2conservedintro}
	\lim_{N\to\infty}
	\Bigg(
	\frc1{L}\int_{I} \dd x\,\mathtt P_k[u_{\bs\chi,\bs y}](x)
	-
	\frc1{L}\sum_{i:x_i(\bs\chi,\bs y)\in I}\chi_i^{2k+1}
	\Bigg)
	=0.
\eeq
\end{theorem}

This gives a rigorous version of \eqref{Prho}. The theorem arises from Corollaries \ref{corolfc} and \ref{corolvanishfc}, which rigorously implement the quasiparticle criterium \eqref{loose}. It allows us to evaluate all such fluid-cell means in terms of the full-line integral of the fluid-cell projection. In particular, for conserved densities, as the full integral is the full conserved quantity, we get the explicit form of the latter purely in terms of spectral parameters. The theorem identifies the precise regularity assumptions that must be fulfilled in order for the fluid-cell means to be given in terms of solitons lying within the fluid cell.

From this, we obtain the rigorous version of \eqref{euler} (at $t=0$); the Euler-scaling case \eqref{euler} is $\Lambda=1$:
\begin{theorem}\label{theoweakintro} Consider a soliton gas \eqref{seqintro} which has regular spectral parameters and that is away from the condensate regime, Eqs.~\eqref{c1intro} and \eqref{c4intro}.
%
%
Let $k\in\Z_+$ and $\Lambda>0$. Then, for every Schwartz function $f:\R\to\R$,
\beq
	\lim_{N\to\infty}
	\Bigg(\int \dd x\,f(x)\,\mathtt P_k[u_{\bs\chi,\bs y}](N^\Lambda x)
	-
	\frc1{N^\Lambda}\sum_{i=1}^N f(x_i(\bs\chi,\bs y)/N^\Lambda)\chi_i^{2k+1}
	\Bigg)=0.
\eeq
\end{theorem}

\subsection{Explicit example: dilute soliton gases}\label{ssectdilute}

In this section we construct an explicit sequence of spectral and impact parameters, \eqref{seqintro}, for which we show that all assumptions of the theorems above are satisfied. This is the case of {\em dilute} soliton gases: informally, a dilute soliton gas is one whose density vanishes as $N\to\infty$. Here, we consider vanishing as $N^{-\ep}$ for $\ep>0$ as small as desired. This example shows how our assumptions work, and that they are natural and can be satisfied. We believe this is a non-trivial example, as even though the density vanishes, its vanishing is very slow, and interaction effects between solitons are still relevant. In particular, the map $\bs x(\bs\chi,\bs y)$ remains non-trivial, e.g.~the ordering of the quasiparticle positions is different from that of the impact parameters $\bs y$.

Consider any sequence $N\mapsto \bs\chi\in\Rpu^N$ satisfying \eqref{c1intro}. For instance, take the regular distribution
\beq\label{exchi}
	\chi_i = \chi_* + (C-\chi_*) i/N,\quad i=1,2,\ldots,N
\eeq
with $C\geq \chi_*+1>1$. Then $\min_{i\neq j\in\dbra 1,N\dket}|\chi_i-\chi_j| = (C-\chi_*)/ N\geq 1/N$ so for every $\eta>0$ we can choose $A = \max_{N\in\N} \log(N)N^{-\eta}<\infty$.

Consider a sequence $N\mapsto \t{\bs y}\in\R^N$. Using this, we construct a dilute soliton gas $N\mapsto \bs\chi,\bs y$ as follows. Let $\t{\bs x} = \bs x(\bs\chi,\t{\bs y})$, let $i_k$ be such that $\t x_{i_{k+1}}>\t x_{i_k}$ for $k=1,2,\ldots,N-1$, choose $R,\ep>0$, and set
\beq
	y_{i_k} = \t y_{i_k}+R(k-N/2) N^\ep.
\eeq
Then
\beq
	x_{i_k} = \t x_{i_k}+R(k-N/2) N^\ep
\eeq
solves \eqref{bethe}, which is easy to show by noticing that, for every $k\neq l$, we have $\sgn(x_{i_k}-x_{i_l}) = \sgn(\t x_{i_k}-\t x_{i_l} + R(k-l)N^\ep)= \sgn(\t x_{i_k}-\t x_{i_l})$ as, by construction, $\sgn(\t x_{i_k}-\t x_{i_l}) = \sgn(k-l)$. As $\t{\bs x}$ satisfies \eqref{separationsimplified}, then so does $\bs x$ because it is less dense. For almost every $R>0$ (with respect to the Lebesgue measure), the non-coincidence condition \eqref{noncoincidence} holds, and hence by the uniqueness statement of Theorem \ref{theosbe}, we have
\beq
	\bs x = \bs x(\bs\chi,\bs y).
\eeq
Below we will have to choose $R$ large enough.

We now show that the soliton gas is away from its condensate regime, Eq.~\eqref{c4intro}. By Eq.~\eqref{boundvarphi} shown below as a consequence of Assumption \ref{asssp}, setting $\beta=0$ and $\alpha/2=\eta$, we have
\beq
	\max_{j=1\atop j\neq i}^N |\varphi_{ij}|
	\leq \frc{1}{\chi_*} (AN^\eta + \log(2C))
	\leq A'N^\eta\quad \forall\;N\in\N
\eeq
for some $A'>0$ that depends on $\eta$. We note that $x_i$'s are separated by nearest-neighbour distances of $RN^\ep$ or more. Then, in order to have at least one term on the left-hand side in \eqref{c4intro}, $x_*$ must be a distance at least $RN^\ep$ from $x_i$, so that for all $i$ and $x_*$ we have
\beq
	\sum_{j:x_j\in[x_*,x_i)}|\varphi_{ij}| \leq \frc{A'}R|x_i-x_*|N^{\eta-\ep}.
\eeq
Choosing $R$ large enough,
\beq
	R>A',
\eeq
and $\ep<\eta$, we have
\beq
	\frc{A'}R|x_i-x_*|N^{\eta-\ep} \leq W|x_i-x_*|,\quad W = A'/R\in[0,1).
\eeq
This shows \eqref{c4intro} for all $i$ and $x_*$, which more than is required.

Therefore, we have shown that the dilute soliton gas $N\mapsto \bs\chi,\,\bs y$ has regular spectral parameters and that it is away from the condensate regime. Hence all theorems above hold for this gas.

\section{Tau functions and soliton positions}\label{secmain}

The KdV equation \eqref{kdv} has many types of solutions, which can be obtained for instance by the method of inverse scattering \cite{babelon2003introduction}. Our analysis is based on tau functions instead, see \cite{hietarinta2007introduction} for a general discussion of their meaning and significance. We are interested in the $N$-soliton solution, with $N$ becoming large. In this section, we introduce the concepts and basic results that will be needed to prove our main theorems. The section is as pedagogical as possible, and we leave the longer, technical proofs to Section \ref{secproofs}.

\subsection{Tau-function representations of the KdV multi-soliton solution}\label{ssecttau}

We start by writing down the explicit $N$-soliton solution in terms of so-called tau functions. The form \eqref{tau}, below, of the $N$-soliton tau function is standard, but not unique, see e.g.~\cite{gardner1974korteweg}. We obtain other forms of tau functions, as expressed in Lemma \ref{lemtaudecomp}, which play a crucial role in our study and which have a clear interpretation in terms of the factorised scattering theory of the $N$-soliton solutions to the KdV equation.

Recall that we use $N\in\N,\,\bs y \in\R^N$ (impact parameters) and $\bs\chi\in\Rpu^N$ (spectral parameters, see Eq.~\eqref{Rpudef}), as well as $v_i = 4\chi_i^2,\ i=1,\ldots,N$. We write
\beq\label{yt}
	\bs y^{(t)} = \bs y + \bs v t, \quad t\in\R
\eeq
for the time-evolved impact parameters.

The $N$-soliton solution to the KdV equation \eqref{kdv} can be expressed in terms of $\bs \chi$ and $\bs y^{(t)}$ as (its dependence on $N$ is within the understanding that $\bs \chi,\bs y^{(t)}\in\R^N$)
\beq\label{utau}
	u(x,t) = u_{\bs \chi,\bs y^{(t)}}(x):=2\p_x^2 \log \tau_{\bs\chi,\bs y^{(t)}}(x)
\eeq
where the tau function can be written as
\beq\label{tau}
	\tau_{\bs\chi,\bs y}(x) = \det(\Psi(x)^2 + \omega),\quad  x\in\R
\eeq
with $\Psi(x),\,\omega\in\Mat(N)$ being $N$ by $N$ real matrices given by (we keep their dependence on $\bs \chi,\,\bs y$ implicit for lightness of notation)
\beq\label{Psiomega}
	\Psi(x) = \diag(\Psi_i(x))_i,\quad \Psi_i(x) = \re^{\chi_i(x-a_i)},\quad
	\omega_{ij} = \frc{2\sqrt{\chi_i\chi_j}}{\chi_i+\chi_j}.
\eeq
Here the time-evolved ``na\"ive'' impact parameters $a_i$ are related to the impact parameters $y_i$ and spectral parameters $\bs\chi$ via (see \eqref{naiveimpacts}, \eqref{Xextremeintro})
\beq\label{ay}
	a_i = y_i - \frc12 \sum_{j=1\atop j\neq i}^N \varphi_{ij} = a_i^\emptyset = X_i^+(\bs\chi,\bs y)
\eeq
where (see Eq.~\eqref{factscat})
\beq\label{varphi}
	\varphi_{ij} = \frc1{\chi_i} \log\Big|\frc{\chi_i-\chi_j}{\chi_i+\chi_j}\Big|.
\eeq
Note that
\beq\label{varphineg}
	\varphi_{ij}<0\quad\forall\ i\neq j.
\eeq

The function \eqref{utau} is smooth, and in fact it is a Schwartz function. One may refer to $\Psi_i(x)$ as the {\em partial wave} of the soliton $i$. As mentioned, the choice of $\bs y$ as impact parameters, instead of the more standard $\bs a$, with the complicated-looking relation \eqref{ay}, is because $\bs y$ has a more universal physical meaning from the viewpoint of scattering theory as we described in Section \ref{ssectbackground} and develop below.

We often concentrate on the time slice $t=0$; results for other times are simply obtained by linear shifts of $\bs y$. The set of all multi-soliton solutions is $\mathfrak S$, Eq.~\eqref{defiS}.

The function $u_{\bs \chi,\bs y}$ can be represented in terms of various tau functions, other than \eqref{tau}. This is because the right-hand side of \eqref{utau} is unchanged under the replacement $\tau_{\bs \chi,\bs y^{(t)}}(x)\to e^{A^{(t)}x+B^{(t)}}\tau_{\bs \chi,\bs y^{(t)}}(x)$, for any $A^{(t)},\,B^{(t)}\in\R$ (which may depend on $t$). This gives us an equivalence relation on tau functions\footnote{For terminology, we use ``the tau function'' in oder to describe the equivalence class of $\tau$, Eq.~\eqref{tau}, and ``a tau function'' or ``a tau-function representation'' or ``a form of the tau function'' for an element of the class.}, as functions of $x$:
\beq\label{equiv}
	\tau_{\bs \chi,\bs y}\equiv \tau' \quad \mbox{iff}\quad 
	\tau'(x)= e^{Ax+B}\tau_{\bs \chi,\bs y}(x)\ \forall x\in\R \quad \mbox{for some  $A,\,B\in\R$}
\eeq
and
\beq
	u_{\bs \chi,\bs y}(x,t) = 2\p_x^2 \log \tau'(x)\quad \forall \ \tau'\equiv \tau_{\bs \chi,\bs y}.
\eeq

It is convenient to expand the determinant form \eqref{tau} in terms of partial waves $\Psi_i(x)$, using the Cauchy structure of the matrix $\omega$, Eq.~\eqref{Psiomega}. In particular, there is a natural choice of tau function, within the equivalence class of \eqref{tau}, for every choice of subset $s\subseteq\dbra 1,N\dket$, where the set $s$ represents, as we explain below, ``out-solitons'', while its complement $\dbra 1,N\dket\setminus s$ represents ``in-solitons''. These tau functions are expanded into partial waves as follows:
\begin{lemma}[In-out representations] \label{lemtaudecomp}
For every $N\in\N$, $\bs\chi\in\Rpu^N,\,\bs y\in\R^N$ and  $s\subseteq \dbra 1,N\dket$, 
\beq\label{taudecomp}
	\tau_{\bs \chi,\bs y} \equiv \tau_{\bs \chi,\bs y}^s,\quad \tau_{\bs \chi,\bs y}^s(x):=\frc1{S_s}\sum_{p\subseteq s,\,q\subseteq \dbra 1,N\dket\setminus s}
	S_{(s\setminus p)\cup q}
	\prod_{i\in p}\Psi_i(x)^2 \prod_{j\in q}\Psi_j(x)^{-2}
\eeq
where
\beq\label{Sp}
	S_{r} := \prod_{i,j\in r\atop i\neq j}|S_{ij}|
	=\prod_{i,j\in r\atop i<j}
	S_{ij}^2,
	\quad S_{ij}:=\frc{\chi_i-\chi_j}{\chi_i+\chi_j},
	\quad  r\subseteq\dbra 1,N\dket.
\eeq
Further, defining
\beq\label{xist}
	a_i^{s} := y_i + \frc12\sum_{j=1\atop j\neq i}^N \sgn_s(j)\varphi_{ij}
	\quad(i\in\dbra 1,N\dket),\quad
	a_i^{s,r} := y_i + \frc12\sum_{j=1 \atop j\not\in r}^N  \sgn_s(j)\varphi_{ij}
	\quad (i\in r)
\eeq
where $\sgn_s(j)$ is defined in \eqref{sgns}, we have the following two equivalent expressions:
\beqa\label{tauresult}
	\tau_{\bs \chi,\bs y}^s(x)&=&
	\sum_{r\subseteq \dbra 1,N\dket}\ 
	\exp \sum_{i\in r} 2\chi_i\sgn_s(i)(x-a_i^{s})\ 
	\exp \sum_{i,j,\in r,\,i\neq j} \sgn_s(i)\sgn_s(j)
	\log|S_{i,j}|\\
	&=&\label{tauresult2}
	\sum_{r\subseteq \dbra1,N\dket}\ 
	\exp \sum_{i\in r} 2\chi_i\sgn_s(i)(x-a_i^{s,r}).
\eeqa
\end{lemma}
The proof is provided in Sec.~\ref{ssecttaurep}; in order to provide an alternative to the more standard expansion in minors, we use a recursive method. We note that for every function $t\mapsto s(t)\subseteq\dbra 1,N\dket$, we have
\beq\label{tauprimet}
	\tau_{\bs \chi,\bs y^{(t)}} \equiv \tau_{\bs \chi,\bs y^{(t)}}^{s(t)}.
\eeq
That is, in the expressions in Lemma \ref{lemtaudecomp}, we can give an additional $t$-dependence to the set $s$ besides the linear displacement of impact parameters, and the result still provides a $t$-dependent tau function that gives the $N$-soliton solution $u_{\bs\chi,\bs y^{(t)}}$ to the KdV equation.

The explicit expansion of the determinant, such as in the in-out representation above, directly gives a simple bound on the $N$-soliton field itself. This bound is weak, for instance the $n$th derivative grows as $N^{n+2}$ for bounded spectral parameters, while it is known, for $n=0$, that the field $u$ is bounded in this case \cite{lundina1985compactness,gesztesy1992limits}. However here we also have bounds on the derivatives, $n\geq 1$, which is important for our results, and when combined with our local projection result Theorem \ref{theo}, this will give a much tighter bound (while still weaker than the one quoted on $u$ itself, it will be enough), Corollary \ref{corolbound2}. 
\begin{lemma}[Rough bound on the KdV field]\label{lembound1} In the context of Lemma \ref{lemtaudecomp}, we have
\beq\label{bound1}
	|\p_x^nu_{\bs\chi,\bs y}(x)|
	\leq 
	2(n+2)!\sum_{m_1,\ldots,m_{n+2}\geq 0\atop \sum_j jm_j= n+2}
	\frc{\Big(\sum_j m_j-1\Big)!}{\prod_{j=1}^{n+2} j!^{m_j}m_j!}
	\Big(2\sum_i\chi_i\Big)^{n+2}
\eeq
for all $n\in\Z_+$ and $x\in\R$.
\end{lemma}
\proof
From \eqref{utau} and Faa Di Bruno's formula, we have
\beq
	\p_x^n u_{\bs\chi,\bs y}(x) = 2(n+2)!\sum_{m_1,\ldots,m_{n+2}\geq 0\atop \sum_j jm_j= n+2}
	(-1)^{\sum_j m_j-1}\Big(\sum_j m_j-1\Big)!\prod_{j=1}^{n+2} \frc{(\p_x^j \tau_{\bs\chi,\bs y}(x))^{m_j}}{j!^{m_j}m_j!\tau_{\bs\chi,\bs y}(x)^{m_j}}.
\eeq
For convenience, let us use \eqref{taudecomp} with $s=\dbra1,N\dket$, so that $q=\emptyset$ in every term. Consider the term with $p$ in \eqref{taudecomp}. Its $j$th derivative $\p_x^j$ is bounded as
\beq
	\Bigg(2\sum_{i\in p} \chi_i\Bigg)^j \,\frc{S_{s\setminus p}}{S_s} \prod_{i\in p}\Psi_i(x)^2 \leq \Bigg(2\sum_{i=1}^N \chi_i\Bigg)^j \,\frc{S_{s\setminus p}}{S_s} \prod_{i\in p}\Psi_i(x)^2.
\eeq
Hence
\beq
	\p_x^j \tau_{\bs\chi,\bs y}(x) \leq \Bigg(2\sum_{i=1}^N \chi_i\Bigg)^j\tau_{\bs\chi,\bs y}(x).
\eeq
The result follows.
\eproof

As a check, and in order to gain physical intuition, it is easy to see that we get the correct scattering shifts at positive and negative infinite times. Take \eqref{tauresult2}. Let us follow soliton $i_*$ and concentrate on the position $x = \b x + v_{i_*} t$ to large times $t\to\infty$. Then it is convenient to take $s= \dbra i_*,N\dket$, as, accounting for the order $v_1<v_2<\ldots<v_N$, only two terms do not vanish, those with $r=\{i_*\}$ and $r=\emptyset$, and both are finite in the limit. Hence we find
\beq
	\lim_{t\to\infty} \tau_{\bs \chi,\bs y^{(t)}}^s(\b x + v_{i_*} t)= 
	\exp\Big[2\chi_{i_*}\Big(\b x-y_{i_*}-\frc12\sum_{j\neq i_*}\sgn(j-i_*) \varphi_{i_*j}\Big)\Big]+1
	=
	\exp\Big[2\chi_{i_*}\Big(\b x-x_{i_*}^+\Big)\Big]+1
\eeq
with {\em outgoing impact parameters} from \eqref{outgoingincoming}. This is a tau function for the 1-soliton solution centred on $x_{i_*}^+$. A similar calculation can be done for all spatial derivatives of the tau function, which also are finite in this limit, and hence the soliton solution \eqref{utau} converges accordingly. On the other hand with $t\to-\infty$ we choose $s = \dbra 1,i_*\dket$ and get
\beq
	\lim_{t\to-\infty} \tau_{\bs \chi,\bs y^{(t)}}^s(\b x + v_{i_*} t) =
	\exp\Big[2\chi_{i_*}\Big(\b x-y_{i_*}+\frc12\sum_{j\neq i_*}\sgn(j-i_*) \varphi_{i_*j}\Big)\Big]+1
	=
	\exp\Big[2\chi_{i_*}\Big(\b x-x_{i_*}^-\Big)\Big]+1
\eeq
with {\em incoming impact parameters} from \eqref{outgoingincoming}. Now we get the tau function for the 1-soliton solution centred on $x_{i_*}^-$. The difference of impact parameters gives the scattering shift of this soliton, where we now write the sign function in terms of the velocities for more physical clarity, see Eq.~\eqref{factscat}:
\beq
	\Delta x_{i_*} = x_{i_*}^+ - x_{i_*}^- = \sum_{j=1\atop j\neq i_*}^N
	\sgn(v_j-v_{i_*})\varphi_{i_*j}.
\eeq
We see that soliton $i_*$ has been affected by a shift to the right by the positive quantity $-\varphi_{i_*j}$ upon scattering with slower solitons $j:v_j<v_{i_*}$, and by a shift to the left by the positive quantity $-\varphi_{i_*j}$ upon scattering with faster solitons $j:v_j>v_{i_*}$. This is in agreement with the known factorised scattering theory of the KdV equation \cite{gardner1974korteweg}, and shows \eqref{uasymp} with \eqref{outgoingincoming}. Note that we also recover \eqref{yaverage}. A similar calculation gives \eqref{uasympgen}.

We see above that a convenient choice of $s$ allowed us to make sure that the tau function and its derivatives converge to a finite, non-zero value in the large-time limit. In this way the evaluation of the soliton solution is rendered simple, as the derivative of the logarithm of the tau function gives the ratio
\beq
	\p_x \log\tau = \frc{\tau \p_x^2 \tau - (\p_x\tau)^2}{\tau^2}
\eeq
and we just have to take ratios of the finite limits. Different choices of $s$ would have given the same result (thanks to the equivalence \eqref{equiv}), however the calculation of the soliton solution for the KdV field requires more work because the numerator and denominator in \eqref{utau} either diverge or go to zero. A clever choice of set $s$ is at the basis of our fluid-cell projection results.

Using factorised scattering ideas, the physical interpretation of \eqref{tauresult}, \eqref{tauresult2} is now clear: We choose a set $s$ representing solitons that naturally ``go to $\infty$'', in the sense that their corresponding partial waves vanish as $y_i\to\infty$ -- these are our ``out-solitons''; while solitons in the complement $\dbra 1,N\dket\setminus s$ ``go to $-\infty$'' -- these are our ``in-solitons''. Then, in the way of writing \eqref{tauresult2}, we sum over all possible subsets of solitons $r$ that are present, i.e.~that have ``not yet been sent to infinity.'' But these solitons have their coordinates shifted, according to the scattering picture, by half of the scattering shifts due to all solitons that have been sent to infinity. This accounts for ``half of the effect'' of the full scattering, as the solitons that are present have been subjected only to half-trajectories, on times $[0,\infty)$ or $(-\infty,0]$, of the solitons sent to infinity. In \eqref{tauresult} we shift by half the scattering shift due to all solitons, even those not ``sent to infinity'', and counterbalance with a self-interaction term within the group of solitons $r$. This physical explanation is what justifies our nomenclature {\em in-out representations} for the tau functions in Lemma \ref{lemtaudecomp}. Note how these expressions, with this physical interpretation, are similar to what is seen in the coordinate Bethe ansatz for wave functions (see e.g.~\cite{korepin1997quantum}).

\subsection{Semiclassical Bethe equations and the centred form of the tau function}\label{ssectsbe}

As mentioned, it has been conjectured that appropriate solutions to the SBE \eqref{bethe} determine meaningful soliton positions \cite{doyon2024new}. We now introduce these equations and the main notions related to them, in particular showing how they allow us to obtain a form of the tau function where the strength of every term can be controlled by the positions $x_i$ solving the SBE. In Section \ref{ssectsberesults}, we then study the SBE in more depth, finding explicit solutions and specifying some of their main properties. Seeing a solution to the SBE as specifying solitons' positions, there are associated notions of imprecisions, which we study in Section \ref{ssectimp}. Making an appropriate choice of solution to the SBE, in Section \ref{ssectprojections} we define fluid-cell projections. Thus, solutions to SBE will be our quasiparticle positions, and we will obtain projection theorems under relatively simple, algebraic conditions on such positions. In Sections \ref{ssectsberesults} and \ref{ssectimp}, we establish slightly more general and in-depth results than what is strictly necessary for our main theorems, in order to provide a fuller understanding of the SBE.

Recall the conventions \eqref{defab} and \eqref{argminmax}. In particular, when we write $[a,b)$, the interval is non-empty both for $a> b$ and $a< b$ (it is empty if $a=b$). We note that the results of this subsection hold for any matrix $(\varphi_{ij})_{i,j=1}^N$ with $\varphi_{ij}<0\;\forall\;i\neq j$, in place of the matrix \eqref{varphi}.

\begin{defi}[Semiclassical Bethe vectors]\label{defibethe}
Let $N\in\N$, $\bs \chi\in\Rpu^N,\,\bs y\in\R^N$. A {\em Semiclassical Bethe vector} for $\bs \chi,\,\bs y$ is a solution $\bs x\in\Rne^N$ to the system of SBE:
\beq\label{bethemain}
	y_i = x_i +\frc12\sum_{j=1\atop j\neq i}^N \sgn(x_i-x_j)\varphi_{ij},\quad i=1,\ldots,N.
\eeq
We define the {\em contraction map} $\mathcal C_{\bs \chi}: \Rne^N\to\R^N$ via Eq.~\eqref{bethemain} as
\beq\label{mapC}
	\bs y = C_{\bs\chi}(\bs x)\quad\forall\;\bs\chi\in\Rpu^N.
\eeq
\end{defi}
Theorem \ref{theosbe} below guarantees that a solution in $\Rne^N$ to the SBE always exist (hence $\mathcal C_{\bs\chi}$ has a right-inverse), and gives explicit solutions. An appropriate semiclassical Bethe vector will be our proposal for {\em quasiparticle positions}, representing effective positions of solitons in a multi-soliton KdV field. The fact that it is a good proposal is justified by our main projection theorems.

Note that in the cases where $\varphi_{ij}\geq0$, and with a differentiable regularisation of the sign function, it is shown in \cite{doyon2026generalised} that \eqref{bethemain} has a unique solution, giving a diffeomorphism $\R^N\to\R^N$. However, for the KdV equation we have $\varphi_{ij}<0$, hence this result does not apply, and there are, in general, many solutions.

An important notion is that of {\em locally contracted vectors} $\bs X$ at a point $x_*\in\R$, and the associated ``distances'' of their coordinates to $x_*$. The coordinate $X_i$ is obtained by subtracting the scattering jumps $\varphi_{ij}$ corresponding to interactions of $i$ with solitons at positions $x_j$ lying between $x_*$ and $x_i$.
\begin{defi}[Locally contracted vector at $x_*\in\R$]\label{definatmag}
Let $N\in\N,\,\bs\chi\in\Rpu^N,\,\bs x\in\Rne^N$. The locally contracted vector at $x_*$ for $\bs\chi,\,\bs x$ results from the {\em partial contraction map at $x_*$}, $\h{\mathcal C}_{\bs\chi,x_*}:\Rne^N\to\R^N$:
\beq\label{Xnatural}
	\bs X = \h{\mathcal C}_{\bs\chi,x_*}(\bs x),\quad
	X_i = x_i + \sgn(x_i-x_*)\sum_{j:x_j\in[x_*,x_i)\,(x_i
	\geq x_*)\atop j:x_j\in(x_i,x_*)\,(x_i< x_*)}\varphi_{ij}.
\eeq
The {\em contracted distances} to $x_*$ are
\beq\label{contracteddist}
	D_i = (X_i - x_*)\sgn(x_i-x_*).
\eeq
\end{defi}
We note that, because $x_i\neq x_j\,\forall\, i\neq j$, if $x_*\in[x_i-\ep,x_i+\ep]$ for $\ep>0$ small enough, then there are no contraction terms in \eqref{Xnatural}, whence $X_i=x_i$. That is, for all $i\in\dbra 1,N\dket$, we have the fixed-point property
\beq\label{fixedpoint}
	\big(\h{\mathcal C}_{\bs\chi,x_i}(\bs x)\big)_i  = x_i,\quad
	D_i=0\ \mbox{for}\ x_*=x_i
\eeq
as well as
\beq\label{fixedneigh}
	\big(\h{\mathcal C}_{\bs\chi,x_*}(\bs x)\big)_i  = x_i,\quad
	D_i=|x_i-x_*|\quad \forall\; x_*\in[x_i-\ep,x_i+\ep],\ \ep\mbox{ small enough.}
\eeq
Eq.~\eqref{fixedpoint} guarantees that the value of the sign function at 0 in \eqref{contracteddist} is irrelevant. We also note the contraction property
\beq\label{contractionprop}
	D_i \leq |x_i-x_*|\quad\forall\,i\in\dbra 1,N\dket.
\eeq

In evaluating coordinate $X_i$ in Definition \ref{definatmag}, we take away all interaction effects between $x_*$ and $x_i$. Hence, heuristically, these only partially account for the collision rate ansatz, accumulating scattering shifts as if the soliton $i$ were at the position $x_*$. Further away from $x_*$, positions $X_i$'s do not fully account for all collisions, by contrast to positions $x_i$'s satisfying the SBE \eqref{bethe}, and thus are further away from true soliton positions. Much like the SBE maps free coordinates $y_i$'s to interacting coordinates $x_i$'s, contracted positions, as functions of $y_j$'s, are still {\em free with respect to the point $x_*$}. Because the scattering shift is negative, Eq.~\eqref{varphineg}, this means that coordinates $X_i$ around $x_*$ are contracted towards $x_*$ with respect to true solitons' coordinates $x_i$'s solving the SBE. Geometrically (see \cite{doyon2018geometric} for the geometric viewpoint on GHD on which this interpretation is based), the resulting {\em contracted space} is affected by a metric that is further from the flat metric the further we are from $x_*$. We note that the contracted distances are related to the parameters appearing in the in-out representation, Lemma \ref{lemtaudecomp}, as is clear from the proof of Lemma \ref{lemtaumain} (Eq.~\eqref{aisXi}):
\beq
	X_i = a_i^{\{i:x_i\geq x_*\}}.
\eeq

The associated contracted distances $D_i$ are not true distances, as they can be negative. However, as we will show (Lemma \ref{lemnat}), they cannot be ``too'' negative, and for certain solutions (Theorem \ref{theosbe}), they are non-negative for all $x_i>x_*$ (or all $x_i<x_*$). The physical meaning of these is that the less positive $D_i$ is, the more the soliton ``condenses at $x_*$'', in the sense of soliton condensates \cite{el2021soliton,suret2024soliton}. Our picture is that when the distance $D_i$ is negative, the position $X_i$ lies in a {\em dark space} beyond $x_*$ -- as if some space needed to be added in order to account for them, and included within the contracted space. This picture makes sense for our definition of a local projection, Definition \ref{defimap}, at the basis of our main local-projection Theorem \ref{theo}: we consider solitons lying on some interval $I$ in contracted space, which includes the full dark space, no matter how far from $x_*$ in the wrong direction a contracted coordinate $X_i$ is.


Our main technical justification for introducing the semiclassical Bethe vectors and in particular their contracted distances, is their relation with tau functions. Combining the notion of contracted distances with a judicious choice of the tau-function representation \eqref{tauresult}, \eqref{tauprimet}, we obtain the following (here we use one result from Lemma \ref{lemnat} below).
\begin{lemma}[Centred form] \label{lemtaumain}
Let $N\in\N$, $\bs\chi\in\Rpu^N,\,\bs y,\,\Delta\bs y\in\R^N$, and  $x_*\in\R$. Let $\bs x$ be a semiclassical Bethe vector for $\bs\chi,\,\bs y$, Definition \ref{defibethe}, and let $\bs D$ be the vector of contracted distances at $x_*$, Definition \ref{definatmag}. Then under the equivalence \eqref{equiv}, we have
\beq\label{taueqcentred}
	\tau_{\bs\chi,\bs y+\Delta\bs y}  \equiv \tau_{\bs\chi,\bs x,\Delta\bs y,x_*}^\#
\eeq
where
\beq
	\tau_{\bs\chi,\bs x,\Delta\bs y,x_*}^\#(x) := \sum_{r\subseteq \dbra1,N\dket} 
	\exp\Big[-\sum_{i\in r} 2\chi_i\big(
	D_i
	-\sgn(x_i-x_*)(x-x_*-\Delta y_i)\big)+ R_r\Big]
	\label{taumain}
\eeq
where
\beq\label{eR}
	R_r := \sum_{i,j\in r,\,i\neq j}
	\sgn(x_i-x_*)
	\sgn(x_j-x_*) \log|S_{i,j}|
\eeq
(recall \eqref{varphi} and \eqref{Sp}).
\end{lemma}
\proof We use the in-out representation, Lemma \ref{lemtaudecomp}, putting the $\Delta \bs y$ shift out of $a_i^s$ in \eqref{tauresult}:
\beq
	\tau_{\bs \chi,\bs y+\Delta \bs y}^s(x)=
	\sum_{r\subseteq \dbra 1,N\dket}\ 
	\exp \sum_{i\in r} 2\chi_i\sgn_s(i)(x-a_i^{s}-\Delta y_i)\ 
	\exp \sum_{i,j,\in r,\,i\neq j} \sgn_s(i)\sgn_s(j)
	\log|S_{i,j}|
\eeq
with the first of Eq.~\eqref{xist}. Let $s=\{i:x_i\geq x_*\}$. Then, recalling our convention \eqref{sgn}, we have
\beq
	\sgn_s(i) = \sgn(x_i-x_*).
\eeq
We then obtain
\beq
	 \sum_{i,j,\in r,\,i\neq j} \sgn_s(i)\sgn_s(j)
	\log|S_{i,j}| = R_r
\eeq
and from \eqref{xist} and \eqref{Xiyi} below, 
\beq\label{aisXi}
	a_i^{s} = X_i.
\eeq
Hence we have
\beqa
	\sgn_s(i) (x-x_i^s-\Delta y_i)&=&\sgn(x_i-x_*)(x-X_i-\Delta y_i)\n
	&=&
	\sgn(x_i-x_*)(x-x_*-\Delta y_i-(X_i-x_*))\n
	&=&-D_i+\sgn(x_i-x_*)(x-x_*-\Delta y_i).
\eeqa
\eproof

\medskip

One aspect of the result \eqref{taueqcentred} is that the left hand-side is a function that depends only on $\bs y+\Delta \bs y$, while the right-hand side is a function that depends on $\bs y$, through $\bs x$, and on $\Delta \bs y$, separately (as well as on $x_*$). Of course, the two functions are different, however the Lemma says that they are equivalent -- they give rise to the same KdV field. The way we choose to separate impact parameters into the sum $\bs y+\Delta \bs y$ is up to us, depending on the application. We may simply set $\Delta \bs y=\bs 0$, and we will denote
\beq
	\tau_{\bs\chi,\bs x,x_*}^\#:=\tau_{\bs\chi,\bs x,\bs 0,x_*}^\#.
\eeq
Below we will use $\tau_{\bs\chi,\bs x,\Delta\bs y,x_*}^\#$ with $\Delta \bs y = \bs v t$ in order to establish the local form for small times $t$ in Theorem \ref{theo}.

Another comment is about time evolution. As a result of the lemma, for every $\R\ni t\mapsto \bs x^{(t)}\in\R^N$ satisfying ${\mathcal C}_{\bs \chi}(\bs x^{(t)}) = \bs y^{(t)}\;\forall\;t\in\R$, we have $\tau_{\bs\chi,\bs y^{(t)}}  \equiv \tau_{\bs\chi,\bs x^{(t)},x_*}^\#\;\forall\;t\in\R$. Then, Lemma \ref{lemtaumain} along with
\beq\label{discussioncentredform}
	u_{\bs \chi, \bs y^{(t)}} (x) = 2\p_x^2\log\tau^\#_{\bs\chi,\bs x^{(t)},x_*}(x),\quad
	\bs x^{(t)} \mbox{ solves \eqref{bethet}},
\eeq
gives a crucial representation of the $N$-soliton solution that we will use in order to establish Theorem \ref{theo}. It gives the tau function for the KdV solution in its {\em centred form} $\tau^\#_{\bs\chi,\bs x,x_*}(x)$: every term in its partial wave expansion, parametrised by the subset $r$, is now strictly decaying as a function of the contracted distances. The tau function is ``centred'' around the observation point $x_*$, being dominated, at all times $t$, by solitons whose contracted distances are not too large. Of course, the representation \eqref{taumain} is valid for all values of $x$, hence gives the correct $N$-soliton solution even for values of $x$ that are far from $x_*$. However, as we will see below, this centred representation allows us to control what we will refer to as local projections, where solitons far from $x_*$ in the sense of contracted distances, are taken away, making sure that in doing this, the KdV field and its derivatives change very little for all values of $x$ that are near enough to $x_*$. It is in this sense that this is a good representation around $x_*$, which focusses only the solitons that matter for the KdV field around it. Formalising this is the purpose of our main technical theorem, Theorem \ref{theo}.

\ssect{Properties of semiclassical Bethe vectors}\label{ssectsberesults}

In this section we establish the main properties of semiclassical Bethe vectors, Definition \ref{defibethe}.

Interestingly, we find that there are always solutions to the SBE, which can be explicitly constructed by a finite algorithm, that satisfy at least one of two one-sided density conditions, Eqs.~\eqref{distexist} and \eqref{distexist2}. These conditions naturally generalise the hard-rod non-overlapping conditions, and are ``microscopic'' versions of the condition according to which soliton condensates correspond to the densest soliton configurations \cite{el2021soliton,suret2024soliton}. They guarantee that contracted distances $D_i$ are non-negative, either for all $x_i\geq x_*$, or all $x_i\leq x_*$.

We define ``regular'' points as those for which the density conditions hold on both sides of it; these points have particularly nice properties. We define a notion of  ``low density'' such that regular points are almost everywhere. These notions are useful to characterise solutions to the SBE, however we will not use them explicitly in our main results.


As we will see below, for certain points $x_*$, where no contracted coordinate lies in the dark space (i.e.~$D_i<0$), the locally contracted vector is an example of a {\em regular magnifying-glass vector}. Magnifying-glass vectors can be defined more generally, independently from the SBE, and also represent solitons' positions for a contracted metric at $x_*$, see Appendix \ref{appmag}; it is sufficient to restrict to ``regular'' magnifying-glass vectors here, where all coordinates are  different from $x_*$. In Appendix \ref{appmag} we partially develop the general theory of magnifying-glass vectors, and associated representations of tau functions and local projection theorem.

\begin{defi}[Regular magnifying-glass vectors at $x_*\in\R$]\label{defimagstr}
Let $N\in\N$, $\bs \chi\in\Rpu^N,\,\bs y\in\R^N$. A {\em regular magnifying-glass vector at $x_*$} for $\bs\chi,\,\bs y$ is a solution $\bs X\in\R^N,\;x_i\neq x_*\,\forall\,i$ to the following system of equations:
\beq\label{eqxymainXstr}
	y_i = X_{i} - \frc12\sum_{j=1\atop j\neq i}^N\sgn(X_j-x_*)\varphi_{ij},\quad i=1,\ldots,N.
\eeq
We define the {\em partial contraction map away from $x_*$}, $\mathcal C_{\bs \chi,x_*}: \{\bs X\in\R^N:x_i\neq x_*\;\forall\;i\}\to\R^N$, as
\beq\label{mapC}
	\bs y = C_{\bs\chi,x_*}(\bs X)\quad\forall\;\bs\chi\in\Rpu^N,\,x_*\in\R.
\eeq
\end{defi}


We observe that any semiclassical Bethe vector and regular magnifying-glass vector is of the form
\beq\label{formvectors}
	x_i = y_i + \frc12\sum_{j=1\atop j\neq i}^N\sigma_j^{(i)}\varphi_{ij},
	\quad
	X_i = y_i + \frc12\sum_{j=1\atop j\neq i}^N\sigma_j\varphi_{ij},	\quad
	\sigma_j^{(i)},\,\sigma_j\in\{\pm\}.
\eeq
Let us define {\em right- and left-extremal positions} $X_i^\pm$ and their maximum and minimum values:
\beq\label{condxstar}
	X_i^\pm(\bs\chi,\bs y) := y_i \mp \frc12\sum_{j=1\atop j\neq i}^N\varphi_{ij},\quad	
	x_-(\bs\chi,\bs y)= \min_i\big(X_i^-\big),\quad
	x_+(\bs\chi,\bs y) = \max_i\big(X_i^+\big).
\eeq
Clearly $x_-(\bs\chi,\bs y)<x_+(\bs\chi,\bs y)$ and we define the {\em core} of the KdV field characterised by $\bs\chi,\bs y$ as
\beq\label{core}
	\core(\bs\chi,\bs y) = [x_-(\bs\chi,\bs y),x_+(\bs\chi,\bs y)].
\eeq
We will show in Corollary \ref{corolextent}  that this determines where the KdV field is supported. It is easy to see that
\beq\label{sbeboundinv}
	[\min_i X_i,\max_i X_i],\ [\min_i x_i,\max_i x_i]\subseteq\core(\bs\chi,\bs y) .
\eeq
Note also that $X_i^+(\bs\chi,\bs y) = a_i$ from \eqref{ay}. That is, the na\"ive impact parameters $a_i$'s are magnifying-glass soliton positions for a large enough observation points. Similarly, $X_i^-(\bs\chi,\bs y)$ corresponds to the na\"ive impact parameters in the ``opposite'' formulation of tau functions obtained under the space-time parity transformation, $(x,t),a_i\mapsto (-x,-t),-a_i$.

For the proofs below, it is convenient to establish the following basic properties.
\begin{lemma}[Properties of regular magnifying-glass vectors] \label{lemdstr} In the context of Definition \ref{defimagstr}, whenever $x_*< x_-(\bs\chi,\bs y)$, resp.~$x_*> x_+(\bs\chi,\bs y)$, the regular magnifying-glass vector exists and is unique, and is given by
\beq\label{valuesextremepositions}
	X_i = X_i^-(\bs\chi,\bs y),
	\quad\mbox{resp.}\quad
	X_i = X_i^+(\bs\chi,\bs y).
\eeq
If $x_*> x_-(\bs\chi,\bs y)$, then for any regular magnifying-glass vector $\bs X$ there exists $i\in\dbra 1,N\dket$ such that $X_i\in[x_-(\bs\chi,\bs y),x_*)$. Similarly, if $x_*< x_+(\bs\chi,\bs y)$, then there exists $i\in\dbra 1,N\dket$ such that $X_i\in(x_*,x_+(\bs\chi,\bs y)]$. If $x_*=x_\pm(\bs\chi,\bs y)$, there are no regular magnifying-glass vectors.
\end{lemma}
\proof Let $x_*<x_-(\bs\chi,\bs y)$. Then because of \eqref{sbeboundinv} we have $\sgn(X_j-x_*)=1$ in \eqref{eqxymainXstr} and therefore we must have $X_i = X_i^-(\bs\chi,\bs y)\,\forall\, i$, and this is a solution. A similar argument holds for $x_*>x_+(\bs\chi,\bs y)$. If $x_*\geq x_-(\bs\chi,\bs y)$, then assume there are no $i$'s such that $X_i\leq x_*$. Then again $\sgn(X_j-x_*)=1$ in \eqref{eqxymainXstr} and $X_i =  X_i^-(\bs\chi,\bs y)\,\forall i$, therefore $\min_i X_i = x_-(\bs\chi,\bs y)$, a contradiction. Hence, again considering \eqref{sbeboundinv}, there must exist $i$ such that $X_i\in[x_-(\bs\chi,\bs y),x_*]$, and as this is a regular magnifying-glass vector, $X_i\neq x_*$, therefore $X_i\in[x_-(\bs\chi,\bs y),x_*)$. If $x_*=x_-(\bs\chi,\bs y)$, this is impossible, hence there cannot be a solution. A similar argument holds for $x_*\leq x_+(\bs\chi,\bs y)$. \eproof

\medskip

Given a semiclassical Bethe vector $\bs x$ with respect to $\bs\chi,\,\bs y$, we say that $x_*\in\R$ is a {\em regular point} for it if
\beq\label{condreg}
	|x_i-x_*|> -\sum_{j:x_j\in[x_*,x_i)} \varphi_{ij}\quad \forall\; i\in\dbra 1,N\dket.
\eeq
We denote by
\beq\label{reg}
	\reg_{\bs \chi}(\bs x) = \Big\{x_*: |x_i-x_*|> -\sum_{j:x_j\in[x_*,x_i)} \varphi_{ij} \quad \forall\; i\in\dbra 1,N\dket\Big\}
\eeq
the set of all regular points of $\bs x$. It is clear that this has the form of a disjoint union of open intervals $I_p,\,p=1,2,\ldots p_{\rm max}$,
\beq
	\reg_{\bs\chi}(\bs x) = \cup_p I_p,
\eeq
with the properties that $x_i\not\in I_p$ for all $i,p$, and where for each nearest-neighbour coordinates $x_i<x_j,\,\not\exists k:x_k\in(x_i,x_j)$, there may be at most one $p$ such that $I_p\subseteq(x_i,x_j)$.

We say that $\bs x$ has {\em low density} if
\beq\label{lowdensity}
	|x_i-a|\geq -\sum_{j:x_j\in[a,x_i)} \varphi_{ij} \quad \forall\  i\in\dbra 1,N\dket,\ a\in\R.
\eeq
It is immediate from our definitions that for a low-density semiclassical Bethe vector $\bs x$ with respect to $\bs\chi,\,\bs y$, we have
\beq\label{reglowdensity}
	\reg_{\bs \chi}(\bs x) = \R\setminus \bs x\quad (\mbox{low density}).
\eeq

The next lemma establishes important results for locally contracted vectors that we will need for our main theorems, and make the connection with regular magnifying-glass vectors.
\begin{lemma}[Properties of locally contracted vectors]\label{lemnat} Let $\bs x$ be a semiclassical Bethe vector for $\bs\chi\in\Rpu^N,\,\bs y\in\R^N$. For every $x_*\in\R$, the locally contracted vector $\bs X$ at $x_*$ takes the form
\beq\label{Xiyi}
	X_i  = y_i + \frc12\sum_{j\neq i}\sgn(x_j-x_*)\varphi_{ij}.
\eeq
For every $x_*<x_-(\bs\chi,\bs y)$, resp.~$x_*> x_+(\bs\chi,\bs y)$, we have
\beq\label{valuesextremepositionsnat}
	X_i = X_i^-(\bs\chi,\bs y),
	\quad\mbox{resp.}\quad
	X_i = X_i^+(\bs\chi,\bs y)\quad\forall\,i\in\dbra 1,N\dket.
\eeq
If $x_*=x_-(\bs\chi,\bs y)$, resp.~$x_*= x_+(\bs\chi,\bs y)$, there exists $i$ such that $X_i=x_-(\bs\chi,\bs y)$, resp.~$X_i= x_+(\bs\chi,\bs y)$, and this also holds for $x_*$ in a neighbourhood of $x_-(\bs\chi,\bs y)$, resp.~$x_+(\bs\chi,\bs y)$.

Let $x_*^+\in\reg_{\bs\chi}(\bs x),\, x_*^+\geq x_*$ be the smallest regular point greater than or equal to $x_*$, and $x_*^-\in\reg_{\bs\chi}(\bs x),\, x_*^-\leq x_*$ be the largest regular point smaller than or equal to $x_*$. Then
\beq\label{maxminnatural}
	\max_{i:x_i\leq x_*} X_i < x_*^+, \quad
	\min_{i:x_i> x_*} X_i > x_*^-,\quad
	D_i >-\max(|x_*^+-x_*|, |x_*^--x_*|)\ \forall \;i\in\dbra 1,N\dket.
\eeq
If $x_*\in\reg_{\bs\chi}(\bs x)$, then
\beq\label{Xiregular}
	X_i\neq x_*,\quad \sgn(X_i-x_*) = \sgn(x_i-x_*),\quad
	D_i>0 \quad \forall\;i\in\dbra 1,N\dket
\eeq
and as a consequence $\bs X$ is a regular magnifying-glass vector, Definition \ref{defimagstr}, giving a factorisation of the contraction map into partial contraction maps:
\beq
	\mathcal C_{\bs\chi,x_*}(\h{\mathcal C}_{\bs\chi,x_*}(\bs x))
	=
	\mathcal C_{\bs\chi}(\bs x)\quad\forall\;x_*\in\reg_{\bs\chi}(\bs x).
\eeq
\end{lemma}
\proof We denote $s_+ = \{i:x_i\geq x_*\}$, $s_- = \{i:x_i< x_*\}$. With our convention \eqref{sgn}, the following general relation holds for every $\bs x\in\Rne^N,\,x_*\in\R$ and $i\in\dbra 1,N\dket$:
\beq\label{specialrelationmain}
	\sgn(x_i-x_*)\sum_{j:x_j\in[x_*,x_i)\,(i\in s_+)\atop j:x_j\in(x_i,x_*)\,(i\in s_-)}\varphi_{ij}
	=
	\frc12\sum_{j=1\atop j\neq i}^N\Big(
	\sgn(x_j-x_*)\varphi_{ij} - \sgn(x_j-x_i)\varphi_{ij}\Big).
\eeq
Then combining \eqref{Xnatural} and \eqref{bethemain} and using $\sgn(x_i-x_j)+\sgn(x_j-x_i)=0$ (as $\bs x\in\Rne^N$), we obtain \eqref{Xiyi}. The statement \eqref{valuesextremepositionsnat} follows using \eqref{sbeboundinv}.

Let $x_*$ be in a small enough neighbourhood of $x_-(\bs\chi,\bs y)$. By \eqref{sbeboundinv} there are two possibilities. One is that $x_j>x_-(\bs\chi,\bs y)\,\forall\,j$. Then by \eqref{Xiyi} we have $X_i=X_i^-(\bs\chi,\bs y)\,\forall\,i$ and therefore $x_-(\bs\chi,\bs y)=\min_j X_j$, hence there exists $i$ such that $X_i = x_-(\bs\chi,\bs y)$. The other is that there exists $i$ such that $x_i=x_-(\bs\chi,\bs y)$, and $x_j>x_-(\bs\chi,\bs y)\,\forall\,j\neq i$. Then
 by the fixed-point property \eqref{fixedneigh}, $X_i =(\h{\mathcal C}_{\bs\chi,x_*\in[x_i-\ep,x_i+\ep]}(\bs x))_i = x_i = x_-(\bs\chi,\bs y)$. A similar argument holds for $x_*=x_+(\bs\chi,\bs y)$. This shows the statement after \eqref{valuesextremepositionsnat}.

Let $x_*\in\reg_{\bs\chi}(\bs x)$. Then in \eqref{Xnatural} we have $\Big|\sum_{j:x_j\in[x_*,x_i)\,(i\in s_+)\atop j:x_j\in(x_i,x_*)\,(i\in s_-)}\varphi_{ij}\Big|< |x_i-x_*|$, and hence $X_i\neq x_*$ and $\sgn(X_i-x_*)=\sgn(x_i-x_*)$, which with \eqref{contracteddist} gives \eqref{Xiregular} and, with \eqref{Xiyi} that we have established, the fact that $\bs X$ is a regular magnifying-glass vector.

Let $x_*\in\R$ and $i\in s_+$. We have $x_i\geq x_*^-$, and further $\{j:x_j\in[x_*,x_i)\}\subseteq \{j:x_j\in [x_*^-,x_i)\}$. Therefore, recalling that $\varphi_{ij}\leq 0$,
\beq
	\sum_{j:x_j\in[x_*,x_i)}\varphi_{ij}\geq
	\sum_{j:x_j\in[x_*^-,x_i)}\varphi_{ij}
\eeq
so that by the regularity condition,
\beqa
	X_i-x_*^- &=& x_i - x_*^- +\sum_{j:x_j\in[x_*,x_i)}\varphi_{ij}\n
	&\geq & x_i-x_*^-+\sum_{j:x_j\in[x_*^-,x_i)}\varphi_{ij}\n
	&>&0. \quad \mbox{(regularity of $x_*^-$)}
\eeqa
This implies the second relation in \eqref{maxminnatural}. A similar argument holds for $i\in s_-$, in which case $x_i< x_*^+$, giving $X_i<x_*^+$ and the first relation in \eqref{maxminnatural}. The third relation of \eqref{maxminnatural} then follows from \eqref{contracteddist}. \eproof

\medskip


The next results give properties of semiclassical Bethe vectors and their contracted vectors. In particular, we show, as mentioned, that semiclassical Bethe vectors always exist, and also how explicit solutions can be written that satisfy a separation condition that generalises that of the hard rods (which is the case $\varphi_{ij}<0$ independent of $i,j$). We also show how, if this separation condition holds
along with a non-coincidence condition, the solution is unique.

We first need the following algorithms, whose entries are an integer number $N\in\N$ and two vectors, $\bs\chi \in \Rpu^N$ and $\bs X\in \R^N$ -- the information of $N$ is within the space in which $\bs\chi,\,\bs y$ lie.

\begin{algo}[${\tt Rightward}(\bs\chi \in \Rpu^N,\bs X\in\R^N)$]\label{rightward}
{\ }\\
Let $s=\dbra 1,N\dket$. As long as $s\neq\emptyset$:
\bi
\item[\tt1-] Set
\beq\label{inpminmain}
	l \in \argmin_{i\in s}
	\Big(X_i - \sum_{j\in \dbra 1,N\dket \setminus s} \varphi_{i j}\Big),\quad
	x_{l} = X_l- \sum_{j\in \dbra 1,N\dket\setminus s}\varphi_{lj}.
\eeq
\item[\tt2-] Replace $s$ by $s\setminus \{l\}$.
\item[\tt3-] Repeat.
\ei
Return the set of all vectors $\bs x\in\R^N$ corresponding to all possible combinations of choices in \eqref{inpminmain}.
\end{algo}

\begin{algo}[${\tt Leftward}(\bs\chi \in \Rpu^N,\bs X\in\R^N)$]
\label{leftward}
{\ }\\
Let $s=\dbra 1,N\dket$. As long as $s\neq\emptyset$:
\bi
\item[\tt1-] Set
\beq\label{inpmaxmain}
	l \in \argmax_{i\in s}
	\Big(X_i + \sum_{j\in \dbra 1,N\dket\setminus s} \varphi_{i j}\Big),\quad
	x_{l} = X_l+ \sum_{j\in \dbra 1,N\dket\setminus s}\varphi_{lj}.
\eeq
\item[\tt2-] Replace $s$ by $s\setminus \{l\}$.
\item[\tt3-] Repeat.
\ei
Return the set of all vectors $\bs x\in\R^N$ corresponding to all possible combinations of choices in \eqref{inpmaxmain}.
\end{algo}

\begin{theorem}[Existence, uniqueness, and explicit forms of semiclassical Bethe vectors]\label{theosbe}
Let $N\in\N,\,\bs\chi\in\Rpu^N,\,\bs y\in\R^N$. Then there exists at least one solution $\bs x\in\Rne^N$ to the system of SBE \eqref{bethemain}.

One non-empty set of such solutions is $\bs x \in {\tt Rightward}(\bs\chi,\bs X^-(\bs\chi,\bs y))$. For this set,
\beq\label{miniexist}
	\min_{i\in\dbra 1,N\dket} x_i = x_-(\bs\chi,\bs y),
\eeq
and
\beq\label{distexist}
	x_i-a\geq -\sum_{j:x_j\in[a,x_i)} \varphi_{ij} \quad \forall\  a
	\leq x_i,\ i\in\dbra 1,N\dket
\eeq
and as a consequence, $D_i\geq 0$ for all $x_i\geq x_*$ (Definition \ref{definatmag}). Any solution $\bs x$ that satisfy \eqref{distexist} lies in ${\tt Rightward}(\bs\chi,\bs X^-(\bs\chi,\bs y))$, and if further the non-coincidence condition holds,
\beq\label{noncoincidence}
	y_i + \frc12\sum_{j\neq i}\sigma_j\varphi_{ij}
	\neq
	y_{i'} + \frc12\sum_{j\neq i'}\sigma_j\varphi_{i'j} \quad
	\forall \,i\neq i'\in\dbra 1,N\dket,\ \bs\sigma\in\{+,-\}^{\dbra 1,N\dket},
\eeq
then $\bs x$ is unique.
%

A possibly different, non-empty set of solutions is $\bs x \in {\tt Leftward}(\bs\chi,\bs X^+(\bs\chi,\bs y))$.  For this set,
\beq
	\max_{i\in\dbra 1,N\dket} x_i = x_+(\bs\chi,\bs y),
\eeq
and
\beq\label{distexist2}
	x_i-a\leq \sum_{j:x_j\in[a,x_i)} \varphi_{ij} \quad \forall\  x_i\leq a,\ i\in\dbra 1,N\dket.
\eeq
As a consequence, $D_i\geq 0$ for all $x_i\leq x_*$ (Definition \ref{definatmag}). Any solution $\bs x$ that satisfy \eqref{distexist2} lies in ${\tt Leftward}(\bs\chi,\bs X^+(\bs\chi,\bs y))$, and if further \eqref{noncoincidence} holds, then $\bs x$ is unique.

\end{theorem}

In order to show the above theorem and to understand better the SBE, we establish a relation between regular magnifying-glass vectors and semiclassical Bethe vectors. This lemma shows, in quite general terms, how the regularity \eqref{reg}, which imposes conditions on positions $x_i$ to be far enough from an observation point $x_*$, connects with the existence of a regular magnifying-glass vector at $x_*$, where the only condition is that $X_i$ be different from $x_*$. 
\begin{lemma}[Relation between regular magnifying-glass and semiclassical Bethe vectors]\label{lemrelationbethe}
Let $N\in\N,\bs\chi\in\Rpu^N,\,\bs y\in\R^N,\,x_*\in\R$ and $\varep>0$. The following statements are equivalent:
\bi
\item[(1)] There is a solution $\bs X\in\R^N$ to \eqref{eqxymainXstr} such that
\beq\label{xixstarpoint1}
	|X_i-x_*|\geq \varep\quad\forall\;i.
\eeq
\item[(2)] There is a solution $\bs x\in\Rne^N$  to \eqref{bethemain} such that
\beq\label{xixstarpoint2}
	|x_i-x_*|\geq -\sum_{j:x_j\in[x_*,x_i)} \varphi_{ij} + \varep\quad \forall\; i\in\dbra 1,N\dket.
\eeq
\ei
The solutions may be obtained from one another, as follows.
\bi
\item$(1)\to(2)$ We define a map $\hat{\mathcal E}_{\bs\chi,x_*}:\{\bs X\in\R^N: X_i\neq x_*\,\forall \,i\} \to\Rne^N$ by $\bs x = \hat{\mathcal E}_{\bs\chi,x_*}(\bs X)$ given as a choice of (recall Eq.~\eqref{setext})
\beq\label{relationconstruction}
	(x_i)_{i\in s_+} \in {\tt Rightward}(\bs\chi_{s_+},\bs X_{s_+}),\quad
	(x_i)_{i\in s_-} \in {\tt Leftward}(\bs\chi_{s_-},\bs X_{s_-})
\eeq
where $s_\pm = \{i\in\dbra 1,N\dket:X_i\gtrless x_*\}$. If $\bs X$ satisfies Point (1), then $\bs x = \hat{\mathcal E}_{\bs\chi,x_*}(\bs X)$ satisfies Point (2). With this construction,
\beq\label{Xabound}
	|x_i-a|\geq -\sum_{j:x_j\in[a,x_i)} \varphi_{ij} \quad \forall\  a\in[x_*,x_i),\,i\in\dbra 1,N\dket
\eeq
and
\beq\label{sgnmainlemma}
	\sgn(x_i-x_*) = \sgn(X_i-x_*).
\eeq
\item $(2)\to(1)$. We consider the restriction of the map $\hat{\mathcal C}_{\bs\chi,x_*}$, Eq.~\eqref{Xnatural}, to $\{\bs x\in\Rne^N,\,x_i\neq x_*\,\forall \,i\} \to\R^N$, that is $\bs X = \hat{\mathcal C}_{\bs \chi,x_*}(\bs x)$ given as
\beq\label{Xxrelation}
	X_i = x_i + \sgn(x_i-x_*)\sum_{j:x_j\in[x_*,x_i)}\varphi_{ij}\quad
	\forall\ i\in\dbra 1,N\dket.
\eeq
If $\bs x$ satisfies Point (2), then $\bs X = \hat{\mathcal C}_{\bs\chi,x_*}(\bs x)$ satisfies Point (1), and \eqref{sgnmainlemma} holds. Further, if $\bs X$ satisfies Point (1), then $\hat{\mathcal C}_{\bs\chi,x_*}(\hat{\mathcal E}_{\bs\chi,x_*}(\bs X)) = \bs X$.
\ei
\end{lemma}

Theorem \ref{theosbe} and Lemma \ref{lemrelationbethe} are proven in Section \ref{ssectbetheproof}.

It is interesting to observe that \eqref{Xabound} implies the following: for every $i,j$ with $x_*<x_j<x_i$, and for every $\{l_0,l_1,\ldots,l_m\}\subset\N$ with $x_j=x_{l_0}<x_{l_1}<\ldots<x_{l_m}<x_{l_{m+1}} = x_i$, we have
\beq\label{difijcons}
	|x_i-x_j|\geq -\sum_{n=0}^m
	\sum_{x_k\in[x_{l_n},x_{l_{n+1}})}\varphi_{l_{n+1}k}.
\eeq
The same formula holds for $x_i<x_j<x_*$ and $x_i=x_{l_{m+1}}<x_{l_m}<\ldots<x_{l_1}<x_{l_{0}} = x_j$. That is, we can partition the set of positions between $x_i$ and $x_j$ and bound the distance by sums over scattering shifts associated to the various parts of the partition.

Note how if both \eqref{distexist} and \eqref{distexist2} hold, then the solution lies both in ${\tt Rightward}$ and ${\tt Leftward}$. This is the case which we referred to above as low-density, Eq.~\eqref{lowdensity}.

Note also that at regular points, Lemma \ref{lemrelationbethe} can be applied for any
$0\leq \varep \leq \min_i \big(|x_i-x_*|+\sum_{j:x_j\in[x_*,x_i)} \varphi_{ij}\big)$, and it
implies that the locally contracted vectors \eqref{Xnatural} 
are indeed regular magnifying-glass vectors: they satisfy \eqref{eqxymainXstr} (as we have shown independently in Lemma \ref{lemnat}). As mentioned, Lemma \ref{lemrelationbethe} gives a more general relation between regular magnifying-glass vectors and semiclassical Bethe vectors.

We now show that all points outside the core are regular, and that points outside $[\min_i x_i,\max_i x_i]$ but within the core cannot be regular.
\begin{lemma}[Regularity outside the core]\label{extregular} Let $\bs x$ be a semiclassical Bethe vector for $\bs\chi\in\Rpu^N,\,\bs y\in\R^N$. Then
\beq
	\reg_{\bs\chi}(\bs x)\setminus [\min_i x_i,\max_i x_i]
	=
	\R\setminus\core(\bs\chi,\bs y).
\eeq
\end{lemma}
\proof We show that
\beq\label{twostat}
	\R\setminus \core(\bs\chi,\bs y)
	\subseteq \reg_{\bs\chi}(\bs x),\quad
	\reg_{\bs\chi}(\bs x)\setminus [\min_i x_i,\max_i x_i]
	\subseteq
	\R\setminus\core(\bs\chi,\bs y).
\eeq
Because of \eqref{sbeboundinv}, the first statement gives  $\R\setminus \core(\bs\chi,\bs y)\subseteq \reg_{\bs\chi}(\bs x)\setminus [\min_i x_i,\max_i x_i]$, and therefore this implies the statement of the present lemma.

Consider the first statement in \eqref{twostat}. Let $x_*<x^-(\bs\chi,\bs y)$, Eq.~\eqref{condxstar}. Let $i_1,\ldots,i_N\in\N$ be the unique indices such that $x_{i_1}<x_{i_2}<\ldots<x_{i_N}$. Let $n\in\dbra1,N\dket$. Then,
\beq
	x_{i_n} = y_{i_n} + \frc12 \sum_{j\neq i_n} \sigma_j\varphi_{i_nj},\quad
	\sigma_j=\lt\{\ba{ll}
	-1 & (j\in\{i_1,\ldots,i_{n-1}\})\\
	1 & (\mbox{otherwise}).
	\ea\rt.
\eeq
Therefore,
\beq
	x_{i_n}-x_* > x_{i_n}-x^-(\bs\chi,\bs y)=
	x_{i_n}-\min_i \Big(y_i + \frc12\sum_{j\neq i}\varphi_{ij}\Big)
	\geq
	x_{i_n}-\Big(y_{i_n} + \frc12\sum_{j\neq i_n}\varphi_{i_nj}\Big)
	=
	-\sum_{k=1}^{n-1}\varphi_{i_ni_k}.
\eeq
Hence the condition on the right-hand side of \eqref{reg} holds. As this is true for all $n$, we have $x_*\in\reg_{\bs\chi}(\bs x)$. A similar analysis holds for $x_*>x^+(\bs\chi,\bs y)$.

Consider the second statement in \eqref{twostat}. Let $x_*\in\reg_{\bs\chi}(\bs x)$ with $x_*<\min_i x_i$. Then \eqref{xixstarpoint2} holds for some $\varep>0$ small enough, and by Lemma \ref{lemrelationbethe}, there is a regular magnifying-glass vector $\bs X$ such that $X_i>x_*$ for all $i$. By Lemma \ref{lemdstr}, it must be that $x_*<x^-(\bs\chi,\bs y)$. A similar analysis applies for $x_*>\max_i x_i$.
\eproof

\medskip

Finally, it is clear that the locally contracted vectors are monotonic in $x_*$, an immediate consequence of Eq.~\eqref{Xnatural}. That is, $X_i$ is either stationary, or ``carried'' along with $x_*$, as $x_*$ moves, in accordance with the fact that there is a contraction towards $x_*$. We state this here for completeness although we will not make explicit use of it.
\begin{corol}[Monotonicity of locally contracted vectors] \label{corolmono} Let $\bs x$ be a semiclassical Bethe vector, and for every $x_*\in\R$ let $\bs X(x_*)$ be the locally contracted vector, Definition \ref{definatmag}. Then
\beq
	X_i(x_*')-X_i(x_*)\geq 0 \quad \forall\ x_*,x_*'\in \R,\,x_*'\geq x_*.
\eeq
\end{corol}

\subsection{Imprecision of solitons' positions}\label{ssectimp}

As mentioned, solutions to the SBE are not unique in general. Intuitively, this non-uniqueness corresponds to the imprecision in locating solitons in a dense gas -- there is no unique way of locating a soliton, because in high densities, interaction may effectively make it extended (uniqueness can be achieved by imposing certain conditions, Theorem \ref{theosbe}). For our purposes, it is convenient to simply choose one solution to the SBE, a choice of our quasiparticle positions. We will see in Section \ref{ssectprojections} how the ${\tt Rightward}$ algorithm is a good basis. 

Here, we explain how to observe the imprecision of locating solitons by defining various measures of imprecision directly from a choice of semiclassical Bethe vectors. These all have the property of vanishing at low densities, Eq.~\eqref{lowdensity}, and, crucially, they allow us to control the relation between the core, Eq.~\eqref{core}, which determines the support of a multi-soliton solution, and the maximal and minimal values of positions $x_i$ from the SBE.

In this subsection, we fix $N\in\N,\,\bs\chi\in\Rpu^N,\,\bs y\in\R^N$ and a semiclassical Bethe vector $\bs x\in\Rne^N$ with respect to these, Definition \ref{defibethe}. For every $x_*\in\R$, we denote $X_i(x_*)$ and $D_i(x_*)$ the associated locally contracted vector and contracted distances respectively, Definition \ref{definatmag}. We omit the explicit dependence on $\bs\chi,\,\bs x$ of the objects we define.

We first define the {\em irregularity length} as the maximal distance between nearest neighbours in the set of regular points (see Eqs.~\eqref{reg} and \eqref{neigh}),
\beq\label{convirreg}
	\delta
	:=
	\neigh(\reg_{\bs\chi}(\bs x)).
\eeq
The smaller the irregularity length, the more regular the vector is. This is our first notion of a global measure of the ``imprecision'' of solitons' positions.  It is immediate from our definitions that for a low-density $\bs x$, Eq.~\eqref{lowdensity}, we have
\beq
	\delta = 0\quad (\mbox{low density}).
\eeq

Second, we define the {\em dark displacement}, as the maximal amount by which the contracted distances $D_i$ may get negative:
\beq\label{varepdef}
	\varep:= \sup_{i\in\dbra 1,N\dket,\,x_*\in\R}
	D_i(x_*)\Theta(-D_i(x_*))
\eeq
where $\Theta$ is Heaviside's steps function. This indicates by how much the contraction goes beyond $x_*$, in the dark space we discussed in Section \ref{ssectsbe} -- how large the dark space is. It stays within the cells delimited by regular points, by Lemma \ref{lemnat}, and is a slightly more accurate measure of imprecision than $\delta$. At lower densities, it is a smaller number, or even 0. In particular, because of Eq.~\eqref{maxminnatural} from Lemma \ref{lemnat}, and by \eqref{reglowdensity}, at low densities we have
\beq
	\varep = 0\quad (\mbox{low density}).
\eeq

Our final notion of imprecision is in fact a family, parametrised by a local scale $\Delta X\geq 0$ in the contracted metric. We would like the true, effective position of $X_i\pm\Delta X$, from which we will obtain an imprecision $\Delta x\geq 0$. Inspired by the fixed-point property \eqref{fixedpoint}, a natural way is to define them by scanning the values of $x_*$ from $x_*\ll x_-(\bs\chi,\bs y)$, within $\R$ or a subset of $\R$, and looking for the first instance when the locally contracted position $X_i$ at $x_*$ crosses $x_*+\Delta X$, and the last instance when it crosses $x_*-\Delta X$. These give us an interval in real space, the {\em effective neighbourhood} $K_i(\Delta X)$ of $x_i$. We define the {\em effective imprecision} $\Delta x(\Delta X)$ as the length of the longest interval. This is expressed in the following definition:
\begin{defi}[Effective neighbourhoods and imprecision for local scale $\Delta X\geq 0$]\label{defineigh}
Let $\bs x$ be a semiclassical Bethe vector for $\bs\chi\in\Rpu^N,\,\bs y\in\R^N$, and for every $x_*\in\R$ let $\bs X(x_*)$ be the locally contracted vector, Definition \ref{definatmag}. Define, for all $i\in\dbra 1,N\dket$,
\beq\label{xilr}\begin{aligned}
	x_i^{\rm left}&:= \sup\{x: X_i(x_*)\geq x_*+\Delta X \;\forall\; x_*< x\}\\
	x_i^{\rm right}&:= \inf\{x: X_i(x_*)\leq x_*-\Delta X \;\forall\; x_*> x\}.
	\end{aligned}
\eeq
The solitons' effective neighbourhoods are
\beq
	K_i(\Delta X)
	:=
	[x_i^{\rm left},x_i^{\rm right}],
\eeq
and the (global) effective imprecision is
\beq
	\Delta x(\Delta X) := \max_{i\in\dbra 1,N\dket} \big(x_i^{\rm right}-x_i^{\rm left}\big).
\eeq
\end{defi}

It is clear that $x_i^{\rm left}$ and $x_i^{\rm right}$ exist in $\R$ by Lemma \ref{lemnat} (in particular Eq.~\eqref{valuesextremepositionsnat}).
It is also clear that we must have $x_i^{\rm left} \leq x_i^{\rm right}$. Note how $\Delta X$ is a {\em radius} around $x_*$ in the contracted space, while $\Delta x$ is the corresponding {\em diameter} of the cell in real space. It is convenient to choose the factors of 2 as such.

The quantities defined in Definition \ref{defineigh} have some simple properties and relations, as follows.
\begin{lemma}[Basic relations on effective neighbourhoods and imprecisions]\label{lembasicimp} For all $i\in\dbra 1,N\dket$,
\beq\label{inclneigh}
	K_i(\Delta X)\subseteq
	K_i(\Delta X')
	\quad \forall\;\Delta X'\geq \Delta X.
\eeq
In particular
\beq\label{inclimp}
	\Delta x(\Delta X)\leq
	\Delta x(\Delta X')
	\quad \forall\;\Delta X'\geq \Delta X.
\eeq
Further,
\beq\label{xiinK}
	x_i\in K_i(0);
\eeq
if $\bs x\in{\tt Rightward}(\bs\chi,\bs X^-(\bs\chi,\bs y))$, resp.~$\bs x\in{\tt Leftward}(\bs\chi,\bs X^+(\bs\chi,\bs y))$, then
\beq\label{leftrightpoints}
	K_i(0) = [x_i,x_i^{\rm right}],\quad \mbox{resp.}\quad
	K_i(0) \subset [x_i^{\rm left},x_i];
\eeq
and at low density, Eq.~\eqref{lowdensity},
\beq\label{Kilowdensity}
	K_i(0) = \{x_i\},\quad \Delta x(0)=0\quad \mbox{(low density)}.
\eeq
Finally,
\beq\label{minDeltax}
	\Delta x(\Delta X)\geq 2\Delta X
	\quad \forall\, \Delta X\geq 0.
\eeq
\end{lemma}
\proof Eqs.~\eqref{inclneigh}, \eqref{inclimp} are immediate from Definition \ref{defineigh}. Eqs.~\eqref{xiinK} follow from \eqref{fixedneigh}.
By Theorem \ref{theosbe} if $\bs x\in{\tt Rightward}(\bs\chi,\bs X^-(\bs\chi,\bs y))$ we have $X_i(x_*)\geq x_*$ for all $x_*\leq x_i$, similarly for the ${\tt Leftward}$ algorithm, hence using \eqref{xiinK}, Eq.~\eqref{leftrightpoints} follows. Eq.~\eqref{Kilowdensity} follows from
\eqref{reglowdensity} and \eqref{Xiregular}.
For all $x_*\leq x_i^{\rm left}$ we have $x_*\leq x_i$ and therefore $X_i(x_*)-x_* = D_i \leq x_i-x_*$ by \eqref{contractionprop}, and also $X_i(x_*)-x_*\geq \Delta X$, therefore $x_i-x_* \geq \Delta X$. Hence $x_i-x_i^{\rm left}\geq \Delta X$. A similar argument gives $x_i^{\rm right}-x_i \geq \Delta X$, and this implies \eqref{minDeltax}. \eproof

\medskip

The interval $K_i(0)$ is a natural notion of where the soliton $i$ is supported, and $\Delta x(0)$ a natural maximal soliton length. For $\Delta X$ large, then $K_i(\Delta X)$ instead has the interpretation of the overall range in real space corresponding to $[x_*-\Delta X,\,x_*+\Delta X]$ in contracted space. This encodes a change of metric on large scales, at the position $x_i$ where the soliton lies. Therefore $\Delta x(\Delta X)$ gives an overall maximal change of length associated to this change of metric.

\begin{lemma}[Relations between position shifts]\label{lemshift}
In the context of Definition \ref{defineigh}, denote $\Delta x = \Delta x(\Delta X)$.
Then
\beq\label{localitycondition}
\begin{aligned}
	&\forall\ i\in\dbra1,N\dket,\,x_*\in\R,\mbox{ the following are true}:\\
	&
	x_i< x_*-\Delta x\Rightarrow
	X_i(x_*)\leq x_*-\Delta X,\quad
	x_i> x_*+\Delta x\Rightarrow
	X_i(x_*)\geq x_*+\Delta X.
	\end{aligned}
\eeq
\end{lemma}
\proof Choose $i$ and $x_*$ such that $x_i<x_*-\Delta x$. Then, using \eqref{inclneigh} and \eqref{xiinK},
\beq
	x_* > x_i + \Delta x \geq x_i^{\rm left} + \Delta x\geq  x_i^{\rm right}.
\eeq
As this is a strict inequality, the condition of the infimum defining $x_i^{\rm right}$ in \eqref{xilr} is satisfied, and therefore $X_i(x_*)\leq x_*-\Delta X$. This shows the first implication.
A similar arguments holds for the other implication.
\eproof

\medskip

Finally, we note that we have three natural measures of imprecision:
\beq
	\delta,\ \varep,\ \Delta x(0).
\eeq
It is convenient to define the {\em outer imprecision} as
\beq\label{minimp}
	\Delta^{\rm out} = \min\{\delta,\varep,\Delta x(0)\}.
\eeq
These various measures of imprecision bound the differences between $\min_i x_i$, $\max_i x_i$ that the boundaries of the core, and bound each other, as follows.
\begin{lemma}[Bounds on imprecisions]\label{lemimp} Recall \eqref{sbeboundinv} and define $\Delta^\pm$ as
\beq\label{sbebound}
	\core(\bs\chi,\bs y) = [\min_i x_i - \Delta^-, \max_i x_i + \Delta^+].
\eeq
Then
\beq
	\Delta^- = \inf_{x_*\in\reg_{\bs\chi}(\bs x):\atop x_*<\min_i x_i}(\min_i x_i - x_*)\geq 0,\quad
	\Delta^+ = \inf_{x_*\in\reg_{\bs\chi}(\bs x):\atop x_*>\max_i x_i}(x_*-\max_i x_i)\geq 0,
	\label{boundvar}
\eeq
and
\beq\label{bounds}
	\Delta^\pm\leq\varep\leq \delta,\quad
	\Delta^\pm\leq \Delta x(0).
\eeq
Further,
\beq
	\Delta^\pm\leq \Delta^{\rm out}.
\eeq
\end{lemma}
\proof Eq.~\eqref{boundvar} follows from By Lemma \ref{extregular}. $\varep\leq \delta$ follows from Lemma \ref{lemnat}, in particular the third relation in Eq.~\eqref{maxminnatural}.
Note that for all $x_*<\min_i x_i$, we have from \eqref{Xnatural}, for every $i$,
\beq
	X_i = x_i + \sum_{j:x_j\in[\min_k x_k,x_i)}\varphi_{ij}.
\eeq
That is, $X_i$ does not depend on $x_*$. Now take $x_*=\min_k x_k-\ep$ for $\ep>0$, noting that all contracted positions are contracted towards the left. Taking $\ep$ as small as desired, by definition of $\varep$ we obtain
\beq
	X_i \geq \min_k x_k-\varep\quad \forall\;i.
\eeq
Taking $\ep$ large enough so that $x_*<\min_k x_k-\varep$, we have $X_i = X_i^-(\bs\chi,\bs y)$. Hence $x_-(\bs\chi,\bs y)\geq \min_k x_k - \varep$, which implies $\Delta^-\leq \varep$. A similar argument gives $\Delta^+\leq \varep$.

Choose $\Delta X =0$ and $x_*= x_-(\bs\chi,\,\bs y)+\ep$ for $\ep>0$ small enough. Then by Lemma \ref{lemnat}, there exists $i$ such that $X_i(x_*)= x_-(\bs\chi,\,\bs y)$ hence $X_i(x_*)<  x_*+\Delta X$, so that by the opposite of the second implication in Lemma \ref{lemshift}, we have $x_i\leq x_*+\Delta x(0)$. Therefore, $\min_i(x_i)\leq x_*+\Delta x(0) = x_-(\bs\chi,\,\bs y)+\Delta x(0)+\ep$, giving $\Delta^-\leq \Delta x(0)+\ep$. As this holds for all $\ep>0$, we conclude that $\Delta^-\leq \Delta x(0)$. A similar argument gives $\Delta^+\leq \Delta x(0)$.
\eproof

\medskip

\subsection{Quasiparticles and fluid-cell projections}\label{ssectprojections}

It is natural to consider a reduction of the $N$-soliton field where we keep only a subset of solitons. It is also natural to take this subset to be that of solitons whose positions, forming a semiclassical Bethe vector, lie within a certain interval. These are {\em projections}, and our main theorems will state that such  projections do not change too much the KdV field on regions where solitons are kept.

But how do we take away solitons in a way that does not affect the field? Just erasing columns and rows in $\Psi(x)^2+\omega$, Eqs.~\eqref{utau}, \eqref{tau}, is in general incorrect. We start with basic lemmas that explain how to do this.


Let $N\in\N$. For any ordered pair of disjoint subsets $s_+,\,s_-\subseteq\dbra 1,N\dket$, $s_+\cap s_-=\emptyset$, we denote
\beq\label{limz}
	\lim_{\bs z:s_+,s_-} \quad\mbox{the limit}\quad
	\lt\{
	\ba{ll}
	z_i\to\infty& (i\in s_+)\\
	z_i\to-\infty& (i\in s_-)\\
	z_i=0 & (\mbox{otherwise})
	\ea\rt\}\quad
	\mbox{in any order.}
\eeq
Let
\beq\label{sspsm1}
	s = \dbra 1,N\dket\setminus (s_+\cup s_-).
\eeq
Given $\bs \chi\in\Rpu^N,\,\bs y\in\R^N$, we define for every such ordered pair $s_+,\,s_-$ the vector $\bs y^{(s_+,s_-)}\in \R^{|s|}$ by
\beq\label{yispsm}
	y_i^{(s_+,s_-)} = y_i + \frc12\sum_{j\in s_+} \varphi_{ij}
	-
	\frc12\sum_{j\in s_-} \varphi_{ij},
	\quad
	i \in s
\eeq
(we keep the induced indexing, $i\in s$, to describe the elements of the vector $\bs y\in\R^{|s|}$). Note that $\bs y^{(s_+,s_-)}$ depends implicitly on $\bs\chi,\,\bs y$, being a $\bs\chi$-dependent shift of $\bs y$. This satisfies a natural consistency condition: for any other ordered pair $s_+',s_-'\subseteq s,\,s_+'\cap s_-' = \emptyset$, we have
\beqa
	\big(\bs y^{(s_+,s_-)}\big)^{(s_+',s_-')}
	&=&
	\Big(y_i + \frc12\sum_{j\in s_+} \varphi_{ij}
	-
	\frc12\sum_{j\in s_-} \varphi_{ij}\Big)_i^{(s_+',s_-')}\n
	&=&
	\Big(y_i + \frc12\sum_{j\in s_+\cup s_+'} \varphi_{ij}
	-
	\frc12\sum_{j\in s_-\cup s_-'} \varphi_{ij}\Big)_i\n
	&=&\bs y^{(s_+\cup s_+',s_-\cup s_-')}.
	\label{consistency}
\eeqa
Recall from Eq.~\eqref{setext} that
\beq
	\bs \chi_s = (\chi_i)_{i\in s}.
\eeq
\begin{lemma}[Impact parameters in the limit] \label{lemproj}
Let $N\in\N,\,\bs \chi\in\Rpu^N,\,\bs y\in\R^N$, let $s_+,\,s_-\subseteq\dbra 1,N\dket$ be an ordered pair of disjoint subsets, and $s$ be given by \eqref{sspsm}. Then for every $x\in\R$ and every integer $n\geq0$,
\beq\label{limextract}
	\lim_{\bs z:s_+,s_-}
	\p_x^n u_{\bs \chi,\bs y+\bs z}(x)
	=
	\p_x^n u_{\bs \chi_s,\bs y^{(s_+,s_-)}}(x)
\eeq
where the order of limits is not important.
\end{lemma}
\proof This is immediate from the in-out representation, Lemma \ref{lemtaudecomp}. Consider the representation $\tau_{\bs \chi,\bs y}^{s_+}$. Then in \eqref{tauresult} the limit $\lim_{\bs z:s_+,s_-}$ can be done on every term, and only those with $r\subseteq s$ remain. Writing (see \eqref{xist})
\beq
	x_i^{s_+} = y_i + \frc12\sum_{j=1\atop j\neq i}^N \sgn_{s_+}(j)\varphi_{ij}
	=
	y_i^{(s_+,s_-)}
	- \frc12\sum_{j\in s\atop j\neq i} \varphi_{ij}
\eeq
we obtain $\tau_{\bs \chi_s,\bs y^{(s_+,s_-)}}^{\emptyset}$. The same can be done for every derivative $\p_x^n \tau_{\bs \chi,\bs y}^{s_+}$, as $\p_x$ merely brings a factor $\sum_{i\in r}2\chi_i \sgn_{s_+}(i)$ in every term.
\eproof

\medskip


In addition to describing the result of the limit above, the modified impact parameters $\bs y^{(s_+,s_-)}$ also correspond to {\em the extraction of solitons in $s$}. Here the picture is that the solitons we ``extract'' are those we keep and concentrate on -- as this is what we want to do with the fluid-cell projection.

Let $x_*\in\R$. We say that $s_+,\,s_-\subseteq\dbra 1,N\dket$ are {\em separated by $x_*$} (with respect to $\bs x\in\Rne^N$) if there exists
\beq\label{intervalreal}
	I=[I_-,I_+],\quad I_-\leq I_+\in\R,
\eeq
with $x_*\in I$ such that
\beq
	s_+ = \{i:x_i> I_+\},\quad
	s_- = \{i:x_i< I_-\}.
\eeq


%
\begin{lemma}[How to extract solitons]\label{lemextractsbe} Let $N\in\N$, and let $\bs x\in\Rne^N$ be a semiclassical Bethe vetor for $\bs \chi\in\Rpu^N,\,\bs y\in\R^N$, Definition \ref{defibethe}. Let $x_*\in\R$, and $s_+, \,s_-\subseteq\dbra 1,N\dket$ be separated by $x_*$. With $s$ given by \eqref{sspsm1}, then $\bs x_s$ is a semiclassical Bethe vector for $\bs\chi_s,\,\bs y^{(s_+,s_-)}$.
\end{lemma}
\proof Because $x_j<x_i$ for all $j\in s_-,\,i\in s$, then in \eqref{bethemain} with such $i$, we can bring the terms with such $j$ to the left-hand side. This gives the second sum in \eqref{yispsm}. The first sum is obtained similarly.
\eproof
\medskip

Another way of stating the above lemma is as follows. For $\bs\chi\in\Rpu^N$, let $\mathcal E_{\bs\chi}$ be a right-inverse of $\mathcal C_{\bs\chi}$,
\beq\label{mapEsbe}
	\mathcal C_{\bs \chi}\circ \mathcal E_{\bs \chi} = {\rm id}.
\eeq
This is an {\em expansion map}. By Theorem \ref{theosbe}, it exists. For $s\subset\dbra 1,N\dket$,  define
\beq
	\bs y^{(s)} := \mathcal C_{\bs \chi_s}(\mathcal E_{\bs\chi}(\bs y)_s).
\eeq
Then the lemma says that if $s_+,\,s_-$ are separated by $x_*$, we have
\beq\label{ysyspsm}
	\bs y^{(s)} = \bs y^{(s_+,s_-)}.
\eeq
The map $(\bs\chi,\bs y)\mapsto (\bs\chi_s,\bs y^{(s)})$ simply executes the following operations: it expands the impact parameters $\bs y$ to obtain positions of solitons $\bs x = \mathcal E_{\bs\chi}(\bs y)$; projects out solitons that are not in the set $s$, that is taking $\bs \chi_s = (\chi_i)_{i\in s}$ and $\bs x_s = (x_i)_{i\in s}$; and contracts the remaining  positions to get the impact parameters $\bs y^{(s)}$ that correspond to them. The spectral parameters $\bs\chi_s$ of the restriction are a subvector of the original vector of $N$ spectral parameters $\bs\chi$, but the {\em impact parameters $\bs y^{(s)}$ are different from the subvector $\bs y_s$}. This is what keep the remaining positions unchanged -- this ``extracts'' the set $s$ of solitons from the solitons $(\bs \chi,\bs y)$. Below,  we will also use a notation that emphasises the interval $I$, Eq.~\eqref{intervalreal}:
\beq\label{datafcgen}
	\bs\chi^{(I)} := \bs\chi_s,\quad
	\bs y^{(I)}:=\bs y^{(s_+,s_-)} = \bs y^{(s)}.
\eeq

According to the proposed construction of Eqs.~\eqref{Omega}, \eqref{mapU} and \eqref{fcupsilon}, it is natural to suggest for the map $\mathcal U$
\beq
	\mathcal U(\bs\chi,\bs y) = (\bs\chi,\bs x)\ :\ 
	\bs x=\mathcal E_{\bs\chi}(\bs y),
\eeq
for some choice of right-inverse $\mathcal E_{\bs\chi}$, corresponding to a choice of solution to \eqref{bethemain}. But there are possibly many solutions to \eqref{bethemain}. The set of solutions we choose determines the image of $\mathcal U$,  the set of quasiparticle coordinates:
\beq
	\mathsf Q= \coprod_{N=0}^\infty
	\{(\bs\chi,\mathcal E_{\bs \chi}(\bs y):\bs \chi\in\Rpu^N,\,\bs y\in\R^N\}.
\eeq
Once such a choice is made, $\mathcal U$ is well defined, and its inverse, acting as
\beq
	\mathcal U^{-1}(\bs\chi,\bs x) = (\bs \chi,\mathcal C_{\bs\chi}(\bs x)),
\eeq
exists on the image $\mathsf Q$ of $\mathcal U$, hence $\mathcal U$ is a bijection $\mathsf S\to\mathsf Q$. For \eqref{fcupsilon} to work, we must be able to evaluate $\mathcal U^{-1}:\mathsf Q \to \mathsf S$ on $\mathcal P_I\mathsf Q$, where the projection $\mathcal P_I$ is defined in \eqref{fcupsilon} for every interval $I$, Eq.~\eqref{intervalreal}. That is, the projection must map $\mathsf Q$ into $\mathsf Q$:
\beq\label{projproperty}
	\mathcal P_I\mathsf Q\subseteq \mathsf Q\quad \forall\;I.
\eeq
This is a constraint on our choice of $\mathcal U$. With such a choice of $\mathcal U$, then the extraction of solitons is also a projection: with
\beq\label{sspsm}
	s = \{i:x_i\in I\},\quad s_-:=\{i:x_i<I_-\},\quad s_+:=\{i:x_i>I_+\},
\eeq
we define
\beqa
	\hat {\mathcal P}_I \quad : \qquad \mathsf S&\to& \mathsf S\n
	(\bs \chi,\bs y)&\mapsto& (\bs \chi_s,\bs y^{(s_+,s_-)})
	\label{hatprojection}
\eeqa
and this has the form
\beq\label{hatpU}
	\hat {\mathcal P}_I = \mathcal U^{-1}\mathcal P_I\mathcal U.
\eeq

Here, we choose Algorithm \ref{rightward}, as by Theorem \ref{theosbe} this provides a choice of right-inverse (a similar construction can be done with Algorithm \ref{leftward}). By \eqref{leftrightpoints}, it gives the positions of the left-most point $x_i$ of the soliton's support $K_i(0)$, hence it is a ``left anchor''. Two aspects are to be considered: (i) we must show the projection property \eqref{projproperty}, and (ii) we must deal with the non-coincidence condition \eqref{noncoincidence} in order to make the result of the algorithm unique. We now do this, with a slightly modified algorithm.

Recall the left-extremal positions $\bs X^-(\bs\chi,\bs y)$, Eq.~\eqref{condxstar}. The data of this algorithm is $N\in\N$ (implicit) and $\bs \chi\in\Rpu^N,\,\bs y\in\R^N$. 
\begin{algo}[${\tt Leftanchor}(\bs\chi \in \Rpu^N,\bs y\in\R^N)$]\label{leftanchor}
{\ }\\
Let $\bs X = \bs X^-(\bs\chi,\bs y)$ and $s=\dbra 1,N\dket$. As long as $s\neq\emptyset$:
\bi
\item[\tt1-] Let $z_i = X_i - \sum_{j\in \dbra 1,N\dket \setminus s} \varphi_{i j},\,i\in s$. If the set $\{z_i\}_{i\in s}$ has a unique minimum, then set $l = \argmin_{i\in s}(z_i)$, and otherwise, $l =\argmin_{j\in \argmin_{i\in s}(z_i)}(\chi_j)$. Then set $x_l = z_l$.
\item[\tt3-] Replace $s$ by $s\setminus \{l\}$.
\item[\tt4-] Repeat.
\ei
Return $\bs x\in\R^N$.
\end{algo}
In step 1, we have resolved the coincidences by using the spectral parameters, which are all different from each other. Crucially, we note that this is equivalent to
\beq\label{limitlalgo}
	l = \lim_{\ep\to 0^+}\argmin_{i\in s}(z_i+v_i\ep)
\eeq
where $\bs v$ is the vector of velocities, Eq.~\eqref{timeevol}. 

We now define our quasiparticle positions using this algorihm.

\begin{defi}[Quasiparticle positions] \label{defiqp} Let $N\in\N,\,\bs \chi\in\Rpu^N,\,\bs y\in\R^N$. The vector of quasiparticle positions $\bs x(\bs\chi,\bs y)$ is defined as 
\beq
	\bs x(\bs\chi,\bs y) = {\tt Leftanchor}(\bs\chi,\bs y).
\eeq
This defines the map $\mathcal U$ as $\mathcal U(\emptyset,\emptyset) = (\emptyset,\emptyset)$ and
\beq
	\mathcal U(\bs\chi,\bs y) = (\bs\chi,\bs x(\bs\chi,\bs y)).
\eeq
The space of quasiparticle data is
\beq
	\mathsf Q = \{(\emptyset,\emptyset)\}\cup \{(\bs\chi,\bs x(\bs\chi,\bs y)):N\in\N,\,\bs \chi\in\Rpu^N,\,\bs y\in\R^N\}.
\eeq
\end{defi}

\begin{theorem}[Quasiparticles, projection and SBE bijection] \label{theodesc} Take the context of Definition \ref{defiqp}. Then $\bs x(\bs\chi,\bs y)\in\Rne^N$ is the unique vector such that
\beq\label{defqp2}
	\bs x(\bs\chi,\bs y)+\ep\bs v \in
	{\tt Rightward}(\bs\chi,\bs X^-(\bs\chi,\bs y) + \ep\bs v)\quad
	\mbox{for all $\ep\geq 0$ small enough.}
\eeq
In particular, if the non-coincidence condition \eqref{noncoincidence} holds, then $\bs x(\bs\chi,\bs y)$ is the unique vector
\beq\label{defqp1}
	\bs x(\bs\chi,\bs y) \in {\tt Rightward}(\bs\chi,\bs X^-(\bs\chi,\bs y)).
\eeq
The quasiparticle positions $\bs x(\bs\chi,\bs y)$ form a semiclassical Bethe vector for $\bs \chi,\,\bs y$, Definition \ref{defibethe}. They satisfy the minimal-distance condition \eqref{distexist},
\beq\label{separation}
	x_i(\bs\chi,\bs y)-a\geq -\sum_{j:x_j(\bs\chi,\bs y)\in[a,x_i(\bs\chi,\bs y))} \varphi_{ij} \quad \forall\  a
	\leq x_i(\bs\chi,\bs y),\ i\in\dbra 1,N\dket.
\eeq
In particular, for all $i,j$ with $x_i\neq x_j$,
\beq\label{minseparation}
	|x_i(\bs\chi,\bs y)-x_j(\bs\chi,\bs y)|\geq \min\{|\varphi_{ij}|,|\varphi_{ji}|\}.
\eeq
Further, for every $I$ as in \eqref{intervalreal}  (see Eq.~\eqref{fcupsilon} for the definition of $\mathcal P_I$),
\beq\label{projectionQ}
	\mathcal P_I\mathsf Q\subseteq \mathsf Q.
\eeq
With \eqref{sspsm} we have
\beq\label{Pexplicit}
	\mathcal P_I(\bs\chi,\bs x(\bs\chi,\bs y)) =(\bs\chi_s,\bs x(\bs\chi_s,\bs y^{(s_+,s_-)}))
\eeq
so that, with \eqref{hatprojection},
\beq\label{projectionexchange}
	\mathcal P_I \mathcal U = \mathcal U \h{\mathcal P}_I.
\eeq
Finally (recall $\bs v$ from Eq.~\eqref{factscat}, and the $\varphi$-exclusion space $\R_\varphi^N$, Eq.~\eqref{exclusion}),
\beq\label{Qexact}
	\mathsf Q = \{(\emptyset,\emptyset)\}\cup\Big\{(\bs\chi,\bs x):
	N\in\N,\,\bs \chi\in\Rpu^N,\,\bs x\in\Rne^N,\, \exists\,\ep>0\,|\,\bs x + \bs v t \in\R_\varphi^N
	\ \forall\ t\in[0,\ep]\Big\}.
\eeq
In particular, for every $\bs\chi\in\Rpu^N$, the SBE \eqref{bethemain} gives a bijection
\beq
	\mathcal C_{\bs\chi} \quad :\quad
	\bs x\in \Rne^N: \exists\,\ep>0\,|\,\bs x + \bs v t \in\R_\varphi^N
	\ \forall\ t\in[0,\ep]\quad
	\mapsto
	\quad \bs y\in\R^N.
\eeq
\end{theorem}
\proof Let $g_i^{\bs\sigma}:= y_i + \frc12\sum_{j\neq i}\sigma_j\varphi_{ij},\,i\in\dbra 1,N\dket,\, \bs\sigma\in\{+,-\}^{\dbra 1,N\dket}$. We note that at every iteration in the algorithm ${\tt Leftanchor}(\bs\chi,\bs y)$, with the corresponding $s\subseteq\dbra 1,N\dket$ and $\{z_i\}_{i\in s}$, we have $z_i = g_i^{\bs \sigma}$ for $\bs\sigma =\sgn_s$. Also, clearly ${\tt Leftanchor}(\bs\chi,\bs y)\in{\tt Rightward}(\bs\chi,\bs X^-(\bs\chi,\bs y))$.

If condition \eqref{noncoincidence} holds, then ${\tt Rightward}(\bs\chi,\bs X^-(\bs\chi,\bs y))$ is a singleton which solves \eqref{bethemain}, by Theorem \ref{theosbe}, and further, by the above observation, $\{z_i\}_{i\in s}$ has a unique minimum, and we see that ${\tt Leftanchor}(\bs\chi,\bs y)$ is the unique element of ${\tt Rightward}(\bs\chi,\bs X^-(\bs\chi,\bs y))$. This shows \eqref{defqp1} and \eqref{separation} and the fact that $\bs x(\bs\chi,\bs y)$ is semiclassical Bethe vector for $\bs \chi,\,\bs y$, Definition \ref{defibethe}.

More generally, because $\chi_i\neq\chi_j\;\forall\;i\neq j$, then for every $\bs\sigma$ we have $g_i^{\bs\sigma}+v_i\ep\neq g_{i'}^{\bs\sigma}+v_{i'}\ep \;
\forall\,i\neq i'\in\dbra 1,N\dket$ whenever $\ep>0,\,\ep<\min_{i,j:g_i^{\bs\sigma}\neq g_j^{\bs\sigma}}((g_i^{\bs\sigma}-g_j^{\bs\sigma})/(v_i- v_j))$. Therefore ${\tt Rightward}(\bs\chi,\bs X^-(\bs\chi,\bs y)+ \ep\bs v)$ is a singleton, and by the linearity of \eqref{inpminmain}, for all such $\ep$ the unique element has the form $\bs x + \ep\bs v$ for some $\bs x\in\R^N$. It is clear that the limit taken in \eqref{limitlalgo} exactly gives this vector, hence the definition of $l$ in step 1 of ${\tt Leftanchor}(\bs\chi,\bs y)$ also gives this vector. This shows \eqref{defqp2}.

In the case where the non-coincidence condition is broken, we have only shown $\bs x(\bs\chi,\bs y)\in\R^N$. Consider this case.

By Theorem \ref{theosbe}, choosing any $a=x_j<x_i$ in the notation of \eqref{distexist}, we find that for all such $\ep$ as above, the vector $\bs x$ above satisfies $x_i-x_j\geq-\ep(v_i-v_j)-\varphi_{ij}\;\forall\,i,j:x_i-x_j\geq -\ep(v_i-v_j)$. Therefore there are $a,a'>0$ such that $x_i-x_j\geq-\varphi_{ij}-a\ep\,\forall\,i,j:x_i>x_j-a'\ep$, and taking $\ep>0$ as small as desired, we find $\bs x(\bs\chi,\bs y)\in\Rne^N$.

By Theorem \ref{theosbe}, because $\bs X^-(\bs\chi,\bs y)+ \ep\bs v = \bs X^-(\bs\chi,\bs y+ \ep\bs v)$, the vector $\bs x + \ep\bs  v$ solves \eqref{bethemain} with $\bs y$ replaced by $\bs y + \bs v \ep$ for all $\ep$ as above, and as $\bs x \in\Rne^N$ then for all $\ep>0$ small enough, we have $\sgn(x_i-x_j+\ep(v_i-v_j)) = \sgn(x_i-x_j)$, hence $\bs x$ solves \eqref{bethemain}. This shows that $\bs x(\bs\chi,\bs y)$ is semiclassical Bethe vector for $\bs \chi,\,\bs y$, Definition \ref{defibethe}.

Further, using again \eqref{distexist},  $x_i-a\geq-\ep v_i-\sum_{j:x_j+\ep v_j\in[a,x_i+\ep v_i)}\varphi_{ij}$ for all $a\leq x_i+\ep v_i$. As $\bs x+\ep\bs  v\in\Rne^N$ for all $\ep>0$ by the fact that $\bs x+\ep \bs v \in{\tt Rightward}(\bs\chi,\bs X^-(\bs\chi,\bs y) + \ep \bs v)$, and also for $\ep=0$ by the result established above, then for all $\ep$ as above, the condition $x_j+\ep v_j < x_i+\ep v_i$ is equivalent to $x_j<x_i$. Hence we have $x_i-a\geq-\ep v_i-\sum_{j:x_j\in[a-\ep v_j,x_i)}\varphi_{ij}$ for all $a\leq x_i+\ep v_i$. The limit $\ep\to0^+$ of this condition gives \eqref{separation}.

We now show \eqref{projectionQ} with \eqref{Pexplicit}.

First, we show that for any $s_+,s_-\in \dbra 1,N\dket,\,s = \dbra 1,N\dket\setminus s_+\cup s_-$ with $s_+\cap s_-=\emptyset$, if the non-coincidence condition \eqref{noncoincidence} holds for $\bs\chi,\,\bs y$, then it also does for $\bs\chi_s,\,\bs y^{(s_+,s_-)}$. This follows from the fact that for every $\bs\sigma\in\{+,-\}^s$ and $i\in s$:
\beq
	y_i^{(s_+,s_-)} + \frc12\sum_{j\in s\atop j\neq i}
	\sigma_j\varphi_{ij}
	=
	y_i + \frc12\sum_{j\in s_+} \varphi_{ij}
	-
	\frc12\sum_{j\in s_-} \varphi_{ij}
	+ \frc12\sum_{j\in s\atop j\neq i}
	\sigma_j\varphi_{ij}
	= g_i^{\bs\sigma},\quad
	\sigma_k = \pm\,(k\in s_\pm).
\eeq

Let $\bs x\in\mathsf Q$. Then $\bs x = \bs x(\bs\chi,\bs y)$ for some $N\in\N,\,\bs\chi\in\Rpu^N,\,\bs y\in\R^N$. Let $s = \{i:x_i\in I\},\,s_-:=\{i:x_i<I_-\},\, s_+:=\{i:x_i>I_+\}$. Then $s_+,\,s_-$ are separated by any point in $I$ and by Lemma \ref{lemextractsbe}, $\bs x_s\in\Rne^{|s|}$ solves \eqref{bethemain} for $\bs \chi_s,\,\bs y^{(s_+,s_-)}$. Also for all $i\in s,\,a\leq x_i$, by \eqref{separation}
\beq
	x_i - a \geq -\sum_{j\in \dbra 1,N\dket:x_j\in [a,x_i)}\varphi_{ij}
	\geq -\sum_{j\in s:x_j\in [a,x_i)}\varphi_{ij}.
\eeq
Hence by Theorem \ref{theosbe},
\beq
	\bs x_s\in {\tt Rightward}\Big(\bs\chi_s,\bs X^-\big(\bs\chi_s,\bs y^{(s_+,s_-)}\big)\Big).
\eeq
If the non-coincidence condition \eqref{noncoincidence} holds, by the result above it also holds for $\bs\chi_s,\bs y^{(s_+,s_-)}$, so there is a unique element in ${\tt Rightward}\Big(s,\bs X^-\big(\bs\chi_s,\bs y^{(s_+,s_-)}\big)\Big)$, and by \eqref{defqp1} established above, we have
\beq
	\mathcal P_I\bs x(\bs\chi,\bs y) = \bs x(\bs\chi,\bs y)_s = \bs x(\bs\chi_s,\bs y^{(s_+,s_-)})\in\mathsf Q.
\eeq
Otherwise, then let $s = \{i:x_i+\ep\chi_i\in I\},\,s_-:=\{i:x_i+\ep\chi_i<I_-\},\, s_+:=\{i:x_i+\ep\chi_i>I_+\}$. For all $\ep>0$ small enough, these do not depend on $\ep$. Then by the same argument as above, using the fact that
\beq
	\bs X^-(\bs\chi,\bs y)+\ep\bs \chi= \bs X^-(\bs\chi,\bs y+\ep\bs \chi),
\eeq
we have, for all $\ep>0$ small enough,
\beq
	\bs x(\bs\chi,\bs y)_s+\ep\bs\chi_s\in {\tt Rightward}\Big(\bs\chi_s,\bs X^-\big(\bs\chi_s,(\bs y+\ep\bs\chi)^{(s_+,s_-)}\big)\Big)
	=
	{\tt Rightward}\Big(\bs\chi_s,\bs X^-\big(\bs\chi_s,\bs y^{(s_+,s_-)}\big)+\ep\bs\chi_s\Big)
\eeq
and by \eqref{defqp2}, we again have
\beq
	\mathcal P_I\bs x(\bs\chi,\bs y) = \bs x(\bs\chi,\bs y)_s = \bs x(\bs\chi_s,\bs y^{(s_+,s_-)})\in\mathsf Q.
\eeq

Finally, we show \eqref{Qexact}. We note that if $\bs x+\bs v t\in\R_\varphi^N$ (Eq.~\eqref{exclusion}), then $\bs x+\bs v t$ satisfies \eqref{distexist}. Further, $\bs x+\bs v t$ is a semiclassical Bethe vector for $\bs \chi,\,\bs y(t)=\mathcal C_{\bs\chi}(\bs x + \bs v t)$, Definition \ref{defibethe}. Hence by Theorem \ref{theosbe}, $\bs x + \bs v t\in {\tt Rightward}(\bs\chi,\bs X^-(\bs\chi,\bs y(t)))$. By an argument similar to that made above if $\bs x+\bs v t\in\R_\varphi^N$ for all $v\in[0,\ep]$, then $\sgn(x_i+v_i t-(x_j + v_j t)) = \sgn(x_i-x_j)$ for all $v$ small enough. Hence $\bs y(t) = \bs y(0) + \bs v t$, and we have $\bs x + \bs v t\in {\tt Rightward}(\bs\chi,\bs X^-(\bs\chi,\bs y(0))+\bs v t)$ for all $t\geq 0$ small enough. As we have shown that this is unique and given by $\bs x = \bs x (\bs \chi,\bs y(0))$, we conclude that $(\bs \chi,\bs x)\in \mathsf Q$. This shows \eqref{Qexact}.
\eproof

\medskip

From this lemma, we observe that {\em time-evolved coordinates}
\beq
	\bs x^{(t)}(\bs\chi,\bs y) = \bs x(\bs\chi,\bs y+\bs vt)
\eeq
{\em are right-continuous at all $t\in\R$}.

Note how the  relation \eqref{consistency} guarantees that \eqref{Pexplicit} makes sense: for every $a\leq c\leq d\leq b$, with $s= \{i:x_i\in [a,b]\}$, $s'= \{i:x_i\in [c,d]\}\subseteq s$, the relation $({\bs x}(\bs\chi,\bs y)_s)_{s'} = {\bs x}(\bs\chi,\bs y)_{s'}$ is consistent with $\bs x(\bs\chi,\bs y)_s = \bs x(\bs\chi_s,\bs y^{(s_+,s_-)})$. Setting $s_-'=s\cap \{i:x_i<c\}$, $s_+'=s\cap \{i:x_i>d\}$, we have that $s_-\cup s_-' = \{i:x_i< c\}$, $s_+\cup s_+' = \{i:x_i>d\}$ and
\beqa
	({\bs x}(\bs\chi,\bs y)_s)_{s'} &=&
	{\bs x}(\bs\chi_s,\,\bs y^{(s_+,s_-)})_{s'}\n
	&=&
	{\bs x}\Big((\bs\chi_s)_{s'},\,\big(\bs y^{(s_+,s_-)}\big)^{(s_+',s_-')}\Big)
	=
	{\bs x}\Big(\bs\chi_{s'},\,\big(\bs y^{(s_+\cup s_+',s_-\cup s_-')}\big)\Big)\n
	&=&{\bs x}(\bs\chi,\bs y)_{s'}.
\eeqa

It is convenient to write
\beq\label{XDconv}
	\bs X(x_*;\bs\chi,\bs y),\quad \bs D(x_*;\bs\chi,\bs y)
\eeq
for the locally contracted vector and the vector of contracted distances at $x_*\in\R$ for $\bs\chi,\,{\bs x}(\bs\chi,\bs y)$ (Definition \ref{definatmag}). We recall the fixed-point property \eqref{fixedpoint},
\beq\label{fixedpointgen}
	X_i(x_i(\bs\chi,\bs y);\bs\chi,\bs y) = x_i(\bs\chi,\bs y).
\eeq
We also specialise the various imprecision measures introduced above to $\bs x(\bs\chi,\bs y)$, denoting these as follows: for every $N\in\N$, $\bs\chi\in\Rpu^N,\,\bs y\in\R^N$,
\beq\label{conveverything}\begin{aligned}
	& \reg(\bs\chi,\bs y) = \reg_{\bs\chi}({\bs x}(\bs\chi,\bs y))\quad \mbox{from Eq.~\eqref{reg}},\\
	& \varep(\bs\chi,\bs y) = \varep\quad \mbox{from Eq.~\eqref{varepdef}},\quad
	\delta(\bs\chi,\bs y) = \neigh(\reg(\bs\chi,\bs y))\quad\mbox{from Eq.~\eqref{convirreg} for $\bs x = {\bs x}(\bs\chi,\bs y)$}\\
	& K_i(\Delta X;\bs\chi,\bs y) = K_i(\Delta X),\quad
	\Delta x(\Delta X;\bs\chi,\bs y) = \Delta x(\Delta X)\quad
	\mbox{from Def.~\ref{defineigh} for $\bs x = {\bs x}(\bs\chi,\bs y)$},\\
	& \Delta^{\rm out}(\bs\chi,\bs y) = \Delta^{\rm out}
	\quad \mbox{from Eq.~\eqref{minimp} for $\bs x = {\bs x}(\bs\chi,\bs y)$.}
	\end{aligned}
\eeq

Recall the interval $I$, Eq.~\eqref{intervalreal}. This is what we refer to as the {\em fluid cell}. We are now ready to define the fluid-cell projection, which simply takes out all solitons whose positions are not within the fluid cell.
\begin{defi}[Fluid-cell projection] \label{defimapfc}
Set $I$ as in Eq.~\eqref{intervalreal}, which we refer to as the fluid cell. The fluid-cell projection onto $I$ is the map (see Eq.~\eqref{sspsm})
\beq
	\mathcal F_{I}: \mathfrak S\to \mathfrak S
\eeq 
such that for every $(\bs \chi,\bs y)\in \mathsf S$ (see Eq.~\eqref{Omega}, and Eqs.~\eqref{yispsm}, \eqref{ysyspsm} and \eqref{datafcgen}),
 \beq\label{fluidcelllimitsbe}
	\mathcal F_{I} u_{\bs \chi,\bs y}
	=
	u_{\bs \chi_s,\bs y^{(s_+,s_-)}} = u_{\bs\chi_s,\bs y^{(s)}}
	=
	u_{\bs\chi^{(I)},\bs y^{(I)}}\quad\forall\, x\in\R
\eeq
and $\mathcal F_{I} u_{\emptyset,\emptyset} = \mathcal F_{I} 0 = 0$, where
\beq\label{spmmapfcsbe}
	s_+ = \{i:x_i>I_+\},\quad
	s_- = \{i:x_i<I_-\},\quad  s = \dbra 1,N\dket \setminus (s_+\cup s_-)\quad\mbox{for}\  \bs x={\bs x}(\bs\chi,\bs y).
\eeq
From Lemma \ref{lemproj} we have
\beq\label{notationfc}
	\mathcal F_{I} u_{\bs \chi,\bs y}(x)
	=
	\lim_{\bs z:s_+,s_-}
	u_{\bs \chi,\bs y+\bs z}(x) \quad\forall\;x\in\R.
\eeq
That is, the relation \eqref{fcupsilon} holds: using the definition \eqref{Omega} of $\Omega$, then the relation \eqref{projectionexchange} and the definition \eqref{Moller} of the M\o ller transform $\Upsilon$, we have
\beq\label{fcformal}
	\mathcal F_I = \Omega \h{\mathcal P}_I\Omega^{-1} = \Omega \mathcal U^{-1}\mathcal P_I\mathcal U\Omega^{-1} =  \Upsilon \mathcal P_I \Upsilon^{-1}.
\eeq
\end{defi}
By \eqref{fcformal}, we have the projection property
\beq
	\mathcal F_I^2 = \mathcal F_I.
\eeq

A crucial aspect is that the locally contracted vector and contracted distances behave well under the fluid-cell projection.
\begin{lemma}[Locally contracted vector and distances under projection]\label{lemcontproj} In the context of Definition \ref{defimapfc}, we have for all $N\in\N,\,\bs\chi\in\Rpu^N,\,\bs y\in R^N$,
\beq\label{eqXi}
	X_i\big(x_*;\bs\chi^{(I)},\bs y^{(I)}\big)
	= X_i\big(x_*;\bs\chi,\bs y\big)\quad
	\forall\;x_*\in I,\,i:x_i\in I
\eeq
and
\beq\label{ineqdi}
	D_i\big(x_*;\bs\chi^{(I)},\bs y^{(I)}\big)
	\geq
	D_i\big(x_*;\bs\chi,\bs y\big)
	\quad\forall\; x_*\in \R,\,i:x_i\in I.
\eeq
\end{lemma}
\proof Eq.~\eqref{eqXi} follows immediately from \eqref{Xnatural}. Eq.~\eqref{ineqdi} likewise follows in the case $x_*\in I$. For $x_*\in\R\setminus I$, we have that $X_i\big(x_*;\bs\chi^{(I)},\bs y^{(I)}\big)\leq X_i\big(x_*;\bs\chi,\bs y\big)$ if $x_*>I_+$ because in this case $x_i<x_*$ and therefore the locally contracted vector after projection has less positive terms added in \eqref{Xnatural}; and similarly $X_i\big(x_*;\bs\chi^{(I)},\bs y^{(I)}\big)\geq X_i\big(x_*;\bs\chi,\bs y\big)$ if $x_*>I_-$, which implies the result.
\eproof
\medskip

These imply the following inequalities between our various measures of imprecision:
\beq\begin{aligned}
	&\reg\big(\bs\chi^{(I)},\bs y^{(I)}\big)\supseteq
	\reg\big(\bs\chi,\bs y\big)\\
	&\varep\big(\bs\chi^{(I)},\bs y^{(I)}\big)\leq
	\varep\big(\bs\chi,\bs y\big)\\
	& \delta\big(\bs\chi^{(I)},\bs y^{(I)}\big)\leq
	\delta\big(\bs\chi,\bs y\big),\\
	& \Delta x\big(\Delta X;\bs\chi^{(I)},\bs y^{(I)}\big)\leq
	\Delta x\big(\Delta X;\bs\chi,\bs y\big)\ \forall\;\Delta X\geq 0\\
	&
	\Delta^{\rm out}\big(\bs\chi^{(I)},\bs y^{(I)}\big)
	\leq
	\Delta^{\rm out}\big(\bs\chi,\bs y\big).
	\end{aligned}
	\label{ineqmeasure}
\eeq
That is, under the fluid-cell projection, all measures of imprecision can only decrease, and the set of regular points can only get larger. Then an immediate consequence of Lemma \ref{lemimp} and the last statement of Eq.~\eqref{ineqmeasure} is the following.
\begin{lemma}[The fluid cell bounds the core of the fluid-cell projected field]\label{lemcoresupportfc} In the context of Definition \ref{defimapfc}, we have for all $N\in\N,\,\bs\chi\in\Rpu^N,\,\bs y\in R^N$,
\beq\label{coresupportfcsbe}
	\core(\bs \chi^{(I)},\bs y^{(I)}) \subseteq [
	I_--\Delta^{\rm out}(\bs\chi,\bs y),
	I_++\Delta^{\rm out}(\bs\chi,\bs y)].
\eeq
\end{lemma}

A crucial, if very simple, result is that the centred form of the tau function after projection is  obtained by deleting terms from the centred form of the tau function before projection. This follows immediately from Lemma \ref{lemcontproj}.
\begin{corol}[Projected form of the centred tau function]\label{coroltauproj} In the context of Definition \ref{defimapfc},  let $N\in\N,\,\bs\chi\in\Rpu^N,\,\bs y\in R^N$ and $\Delta\bs y \in \R^{|s|}$. For every $x_*\in I$, the tau function of the projected KdV field $\tau_{\bs\chi^{(I)},\bs y^{(I)}+\Delta\bs y}$ is equivalent, in the sense of \eqref{equiv}, to the restriction to $s$ of the centred tau function $\tau_{\bs\chi,\bs x,\Delta\bs y,x_*}^\#$: with the notation of Definition \ref{lemtaumain},
\beq\label{tauproj}
	\tau_{\bs\chi^{(I)},\bs y^{(I)}+\Delta\bs y}
	\equiv\tau_s,\quad
	\tau_s(x) := \sum_{r\subseteq s} 
	\exp\Big[-\sum_{i\in r} 2\chi_i\big(
	D_i
	-\sgn(x_i-x_*)(x-x_*-\Delta y_i)\big)+ R_r\Big]
\eeq
where $\bs D = \bs D(x_*;\bs\chi,\bs y)$ and $\bs x = {\bs x}(\bs\chi,\bs y)$, with $\tau_\emptyset = 1$.
\end{corol}

In the next Section, when establishing our main theorems we will bound the terms that are omitted in \eqref{tauproj}. In order to do so, we will need control over the contracted distances themselves. For this purpose, we  define {\em local projections with respect to a closed interval $I'\subset \R$ in contracted space}. A fluid-cell projection is local, at some $(\bs\chi,\bs y)$ with respect to $I'$, if solitons projected away all lie outside the interval $I'$ {\em in the contracted metric} (which depends on $(\bs\chi,\bs y)$) at the centre $x_*$ of the interval $I'$. Crucially, the interval $I'$ includes the dark space of contracted positions that went beyond $x_*$, no matter how far they go beyond $x_*$ (see Section \ref{ssectsbe}). Here, we denote the interval
\beq\label{interval}
	I'=[x_*-L/2,x_*+L/2]\subset \R.
\eeq
\begin{defi}[Local projection] \label{defimap}
Let $x_*\in\R$, $L\geq0$ and $I'\subset\R$ be the interval \eqref{interval}. Let $N\in\N,\,\bs\chi\in\Rpu^N,\,\bs y\in R^N$. In the context of Definition \ref{defimapfc}, the fluid-cell projection $\mathcal F_I$ is local at $(\bs\chi,\bs y)$ with respect to $I'$
if, denoting $\bs X=\bs X(x_*;\bs\chi,\bs y)$ and using \eqref{spmmapfcsbe},
\beq
	s_+\subseteq \{i:X_i> x_*+L/2\},\quad
	s_-\subseteq \{i:X_i< x_*-L/2\}.
\eeq
Equivalently, denoting $\bs D=\bs D(x_*;\bs\chi,\bs y)$,
\beq
	D_i> L/2\quad\forall\; i\in \dbra 1,N\dket\setminus s.
\eeq
\end{defi}
Because of the contraction property, we must have $I'\subseteq I$. Combined with the centred form of the tau function from Lemma \ref{lemtaumain} and Corollary \ref{coroltauproj}, Theorem \ref{theo} will show how, under certain conditions, any local projection gives a good local approximation of the KdV field locally. However, the solitons that contribute are chosen using the contracted metric, where precision is lost the further we are from $x_*$. Thus we obtain an approximation of the field on an interval that is smaller then the ``real-space'' interval corresponding to $I$; we only know which solitons {\em may} contribute to the field there, but not which {\em actually do} contribute.

In particular, in order to obtain projection theorems for fluid-cell projections without referring to the local structure, given a fluid-cell projection $\mathcal F_I$, we will need to know for what points $x_*\in I$ it is a local projection for some interval of length $L\geq0$. Our imprecision measures allow us to determine this.
\begin{lemma}[Condition for a fluid-cell projection to be local]\label{lemfclocal} In the context of Definition \ref{defimapfc}, let $N\in\N,\,\bs\chi\in\Rpu^N,\,\bs y\in R^N$. Let $\Delta X\geq 0$. For every $\ep>0$ and
\beq\label{xstarrange}
	x_*\in [I_-+\Delta x(\Delta X+\ep;\bs\chi,\bs y)\,,\,I_+-\Delta x(\Delta X+\ep;\bs\chi,\bs y)]
\eeq
$\mathcal F_I$ is local at $(\bs\chi,\bs y)$ with respect to $I'= [x_*-\Delta X,x_*+\Delta X]$.
\end{lemma}
\proof Let $x_*\in [I_-+\Delta x(\Delta X+\ep;\bs\chi,\bs y)\,,\,I_+-\Delta x(\Delta X+\ep;\bs\chi,\bs y)]$. Then by Definition \ref{defineigh} we have $\dist\Big(x_*,\cup_{i:x_i<I_-} K_i\big(\Delta X+\ep;\bs\chi,\bs y\big)\Big)>0$ and therefore by Lemma \ref{lemshift} 
\beq
	X_i-x_*\geq \Delta X+\ep>\Delta X,\ D_i> \Delta X\ \forall\;i\in s_-;
\eeq
and also $\dist\Big(x_*,\cup_{i:x_i>I_+} K_i\big(\Delta X+\ep;\bs\chi,\bs y\big)\Big)>0$, whence
\beq
	X_i-x_*\leq -\Delta X-\ep<-\Delta X,\ D_i< \Delta X\ \forall\;i\in s_+.
\eeq
\eproof

\medskip


\section{Main results}\label{sectmainresult}

We now have all the tools in order to express our main results. This involves taking $N$ large, with appropriate control over $\bs\chi$ and $\bs y$.

We start with our assumptions in Section \ref{ssecass}: various results will use different subsets of assumptions, as we will specify. All assumptions are, we believe, simple enough, natural, and heuristically justified. Then we express our three groups of main results: the local form of multi-soliton field in Section \ref{ssecmain} -- the main technical result at the basis of our construction --, the spatial extent and bound of multi-soliton fields in Section \ref{ssectsupport} -- interesting consequences of the local form theorem --, and the results on fluid-cell averages in Section \ref{ssectmainfc} -- our main results.

\subsection{Large $N$: assumptions}\label{ssecass}

Our theorems are concerned with bounding the difference between a $N$-soliton KdV field and its fluid-cell projection, a $M$-soliton field with $M\leq N$. The interest is in considering large $N$ and a large interval $I$ on which we project, with the appropriate relation between them, in particular with the size of $I$ increasing with $N$, but staying much less then $N$.

In Section \ref{sectstatement} a soliton gas was simply defined as a sequence of spectral $\bs\chi$ and impact parameters $\bs y$, Eq.~\eqref{seqintro}. Here we extend this concept in order to allow for the number of solitons to change. Therefore, the soliton gas is the sequence
\beq\label{seqgen}
	\N\ni\b N\mapsto N\in\dbra 1,\b N\dket,\, {\bs \chi}\in\Rpu^{N},\,{\bs y}\in\R^{N}.
\eeq

Below, $\b N$ is the control parameter on the assumptions. In our main theorem, it will turn out to play the role of the number of solitons before projection. An important aspect is that our assumptions on sequences of $\b N$-soliton field imply the assumptions on fluid-cell projections thereof, allowing us to have clear projection theorems with assumptions that only refer to the original multi-soliton field. For lightness of notation we keep the dependence on the sequence index $\b N$ implicit.




The assumptions we will consider are of three types: the first concerns spectral parameters, and fixes how near they can be to each other and to 0, and how far they can go towards $\infty$. The second is about the density of solitons around points $x_*\in\R$ in contracted space. This will be expressed as bounds on large-distance density of solitons' displacements. This will allow us to control locally the KdV field, but it does not allow us to control the projection. For this, we need quasiparticles' positions, Definition \ref{defiqp}; these are their positions in ``real'' space and do not change under projection. The third type of assumption is concerned with the change of metric from real to contracted space, as expressed in the behaviour of the effective imprecision $\Delta x(\Delta X;\bs\chi,\bs y)$, and as expressed in terms of the amount of space taken away in this change of metric. We establish some lemmas that connect these assumptions and that uniformly bound the density of solitons in real space; in particular, the third type of assumption will imply the second type.

\medskip
{\em First type.} The following assumption is for a sequence $\b N\mapsto N,\,\bs\chi$.
\begin{assum}[Regularity of spectral parameters]\label{asssp}
There exists $\chi_*>0$, $A>0$, $C>\chi_*$, and $\alpha>0$, $\beta\in[0,\alpha/2]$, such that for every $\b N\in\N$,
\beq
	\label{c1}
	\min_i(\chi_i)\geq \chi_*,\quad
	\min_{i\neq j}|\chi_i-\chi_j|\geq e^{-A \b N^{\alpha/2}},\quad
	\max_i(\chi_i)\leq C\b N^\beta.
\eeq
\end{assum}
The data of this assumption is $\chi_*,\,A,\,C,\,\alpha,\,\beta$. This assumption is simple to implement, and weak enough to be, we believe, widely applicable. Yet, because of Theorem \ref{theosbe}, it already imply a bound on the density of solitons in real space.

\begin{lemma}[Bound on real soliton density on $\R$]\label{lemrealdens} Consider a sequence $\b N\mapsto N,\,\bs\chi,\,\bs y$ such that Assumption \ref{asssp} holds. Then for all $\b N\in\N$,
\beq
	\rho(\bs x):=\sup_{x_*\in\R,\,r\geq 1}\frc{\big|\{i:|x_i-x_*|\leq r\}\big|}{2r} \leq \b\rho \b N^{2\beta},\quad 
	\b \rho = \frc12\Big(1+\frc{C^2}{\chi_*}\Big)
\eeq
where $\bs x =  {\bs x}(\bs\chi,\bs y)$.
\end{lemma}
\proof
We have
\beq
	\min_{i,j}|\varphi_{ij}|
	\geq \frc1{\max_i \chi_i} \min_{i,j:\atop \chi_j>\chi_i}
	\big(\log(1+\chi_i/\chi_j)
	-
	\log(1-\chi_i/\chi_j)\big)
\eeq
Because $f(r):=\log(1+r)-\log(1-r)\geq 2r$ (as can be seen by taking derivatives), 
\beq
	\min_{i,j}|\varphi_{ij}|
	\geq
	\frc2{C}\b N^{-\beta}\,
	\min_{i,j:\atop \chi_j>\chi_i} \frc{\chi_i}{\chi_j}
	\geq
	\frc{2\chi_*}{C^2}\b N^{-2\beta}.
\eeq
Hence by \eqref{minseparation}, for every $x_*\in\R$, $r\geq 1$ and $\b N\in\N$,
\beq
	\big|\{i:|x_i-x_*|\leq r\}\big|\leq 1+\frc{2r}{\min_{ij} |\varphi_{ij}|}
	\leq
	1+\frc{rC^2\b N^{2\beta}}{\chi_*}
	\leq
	\Big(1+\frc{C^2}{\chi_*}\Big)r\b N^{2\beta}.
\eeq
\eproof
\medskip

{\em Second type.} The following assumption
is for  a sequence $\b N\mapsto N,\,\bs\chi,\,\bs y,\,x_*$ and we take $\bs D = 
\bs D(x_*;\bs\chi,\bs y)$.

\begin{assum}[Bound on contracted soliton density with exponent $\eta\geq 0$]\label{assde} There exist $B>0,\,D>0$ and $\mu\geq  0,\,\nu\geq 0$ such that for every $\b N\in\N$:
\beq
	\label{c2}
	\rho_{\bs D}(d)\leq D \b N^{\nu}\quad \forall\  d\geq B\b N^{\eta +\mu}
\eeq
where
\beq
	\rho_{\bs D}(d) := \frc{\big|\{i:D_i\leq d\}\big|}{2d}.
\eeq
\end{assum}
The data of this assumption is $B,\,D,\,\mu,\,\nu$, while $\eta$ is a parameter. Below we will choose it to be related to the data of Assumption \ref{asssp}, as
\beq\label{etachoice}
	\eta = \frc{\alpha}2.
\eeq
Assumption \ref{assde} has a uniform version for $x_*\in\R$,
where the assumptions' data is uniform.
%
%
Assumption \ref{assde} tells us about how much the density of solitons in contracted space is allowed to increase with $\b N$. Normally, for contractions that are not too strong (which happens at low enough density), we expect that $\nu=0$. But here we allow $\nu>0$. Physically, this should allow us to get near to the condensate soliton gas as $\b N\to\infty$; although we leave an analysis of soliton gases within the present formalism to future works.

%

\medskip
{\em Third type.} Finally, the following assumption is for a sequence $\b N\mapsto N,\,\bs\chi,\,\bs y$ and we take $\bs x =  {\bs x}(\bs\chi,\bs y)$.

\begin{assum}[Bound on imprecisions
with exponent $\eta\geq 0$]\label{assimp} There exists $B\geq 1,\,V\geq 2$ and $\mu\geq0,\,\omega\geq 0$ such that for every $\b N\in\N$:
\beq\label{c4}
	\frc{\Delta x(\Delta X;\bs\chi,\bs y)}{\Delta X}\leq V\b N^\omega
	\quad \forall \;\Delta X\geq B\b N^{\eta+\mu}.
\eeq
\end{assum}

\begin{assum}[Bound on spatial contraction with exponent $\eta\geq 0$]\label{asscont} There exists $B\geq 1,\,V\geq 2$ and $\mu\geq0,\,\omega\geq 0$ such that for every $\b N\in\N$:
\beq\label{c5}
	\sum_{j:x_j\in[x_*,x_i)}|\varphi_{ij}|
	\leq \Big(1-\frc2{V\b N^\omega}\Big) |x_i-x_*|\quad
	\forall\ x_*\in\R,\ 
	i\in\dbra 1,N\dket:|x_i-x_*|\geq BV\b N^{\eta+\mu+\omega}.
\eeq
\end{assum}
The data of these assumptions is $B,V,\mu,\omega$, while, again, $\eta$ is a parameter, which we will take to be \eqref{etachoice}. Note that by \eqref{minDeltax} we always have $V\geq 2$ in Assumption \ref{assimp}, and that we set $B\geq 1$ for convenience, as this makes our results below slightly simpler. Assumption \ref{assimp} tells us about the change of metric from the contracted to the real space, while Assumption \ref{asscont} is more explicitly about (a slight simplification of) the contraction term in \eqref{Xnatural}. Again, in situations with low enough densities, we expect the contracted distances to be of the order of a nonzero proportion of the real distances, giving $\omega=0$. But we allow for contracted distances that vanish on large scales as $\b N\to\infty$, which, again, we believe can describe the approach to soliton condensates.

We now establish two simple results, essentially the implications Assumption \ref{asscont} $\Rightarrow$ Assumption \ref{assimp} $\Rightarrow$ Assumption \ref{assde}.

\begin{lemma}[From contracted space to imprecisions]\label{lemconcentration} Consider a sequence $\b N\mapsto N,\,\bs\chi,\,\bs y$ such that Assumption \ref{asscont} holds with exponent $\eta$. Then Assumption \ref{assimp} holds with exponent $\eta$, with the same $B,\,V,\,\mu,\,\omega$.
\end{lemma}
\proof As
\beq
	\sum_{j:x_j\in[x_*,x_i)\,(x_i
	\geq x_*)\atop j:x_j\in(x_i,x_*)\,(x_i< x_*)}|\varphi_{ij}|
	\leq
	\sum_{j:x_j\in[x_*,x_i)}|\varphi_{ij}|
\eeq
(as per our convention \eqref{defab}), the inequality \eqref{c5} implies
\beq\label{condDW}
	D_i(x_*;\bs\chi,\bs y) \geq \frc2{V\b N^\omega} |x_i-x_*|
\eeq
(see Definition \ref{definatmag} and Eq.~\eqref{contracteddist}). Then by \eqref{condDW}, for every $i$ and every $\Delta X\geq B\b N^{\eta+\mu}$, we have that $x_i-x_*\geq V\b N^\omega\Delta X/2$ implies $X_i(x_*;\bs\chi,\bs y)-x_*\geq \Delta X$. Thus by Definition \ref{xilr}, we have
\beq
	x_i^{\rm left}(\Delta X;\bs\chi,\bs y)\geq
	x_i-\frc {V\b N^\omega}2\Delta X
\eeq
for all $\Delta X\geq B\b N^{\eta+\mu}$. A similar argument gives
\beq
	x_i^{\rm right}(\Delta X;\bs\chi,\bs y)\leq
	x_i+\frc {V\b N^\omega}2\Delta X.
\eeq
Hence
\beq
	\frc{\Delta x(\Delta X;\bs\chi,\bs y)}{\Delta X}
	\leq V\b N^\omega\quad\forall \ \Delta X\geq B\b N^{\eta+\mu}.
\eeq
\eproof
%
\begin{lemma}[From real to contracted soliton density]\label{lemrealcontracted} Consider a sequence $\b N\mapsto N,\,\bs\chi,\,\bs y$ such that Assumption \ref{asssp} holds, and such that Assumption \ref{assimp} holds
with exponent $\eta$. Then Assumption \ref{assde} holds  uniformly for $x_*\in \R$
with exponent $\eta$, with $B,\mu$ as in Assumption \ref{assimp}  and
\beq\label{relationcoeff}
	D = V\b\rho,\quad
	\nu =  2\beta+\omega.
\eeq
\end{lemma}
\proof From Lemmas \ref{lemshift} and \ref{lemrealdens}, for all $d,\ep>0$ such that $\Delta x:=\Delta x(d+\ep;\bs\chi,\bs y)\geq 1$,
\beq\label{interproof546}
	|\{i:D_i\leq d\}|\leq|\{i:D_i< d+\ep\}|\leq  |\{i:|x_i-x_*|\leq \Delta x\}|
	\leq 2\b\rho\Delta xN^{2\beta}.
\eeq
As $B\geq 1$ and $\b N\in\N$, then by \eqref{minDeltax}, for every $d\geq B\b N^{\eta+\mu}$ we have $\Delta x\geq 2(d+\ep)> 2B\b N^{\eta+\mu}\geq 2>1$. Then we can use \eqref{interproof546} for every $d\geq B\b N^{\eta+\mu}$ and by Assumption \ref{assimp},
\beq
	|\{i:D_i\leq d\}|\leq
	2V \b\rho \,(d+\ep)\,\b N^{2\beta+\omega} .
\eeq
Taking $\ep\to0^+$, we have for such $d$'s
\beq
	\rho_{\bs D}(d) =\frc{|\{i:D_i\leq d\}|}{2d}
	\leq
	V \b\rho \,\b N^{2\beta+\omega}.
\eeq
%
%
\eproof

\medskip

Finally, a crucial if simple lemma is that our assumptions remain valid after any fluid-cell projection: if the assumptions hold for $\b N\mapsto N,\, \bs\chi,\,\bs y$, then they hold for $\b N\mapsto N',\, \bs\chi',\,\bs y'$ for any sequence $\b N\mapsto I$ of intervals \eqref{intervalreal} and $(\bs\chi',\bs y')= \h{\mathcal P}_I(\bs\chi,\bs y)$ with $N'$ solitons (see Eq.~\eqref{hatprojection}).
\begin{lemma}[Stability of assumptions] \label{lemstable} Let $I$ be as in \eqref{intervalreal}. Assumption \ref{asssp} for $\N\ni \b N\mapsto N,\,\bs\chi\in \Rpu^N$ implies Assumption \ref{asssp} for $\N\ni \b N\mapsto M\in\dbra 1,N\dket,\,\bs\chi^{(I)}\in \Rpu^{M}$. Assumption \ref{assde} for $\N\ni \b N\mapsto N,\,\bs\chi\in \Rpu^N,\,\bs y\in \R^N,\,x_*\in\R$ implies Assumption \ref{assde} for $\N\ni \b N\mapsto M\in\dbra 1,N\dket,\,\bs\chi^{(I)}\in \Rpu^{M},\,\bs y^{(I)}\in \R^{M},\,x_*\in\R$, with the same exponent $\eta$. Assumption \ref{assimp} for $\N\ni \b N\mapsto N,\,\bs\chi\in \Rpu^N,\,\bs y\in \R^N$
implies Assumption \ref{assimp} for $\N\ni \b N\mapsto M\in\dbra 1,N\dket,\,\bs\chi^{(I)}\in \Rpu^{M},\,\bs y^{(I)}\in \R^{M}$ on $\R$
with the same exponent $\eta$. The data of the assumptions are unchanged.
\end{lemma}
\proof The first statement is immediate. For the second statement, we note that using \eqref{ineqdi}, we have
\beq
	\rho_{\bs D(x_*;\bs\chi^{(I)},\bs y^{(I)})}(d)\leq \rho_{\bs D(x_*;\bs\chi,\bs y)}(d),
\eeq
from which the statement follows. For the last statement, we use \eqref{ineqmeasure}.
\eproof

\medskip

In Section \ref{ssectovermain}, for Assumption \ref{asssp} we took the case $\beta=0$, and $\alpha$ as near as desired to 0, Eq.~\eqref{c1intro} (with $\eta = \alpha/2$); for Assumption \ref{assde} we took $\nu=\mu=0$ and $B=1$ in order to write Eq.~\eqref{c2intro}; and for Assumptions \ref{assimp}, \ref{asscont} we took $\omega=\mu=0$ and $B=1$ in \eqref{c4intro}.

We now express our main results, where Assumption \ref{asssp} is used throughout, and Assumptions \ref{assde} and \ref{assimp} are used as required.

\subsection{Local form of multi-soliton field}\label{ssecmain}


The first result concerns the  local form of the sequence of multi-soliton fields \eqref{utau} associated to \eqref{seqgen} as $\b N$ gets large: it turns out that they can be well approximated by multi-soliton fields with less solitons in it. For this result, we consider time evolution, which we denote as $\tau_t$:
\beq
	\tau_t u_{\bs\chi,\bs y} = u_{\bs\chi,\bs y + \bs v t}\quad
	\forall\,t\in \R,\,N\in\N,\,\bs\chi\in\Rpu^N,\,\bs y\in\R^N.
\eeq
We are are looking to approximate $u_{\bs\chi,\bs y + \bs v t}(x)$ by a field with a lesser number of solitons on a small regions of space-time $(x,t)$ around some point $(x_*,0)$ at initial time.

Here we consider sequences as follows.
{\em 
\beq\label{seq}
	\N\ni \b N\mapsto N\in \dbra 1,\b N\dket,\,\bs\chi\in\Rpu^{N},\,\bs y\in\R^{N},\,x_*\in \R,\,L\geq 0\quad \mbox{such that:}
\eeq
the spectral regularity Assumption \ref{asssp} holds; and the bound on density Assumption \ref{assde} holds with exponent $\eta = \alpha/2$.
}

\begin{theorem}[Local form of the multi-soliton field] \label{theo}
There exists $D_{n,m}>0,\kappa_{n,m}>0\ (n,m=0,1,2,\ldots),\,E>0$ and $N_*>0$ depending only on the assumptions' data, with $E$ given in \eqref{Edef}, such that for every sequence \eqref{seq}, for every integer $n,m\geq 0$, and for every sequence of local projections $\b N\mapsto \mathcal L$ at $(\bs\chi,\bs y)$ with respect to $I'$ (Eq.~\eqref{interval}, Definition \ref{defimap}), we have for all $\b N\geq N_*$,
\beq\label{proj}
	\sup_{(x,t)\in J\times K}\Big|\p_x^n\p_t^m \big(u_{\bs \chi,\bs y+\bs v t}(x) - \tau_t \mathcal L u_{\bs \chi,\bs y}(x)\big)\Big| \leq D_{n,m}
	\b N^{\kappa_{n,m}} \exp\Big(E \b N^{\alpha+\mu+\nu}-\chi_*L\Big)
\eeq
where
\beq\label{xregion}
	J:=[x_*- C^{-1}\b N^{\eta-\beta}\,,\,
	x_*+C^{-1}\b N^{\eta-\beta}],
\eeq
\beq\label{tregion}
	K:=[- C^{-3}\b N^{\eta-3\beta}\,,\,
	C^{-3}\b N^{\eta-3\beta}].
\eeq
In particular, setting $L=R \b N^\gamma$ with either
\beq\label{gamma}
	R>\frc{E}{\chi_*},\ \gamma = \alpha+\mu+\nu,\quad\mbox{or}\quad R>0,\ \gamma>\alpha+\mu+\nu,
\eeq
we have, uniformly on sequences \eqref{seq} with fixed assumptions' data, and sequences $\b N\mapsto \mathcal L$ as above,
\beq\label{convuproj}
	\sup_{(x,t)\in J\times K} \Big|\p_x^n \p_t^m\big(u_{\bs \chi,\bs y+\bs v t}(x) - \tau_t\mathcal L u_{\bs \chi,\bs y}(x)\big)\Big| = \mathcal O(\b N^{-\infty}) \quad (\b N\to\infty)\qquad \forall \,n,m\in\Z_+.
\eeq
\end{theorem}
The proof is provided in Sec.~\ref{proofmain}. The idea of the proof is as follows. We take the centred form Lemma \ref{lemtaumain}, and look at terms where $r\not\subseteq s$ for the set $s\supseteq\{i:D_i\leq L/2\}$ of Definition \ref{defimap}. By Corollary \ref{coroltauproj}, these are the terms that are to be bounded. For this purpose, we first bound the ``interacting part''
$R_r$ in a way that does not depend on $r$. Bounding this is good because without this term, the result \eqref{taumain} simplifies to a non-interacting tau function, essentially a product $\prod_{i=1}^N (1+ \re^{-2\chi_i(D_i-\sgn(x_i-x_*)(x-x_*-v_i t))})$. The bound is applied only to the terms with $r\not\subseteq s$, and thus we have a difference of non-interacting tau functions, one where far-solitons are taken away, and this, we can bound by simple techniques, using our assumptions on solitons\ contracted distances and impact parameters. The critical aspect of this proof is to bound $R_r$, given in \eqref{eR}. Using  \eqref{Sp} for $S_{ij}$, formally replacing $\sum_i$ by $\int \dd \chi$, and opening up towards $\chi<0$ by anti-symmetry, this is of the form
\beq\label{ajfge}
	\int \dd \chi\dd\chi'\,F(\chi)F(\chi')\log|\chi-\chi'|
\eeq
for some function $F(\chi)$ with $F(-\chi) = - F(\chi)$. As the regular part of the Fourier transform of $\log|\cdot|$ is strictly negative, the result is negative, hence $\re^{R_r}<1$. Equivalently, $\log|\vec\chi-\vec\chi'|$ is the Greens function for $\nabla^2$ on $L^2(\R^2)$, hence it is a negative semi-definite integral operator on an appropriate space, on which $\nabla^2$ has an inverse. Of course, $R_r$ in \eqref{eR} is not strictly of the form \eqref{ajfge}, and in particular the diagonal part is taken away, as it should (otherwise $R_r$ would not be well defined because of the singularity of the $\log$ function). But we can bring it to that form up to terms that we can control, essentially an integral along the diagonal, and these term grow only as a small enough power law with $\b N$.

The theorem is useful if we choose the local size $L\propto \b N^\gamma$ as above, as in this case the right-hand side of \eqref{proj} decays exponentially with $L$ as $\b N\to\infty$, showing that the projection agrees, for values of $x$ in \eqref{xregion}, with the original KdV field. It is also more meaningful to take $\gamma<1$, or $\gamma$ small enough in such a way that there are solitons lying outside $[x_*-L/2,x_*+L/2]$, as otherwise no solitons are projected out and the left-hand side of \eqref{proj} is zero. With this, the theorem says that {\em as $\b N$ gets large, only solitons whose contracted positions lie within $[x_*-L/2,x_*+L/2]$, i.e.~the set $s=\{i:D_i\leq L/2\}$ -- including the dark space of negative contracted distances -- may influence the KdV field in the region \eqref{xregion}, and the way they do this is by making the field well approximated by a multi-soliton solution described by the corresponding subset of spectral and impact  parameters $(\bs\chi_s,\bs y^{(s_+,s_-)})$ having semiclassical Bethe vectors $\bs x_s$.}

Note that by the constraint $\beta\leq \alpha/2$ in Assumption \ref{asssp}, the spatial interval $J$ does not shrink as $\b N\to\infty$; this is why we imposed this condition. However the time interval $K$ may shrink to the single point $\{0\}$, unless $\beta$ is small enough. It is easy to generalise to other ``time evolutions'' of the KdV hierarchy, where $\eta-3\beta$ is replaced by $\eta-(2m+1)\beta$ for $m=2,3,4,\ldots$.

The above is truly about the local form only. First, we do not have control over the actual physical space that is covered by the contracted space interval $I$ around $x_*$. Second, with \eqref{gamma}, the region \eqref{xregion} is likely to be much smaller than the true spacial region corresponding to the interval $I$ in contracted space around $x_*$: the theorem does not say if all solitons covered by $I$, or at least a proportion of them tending to 1, do, indeed, affect the field in the real-space region \eqref{xregion}. In order to relate fluid-cell averages to averages within the full extent of a projected multi-soliton field, and more generally to implement the quasiparticle criterium \eqref{loose}, so that we have a good control of all spectral parameters that contribute, we must have a better understanding of this matter. This is obtained using corollaries of the above theorem that concerns the support of a multi-soliton solution and how large the multi-soliton may get, which we now express.

\subsection{Spatial extent of, and bound on, multi-soliton field}\label{ssectsupport}

We obtain as corollaries of Theorem \ref{theo} two fundamental results on the structure of the KdV field, which do not involve any fluid-cell projection (but whose proofs are obtained by using appropriate local projections). Both will be useful in order to relate fluid-cell projections to fluid-cell averages, our main theorems. The results here will be presented for spatial derivatives only, but by the KdV equation itself, \eqref{kdv}, one may easily obtain similar bounds involving time derivatives as well. The first result is presented at time 0 only, but its translation to other times is straightforward by linear evolution of $\bs y$.

We first obtain an interesting result about the support of multi-soliton fields. Bounds on the support of the multi-soliton KdV field can be obtained via Riemann-Hilbert problem techniques (see e.g.~\cite{babelon2003introduction}); we believe the following corollary is, however, the first that identifies $\core(\bs\chi,\bs y)$, Eq.~\eqref{core}, has being a good description of the support. For generality, the assumption here considers what happens outside the core, adding a width $w\geq 0$; below we will only need $w=0$.

We consider sequences as follows.
{\em \beq\label{seq2}
	\N\ni \b N\mapsto N\in\dbra 1,\b N\dket,\,\bs\chi\in\Rpu^N,\,\bs y\in\R^N,\,w\geq 0\quad \mbox{such that:}
\eeq
the spectral regularity Assumption \ref{asssp} holds; and the bound on the density Assumption \ref{assde} holds with exponent $\eta = \alpha/2$
at, or near to, both boundaries of the core, i.e.~for $\b N\mapsto N,\,\bs\chi,\,\bs y,\,x_*=x^\pm(\bs\chi,\bs y)\pm w$.
}
\begin{corol}[Spatial extent of the multi-soliton field] \label{corolextent} There exist $D_n>0,\,\kappa_n>0\ (n=0,1,2,\ldots),\,E>0$ and $N_*>0$ depending only on the assumptions' data, with $E$ given in \eqref{Edef}, such that for every sequence \eqref{seq2}, for every integer $n\geq 0$ and for every $\b N\mapsto x\in\R$ with $\dist(x,\core(\bs\chi,\bs y))>w$, we have for all $\b N\geq N_*$,
\beq\label{projvanish}
	\Big|\p_x^n u_{\bs \chi,\bs y}(x)\Big| \leq D_n
	\b N^{\kappa_n} \exp\Big(E\b N^{\alpha+\mu+\nu}-2\chi_*\dist(x,\core(\bs\chi,\bs y))\Big).
\eeq
As a consequence, there exist $Q>0,\,\ep\in(0,2\chi_*)$ such that, with
\beq\label{Ksupport2}
	J := [x_-(\bs\chi,\bs y)-w-Q \b N^{\alpha+\mu+\nu}\,,\,x_+(\bs\chi,\bs y)+w+Q \b N^{\alpha+\mu+\nu}]\supset \core(\bs\chi,\bs y)
\eeq
(recall \eqref{core}) we have uniformly on sequences \eqref{seq2} with fixed assumptions' data,
\beq\label{uvanish}
	\sup_{x\in \R\setminus J} |\re^{\ep \dist(x,\core(\bs \chi,\bs y))}\p_x^n u_{\bs \chi,\bs y}(x)| = \mathcal O(\b N^{-\infty})\quad (\b N\to\infty)\qquad \forall\, n.
\eeq
\end{corol}
\proof
Here we denote $\bs D^+(\bs\chi,\bs y):=\lim_{\ep\to0^+}\bs D(x_+(\bs\chi,\bs y)+\ep;\bs\chi,\bs y)$ and $\bs D(x_*) := \bs D(x_*;\bs\chi,\bs y)$. Let $x_*> x_+(\bs\chi,\bs y)$, which may depend on $\b N$ in an arbitrary fashion. By Lemma \ref{lemnat} locally contracted vectors  are given by $\bs X^+(\bs\chi,\bs y)$ independently of $x_*$, Eq.~\eqref{valuesextremepositions}, and by Lemma \ref{extregular} $x_*$ is regular, so that $\sgn(x_i-x_*)=\sgn(X_i-x_*) = -1$ for all $i$. As this holds for all $x_*>x_+(\bs\chi,\bs y)$, this implies that $\bs D(x_*) = \bs D^+(\bs\chi,\bs y) + (x_*-x_+(\bs \chi,\bs y))$, and further, $D_i^+(\bs\chi,\bs y)\geq 0$ for all $i$.

Now let $x_*> x_+(\bs\chi,\bs y)+w$. Then, for any $d>0$, by the above relations, we have $\rho_{\bs D(x_*)}(d)\leq \rho_{\bs D^+(\bs\chi,\bs y)+w\bs 1_M}(d) = \rho_{\bs D(x_+(\bs\chi,\bs y)+w)}(d) $
whence the assumption of the theorem means that the density condition \eqref{c2} also is fulfilled for the sequence $N\mapsto M,\,\bs\chi,\,\bs y,\,x_*$ with $\eta=\alpha/2$. Note that $D_i(x_*)\geq x_*-x_+(\bs\chi,\bs y)$ because $D_i^+(\bs\chi,\bs y)\geq 0$ for all $i$. Let $0<\ep<1$ and set $L/2 := (1-\ep)(x_*-x_+(\bs\chi,\bs y)) = (1-\ep)\dist(x_*,\core(\bs\chi,\bs y))$ and $I'$ as \eqref{interval}. Consider the fluid-cell projection $\mathcal F_{I}$ where $I = \{x_*\}$. This projects onto $s= \emptyset$ (Definition \ref{defimapfc}) and as $D_i(x_*)>L/2$ for all $i$, this is local at $(\bs\chi,\bs y)$ with respect to $I'$ (Definition \ref{defimap}). Applying Theorem \ref{theo} with $x=x_*$ and $\mathcal L = \mathcal F_{I}$ completes the proof of the statement \eqref{projvanish} for $x$ a distance at least $w$ to the right of $\core(\bs\chi,\bs y)$ and for $(1-\ep)\dist(x_*,\core(\bs\chi,\bs y)$ in place of $\dist(x_*,\core(\bs\chi,\bs y)$. As this holds for all $\ep>0$, we may take the limit $\ep\to0$ on the right hand side and we find \eqref{projvanish}. A similar argument holds for $x$ a distance at least $w$ to the left of $\core(\bs\chi,\bs y)$. The statement \eqref{uvanish} then follows by choosing $\ep\in(0,2\chi_*)$ and writing $2\chi_* = (2\chi_*-\ep) + \ep$, and $Q>E/(2\chi_*-\ep)$.
\eproof

\medskip
This corollary says that the multi-soliton field vanishes uniformly outside a ``slight'' extension of its core. This extension, a skin-effect that fattens its boundaries, is more important if the density, in the contracted space, increases as a faster power law, and if the spectral parameters are closer to each other. We discuss in Sec.~\ref{sssectcond} how this is natural from the physics of solitons and their interaction, and in Sec.~\ref{sssectsg} how the expected situation for a finite-density gas gives that $\mu=\nu=0$ and that $\alpha$ can be take as near to $0$ as desired.

Note that using \eqref{c1} and \eqref{varphi}, \eqref{varphineg}, we have
\beq\label{boundvarphi}
	|\varphi_{ij}| \leq \frc1{\chi_*}
	\big(\max_{ij} \log|\chi_i+\chi_j|
	-\min_{ij} \log(\chi_i-\chi_j)\big)
	\leq 
	\frc1{\chi_*}\big(
	\log(2C) +\beta\log N +
	AN^{\frc{\alpha}2}\big)
\eeq
and bounding each term in the sums in \eqref{condxstar}, we find that there is $R'>0$ such that
\beq
	J' := [\min_i (y_i)-R' N^{1+\frc\alpha2}\,,\,\max_i (y_i)+R' N^{1+\frc\alpha2}]\supset \core(\bs\chi,\bs y).
\eeq
Then,  if
\beq\label{boundalpha}
	\mu+\nu < 1-\frc\alpha 2,
\eeq
which is equivalent to $\alpha+\mu+\nu < 1+\frc\alpha2$, we also have $J'\supset J$, and therefore uniform vanishing, Eq.~\eqref{uvanish}, outside $J'$. Eq.~\eqref{boundalpha} is natural and valid in most applications, as $\mu$ and $\nu$ are expected to be small. Typically, in a finite, nonzero-density configuration, we expect $\max_i(y_i), \min_i(y_i) \propto N$. In this case, we see that the support of the multi-soliton field is within a region whose left and right boundaries can grow at most like $N^{1+\alpha/2}$.

It is immediate by Lemma \ref{lemrealcontracted} and \ref{extregular} that if Assumption \ref{asssp} holds, and Assumption \ref{assimp} holds on $\R$ for exponent \eqref{etachoice}, then Corollary \ref{corolextent} holds for any $w>0$ under the relation \eqref{relationcoeff} between data. In particular, \eqref{uvanish} holds with
\beq\label{Ksupport3}
	J := [x_-(\bs\chi,\bs y)-Q \b N^{\alpha+2\beta+\mu+\omega}\,,\,x_+(\bs\chi,\bs y)+Q \b N^{\alpha+2\beta+\mu+\omega}]\supset \core(\bs\chi,\bs y).
\eeq

The second result we obtain bounds the size of $u_{\bs\chi,\bs y}(x)$ everywhere on its support. This is based on  Lemma \ref{lembound1}. Here we consider sequences as follows.
{\em 
\beq\label{seq2b}
	\N\ni \b N\mapsto N\in \dbra 1,\b N\dket,\,\bs\chi\in\Rpu^N,\,\bs y\in\R^N\quad \mbox{such that:}
\eeq
the spectral regularity Assumption \ref{asssp} holds; and the bound on imprecision  Assumption \ref{assimp} holds
with exponent $\eta = \alpha/2$.
}
\begin{corol}[Bound on the multi-soliton field] \label{corolbound2} There exists $N_*\in\N$ and $T>0$ depending only on the assumptions' data such that for every sequence \eqref{seq2b}, and for every integers $n\geq 0$ and $\b N\geq N_*$, we have
\beq\label{bound2main}
	\sup_{(x,t)\in \R\times K}|\p_x^nu_{\bs \chi,\bs y+\bs v t}(x)|\leq 
	2(n+2)!\,T^{n+2}\,\Biggr(\sum_{m_1,\ldots,m_{n+2}\geq 0\atop \sum_j jm_j= n+2}
	\frc{\Big(\sum_j m_j-1\Big)!}{\prod_{j=1}^{n+2} j!^{m_j}m_j!}\Biggr)\,
	\b N^{(\alpha+5\beta+\mu+2\omega)(n+2)}
\eeq
where $K$ is given in \eqref{tregion}.
\end{corol}
\proof Let $L=R\b N^\gamma$ for some $\gamma>0,\,R>0$. Set $\ep>0$ and $\Delta x:= \Delta x(L/2+\ep;\bs\chi,\bs y)$,
let $\b N\mapsto x_*\in \R$, and set
\beq
	I = [x_*-\Delta x,x_*+\Delta x].
\eeq
By Lemma \ref{lemfclocal}, $\mathcal F_I$ is local at $(\bs\chi,\bs y)$  with respect to $[x_*-L/2,x_*+L/2]$. By Lemma \ref{lemrealcontracted}, Assumption \ref{assde} holds with exponent $\eta$ uniformly for $x_*\in\R$. 
Then Theorem \ref{theo} applies, and by \eqref{relationcoeff}, choosing
\beq
	\gamma = \alpha+\mu+\nu = \alpha+2\beta+\mu+\omega,
	\quad R>E/\chi_*
\eeq
where $E$ is given by \eqref{Edef} with $D = V\b \rho$, the right-hand side of \eqref{proj} vanishes as a stretched exponential as $\b N\to\infty$. As we can choose $N_*$ such that $L> 2B\b N^{\frc{\alpha}2+\mu}$ for all $\b N\geq N_*$ (because $\alpha>0$), then Lemma \ref{lemrealdens} and Assumption \ref{assimp} imply that
\beq
	|\{i:|x_i-x_*||\leq \Delta x\}| \leq 2\Delta x\b\rho\b N^{2\beta}
	\leq  V\b\rho (L+2\ep)\b N^{\omega+2\beta}
	\leq VR'\b \rho \b N^{\gamma+\omega+2\beta}
\eeq
where $R' = 1+2\ep/R\geq 1+2\ep/(R\b N^{\b\gamma})$. With \eqref{c1}, we then have
\beq
	\sum_{i:|x_i-x_*|\leq \Delta x} \chi_i \leq CVR'\b\rho\b N^{\alpha+5\beta+\mu+2\omega}.
\eeq
Inequality \eqref{bound1} gives
\beq
	|\p_x^n\mathcal F_{I}u_{\bs \chi,\bs y+\bs v t}(x)|\leq 
	2(n+2)!\sum_{m_1,\ldots,m_{n+2}\geq 0\atop \sum_j jm_j= n+2}
	\frc{\Big(\sum_j m_j-1\Big)!}{\prod_{j=1}^{n+2} j!^{m_j}m_j!}
	\Big(2CVR'\b\rho \b N^{\alpha+5\beta+\mu+2\omega}\Big)^{n+2}
\eeq
for all $x\in \R$ and integer $n\geq 0$. Then \eqref{proj} and the triangle inequality give the result \eqref{bound2main}, where we can choose $T:=2CVR'\b\rho+\ep$ for some $\ep>0$ (independent of $n$) large enough in order to account for the exponential correction from \eqref{proj}.
\eproof

\subsection{Fluid-cell averages and main projection result}\label{ssectmainfc}

We now consider what happens under fluid-cell projections, Definition \ref{defimapfc}. For this purpose, it is sufficient to concentrate on sequences \eqref{seqgen} with
\beq
	\b N=N,
\eeq
this is the case where the assumptions are the weakest (the parameter $\b N$ determining the bounds is simply the number of solitons $N$). It turns out that the results of Section \ref{ssectsupport} on the support and bound of a multi-soliton field, give us enough control on exactly what solitons contribute on the region where the fluid-cell projection is valid.

Recall Definition \ref{defimapfc} for the fluid-cell projection $\mathcal F_{I}$, which involves taking a limit that sends impact parameters to $\pm\infty$ for all solitons whose quasiparticles' positions $\h x_i(\bs\chi,\bs y)$ are outside the interval $I=[I_-,I_+]$, and which takes the form (see \eqref{datafcgen})
\beq\label{FIprop}
	\mathcal F_{I} u_{\bs\chi,\bs y}= u_{\bs \chi^{(I)},\bs y^{(I)}}
\eeq
where in particular $\bs \chi^{(I)} = \bs\chi_{\{i:x_i\in I\}}$.
%

We first establish three preliminary result on which our main theorems will be based. For the following three collolaries, we consider sequences as follows.
{\em \beq\label{seq3}
	\N\ni N\mapsto \bs\chi\in\Rpu^N,\,\bs y\in\R^N,\, I_-,I_+\in\R\ (I_-\leq I_+)\quad\mbox{such that:}
\eeq
the spectral regularity Assumption \ref{asssp} holds; 
the bound on imprecision Assumption \ref{assimp} with exponent $\eta = \alpha/2$
holds;
and setting
\beq\label{DeltaI}
	\Delta I := VRN^{\alpha+2\beta+\mu+2\omega},\quad
	R>\frc{E}{2\chi_*}
\eeq
where $E$ is given in \eqref{Edef} with $D = V\b \rho$, we have
\beq\label{intervalled}
	I_-+\Delta I < I_+-\Delta I.
\eeq
}

Below, in the proofs we will use
\beq\label{DeltaXmain}
	\Delta X := R'N^\gamma,\quad 
	\gamma = \alpha+2\beta+\mu+\omega,\quad
	\frc{E}{2\chi_*}<R' < R
\eeq
and
\beq
	\Delta x:=\Delta x(\Delta X+(R-R')N^\gamma;\bs\chi,\bs y).
\eeq
We also use the fact that, by Assumption \ref{assimp}, $\Delta x\leq VN^\omega(\Delta X+(R-R')N^\gamma)=VRN^{\alpha+2\beta+\mu+2\omega}= \Delta I$. We set as usual
\beq
	I = [I_-,I_+].
\eeq
Note the various exponents we have encountered until now, in terms of the data of Assumption \ref{asssp} and \ref{assimp}, in non-decreasing order: $\alpha+2\beta+\mu+\omega$ in \eqref{Ksupport3} and as part of the proofs below; $\alpha+2\beta+\mu+2\omega$ in our choice of $\Delta I$, Eq.~\eqref{DeltaI}, and $\alpha+5\beta+\mu+2\omega$ in Corollary \ref{corolbound2}.

The first result is a simple consequence of Theorem \ref{theo}. This tells us what happens {\em within the fluid cell}, the interval $I$. It says that the fluid-cell projection on $I$ agrees with the original KdV field on an interval that is ``slightly'' smaller than $I$, with a variation $\Delta I>0$ taken away at the borders. For it, we need to control the local projection on $\R$.

\begin{corol}[Form of the projected field within the fluid cell]\label{corolfc}
We have, uniformly on sequences \eqref{seq3} with fixed assumptions' data,
\beq\label{convuprojfc}
	\sup_{x\in  [I_-+\Delta I,I_+-\Delta I],} \Big|\p_x^n \big(u_{\bs \chi,\bs y}(x) - \mathcal F_{I} u_{\bs \chi,\bs y}(x)\big)\Big| = \mathcal O(N^{-\infty}) \quad (N\to\infty)\qquad \forall \;n\in\Z_+.
\eeq
\end{corol}
\proof 
Consider $N\mapsto M,\,\bs\chi^{(I)},\bs y^{(I)}$. By Lemmas \ref{lemstable} and \ref{lemrealcontracted}, Assumption \ref{asssp} holds, and Assumption \ref{assde} holds with exponent $\eta$ uniformly for $x_*\in \R$
with $D = V\b\rho,\,\nu =  2\beta+\omega$. Hence we can apply Theorem \ref{theo} uniformly for such $x_*$. By Lemma \ref{lemfclocal}, $\mathcal F_I$ is local at $(\bs\chi,\bs y)$ with respect to $[x_*-\Delta X,x_*+\Delta X]$. So we set $L/2 = \Delta X$ and $x_*\in 
[I_-+\Delta x,I_+-\Delta x]$ and $N\mapsto \mathcal L = \mathcal F_I$ in Theorem \ref{theo}. Because of \eqref{DeltaXmain}, we may use \eqref{convuproj},
and we set $x=x_*$
and $t=0$ in \eqref{convuproj}. Doing this for all $x_*\in
[I_-+\Delta x,I_+-\Delta x]$,
as $\Delta x< \Delta I$ this completes the proof. \eproof

\medskip

The next result is a consequence of Corollary \ref{corolextent} and Lemma \ref{lemcoresupportfc}. It tells us what happens {\em outside the fluid cell}, the interval $I$, after the fluid-cell projection: the projected field vanishes outside a ``slightly'' bigger interval, which can be taken as a positive multiple of $\Delta I$ added on at the borders.

\begin{corol}[Vanishing of the projected field outside the fluid cell]\label{corolvanishfc}
There exist $a\geq 1,\,\ep>0$ depending only on the assumptions' data such that we have, uniformly on sequences \eqref{seq3} with fixed assumptions' data,
\beq\label{uvanishfcN}
	\sup_{x\in \R\setminus [I_--a\Delta I,I_++a\Delta I]} |\re^{\ep \dist(x,I)}\p_x^n \mathcal F_{I} u_{\bs \chi,\bs y}(x)| = \mathcal O(N^{-\infty})\quad (N\to\infty)\qquad \forall\, n\in\Z_+.
\eeq
\end{corol}
\proof
Consider $N\mapsto M,\,\bs\chi^{(I)},\bs y^{(I)}$. By Lemmas \ref{lemstable} and \ref{lemrealcontracted}, Assumption \ref{asssp} holds, and Assumption \ref{assde} holds with exponent $\eta$ uniformly for
$x_*\in\R$
with $D = V\b\rho,\,\nu =  2\beta+\omega$. Hence we can apply Corollary \ref{corolextent}, where we may take $w= 0$.
We therefore have
\beq
	\sup_{x\in \R\setminus J} |\re^{\ep \dist(x,\core(\bs \chi^{(I)},\bs y^{(I)}))}\p_x^n \mathcal F_Iu_{\bs \chi,\bs y}(x)| = \mathcal O( N^{-\infty})\quad ( N\to\infty)\qquad \forall\, n
\eeq
and by Lemma \ref{lemcoresupportfc} $\core(\bs\chi^{(I)},\bs y^{(I)}) \subseteq [I_--\Delta^{\rm out}(\bs\chi,\bs y), I_++\Delta^{\rm out}(\bs\chi,\bs y)]$ so that we can use
\beq\label{JfcNout}
	J = [I_--\Delta^{\rm out}(\bs\chi,\bs y)-QN^{\alpha+2\beta+\mu+\omega},I_++\Delta^{\rm out}(\bs\chi,\bs y)+QN^{\alpha+2\beta+\mu+\omega}].
\eeq
As $\dist(x,\core(\bs \chi^{(I)},\bs y^{(I)}))\geq \dist(x,I)$, and because of the fact that by \eqref{minimp} and Lemma \ref{lembasicimp} (Eq.~\eqref{inclimp}) we have $\Delta^{\rm out}(\bs\chi,\bs y)\leq \Delta x$, there exists $Q>0$ depending only on the assumptions' data such that the result \eqref{uvanishfcN} follows
for
\beq\label{JfcNint}
	J := [I_--\Delta x-QN^{\alpha+2\beta+\mu+\omega},I_++\Delta x+QN^{\alpha+2\beta+\mu+\omega}].
\eeq
With $\Delta x\leq \Delta I$  and $QN^{\alpha+2\beta+\mu+\omega}\leq QN^{\alpha+2\beta+\mu+2\omega} = (Q/VR) \Delta I$ the result follows.
%
 \eproof

\medskip

We have established what happens within the fluid cell, far enough from its boundaries (Corollary \ref{corolfc}), and outside of it, also far enough from its boundaries (Corollary \ref{corolvanishfc}). The third preliminary result tells us how the projected field is bounded {\em everywhere}. This follows from Corollary \ref{corolbound2}, and is crucial for the projection theorems below, as we need to show that the contributions near the boundaries vanish.
\begin{corol}[Uniform boundedness of the projected field]\label{coroboundedfc}
There exists $N_*\in\N$ and $T>0$ depending only on the assumptions' data such that for every sequence \eqref{seq3}, and every integers $n\geq 0$ and $N\geq N_*$, we have
\beq\label{bound2mainfc}
	\sup_{x\in \R}|\p_x^n\mathcal F_Iu_{\bs \chi,\bs y}(x)|\leq 
	2(n+2)!\,T^{n+2}\,\Biggr(\sum_{m_1,\ldots,m_{n+2}\geq 0\atop \sum_j jm_j= n+2}
	\frc{\Big(\sum_j m_j-1\Big)!}{\prod_{j=1}^{n+2} j!^{m_j}m_j!}\Biggr)\,
	N^{(\alpha+5\beta+\mu+2\omega)(n+2)}.
\eeq
\end{corol}
\proof By Lemma \ref{lemstable} and Lemma \ref{ineqmeasure}
the assumptions of Corollary \ref{corolbound2} hold for $N\mapsto M,\,\bs\chi^{(I)},\bs y^{(I)}$ with exponent $\eta$.
The result follows by choosing $t=0\in K\;\forall\, N$. \eproof

\medskip

Finally, we now have all required results in order to establish the main fluid-cell projection theorem. This is based on the following: if $\Delta I$ is small enough compared to $I_+-I_-$, then the first two preliminary results above tell us that the fluid-cell projection is a good approximation of the original KdV multi-soliton field within the bulk of the fluid cell $I$, and that beyond a small fattening of $I$, the projected field vanishes. With the third preliminary result, we can neglect the parts around the boundaries of $I$, within $[I_+-\Delta I,I_+ + a\Delta I]$ and $[I_--a\Delta I,I_- + \Delta I]$. Thus, the fluid-cell average of ``any'' observable, function of the original multi-soliton field and its derivatives within $I$, is given by the fluid-cell average of the projected field. In particular, this tells us that the solitons with spectral parameters $\bs\chi_s = (\chi_i)_{x_i\in I}$ are exactly those that contribute within $I$, up to corrections that vanish as $N\to\infty$, and that fluid-cell averages of densities of the KdV conserved quantities are given by the density of such quantities carried by the solitons within $I$. This is the main property that we need to define the density of states of soliton gases.

We assign a degree $\deg$ to polynomials $P:\R^{n+1}\to\R$ as follows: for formal variables $u_k,\,k=0,\ldots,n$, we have
\beq\label{degree}
	\deg(u_k) = k+2,\quad
	\deg(u_{k_1}\cdots u_{k_l}) = \sum_{j=1}^l \deg(u_{k_l}),
\eeq
and the degree of the polynomial $P(u_0,\ldots,u_n)$ is the maximal degree of all its monomials.

For the following two results, we consider  sequences as follows.
{\em \beq\label{seq4}
	\N\ni N\mapsto \bs\chi\in\Rpu^N,\,\bs y\in\R^N\quad\mbox{such that:}
\eeq
the spectral regularity Assumption \ref{asssp} holds; 
the bound on imprecision Assumption \ref{assimp} with exponent $\eta = \alpha/2$ holds.
}
\begin{theorem}[Fluid-cell mean of local observables as fluid-cell projection] \label{theofc}
Let $N\mapsto I_-<I_+$ and $\lambda>0$, such that
\beq\label{Lgamma}
	L:= I_+-I_- = N^\lambda.
\eeq
Denote
\beq\label{uupDeltax}
	u = u_{\bs\chi,\bs y},\quad u'=\mathcal F_{I}u_{\bs\chi,\bs y}.
\eeq
Then, uniformly on sequences \eqref{seq4} with fixed assumptions' data:

\medskip
\noindent Case I: Let $F:\R^{n+1}\to \R$ be a  bounded Lipschitz function for the $L^1$ norm on $\R^{n+1}$ with $F(\bs 0) = 0$. If
\beq\label{case1}
	\lambda>\alpha+2\beta+\mu+2\omega,
\eeq
then
\beq\label{res1}
	\lim_{N\to\infty} \Bigg(\frc1{L}\int_{I} \dd x\,F(u,\p_x u,\ldots,\p_x^nu)
	- \frc1{L}\int_{-\infty}^{\infty} \dd x\,F(u',\p_x u',\ldots,\p_x^nu')
	\Bigg)=0.
\eeq
Case II: Let $P:\R^{n+1}\to\R$ be a polynomial without constant term. Set $r=\deg(P(u_0,\ldots,u_n))$. If
\beq\label{case2}
	\lambda > (1+r)\alpha+(2+5r)\beta+r\mu+(2+2r)\omega,
\eeq
then
\beq\label{res2}
	\lim_{N\to\infty} \Bigg(
	\frc1{L}\int_{I} \dd x\,P(u,\p_x u,\ldots,\p_x^nu)
	- \frc1{L}\int_{-\infty}^{\infty} \dd x\,P(u',\p_x u',\ldots,\p_x^nu')
	\Bigg)=0.
\eeq
In particular, for the $k$th local conserved density $\mathtt P_k$, Eq.~\eqref{conspoly}, we have
\beq\label{res2conserved}
	\lim_{N\to\infty}
	\Bigg(
	\frc1{L}\int_{I} \dd x\,\mathtt P_k[u](x)
	-
	\frc1L\sum_{i:x_i\in I}\chi_i^{2k+1}
	\Bigg)
	=0,\quad k\in\Z_+.
\eeq
\end{theorem}
We provide the proof of Theorem \ref{theofc} in Section \ref{ssectprooffc}. 

Finally, Theorem \ref{theofc} can be re-expressed, instead, as a more standard weak limit, under slowly-varying Schwartz functions. In the following, the natural ``Euler'' scale is the choice $\Lambda=1$, if it is available given the exponents in the assumptions' data.

\begin{theorem}[Weak limit of conserved densities]\label{theoweak} Denote
\beq
	u = u_{\bs\chi,\bs y},\quad
	\bs x= \bs x(\bs\chi,\bs y).
\eeq
Let $k\in\Z_+$, and, for the polynomial \eqref{conspoly},
\beq
	r = \deg(P_k(u_0,u_1,\ldots)).
\eeq
Let
\beq\label{condLambdacombined}
	\Lambda>(1+2r)\alpha+(2+10r)\beta+2r\mu+(2+4r)\omega.
\eeq
Then, for every Schwartz function $f:\R\to\R$ we have, uniformly on sequences \eqref{seq4} with fixed assumptions' data,
\beq\label{weaklimitmain}
	\lim_{N\to\infty}
	\Bigg(\int \dd x\,f(x)\,\mathtt P_k[u](N^\Lambda x)
	-
	\frc1{N^\Lambda}\sum_{i=1}^N f(x_i/N^\Lambda)\chi_i^{2k+1}
	\Bigg)=0.
\eeq
\end{theorem}
The proof is provided in Section \ref{ssectweak}.

\section{Discussion}\label{sectdisc}

We now heuristically discuss a number of points concerning the above results.


\subsection{Locally contracted vectors and the change of metric}\label{ssectdiscmetric}

A crucial notion in our results is that of locally contracted vectors, Definition \ref{definatmag}. This takes away the interaction effects of particles lying between an observation point $x_*$ and the position $x_i$ of the quasiparticle. This can be interpreted as ``changing the metric'' around $x_*$, with the new positions being more accurate quasiparticle positions the nearer they are to $x_*$. See Figure \ref{figmag}.
\begin{figure}
\bc\includegraphics[width=10cm]{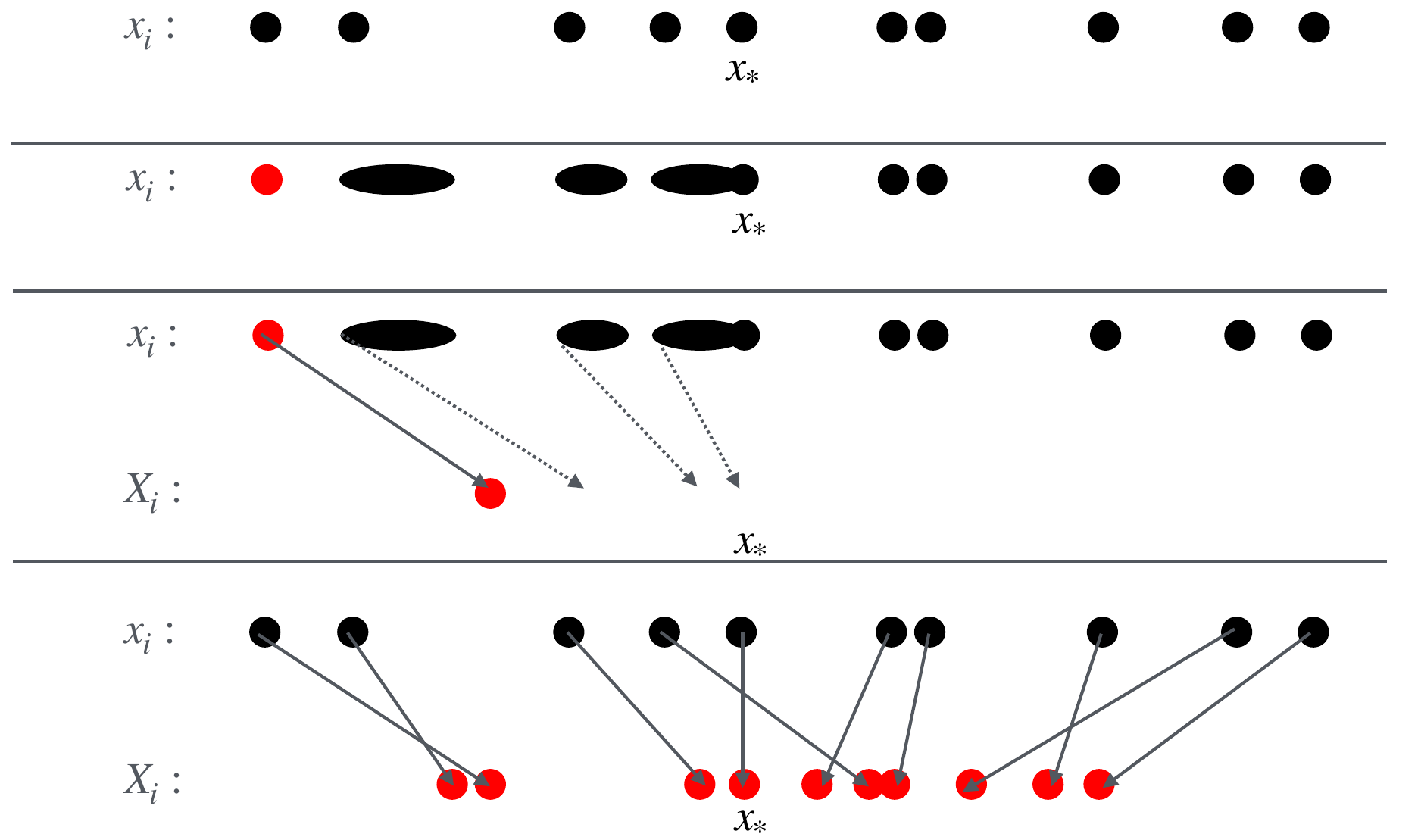}
\ec
\caption{Cartoon illustration of how locally contracted positions $X_i$'s relate to quasiparticle positions $x_i$'s from the SBE \eqref{bethe}. The case illustrated applies to the {\tt rightward} Algorithm \ref{rightward}, as contraction of solitons to the right of $x_*$ do not go beyond $x_*$.}
\label{figmag}
\end{figure}
One can gain intuition into this by informally taking the large-scale limit, where the empirical density becomes continuous:
\beq
	\sum_{i=1}^N \delta(\chi-\chi_i)\delta(x-x_i) \to \rho(\chi,x).
\eeq
Then from \eqref{Xnatural} we have (here $\varphi(\chi,\chi') = \frc1\chi\log\frc{|\chi-\chi'|}{\chi+\chi'}$)
\beq
	X_i = x_i + \int_{x_*}^{x_i} \dd x'\int_0^\infty \dd\chi'\,\rho(\chi',x')\varphi(\chi_i,\chi').
\eeq
Taking small variations, this suggest to define the $\chi$-dependent metric change as
\beq\label{dX}
	\dd X(\chi,x) = \Big(1 + \int_0^\infty \dd \chi'\,\rho(\chi',x)\varphi(\chi,\chi')\Big)\dd x.
\eeq
The quantity in the parenthesis is the ``density of space'', or total density (see e.g.~\cite{doyon2018geometric,doyon2020lecture}) 
\beq
	\rho_{\rm tot}(\chi,x)=1 + \int_0^\infty \dd \chi'\,\rho(\chi',x)\varphi(\chi,\chi'),
\eeq
which represents the space left for quasiparticle to move in, once their effective size, due to interaction, has been taken out. The corresponding coordinate function $X(\chi,x)$ satisfies the observation-point condition
\beq\label{Xxstar}
	X(\chi,x_*) = x_*.
\eeq
Note that the transformation \eqref{dX} is known from GHD \cite{doyon2018geometric,doyon2020lecture,hubner2024new}, and can be used to trivialises the GHD equation.

Crucially, it is known from soliton-gas theory that the space of allowed density of states $\rho(\chi,x)$ are those which guarantee that $\rho_{\rm tot}(\chi,x)\geq 0$. In the homogeneous case $\rho(\chi,x) = \rho(\chi)$, choosing a subset of $\R$ for the support of $\rho(\chi)$ (typically a union of disjoint intervals), and setting $\rho_{\rm tot}(\chi)= 0$, gives a soliton condensate \cite{el2021soliton,suret2024soliton}. Therefore, the condition which we referred to as the gas being away from the condensate regime, \eqref{c4intro}, indeed agrees with this physics: going back to the generically inhomogeneous situation, uniformly away from a condensate we have $\rho_{\rm tot}(\chi,x)\in(0,1]$ uniformly on $\chi,x$, hence
\beq
	\sum_{j:x_j\in[x_*,x_i)}|\varphi_{ij}|
	\to
	\Bigg|\int_{x_*}^{x_i}\dd x'\int_0^\infty \dd\chi' \rho(\chi',x')\varphi(\chi_i,\chi')\Bigg|
	=
	\Bigg|\int_{x_*}^{x_i}\dd x'(1-\rho_{\rm tot}(\chi',x'))\Bigg|
	\leq
	(1-a)|x_i-x_*|
\eeq
where
\beq
	a = \inf_{\chi>0,x\in\R}\rho_{\rm tot}(\chi,x)\in(0,1].
\eeq
In particular, {\em away from condensates, the change of metric from $x$ to $X$ is well behaved}: we have $|X(\chi,x)-x_*| \geq a(\chi)|x-x_*|$ for $a (\chi):= \inf_{x\in\R}\rho_{\rm tot}(\chi,x)> 0$ (and $\sgn(X(\chi,x)-x_*) = \sgn(x-x_*)$), giving
\beq\label{ratioDelta}
	\frc{\Delta x(\Delta X)}{\Delta X} \leq \frc{2}{a}.
\eeq
Our more general condition Assumption \ref{assimp} allows for a soliton gas that approaches the condensate as $N\to\infty$, as determined by the exponent $\omega$.

\subsection{Meaning of the other conditions on spectral and impact parameters}\label{sssectcond}

Above, we have discussed Assumptions \ref{assimp} and \ref{asscont}. Now consider Assumptions \ref{asssp} and \ref{assde}.
What do they mean in terms of the solitons involved?

The minimal value $\chi_*$ in conditions \eqref{c1} is natural. Indeed, as $\chi_i$ decreases, the one-soliton solution \eqref{1sol} becomes supported on a larger and larger domain, and is not ``localisable'' anymore (cannot be given a meaningful position): it cannot be given a specific position in a local cell of length $L = N^{\gamma}$. It is likely that the theorem can be generalised to account for a power law, $\min_i(\chi_i) = N^{-\upsilon},\,\upsilon\geq0$, we leave this for future works.

The minimal value of differences $|\chi_i-\chi_j|$ in \eqref{c1} is involved in the maximal size of the scattering shift two solitons can be subjected to. By \eqref{varphi}, the stretched exponential indicates that the shift cannot be larger than a power law. Again, this guarantees that, even throughout interaction events, solitons can be located.

The maximal value of $\chi_i$'s in \eqref{c1} is also involved in the maximal size of the scattering shift two solitons can be subjected to, however here it is simply a power law, this is a subleading effect compared to the minimal value of differences $|\chi_i-\chi_j|$. For $\chi_i$ large, the one-soliton solution \eqref{1sol} becomes very peaked, hence this is a regularity condition, and used to bound derivatives of the KdV field.

The density requirement \eqref{c2}
is also natural. Indeed, for a density that is too high, or a too large number of solitons within a region, there are too many scattering shifts for solitons to be locatable. The density $\rho_{\bs D}(d)$ for a cell of radius $d$ is with respect to contracted distances, where a metric deformation has occurred. In accordance to the discussion above, a positive exponent $\nu>0$ would correspond to a soliton gas that approaches a condensate, around the observation point $x_*$, and thus away from condensates, we expect $\nu=0$.

The conditions \eqref{c1} on $\bs\chi$ are simple to verify in specific examples. Conditions \eqref{c2}, \eqref{c4}, \eqref{c5} are more indirect as they necessitate the evaluation of quasiparticle positions and their contractions.

Yet, all conditions are expected to be valid in the crucial application of finite-density soliton gases with random spectral and impact parameters, as we now explain.

\subsection{How the assumptions hold for the KdV soliton gas}\label{sssectsg}

It would be good to verify that our assumptions are expected to hold in the most interesting application of our results, that of ``random'' soliton gases, at least at most macroscopic times, besides the dilute case of Sec.~\ref{ssectdilute}. A full rigorous derivation will be provided in a coming work; we give here arguments that suggest that this is the case. These arguments are expressed for the homogeneous soliton gas, as these are simplest to describe, however similar arguments hold for inhomogeneous soliton gases as well. Here, we only discuss Assumptions \ref{asssp} and \ref{assde}, but similar arguments can be made for Assumptions \ref{assimp}, \ref{asscont} following the intuition from Sec.~\ref{ssectdiscmetric}.

A homogeneous KdV soliton gas may be obtained by generating impact parameters $\bs y$ as $N$ i.i.d.~variables uniformly distributed on $[-\ell/2,\ell/2]$, and spectral parameters $\bs\chi$ as $N$ i.i.d.~variables uniformly distributed on some subset of $\R_+$, re-labelling them so that $\bs\chi\in\R_{+\uparrow}$, and taking $\ell\propto N$ and $N\to\infty$. The interest is in the statistics of $u_{\bs\chi,\bs y}(x)$ on finite regions, or on regions that grow with $N$ weakly enough, $x\in[-L/2,L/2]$, $L\ll N$. For inhomogeneous gases, we could consider a Poisson point process for $(y_i,\chi_i)\in\R^2$ with intensity $P(\chi,y)$, or we could distribute $y_i$'s and $\chi_i$'s according to a measure such as that explained in \cite{doyon2026generalised} (see \eqref{measure} below): with a Boltzmann weight of the form $\re^{-\sum_{i=1}^N \beta(x_i(\bs\chi,\bs y)/\ell,\chi_i)}$ for some function $\beta(x,\chi)$. Developing these ideas is left for future works.

Let us argue that for a homogeneous soliton gas, in the limit $N\to\infty$ the conditions from Assumptions \ref{asssp} and \ref{assde} are almost surely satisfied.

First, we may distribute $\chi$ on $[\chi_*,C]$ for some finite nonzero $0<\chi_*<C$. Then, the first and last conditions of \eqref{c1} are satisfied, with $\beta=0$.

Second, with these finite lower and upper bounds, in the above construction the typical values of $\chi_{i+1}-\chi_i$ scale as $1/N$, which is indeed much greater than the stretched exponential form in \eqref{c1}. But, as $N\to\infty$, there may be pairs of solitons whose spectral parameters get very near to each other, hence we need to control these atypical events. As $\Delta _i := N(\chi_{i+1}-\chi_i)$ are, in this case, i.i.d., exponentially distributed variables in the limit $N\to\infty$, with some mean $1/\lambda$, the probability that $\Delta_i\geq \Delta$ is
\beq
	\int_\Delta^\infty \dd w\,\lambda \re^{-\lambda w} = \re^{-\lambda \Delta}.
\eeq
Therefore, for $N$ large the minimum $\min_i \Delta_i =  N\min_{i\neq j}|\chi_i-\chi_j|$ has cumulative probability distribution
\beq
	P\big(N\min_{i\neq j}|\chi_i-\chi_j| < \Delta\big) \sim 1-\re^{-N\lambda \Delta}
\eeq
so that we obtain, for every $\alpha>0$,
\beq
	P\big(\min_{i\neq j}|\chi_i-\chi_j|\geq \re^{-AN^{\alpha/2}}\big) \sim \re^{-N^2\lambda\, \re^{-AN^{\alpha/2}}} \to 1\quad (N\to\infty).
\eeq
That is, in the limit $N\to\infty$, the second condition of \eqref{c1} is almost surely satisfied. Note that we could have replaced the stretched exponential by a power law $N^{-\alpha}$ for $\alpha>2$.

Third, according to the discussion of Section \ref{ssectdiscmetric}, away from a condensate the contracted distances, for any observation point $x_*$, are well distributed. Therefore, again for typical configurations and positions, \eqref{c2} will be satisfied with $\mu=\nu=0$. But again, we need to discuss the atypical configurations. Let us denote by $\rho$ the soliton density in the contracted space in the limit $N\to\infty$. Assuming that a large-deviation principle holds, there is some large-deviation function $I(\Delta\rho)\geq 0,\,I(0)=0$ with $I(\Delta\rho\to\infty)$ quickly enough as $|\Delta\rho|\to\infty$ such that the following probability density is
\beq
	P(\rho_{\bs D}(d) -\rho = \Delta\rho) \sim \re^{-d I(\Delta \rho)}.
\eeq
Therefore taking $D>\rho$
\beqa
	\lefteqn{P(\rho_{\bs D}(d)> D\ \mbox{for some}\  d\geq N^{\alpha/2}) }
	&&\n
	&\lesssim&
	\sum_{d=N^{\alpha/2}}^{\infty}
	\int_{D-\rho}^\infty \dd\Delta\rho\,
	\re^{-d I(\Delta \rho)}
	=
	\int_{D-\rho}^\infty \dd\Delta\rho\,
	\sum_{d=N^{\alpha/2}}^{\infty}
	\re^{-d I(\Delta \rho)}
	=
	\int_{D-\rho}^\infty \dd\Delta\rho\,
	\frc{\re^{-N^{\alpha/2}I(\Delta\rho)}}{1-\re^{I(\Delta\rho)}}\n
	&\to& 0 \quad (N\to\infty)
\eeqa
and, as $\alpha>0$, we see that \eqref{c2} is satisfied almost surely as $N\to\infty$, with $B=1$ and $\mu=\nu=0$.

A similar analysis can be done for $\Delta x(\Delta X)/\Delta X$ where we would find \eqref{ratioDelta} almost surely as $N\to\infty$.

Of course, determining that a large deviation principle holds would require more analysis; but results similar to those above should follow from our explicit Algorithm \ref{leftanchor} along with properties of the variables $\bs\chi,\,\bs y$, and in particular the central limit theorem, as is done in the case of hard rods (see e.g.~\cite{spohn2012large}).

%
Having argued that the conditions should apply in soliton gases, we are now ready to heuristically derive, from our results, the kinetic equation of soliton gases.

\subsection{The soliton gas kinetic equation}\label{ssectghd}

The most interesting application of our theorem is to the kinetic equation /  GHD equation for the KdV soliton gas \cite{suret2024soliton}
\beq\label{ghd}
	\p_t \rho(\chi,x,t) + \p_x \big(v^{\rm eff}(\chi,x,t) \rho(\chi,x,t)\big) = 0
\eeq
where the effective velocity solves the linear integral equation
\beq\label{veff}
	v^{\rm eff}(\chi,x,t)
	=
	4\chi^2
	+
	\int_{\chi_*}^\infty \dd \chi'\,\rho(\chi',x,t) \varphi(\chi,\chi')
	(v^{\rm eff}(\chi',x,t)-v^{\rm eff}(\chi,x,t))
\eeq
and $\varphi(\chi,\chi') = \frc1{\chi}\log\frc{|\chi-\chi'|}{\chi+\chi'}$.

We now provide a non-rigorous derivation of \eqref{ghd}, \eqref{veff} for the Euler-scale weak limit of the empirical density
\beq\label{rhowlim}
	\rho(\chi,x,t) = \wlimu{N\to\infty}\frc1N
	\sum_{i=1}^N \delta(\chi-\chi_i)\delta(x-x_i^{(Nt)}/N)
\eeq
for quasiparticle positions $x_i^{(Nt)}$ evolved to time $Nt$ -- here and below we omit their dependence on $\bs\chi,\,\bs y$. For this purpose, we  follow the steps of \cite{doyon2024new}, which used the SBE \eqref{bethet}:
 \beq\label{betheDeltat}
	y_i +4\chi_i^2 t = x_i^{(t)}+ \frc12 \sum_{j=1\atop j\neq i}^N \sgn(x_i^{(t)}-x_j^{(t)})
	\varphi_{ij}
\eeq

\medskip
\noindent {\em Derivation.}
Dividing by $N$, changing $t\to Nt$ and denoting $\b x_i^{(t)} = x_i^{(Nt)}/N$,
\beq\label{betheDeltat2}
	y_i/N +4\chi_i^2 t = \b x_i^{(t)}+ \frc1{2N} \sum_{j=1\atop j\neq i}^N \sgn(\b x_i^{(t)}-\b x_j^{(t)})\,
	\varphi(\chi_i,\chi_j).
\eeq
Taking the time derivative, we assume that we can write
\beq\label{dotveff}
	\dot{\b x}_i^{(t)} = v^{\rm eff}(\chi_i,\b x_i^{(t)},t)
\eeq
for all $i$'s and for some function $v^{\rm eff}(\chi,x,t)$:
\beqa
	4\chi_i^2 &=& \dot {\b x}_i^{(t)}+ \frc1{2N} \sum_{j=1\atop j\neq i}^N (\dot {\b x}_i^{(t)}-\dot{\b x}_j^{(t)})\delta(\b x_i^{(t)}-\b x_j^{(t)})\,
	\varphi(\chi_i,\chi_j)\n
	&=& v^{\rm eff}(\chi_i,\b x_i^{(t)},t)\n
	&& +\, \frc12 \int_{-\infty}^\infty \dd x'\int_{\chi_*}^\infty\dd \chi'\,\rho(\chi',x',t)
	(v^{\rm eff}(\chi_i,\b x_i^{(t)},t)-v^{\rm eff}(\chi',x',t))\delta(\b x_i^{(t)}-x')\,
	\varphi(\chi_i,\chi')\n
	&=& v^{\rm eff}(\chi_i,\b x_i^{(t)},t)+ \frc12 \int_{\chi_*}^\infty\dd \chi'\,\rho(\chi',\b x_i^{(t)},t)
	(v^{\rm eff}(\chi_i,\b x_i^{(t)},t)-v^{\rm eff}(\chi',\b x_i^{(t)},t))\,
	\varphi(\chi_i,\chi')\no
\eeqa
which, taken for generic $\chi_i\to \chi$ and $\b x_i^{(t)}\to x$, is \eqref{veff}. Taking the time derivative of \eqref{rhowlim} and using again \eqref{dotveff}, we obtain \eqref{ghd}.
\eproof

\section{Proofs}\label{secproofs}

\subsection{Representations of the tau function}\label{ssecttaurep}

We show Lemma \ref{lemtaudecomp}.

\proof
Recall that $\Mat(N)$ is the space of real $N$ by $N$ matrices. For $s\in\dbra 1,N\dket$, we also denote by $\Mat(s)$ the space of real $|s|$ by $|s|$ matrices with the convention that indices run in the set $s$, that is $[A_{ij}]_{i\in s,j\in s}\in \Mat(s)$. For convenience we also denote for $i\in\dbra 1,N\dket$
\beq
	\b\Psi_i(x,t) = \Psi_i(x,t)^{-1}
\eeq
where $\Psi(x,t)$ is the diagonal matrix defined in \eqref{Psiomega}.

Recall the matrix $\omega$ from \eqref{Psiomega}. 
Using the fact that it is a Cauchy-like matrix, its determinant can be evaluated as (see \eqref{Sp})
\beq\label{detomega}
	\det\omega
	=
	\prod_{j<i} S_{ij}^2=S_{\dbra 1,N\dket}.
\eeq
It also has inverse
\beq\label{invomega}
	\omega^{-1} = \gamma\omega\gamma,\quad
	\gamma = \diag(\gamma_i)_{i\in\dbra 1,N\dket},\quad
	 \gamma_i = \prod_{k\neq i} S_{ik}^{-1}.
\eeq
For every $s\subseteq\dbra1,N\dket$ with $\b s = \dbra 1,N\dket\setminus s$, and every diagonal matrix $D\in\Mat(N)$, we define
\beq\label{gammaPsis}
	\gamma_s = \diag\Big(\prod_{k\in s,\,k\neq i}S_{ik}^{-1}\Big)_{i\in s}
	\in\Mat(s),
	\quad
	D_s = \diag(D_i\chi_s(i) + \chi_{\b s}(i))_{i\in\dbra 1,N\dket}\in
	\Mat(N)
\eeq
and similarly for $\b\Psi_s$, where $\chi_s$ is the indicator function of the set $s$. For a matrix $A\in\Mat(N)$, we write its restriction onto the set $s$ as
\beq
	A_{|s} = [A_{ij}]_{i\in s,j\in s}
	\in\Mat(s).
\eeq
Using \eqref{invomega} we have
\beq\label{omegainvs}
	(\omega_{|s})^{-1} = \gamma_s \omega_{|s}\gamma_s.
\eeq

Let $s\subseteq\dbra1,N\dket$. By expanding the following determinant in powers of the independent variables $\Psi_i^2$ ($i\in s$), let us define the coefficients $c_p^s$, rational functions of $\chi_i$'s, as
\beq\label{defc}
	\det(\Psi^2_{|s} + \omega_{|s}) = \sum_{p\subseteq s}
	c_{p}^{s}
	\prod_{i\in p}\Psi_i^2.
\eeq
We note that
\beq\label{cinit}
	c_s^s = 1,\quad c_\emptyset^s = \det\big(\omega|_s\big).
\eeq

Again let $s\subseteq\dbra1,N\dket$, and let $r\subseteq s,\,\b r = s\setminus r$. Writing
\beq
	\b\Psi^2_{|s} = (\b\Psi^2_{|s})_r (\b\Psi^2_{|s})_{\b r}
\eeq
and using $\det((\b\Psi_{|s}^2)_{\b r}) = \det(\b\Psi_{|\b r}^2)$ and $\det(\Psi_{|s}^2) = \det(\Psi_{|r}^2)\det(\Psi_{|\b r}^2)$, we have
\beqa
	\det(\Psi^2_{|s} + \omega_{|s}) &=&
	\det\Big(\b\Psi^2_{|s} + (\omega_{|s})^{-1}\Big)\det(\omega_{|s})\det(\Psi^2_{|s})
	\n &=&
	\det\Big((\b\Psi^2_{|s})_{ r} + (\Psi_{|s})_{\b r}(\omega_{|s})^{-1}(\Psi_{|s})_{\b r}\Big)\det(\omega_{|s})\det(\Psi^2_{|r}).
\eeqa
Therefore, setting $\Psi_i=0\,\forall \,i\in \b r$ we have
\beqa
	\lim_{\Psi_i=0\,\forall \,i\in \b r}\det(\Psi^2_{|s} + \omega_{|s})
	&=&
	\det\Big(\b\Psi^2_{| r} + (\omega_{|s})^{-1}_{\quad| r}\Big)\det(\omega_{|s})\det(\Psi^2_{| r})\n
	&=&
	\det\Big(\Psi^2_{|  r} + ((\omega_{|s})^{-1}_{\quad| r})^{-1}\Big)\det(\omega_{|s})\det((\omega_{|s})^{-1}_{\quad| r}).
\eeqa
Using \eqref{omegainvs} we now evaluate
\beq\label{ghsa}
	(\omega_{|s})^{-1}_{\quad | r} =
	(\gamma_s\omega_{|s}\gamma_s)_{| r}
	=
	\gamma_{s| r}\,\omega_{| r}\,\gamma_{s|{ r}}
\eeq
from which we get
\beq
	((\omega_{|s})^{-1}_{\quad| r})^{-1} 
	=
	\gamma_{s|{r}}^{-1}\,\gamma_{ r}\omega_{| r}\gamma_{r}\,
	\gamma_{s|r}^{-1}.
\eeq
Therefore
\beq
	\lim_{\Psi_i=0\,\forall \,i\in \b r}\det(\Psi^2_{|s} + \omega_{|s})
	=
	\det\Big(\gamma_{s| r}^{2}\gamma_{ r}^{-2} \Psi^2_{| r} + \omega_{| r}\Big)\det(\omega_{|s})\det((\omega_{|s})^{-1}_{\quad|r})
	\det(\gamma_{s| r}^{-2}\gamma_{ r}^2)
\eeq
and, using \eqref{defc}
\beq\label{mrel}
	\sum_{p\subseteq r}c_{p}^{s}
	\prod_{i\in p}\Psi_i^2
	= \sum_{p\subseteq r}
	c_p^{ r} \prod_{i\in p} (\gamma_{s| r}^{2}\gamma_{ r}^{-2})_i
	\det(\omega_{|s})\det((\omega_{|s})^{-1}_{\quad| r})
	\det(\gamma_{s| r}^{-2}\gamma_{ r}^2)
	\,\prod_{i\in p}\Psi_i^2.
\eeq
Using \eqref{ghsa}, \eqref{detomega} and \eqref{gammaPsis}  we find
\beqa
	\gamma_{s|r}^{2}\gamma_{r}^{-2} &=&
	\diag\Big(
	\prod_{k\in \b r}S_{ik}^{-2}
	\Big)_{i\in r}
	\\
	\det(\gamma_{s| r}^{-2}\gamma_{ r}^{2}) &=&
	\prod_{i\in r,\,k\in \b r}S_{ik}^{2}
	\\
	\det((\omega_{|s})^{-1}_{\quad | r}) &=&
	\prod_{i\in r, k\in s\atop i\neq k}S_{ik}^{-2}
	\prod_{i,j\in  r\atop i<j} S_{ij}^2
	\\
	\det(\omega_{|s}) &=&
	\prod_{i,j\in s\atop i<j} S_{ij}^2.
\eeqa
Using these expressions, \eqref{mrel} gives the following recursion relation, valid for all $p\subseteq  r\subseteq s\subseteq\dbra 1,N\dket$:
\beq\label{crecur}
	c_{p}^{s}
	= c_p^{ r}
	\prod_{i\in  r\setminus p,\,j\in s\setminus  r} S_{ij}^{2}
	\prod_{i,j\in s\setminus  r\atop i<j} S_{ij}^2.
\eeq
The case $p=r$, along with the first relation of \eqref{cinit}, gives
\beq\label{csol}
	c_p^s = \prod_{i,j\in s\setminus  p\atop i<j} S_{ij}^2 = S_{s\setminus p},
\eeq
which one can check is fulfils \eqref{crecur} for all $p\subseteq  r\subseteq s\subseteq\dbra 1,N\dket$ and  the second relation in \eqref{cinit}.

Recall the equivalence relation \eqref{equiv}. According to this, setting $\b s = \dbra 1,N\dket\setminus s$, we have
\beq
	\tau = \det(\Psi^2 + \omega) \equiv S_s^{-1}\det(\Psi^2 + \omega)\det\b\Psi_{\b s}^2
	=
	\sum_{p\subseteq s,\,q\subseteq\b s}
	S_s^{-1}c_{p\cup (\b s\setminus q)}^{\dbra 1,N\dket}
	\prod_{i\in p}\Psi_i^2 \prod_{j\in q}\b\Psi_j^2.
\eeq
Along with \eqref{csol}, this shows \eqref{taudecomp} with \eqref{Sp}. 

Choose a subset $s\subseteq\dbra1,N\dket$ and consider the representation \eqref{taudecomp}. Write (see Eq.~\eqref{xist})
\beq
	a_i^{s} = y_i + \frc12\sum_{j\neq i} \sgn_s(j)\varphi_{ij}.
\eeq
Then, setting $r=p\cup q$ in \eqref{taudecomp} and recalling \eqref{Psiomega}, \eqref{ay} and \eqref{varphi}, we have
\beqa
	\tau_s(x,0) &=&
	\sum_{p\subseteq s,\,q\subseteq \b s}
	\frc{S_{(s\setminus p)\cup q}}{S_s}
	\exp \Big[\sum_{i\in r} 2\chi_i\sgn_s(i)(x-y_i)\Big]
	\prod_{i\in r,\,j\in\dbra1,N\dket\atop i\neq j}|S_{ij}|^{\sgn_s(i)}
	\n
	&=&
	\sum_{p\subseteq s,\,q\subseteq \b s}
	\frc{S_{(s\setminus p)\cup q}}{S_s}
	\exp \Big[\sum_{i\in r} 2\chi_i\sgn_s(i)(x-a_i^s)\Big]
	\prod_{i\in r,\,j\in\dbra1,N\dket\atop i\neq j}|S_{ij}|^{\sgn_s(i)(1+\sgn_s(j))}
	\n
	&=&
	\sum_{p\subseteq s,\,q\subseteq \b s}
	\frc{S_{(s\setminus p)\cup q}}{S_s}
	\exp \Big[\sum_{i\in r} 2\chi_i\sgn_s(i)(x-a_i^s)\Big]
	\prod_{i\in r,\,j\in s\atop i\neq j}|S_{ij}|^{2\sgn_s(i)}.
\eeqa
Denoting $s_1= p,\,s_2 = s\setminus p,\,s_3=q,\,s_4=\b s\setminus q$, and $\prod_{i\in s_a,\,j\in s_b\atop i\neq j}|S_{ij}| = A_{ab}$, we have
\beq
	\prod_{i\in r,\,j\in s\atop i\neq j}|S_{ij}|^{2\sgn_s(i)}
	=A_{11}^2A_{12}^2A_{31}^{-2} A_{32}^{-2}
\eeq
and
\beq
	S_{(s\setminus p)\cup q}
	=
	\prod_{i,j\in (s\setminus p)\cup q\atop i\neq j}
	|S_{ij}|
	= A_{22}A_{23}A_{32}A_{33} = A_{22}A_{32}^2A_{33}
\eeq
giving
\beq
	S_{(s\setminus p)\cup q}
	\prod_{i\in r,\,j\in s\atop i\neq j}|S_{ij}|^{2\sgn_s(i)}
	= A_{11} A_{12}^2 A_{22}\ A_{11} A_{33} A_{31}^{-2}
	= S_s \prod_{i,j\in r\atop i\neq j}
	|S_{ij}|^{\sgn_s(i)\sgn_s(j)}.
\eeq
Thus we find \eqref{tauresult}. Eq.~\eqref{tauresult2} is obtained simply by absorbing the second exponential factor within the definition of $a_i^{s,r}$.
\eproof

\subsection{Semiclassical Bethe equations}\label{ssectbetheproof}

In this subsection we prove Theorem \ref{theosbe} and Lemma \ref{lemrelationbethe}.

First, in order to prove the lemma, we start with an intermediate lemma.

\begin{lemma}\label{lembetheintermediate} Let $N\in\N,\,\bs X\in\R^N,\,x_*\in\R$ and assume that $|X_i-x_*|\geq \varep\,\forall \,i$. There exists a solution $\bs x\in\Rne^N$ to the system of equations
\beq\label{xXinterm}
	x_i = X_i - \sgn(X_i-x_*)\sum_{j:x_j\in [x_*,x_i)}\varphi_{ij}
\eeq
with
\beq\label{signproperty}
	\sgn(X_i-x_*) = \sgn(x_i-x_*)\quad \forall\;i\in\dbra 1,N\dket.
\eeq
A solution is given by
\beq\label{solulemma}
	(x_i)_{i\in s_+} = {\tt Rightward}(s_+,(X_i)_{i\in s_+}),\quad
	(x_i)_{i\in s_-} = {\tt Leftward}(s_-,(X_i)_{i\in s_-})
\eeq
where
\beq
	s_\pm = \{i\in\dbra 1,N\dket:X_i\gtrless x_*\}.
\eeq
It has the following properties: for all $i\in\dbra 1,N\dket$,
\beq\label{xixstartlemma}
	|x_i-x_*|\geq -\sum_{j:x_j\in[x_*,x_i)} \varphi_{ij}+\varep
\eeq
and
\beq\label{xixlemma}
	|x_i-a|\geq -\sum_{j:x_j\in[a,x_i)} \varphi_{ij} \quad \forall\  a\in[x_*,x_i].
\eeq
\end{lemma}
\proof 
Execute the following algorithm.
\bi
\item[\tt1-] If $s_+$ is non-empty, set $i_1^+ \in \argmin_{i\in s_+}X_i$ and let
\beq
	x_{i_1^+} = X_{i_1^+}.
\eeq
Otherwise, pass to {\tt 3-}.
\item[\tt2-]
Sequentially for $n=2,3,\ldots$, and as long as $\{i_1^+,\ldots,i_{n-1}^+\}\neq s_+$: set
\beq\label{inpmin}
	i_n^+ \in \argmin_{i\in s_+\setminus\{i_1^+,\ldots,i_{n-1}^+\}}
	\Big(X_i - \sum_{m=1}^{n-1} \varphi_{ii_m^+}\Big),\quad
	x_{i_n^+} = X_{i_n^+}- \sum_{m=1}^{n-1} \varphi_{i_n^+i_m^+}.
\eeq
\item[\tt3-] If $s_-$ is non-empty, set $i_1^- \in \argmax_{i\in s_-}X_i$ and let
\beq
	x_{i_1^-} = X_{i_1^-}.
\eeq
Otherwise, stop here.
\item[\tt4-] Sequentially for $n=2,3,\ldots$, and as long as $\{i_1^-,\ldots,i_{n-1}^-\}\neq s_-$: set
\beq
	i_n^- = \argmax_{i\in s_-\setminus\{i_1^-,\ldots,i_{n-1}^-\}}
	\Big(X_i + \sum_{m=1}^{n-1} \varphi_{ii_m^-}\Big),\quad
	x_{i_n^-} = X_{i_n^-}+ \sum_{m=1}^{n-1} \varphi_{i_n^-i_m^-}.
\eeq
\ei
Because $\varphi_{ij}<0$, then $x_i\geq X_i\geq x_*+\varep$ for all $i\in s_+$, and $x_i \leq X_{i}\leq x_*-\varep$ for all $i\in s_-$. Hence \eqref{signproperty} holds. Also, for all $n,n'\in\N$ with $n'> n$,
\beq\label{hduai1}
	x_{i_{n'}^+}-x_{i_n^+} = \Big(X_{i_{n'}^+}-\sum_{m=1}^{n-1}\varphi_{i_{n'}^+i_m^+}\Big) -
	\Big(X_{i_{n}^+}-\sum_{m=1}^{n-1}\varphi_{i_{n}^+i_m^+}\Big)
	 - \sum_{m=n}^{n'-1} \varphi_{i_{n'}^+i_m^+}
	 \geq - \sum_{m=n}^{n'-1} \varphi_{i_{n'}^+i_m^+}>0
\eeq
because we take the minimum in \eqref{inpmin}.
Similarly
\beq\label{hduai2}
	x_{i_{n'}^-}-x_{i_n^-} = \Big(X_{i_{n'}^-}+\sum_{m=1}^{n-1}\varphi_{i_{n'}^-i_m^-}\Big) -
	\Big(X_{i_{n}^-}+\sum_{m=1}^{n-1}\varphi_{i_{n}^-i_m^-}\Big)
	 +\sum_{m=n}^{n'-1} \varphi_{i_{n'}^-i_m^-}
	 \leq \sum_{m=n}^{n'-1} \varphi_{i_{n'}^-i_m^-}<0.
\eeq
These orderings show that
\beq
	x_{i_n^+} = X_{i_n^+} - \sum_{m=1}^{n-1} \varphi_{i_n^+i_m^+}
	= X_{i_n^+} - \sgn(X_{i_n^+}-x_*)\sum_{j:x_j\in[x_*,x_{i_n^+})} \varphi_{i_n^+j}
\eeq
and
\beq
	x_{i_n^-} = X_{i_n^-} + \sum_{m=1}^{n-1} \varphi_{i_n^-i_m^-}
	= X_{i_n^-} - \sgn(X_{i_n^-}-x_*)\sum_{j:x_j\in(x_{i_n^-},x_*]} \varphi_{i_n^-j}
\eeq
This shows \eqref{xXinterm} and \eqref{xixstartlemma}. The inequality
\eqref{xixlemma} holds thanks to \eqref{hduai1} and \eqref{hduai2}, and \eqref{solulemma} holds because the algorithm above is exactly Algorithm \ref{rightward} followed by Algorithm \ref{leftward}.
%
\eproof

\medskip
\noindent {\em Proof of Lemma \ref{lemrelationbethe}}:
Consider $(1)\Rightarrow(2)$. The statement of Point (1) of Lemma \ref{lemrelationbethe} implies the assumptions of Lemma \ref{lembetheintermediate}. Consider the resulting $\bs x$ given by \eqref{solulemma}. By \eqref{xixstartlemma} this satisfies \eqref{xixstarpoint2} from Point (2). Note that by \eqref{specialrelationmain} and the fact that $x_i\neq x_*$ for every $i\in\dbra 1,N\dket$ we have
\beq\label{specialrelation}
	\sgn(x_i-x_*)\sum_{j:x_j\in[x_*,x_i)}\varphi_{ij}
	=
	\frc12\sum_{j=1\atop j\neq i}^N\Big(
	\sgn(x_j-x_*)\varphi_{ij} - \sgn(x_j-x_i)\varphi_{ij}\Big).
\eeq
Therefore, by \eqref{xXinterm}, \eqref{signproperty}, and \eqref{eqxymainXstr} (as part of the statement of Point (1)), we have, for all $i$,
\beqa
	x_i &=& X_i - \frc12\sum_{j=1\atop j\neq i}^N\Big(
	\sgn(x_j-x_*)\varphi_{ij} - \sgn(x_j-x_i)\varphi_{ij}\Big)
	\n
	&=& y_i +\frc12 \sum_{j=1\atop j\neq i}^N\sgn(x_j-x_i)\varphi_{ij},
	\label{hghgf}
\eeqa
whence \eqref{bethemain} holds as $\sgn(x_j-x_i) = - \sgn(x_i-x_j)$ because $\bs x\in\Rne^N$. This completes the proof of $(1)\Rightarrow(2)$, and Lemma \ref{lembetheintermediate} also gives \eqref{relationconstruction} and \eqref{Xabound}. Consider $(2)\Rightarrow(1)$. Set
\beq
	s_i = \sgn(x_i-x_*).
\eeq
By \eqref{xixstarpoint2}, for all $i$,
\beqa
	s_i &=& 
	\sgn\Big(x_i-x_*+\sgn(x_i-x_*)\sum_{j:x_j\in[x_*,x_i)} \varphi_{ij} \Big).
\eeqa
Further, by \eqref{bethemain} and \eqref{specialrelation},
\beqa
	\lefteqn{x_i-x_*+\sgn(x_i-x_*)\sum_{j:x_j\in[x_*,x_i)} \varphi_{ij} }&&\n
	&=& y_i-x_*+\frc12\sum_{j=1\atop j\neq i}^N \sgn(x_j-x_i)\varphi_{ij}+\sgn(x_i-x_*)\sum_{j:x_j\in[x_*,x_i)} \varphi_{ij}\n
	&=& y_i-x_*+\frc12\sum_{j=1\atop j\neq i}^N s_j\varphi_{ij}.
	\label{pqoeiu}
\eeqa
Therefore, we may set
\beq\label{dhahus}
	X_i = y_i + \frc12\sum_{j=1\atop j\neq i}^N s_j\varphi_{ij}
\eeq
and we find that
\beq
	\sgn(X_i-x_*) = s_i
\eeq
which implies
\beq
	X_i = y_i + \frc12\sum_{j=1\atop j\neq i}^N \sgn(X_j-x_*)\varphi_{ij}.
\eeq
Hence $\bs X\in\R^N$ solves \eqref{eqxymainXstr}. Further, again by \eqref{xixstarpoint2},
\beq
	\Big|x_i-x_*+\sgn(x_i-x_*)\sum_{j:x_j\in[x_*,x_i)} \varphi_{ij}\Big| \geq \varep
\eeq
and therefore using the result \eqref{pqoeiu}, we have
\beq
	|X_i -x_*|=
	\Big|y_i-x_*+\sum_{j=1\atop j\neq i}^N s_j\varphi_{ij}\Big|
	\geq\varep.
\eeq
This implies \eqref{xixstarpoint1} and $\sgn(X_j-x_*) = \sgn_\varep(X_j-x_*)$ and completes the proof of $(2)\Rightarrow(1)$.

Combining \eqref{dhahus} with \eqref{specialrelation} and \eqref{bethemain}, we obtain \eqref{Xxrelation}. Finally, we observe that if $\bs X$ satisfies Point 1, then the constructed $\bs x=\hat{\mathcal E}_{\bs \chi,x_*}(\bs X)$ satisfies $x_i =  X_i - \sgn(x_i-x_*)\sum_{j:x_j\in [x_*,x_i)}\varphi_{ij}$ (by \eqref{xXinterm} and \eqref{signproperty}). Hence $\bs X' = \hat{\mathcal C}_{\bs \chi,x_*}(\bs x)$ is given by $X_i' = x_i + \sgn(x_i-x_*)\sum_{j:x_j\in[x_*,x_i)}\varphi_{ij} = X_i$.
\eproof

\medskip

We note that if $\bs x$ satisfies Point (2), then we do not necessarily have $\hat{\mathcal E}_{\bs\chi,x_*}(\hat{\mathcal C}_{\bs\chi,x_*}(\bs x)) = \bs x$; this is due in part to the possibly many elements of $\argmin,\,\argmax$ used in Algorithms \ref{rightward}, \ref{leftward}, but most importantly to the possible existence of other solutions not given by Algorithms \ref{rightward}, \ref{leftward}.

This ambiguity is partially lifted by the set of relations \eqref{Xabound}, which is stronger than \eqref{xixstarpoint2}. That is, it guarantees that Algorithms \ref{rightward} (towards the right) and \ref{leftward} (towards the left) give the full set of possibilities. The lack of uniqueness is then only due to the possible multiplicity of $\argmin,\,\argmax$, and this can be avoided with the non-coincidence condition \eqref{noncoincidence}.

\begin{lemma}[Uniqueness of Rightward and Leftward solutions] \label{lemuniq} Let $N\in\N,\bs\chi\in\Rpu^N,\,\bs y\in\R^N$. Let $\bs x\in\Rne$ be a solution to \eqref{bethemain}. Choose $x_*\in\R\setminus \bs x$ and let $s_\pm = \{i:x_i\gtrless x_*\}$ and $\bs X = \h{\mathcal C}_{\bs\chi,x_*}(\bs x)$ (Eq.~\eqref{Xxrelation}).
If
\beq\label{condr}
	x_i-a\geq- \sum_{j:x_j\in[a,x_i)} \varphi_{ij} \quad \forall\; a\in\R,\,i\in\dbra 1,N\dket:\ x_*\leq a< x_i,
\eeq
then
\beq\label{right}
	(x_i)_{i\in s_+} \in {\tt Rightward}(s_+,(X_i)_{i\in s_+}).
\eeq
In particular, if $x_*<x^-(\bs\chi,\bs y)$, then
\beq\label{specialright}
	\bs x \in {\tt Rightward}(\dbra 1,N\dket,\bs X^-(\bs\chi,\bs y)).
\eeq
If
\beq\label{condl}
	x_i-a\leq\sum_{j:x_j\in[a,x_i)} \varphi_{ij} \quad \forall\; a\in\R,\,i\in\dbra 1,N\dket:\ x_i<a\leq x_*,
\eeq
then
\beq\label{left}
	(x_i)_{i\in s_-} \in {\tt Leftward}(s_-,(X_i)_{i\in s_-}).
\eeq
In particular, if $x_*>x^+(\bs\chi,\bs y)$, then
\beq\label{specialleft}
	\bs x \in {\tt Leftward}(\dbra 1,N\dket,\bs X^+(\bs\chi,\bs y)).
\eeq
If the non-coincidence condition \eqref{noncoincidence} holds, then in \eqref{right}, \eqref{specialright}, \eqref{left}, \eqref{specialleft}, the set is composed of a single element.
\end{lemma}
\proof We provide the proof of \eqref{left}, \eqref{specialleft}; for \eqref{right}, \eqref{specialright} the proof is similar. Is $s_-=\emptyset$ there there is nothing to prove, so we assume $s_-\neq\emptyset$. Let $(i_1,\ldots,i_N)$ be the sequence of indices such that $x_{i_1}>x_{i_2}>\ldots>x_{i_N}$. For all $k\in\dbra 0,N\dket$, define
\beq
	X_i^{(k)} := y_i
	+ \frc12
	\sum_{p=1\atop i_p\neq i}^k  \varphi_{i i_p}
	-\frc12\sum_{p=k+1\atop i_p\neq i}^N \varphi_{ii_p}.
\eeq
Then by \eqref{bethemain} and the stated ordering,
\beq\label{xikinterm}
	x_{i_{k}}
	= X_{i_{k}}^{(k-1)},\quad k=1,2,\ldots,N.
\eeq
Let $p\in\dbra 1,N\dket$ be the smallest integer such that $x_{i_p}<x_*$; it exists because $s_-\neq\emptyset$. Then $i\in s_-$ implies $i=i_k$ for some $k\geq p$. By \eqref{Xxrelation} and \eqref{xikinterm} we have, for all $k\geq p$,
\beq\label{Xikinterm}
	X_{i_k} = x_{i_k} - \sum_{q=p}^{k-1} \varphi_{i_ki_{q}} = X_{i_k}^{(p-1)}.
\eeq

Let $\bs x' \in {\tt Leftward}(s_-,(X_i)_{i\in s_-})$. We now proceed by recursion. Assume that the sequential ${\tt Leftward}$ construction, Algorithm \ref{leftward}, gives for the first $n$ steps the values $x_{i_k}'=x_{i_k}$, $k=p,\ldots,p+n-1$. We allow the case $n=0$, where no value is obtained, in which case the assumption is trivially satisfied (the starting point of the recursion). Then by \eqref{inpmaxmain} and \eqref{Xikinterm} we wish to evaluate
\beq
	\argmax_{k\in\dbra p+n,N\dket}
	\Big(X_{i_k} + \sum_{l=p}^{p+n-1} \varphi_{i_k i_l}\Big)=
	\argmax_{k\in\dbra p+n,N\dket}
	\Big(X_{i_k}^{(p+n-1)}\Big).
\eeq
For every $m\geq p+n+1$, by the assumption \eqref{condl} of the Lemma with $i=i_{m}$ and $a=x_{i_{p+n}}$,
\beq
	x_{i_m} - x_{i_{p+n}} \leq \sum_{k=p+n}^{m-1} \varphi_{i_mi_k}.
\eeq
Therefore we have, for every $m\geq p+n+1$,
\beq
	0 \geq x_{i_m}-x_{i_{p+n}} -\sum_{k=p+n}^{m-1} \varphi_{i_mi_k}
	=
	X_{i_m}^{(m-1)} - X_{i_{p+n}}^{(p+n-1)}
	-\sum_{k=p+n}^{m-1} \varphi_{i_mi_k}
	=
	X_{i_m}^{(p+n-1)} - X_{i_{p+n}}^{(p+n-1)}.
\eeq
Hence $i_{p+n}\in \argmax_{k\in\dbra p+n,N\dket}\Big(X_{i_k}^{(p+n-1)}\Big)$, so there is a choice such that $x_{i_k}'=x_{i_k}$, $k=p,\ldots,p+n$. This completes the recursion.

If $x_*>x^+(\bs\chi,\bs y)$ then by \eqref{sbeboundinv}, Eq.~\eqref{condl} implies \eqref{xixstarpoint2} for some $\varep>0$, and by Lemma \ref{lemrelationbethe}, $\bs X$ solves \eqref{eqxymainXstr}. As by Lemma \ref{lemdstr} this solution is unique and given by $\bs X = \bs X^+(\bs\chi,\bs y)$, the result \eqref{specialleft} follows.

Finally, under \eqref{noncoincidence}, then all $\argmin$, $\argmax$ are singletons.
\eproof

\medskip
\noindent {\em Proof of Theorem \ref{theosbe}}: By Lemma \ref{lemdstr}, we note that all $x_* \leq x_-(\bs\chi,\bs y)-\varep$ satisfy Point (1) of Lemma \ref{lemrelationbethe}. Therefore Point (2) holds. For these choices of $x_*$, we have $X_i = X_i^-(\bs\chi,\bs y)$ (Lemma \ref{lemdstr}), and  we may make the construction \eqref{relationconstruction} with $s_+ = \dbra 1,N\dket$, $s_-=\emptyset$. This gives $\bs x = {\tt Rightward}(\dbra 1,N\dket,\bs X^-(\bs\chi,\bs y))$. Further, this satisfies \eqref{Xabound}, which implies \eqref{distexist}. This construction gives $\min_{i\in\dbra 1,N\dket} x_i = \min_{i\in\dbra 1,N\dket} X_i^-(\bs\chi,\bs y) = x_i^-$, hence \eqref{miniexist} holds. A similar argument holds for $x_* = x_+(\bs\chi,\bs y)+\varep$, given the statements for the ${\tt Leftward}$ solution. Finally, Lemma \ref{lemuniq} provides the uniqueness. \eproof

\medskip

\subsection{Local form of the multi-soliton field}\label{proofmain}

We now prove simultaneously Theorems \ref{theo} and \ref{theomg}. For Theorem \ref{theo}, $D_i$'s are contracted distances and $\sigma_i=\sgn(x_i-x_*)$. For Theorem \ref{theomg}, we have soliton displacements $d_i$, and below we simply set
\beq
	D_i=|d_i|,\quad \sigma_i = \sgn(d_i)\quad\mbox{(Thm \ref{theomg})}.
\eeq
We are given the sequence \eqref{seq} for Theorem \ref{theo}, resp.~\eqref{seqmg} for Theorem \ref{theomg}. We have the related sequence $\b N\mapsto \bs D$, resp.~$\b N\mapsto \bs d$, which satisfy the assumptions of the theorem, within which the quantities $A,\,B,\,C,\,D,\,U,\,\chi_*,\,\alpha,\,\beta,\,\mu,\,\nu,\,\sigma$ are also fixed. For Theorem \ref{theo}, there is no Assumption \ref{assac}, so we set
\beq
	U=\sigma=0\quad\mbox{(Thm \ref{theo})}. 
\eeq

\proof
We consider the centred form, Lemma \ref{lemtaumain}, resp.~Lemma \ref{lemtaumainmg}, for the time-evolved tau function,
\beq
	\tau_{\bs\chi,\bs y^{(t)}} = \tau_{\bs\chi,\bs y + \bs v t}
\eeq
as per \eqref{yt}, choosing
\beq
	\Delta \bs y = \bs v t.
\eeq
We denote $\tau^\# = \tau^\#_{\bs\chi,\bs D,\bs v t,x_*}$, resp.~$\tau^\# = \tau^\#_{\bs\chi,\bs d,\bs v t,x_*}$, and hence we have
\beq\label{centred0}
	\tau^\#(x) = \sum_{r\subseteq \dbra1, N\dket} 
	\exp\Big[-\sum_{i\in r} 2\chi_i\big(
	D_i
	+\sigma_i(x_*-x+v_it + e_i)\big)+ R_r\Big]
\eeq
where 
\beq\label{eR0}
	e_i := \lt\{\ba{ll}
	0& \mbox{(Thm \ref{theo})}\z
	\frc12\sum_{ j\neq i} \big(\sgn(d_j) - \sgn_{\varep}(d_j)\big)\varphi_{ij}
	& \mbox{(Thm \ref{theomg})}
	\ea\rt.,\quad
	R_r := \sum_{i,j\in r,\,i\neq j}
	\sigma_i
	\sigma_j \log|S_{i,j}|.
\eeq

\subsubsection{Bounding $\chi_i e_i$}
We first bound the quantity $\chi_i e_i$ in \eqref{eR0}, which is involved in  Theorem \ref{theomg} (but not Theorem \ref{theo}). We make use of \eqref{sgnepgen}:
\beqa
	|\chi_i e_i|
	&\leq& \frc12\sum_{j\neq i}
	|\sgn(d_j) - \sgn_{\varep}(d_j)|\,|\chi_i\varphi_{ij}|\n
	&\leq& \max_{j\neq k}(|\chi_k\varphi_{kj}|)\;
	|\{j: |d_j|<\varep\}|.
\eeqa
From the 2nd and 3rd bounds in \eqref{c1} we have
\beqa
	\max_{j\neq k}|\chi_k\varphi_{kj}|
	&=&
	\max_{j\neq k}\log\frc{\chi_j+\chi_k}{|\chi_j-\chi_k|}
	\n
	&\leq&
	\log \max_{j\neq k}(\chi_j+\chi_k)-\log\min_{j\neq k}|\chi_j-\chi_k|\n
	&\leq&
	\log (2C\b N^\beta) +A\b N^{\frc\alpha2}.
\eeqa
Therefore, by \eqref{c3},
\beqa
	|\chi_i e_i| &\leq& (A\b N^{\frc{\alpha}2}+\beta\log \b N + \log (2C))U\b N^{\frc{\sigma}2}\n
	&\leq& V\b N^{\frc{\alpha+\sigma}2}
	\label{ebound}
\eeqa
for all $\b N\geq N_*^{(1)}(\alpha)$ where
\beq\label{Nalpha}
	N_*^{(1)}(\alpha):= \mbox{maximal value of $N_*\geq 1$ s.t.}\ N_*^{\frc\alpha 2} = \log N_*\ \mbox{if it exists, and  1 otherwise}
\eeq
(the first case occurs for $\alpha$ near enough to 0), and where
\beq\label{Vvalue}
	V := (A+\beta+\log(2C))U.
\eeq

\subsubsection{Bounding $R_r$} This is a crucial part of the proof. Using the first two bounds of \eqref{c1}, we bound $R_r$, Eq.~\eqref{eR0} as follows: there exists $N_*^{(4)}(A,C,\chi_*,\alpha,\beta)\geq N_*^{(1)}(\alpha)$ such that for all $\b N>N_*^{(4)}(A,C,\chi_*,\alpha,\beta)$ and $r\subseteq\dbra 1,N\dket$,
\beq\label{Rbound}
	R_r \leq (3+2A) \,|r|\, \b N^{\frc\alpha2}.
\eeq

Let
\beq\label{mdef}
	m= \re^{-A\b N^{\alpha/2}},
\eeq
and define the distribution (which are interpreted as signed densities)
\beq
	f(\chi) := \sum_{i\in r\,:\,\sigma_i=1} \delta(\chi-\chi_i)- \sum_{i\in r\,:\,\sigma_i=-1} \delta(\chi-\chi_i).
\eeq
Then, for all $0<\ep\leq m/2$ (which is not an optimal choice but sufficient for our purposes)
\beq
	R_r = \int_{\R_+^2,|\chi'- \chi|>\ep}\dd\chi\dd\chi'\,
	f(\chi)
	f(\chi')
	\log \Big|\frc{\chi-\chi'}{\chi+\chi'}\Big|.
\eeq
Note how, because of the second bound in \eqref{c1}, the condition of excluding the ``fat diagonal'' from the integration region only takes away the diagonal part, as it should from \eqref{eR0}.

Choose $\ep'\leq \min(\ep/4,2\chi_*)$ and let us modify the signed densities using regularised delta-functions, defining the compact-support, bounded function
\beq
	\t f(\chi) := \sum_{i\in r\,:\,\sigma_i=1} \delta_{\ep'}(\chi-\chi_i)- \sum_{i\in r\,:\,\sigma_i=-1} \delta_{\ep'}(\chi-\chi_i),\quad
\eeq
where
\beq
	\delta_{\ep'}(\chi) = \lt\{\ba{ll} 1/\ep' & (|\chi|\leq\ep'/2)\\
	0 & (|\chi|>\ep'/2).
	\ea\rt.
\eeq
Let us define (keeping its dependence on $r$ implicit)
\beq
	\t R := \int_{\R_+^2,|\chi'- \chi|>\ep}\dd\chi\dd\chi'\,
	\t f(\chi)
	\t f(\chi')
	\log \Big|\frc{\chi-\chi'}{\chi+\chi'}\Big|.
\eeq
Because $\ep'\leq \ep/4$, the supports of the regularised delta functions for non-diagonal terms do not overlap with the excluded fat diagonal, and the supports of diagonal terms lies completely within this excluded region, and because $\ep'/\leq\chi_*$, the supports of all terms stay within $\R^2_+$. Further, because $\ep'\leq m/8$, all supports are pairwise disjoint. See Fig.~\ref{figcons} for a visual guide.
\begin{figure}
\bc\includegraphics[width=4cm,height=4cm]{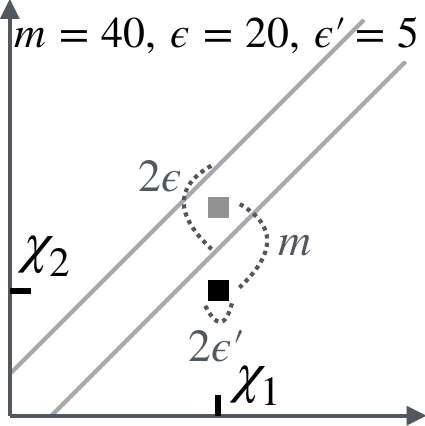}
\ec
\caption{Illustration of the integration of the mollified signed density, $\t R$. The term with $i=1,\;j=2$ is fully included, but that with $i=j=1$ is fully excluded.}
\label{figcons}
\end{figure}

 Hence we have
\beq
	\t R = \sum_{i,j\in r,\,i\neq j}\sigma_i\sigma_j
	\frc1{{\ep'}^2}\int_{\chi_i-\ep'/2}^{\chi_i+\ep'/2}\dd\chi
	\int_{\chi_j-\ep'/2}^{\chi_j+\ep'/2}\dd\chi'\,
	\log \Big|\frc{\chi-\chi'}{\chi+\chi'}\Big|.
\eeq
We evaluate, for all $i\neq j$,
\beqa
	\Delta R_{ij} &:=& \frc1{{\ep'}^2}\int_{\chi_i-\ep'/2}^{\chi_i+\ep'/2}\dd\chi
	\int_{\chi_j-\ep'/2}^{\chi_j+\ep'/2}\dd\chi'\,
	\log \Big|\frc{\chi-\chi'}{\chi+\chi'}\Big|
	-
	\log \Big|\frc{\chi_i-\chi_j}{\chi_i+\chi_j}\Big|
	\n
	&=&
	\frc1{{\ep'}^2}\int_{\chi_i-\ep'/2}^{\chi_i+\ep'/2}\dd\chi
	\int_{\chi_j-\ep'/2}^{\chi_j+\ep'/2}\dd\chi'\,
	\log \Big(\Big|\frc{\chi-\chi'}{\chi+\chi'}\Big| \Big/
	\Big|\frc{\chi_i-\chi_j}{\chi_i+\chi_j}\Big|\Big)
	\n
	&=&
	\frc1{{\ep'}^2}\int_{-\ep'/2}^{\ep'/2}\dd\chi
	\int_{-\ep'/2}^{\ep'/2}\dd\chi'\,
	\Big(\log \Big|1+\frc{\chi-\chi'}{\chi_i-\chi_j}\Big|
	-
	\log \Big|1+\frc{\chi+\chi'}{\chi_i+\chi_j}\Big|\Big)
	\n
	\Rightarrow \ |\Delta R_{ij}| &\leq&
	\log \Big|1+\frc{\ep'}{|\chi_i-\chi_j|}\Big|
	+
	\log \Big|1+\frc{\ep'}{\chi_i+\chi_j}\Big|
	\leq
	2\ep'/m
\eeqa
which is valid for all $\b N\geq N_*^{(2)}$ as follows. We use $\log(1+a)\leq a$ for all $a\geq 0$, and the first two bounds of \eqref{c1}, which give $\chi_i+\chi_j\geq 2\chi_*$ and $|\chi_i-\chi_j|\leq m$, and then we used $2\chi_*\geq m$ for all $\b N\geq \big(A^{-1}|\log (2\chi_*)|\big)^{\frc2\alpha}$. Combining with \eqref{Nalpha} we have the result for all $\b N\geq N_*^{(2)}(A,\alpha)$ with
\beq
	N_*^{(2)}(A,\chi_*,\alpha):=\max\Big\{N_*^{(1)}(\alpha),\big(A^{-1}|\log (2\chi_*)|\big)^{\frc2\alpha}\Big\}.
\eeq
As $R_r- \t R = \sum_{i,j\in r\atop i\neq j} \Delta R_{ij}$ we find
\beq\label{RtR}
	|R_r - \t R| = \Big|\sum_{i,j\in r\atop i\neq j} \Delta R_{ij}\Big| \leq 2|r|^2\ep'/m.
\eeq

Now we evaluate the contribution of the fat diagonal for the expression with mollified densities. The following is a non-negative quantity by the second line:
\beqa
	0\leq \t Q &:=&
	-\int_{\R_+^2,|\chi'- \chi|\leq \ep}\dd\chi\dd\chi'\,
	\t f(\chi)
	\t f(\chi')
	\log \Big|\frc{\chi-\chi'}{\chi+\chi'}\Big|\n
	&=&
	-\sum_{i\in r}
	\frc1{{\ep'}^2}\int_{\chi_i-\ep'/2}^{\chi_i+\ep'/2}\dd\chi
	\int_{\chi_i-\ep'/2}^{\chi_i+\ep'/2}\dd\chi'\,
	\log \Big|\frc{\chi-\chi'}{\chi+\chi'}\Big|
	\quad \mbox{(because $\ep'\leq \ep/4\leq m/8$)}
	\n
	&=&
	\sum_{i\in r}
	\frc1{{\ep'}^2}\int_{-\ep'/2}^{+\ep'/2}\dd\chi
	\int_{-\ep'/2}^{+\ep'/2}\dd\chi'\,
	\Big(\log (2\chi_i+\chi+\chi') - \log|\chi-\chi'|\Big)
	\n
	&\leq &
	\frc{|r|}{{\ep'}^2}\int_{-\ep'/2}^{+\ep'/2}\dd\chi
	\int_{-\ep'/2}^{+\ep'/2}\dd\chi'\,
	\Big(\log (2C\b N^\beta+\ep') - \log|\chi-\chi'|\Big)
	\quad \mbox{(by the third bound in \eqref{c1})}\n
	&=&
	|r|\Big(-\log(\ep')-\frc12 + \log(2C\b N^\beta+\ep')\Big)
	\quad \mbox{(by integration).}
\eeqa
We may take, for all $\b N\geq (A^{-1}\log 8)^{2/\alpha}$,
\beq\label{epprime}
	\ep' = \re^{-2 A \b N^{\alpha/2}}
\eeq
so we now have the condition $\b N\geq N_*^{(3)}(A,\alpha)$ with
\beq
	N_*^{(3)}(A,\chi_*,\alpha):=\max\Big\{N_*^{(2)}(A,\chi_*,\alpha),\big(A^{-1}\log 8\big)^{\frc2\alpha}\Big\}.
\eeq
Then
\beq\label{Qt}
	0\leq \t Q\leq |r|\Big(2A \b N^{\alpha/2} + \log(2C\b N^\beta+\ep')-\frc12\Big).
\eeq

We have
\beqa
	\t R  - \t Q &=& \int_{\R_+^2}\dd\chi\dd\chi'\,
	\t f(\chi)
	\t f(\chi')
	\log \Big|\frc{\chi-\chi'}{\chi+\chi'}\Big|\n
	&=& \frc12
	\int_{\R^2}\dd\chi\dd\chi'\,
	\t F(\chi)
	\t F(\chi')
	\log \big|\chi-\chi'\big|.
\eeqa
where $\t F(\chi) = \t f(\chi),\;\t F(-\chi) = -\t f(\chi)$ for all $\chi\geq 0$. Note that
\beq
	\int \dd\chi\,\t F(\chi)=0.
\eeq
Consider a sequence of real Schwartz functions $\t F_n,\,n=1,2,3,\ldots$, uniformly bounded by a Schwartz function, such that $\int \dd\chi\,\t F_n(\chi)=0$ and $\lim_n\t F_n(\chi) = \t F(\chi)$ pointwise (it is easy to construct such a sequence). Then by the bounded convergence theorem
\beq
	\t R - \t Q = \lim_{n\to\infty}
	\frc12
	\int_{\R^2}\dd\chi\dd\chi'\,
	\t F_n(\chi)
	\t F_n(\chi')
	\log \big|\chi-\chi'\big|.
\eeq
Applying Fourier transform,
\beq
	\int_{\R^2}\dd\chi\dd\chi'\,
	\t F_n(\chi)
	\t F_n(\chi')
	\log \big|\chi-\chi'\big|
	=\int \dd p\,|\widehat{\t F_n}(p)|^2
	\widehat{\log|\cdot|}(p)
\eeq
where $\widehat{\log|\cdot|}(p)$ is a tempered distribution, see e.g.~\cite{Titchmarsh1948}. As $\widehat{\t F_n}(0)=\int\dd\chi\,\t F_n(\chi)=0$ we can restrict to its regular part $-\pi/|p|$, giving
\beq
	\int_{\R^2}\dd\chi\dd\chi'\,
	\t F_n(\chi)
	\t F_n(\chi')
	\log \big|\chi-\chi'\big|
	=
	-\pi\int \dd p\,\frc{|\widehat{\t F_n}(p)|^2}{|p|}\leq 0\quad \forall\;n.
\eeq
Thus we conclude that
\beq\label{RQ}
	\t R \leq \t Q.
\eeq

In order to see more ``fundamentally'' what is going on, recall that $-\log|\vec \chi-\vec\chi'|$, for $\vec\chi,\vec\chi'\in\R^2$, is the (un-normalised) Greens function for the positive semi-definite Laplace operator $-\vec\nabla^2$. Thus, $-(\t R-\t Q)$ is the potential energy of a neutral charge distribution, and this is known to be strictly positive. Mathematically, if $G(\vec\chi)$ is a smooth, compactly supported function such that $\int \dd^2\chi\,G(\vec\chi) = 0$, then it lies in the domain of $(-\nabla^2)^{-\frc12}$, and from this and the Green's function property we can show that $-\int_{\R^2} \dd^2\chi\int_{\R^2}\dd^2\chi'\,
	G(\vec\chi)G(\vec\chi')\log|\vec\chi-\vec\chi'|\geq 0$. We can then take such $G(\chi_1,\chi_2)$ that approximates $\t F(\chi_1)\delta(\chi_2)$ as closely as we wish, and get the result \eqref{RQ}. We see that {\em the logarithmic structure of the two-body scattering shift is crucial}, and the main aspect is that its negative must be the kernel of a positive-definite operator.

Hence, we have found, using $\ep'/m= \re^{-A \b N^{\alpha/2}}$ from \eqref{mdef}, \eqref{epprime} and $|r|\leq N\leq \b N$:
\beqa
	R_r &\leq& \t R + 2|r|\b N\re^{-A \b N^{\alpha/2}}\quad \mbox{(from \eqref{RtR})}\n
	&\leq& \t Q + 2|r|\b N\re^{-A \b N^{\alpha/2}}
	\quad\mbox{(from \eqref{RQ})}\n
	&\leq&
	|r|\Big( 2\b N\re^{-A \b N^{\alpha/2}}
	+
	2A \b N^{\alpha/2} + \log(2C\b N^\beta+\re^{-2A\b N^{\alpha/2}})-\frc12\Big)
	\quad\mbox{(from \eqref{Qt})}\n
	&\leq& (3+2A) |r| \b N^{\alpha/2}
\eeqa
for all $\b N\geq N_*^{(4)}(A,C,\chi_*,\alpha,\beta)$ with
\beqa
	\label{Nalpha4}
	\lefteqn{N_*^{(4)}(A,C,\chi_*,\alpha,\beta)} &&\\
	&:=&\max\{N_*^{(3)}(A,\chi_*,\alpha),\n
	&&\quad \mbox{maximal value of $N_*\geq 1$ if it exists s.t.}\ N_*\re^{-A N_*^{\alpha/2}} = N_*^{\alpha/2},\n
	&& \quad \mbox{maximal value of $N_*\geq 1$ if it exists s.t.}\ 
	\log(2CN_*^\beta+\re^{-2AN_*^{\alpha/2}})-\frc12 = N_*^{\alpha/2}\}.
	\no
\eeqa
This shows \eqref{Rbound}.

\subsubsection{Bounding the sum over terms where some solitons lie outside the local cell.} We now use the two bounds proven above in order to study the centred form of the tau function, and show that solitons far from $x_*$ can be discarded. Let us write
\beq\label{tautaur}
	\tau^\# = \sum_{r\subseteq \dbra1, N\dket} \tau_r
\eeq
with
\beq\label{taur}
	\tau_r(x,t):=\exp\Big[-\sum_{i\in r} 2\chi_i\big(
	D_i
	+\sigma_i(x_*-x+v_i t + e_i)\big)+ R_r\Big].
\eeq
For every $w\geq 0$, let
\beq\label{sw}
	s_w := \{i\in \dbra 1,N\dket:D_i\leq w/2\},\quad \b s_w := \dbra 1,N\dket\setminus s_w.
\eeq

Note that using \eqref{ebound}, we have
\beq
	\Big|-\sum_{i\in r} 2\chi_i\sigma_ie_i\Big| \leq
	2V\,|r|\,\b N^{\frc{\alpha+\sigma} 2}.
\eeq
Combining this with \eqref{Rbound} and using $\b N^{-\sigma}\leq 1$ for $\b N\in\N$, we have
\beq
	\Bigg|-\sum_{i\in r} 2\chi_i
	\sigma_ie_i+ R_r\Bigg|\leq
	W\,|r|\,\b N^{\frc{\alpha+\sigma}2}
\eeq
where
\beq\label{Wdef}
	W:=3 + 2A+ 2V = 3 + 2(1+2U)A+2U\log(2C) + 2U\beta 
\eeq
(recall \eqref{Vvalue} for $V$). For notational convenience we write
\beq\label{MW}
	M:=W\b N^{\frc{\alpha+\sigma}2}.
\eeq
This quantity encodes all interaction effects to the large-$N$ behaviour. We now have
\beqa
	\sum_{r\not\subseteq s_L} \tau_r &\leq& \sum_{r\not\subseteq s_L}
	\exp-\sum_{i\in r} 2\chi_i
	\big(
	D_i
	+\sigma_i(x_*-x+v_i t)\big)\ 
	\exp |r| M\n
	&\leq&
	\sum_{r\not\subseteq s_L}
	\exp\sum_{i\in r} \big(-2\chi_i
	D_i+2\chi_i |x_*-x|+8\chi_i^3 |t|+M\big)\n
	&\leq&
	\sum_{j\in\b  s_L}
	\sum_{r\subseteq\dbra 1,N\dket\setminus \{j\} }
	\exp\sum_{i\in r\cup \{j\}} (-2\chi_i
	D_i+2\chi_i |x_*-x|+8\chi_i^3 |t|+M)\n
	&\leq&
	\sum_{j\in\b s_L}
	\exp(-2\chi_j
	D_j+2\chi_j|x_*-x|+8\chi_i^3 |t|+M)\n &&\times\;
	\sum_{r\subseteq\dbra 1,N\dket }
	\exp\sum_{i\in r} (-2\chi_i
	D_i+2\chi_i |x_*-x|+8\chi_i^3 |t|+M)\n
	&=&
	\sum_{j\in\b s_L}
	\exp(-2\chi_j
	D_j+2\chi_j|x_*-x|+8\chi_i^3 |t|+M)\n &&\times\;
	\prod_{i=1}^{\t N}
	\Big(1+\exp(-2\chi_i
	D_i+2\chi_i |x_*-x|+8\chi_i^3 |t|+M)\Big).
	\label{taucalcu}
\eeqa
We see that in the first line the expression corresponds to the terms, for the difference of tau-functions on $\dbra 1,N\dket$ and on $s_L$, that we would have for non-interacting solitons, up to the correction $e^{|r|M}$. The rest is a set of manipulations that bounds these terms, accounting for this correction. In passing from the 1st to the 2nd line we used $-\sigma_i(x_*-x)\leq |x_*-x|$, $-\sigma_iv_it\leq v_i|t|$ and $v_i = 4\chi_i^2$. In passing from the 2nd to the 3rd line, we used the fact that if $r\subseteq\dbra1,N\dket$ but $r\not\subseteq s_L$ then there must be $j\in\b s_L$ such that $j\in r$; we redefine $r$ by the replacement $r\rightarrow r\cup\{j\}$. Then, we overcount: we add a possibly non-zero number of positive terms, because it is possible to choose $j\in\b s_L,\,r\subseteq\dbra 1,N\dket\setminus \{j\}$ and $j'\in\b s_L,\,r'\subseteq\dbra 1,N\dket\setminus \{j'\}$ with $j\neq j'$ but $r\cup\{j\}=r'\cup\{j'\}$, thus the inequality. This overcounting, however, is not significant in the large-$N$ limit. In passing from the 3rd to the 4th line, we extracted a factor corresponding to the term $i=j$ in the exponential, and we have added yet new positive terms, those for which $r\subseteq\dbra 1, N\dket$ contains $j$, thus the inequality. This is convenient, as it allows us to write the last equality, where we have executed the sum over $r$, and obtained the an expression akin to a non-interacting tau function.

Now, for any $j\in \b s_L$, recalling \eqref{sw}, and by the 1st and 3rd bound of \eqref{c1} along with \eqref{xregion} and \eqref{tregion},
\beq
	2\chi_j
	D_j-2\chi_j |x_*-x|-8\chi_i^3 |t|-M
	\geq
	\chi_* L
	-
	(10+W)\b N^{\frc{\alpha+\sigma}2}.
\eeq
Hence, with $|\b s_L|\leq N\leq \b N$, the first factor in \eqref{taucalcu} can be bounded by
\beq\label{aors1}
	\sum_{j\in\b s_L}
	\exp(-2\chi_j
	D_j+2\chi_j|x_*-x|+8\chi_i^3|t|+M)
	\leq
	\b N\re^{(10 + W)\b N^{(\alpha+\sigma)/2} - \chi_* L
	}.
\eeq
This factor explicitly accounts for the exponential decay due to the soliton being far.

For the second factor, we now use $2\chi_i |x_*-x|+8\chi_i^3|t|+M\leq (10+W)\b N^{\frc{\alpha+\sigma}2}=:2\chi_* T$, so that
\beqa
	\lefteqn{\prod_{i=1}^{ N}
	\Big(1+\exp(-2\chi_i
	D_i+2\chi_i |x_*-x|+8\chi_i^3|t|+M)\Big)} \n
	&\leq&
	\prod_{i=1}^{\t N}
	\Big(1+\exp(-2\chi_*(
	D_i-T))\Big) \n
	&\leq&
	(2\re^{2\chi_*T})^{|s_{Q}|}
	\prod_{i\in \b s_{Q}}
	\Big(1+\exp-2\chi_*(
	D_i-T)\Big),\quad Q:=B'\b N^{\frc{\alpha+\sigma}2+\mu}
	\label{aors2}
\eeqa
where we choose
\beq\label{Bp}
	B' = \max\Big\{B,\frc{10+W}{2\chi_*}\Big\}.
\eeq
In the second inequality, we separated factors for index $i$ with $D_i\leq Q$, for which we used $1+\re^{-2\chi_*(|d_i|-T)}\leq 2\re^{2\chi_*T}$, from factors with $D_i>Q$. From \eqref{c2} and our choice \eqref{Bp}, we have $|s_Q|\leq QDN^\nu = B'DN^{\frc{\alpha+\sigma}2+\mu+\nu}$, hence
\beq\label{aors3}
	(2\re^{2\chi_*T})^{|s_{Q}|} \leq
	\exp\Big[
	\log(2)B'D\b N^{\frc{\alpha+\sigma}2+\mu+\nu}
	+
	(10+W)B'D\b N^{\alpha+\sigma+\mu+\nu}
	\Big].
\eeq
For all $\mu\geq 0$, we have $Q\geq T$ for all $\b N\in\N$ by our choice \eqref{Bp}. Then,
\beqa
	\prod_{i\in \b s_{Q}}
	\Big(1+\exp-2\chi_*(
	D_i-T)\Big)&\leq&
	\prod_{d=Q}^\infty
	\Big(1+\exp-2\chi_* (d-T)\Big)^{|s_{d+1}\setminus s_d|}\n
	&\leq&
	\prod_{d=Q}^\infty
	\Big(1+\exp-2\chi_* (d-T)\Big)^{(d+1)D\b N^\nu}\n
	&\leq&
	\prod_{d=Q-T}^\infty
	\Big(1+\exp-2\chi_* d\Big)^{(d+T+1)D\b N^\nu}\n
	&\leq &
	\re^{KD\b N^\nu + J(T+1)D\b N^\nu}
	\label{aors4}
\eeqa
where
\beq\label{JK}
	K=\log\prod_{d=0}^\infty
	\Big(1+\exp-2\chi_* d\Big)^{d},\quad
	J = \log \prod_{d=0}^\infty
	\Big(1+\exp-2\chi_* d\Big) = \log(-1,e^{-2\chi_*})_\infty
\eeq
are convergent expressions, functions of $\chi_*$ only.

Putting together \eqref{aors1}, \eqref{aors2}, \eqref{aors3}, \eqref{aors4}, we find
\beqa
	\sum_{r\not\subseteq s_L} \tau_r
	&\leq&
	\b N\exp\Big(
	(10+W)\b N^{\frc{\alpha+\sigma}2}+
	\log2\,B'D\b N^{\frc{\alpha+\sigma}2+\mu+\nu} + (10+W)B'D\b N^{\alpha+\sigma+\mu+\nu}\n
	&&\qquad\qquad
	+ (K+J)D\b N^\nu + \frc{J(10+W)}{2\chi_*}D\b N^{\frc{\alpha+\sigma}2 + \nu}
	- \chi_* L
	\Big).
\eeqa
Within the exponential, the leading power of $\b N$ with positive coefficient is $\alpha+\sigma+\mu+\nu$, and for all powers $\b N^a$ that appear, we have $\b N^a\leq\b  N^{\alpha+\sigma+\mu+\nu}$. Hence, with \eqref{Wdef}, \eqref{Bp} and \eqref{JK}, we may set
\beq\label{Edef}
	E:= (10+W)\Big(1+B'D+\frc{JD}{2\chi_*}\Big)
	 + \log 2\,B' D+(K+J)D,
\eeq
and we find that there exists $N_*\in\N$, which depend on all parameters $A,B,C,D,U,\,\chi_* ,\alpha,\beta,\mu,\nu,\sigma$, such that for all $\b N\geq N_*$ and all $L>0$,
\beq\label{resbound}
	\sum_{r\not\subseteq s_L} \tau_r
	\leq
	\b N\exp\Big(E\b N^{\alpha+\sigma+\mu+\nu} - \chi_*L\Big).
\eeq

\subsubsection{Bounding the difference}

We consider a local projection $\mathcal L$, Definition \ref{defimap}, resp.~Definition \ref{defimapmg}. Associated to it, there is $\b N\mapsto s$ such that
\beq
	\{i:D_i\leq L/2\}\subseteq s \subseteq \dbra 1,N\dket
	\quad \forall \; \b N\in\N.
\eeq
We also define $\b s = \dbra 1,N\dket \setminus s$. Let us write
\beq\label{uup}
	u(x,t) = u_{\bs \chi,\bs y+\bs v t},\quad u' = u_{\bs\chi_s,\bs y^{(s)}+\bs v_s t}.
\eeq
By Corollary \ref{coroltauproj}, resp.~Corollary \ref{coroltauprojmg}, we may write
\beq\label{up}
	u'(x,t)
	= 2\p_x^2 \log \tau'(x,y)
\eeq
with
\beq
	\tau'(x,t) = 
	\sum_{r\subseteq s} 
	\exp\Big[-\sum_{i\in r} 2\chi_i\big(
	|d_i|
	+\sgn(d_i)(x_*-x+v_i t+e_i)\big)+ R_r\Big]
\eeq
where
\beq
	e_i := \lt\{\ba{ll}
	0 & \mbox{(Thm \ref{theo})}\z
	\frc12\sum_{j=1\atop j\neq i}^N \big(\sgn(d_j) - \sgn_{\varep}(d_j)\big)\varphi_{ij}
	& \mbox{(Thm \ref{theomg})},
	\ea\rt.
	\quad
	R_r := \sum_{i,j\in r,\,i\neq j}
	\sigma_i
	\sigma_j \log|S_{i,j}|.
\eeq
Hence with \eqref{taur} we find
\beq\label{resproj}
	\tau'(x,t) = \sum_{r\subseteq s} \tau_r(x,t).
\eeq
We also denote
\beq
	\tau(x,t) =
	\sum_{r\subseteq \dbra 1,N\dket} \tau_r(x,t)
\eeq
and we have  (by Lemma \ref{lemtaumain}, resp.~Lemma \ref{lemtaumainmg})
\beq\label{ubtau}
	u(x,t) =2 \p_x^2 \log \tau(x,t).
\eeq
By \eqref{resbound} and $s_L\subseteq s$ (see \eqref{sw}), we then obtain
\beq
	\big|\tau - 
	\tau'\big|
	\leq
	\sum_{r\not\subseteq s}\tau_r
	\leq
	\sum_{r\not\subseteq s_L}\tau_r
	\leq \b N\re^{E\b N^{\alpha+\sigma+\mu+\nu}-\chi_* L}.
\eeq

We are looking to bound
\beq
	|\p_x^n u - \p_x^n u'|
\eeq
for $n=0,1,2,\ldots$. From \eqref{ubtau}, \eqref{up} and Faa Di Bruno's formula, we have
\beq
	\p_x^n u = 2(n+2)!\sum_{m_1,\ldots,m_{n+2}\geq 0\atop \sum_j jm_j= n+2}
	(-1)^{\sum_j m_j-1}\Big(\sum_j m_j-1\Big)!\prod_{j=1}^{n+2} \frc{(\p_x^j \tau)^{m_j}}{j!^{m_j}m_j!\tau^{m_j}}
\eeq
and similarly for $u'$ in terms of $\tau'$. Therefore, there exists $D_n'>0$ such that
\beq\label{dsupr}
	|\p_x^n u-\p_x^n u'| \leq
	D_n'
	\sum_{m_1,\ldots,m_{n+2}\geq 0\atop \sum_j jm_j= n+2}
	\Bigg|\prod_{j=1}^{n+2} \frc{(\p_x^j \tau)^{m_j}}{\tau^{m_j}}
	-
	\prod_{j=1}^{n+2} \frc{(\p_x^j \tau')^{m_j}}{(\tau')^{m_j}}\Bigg|
	\quad \forall \; n=0,1,2,\ldots
\eeq
We bound the difference of derivatives of the two tau functions, using $|r|\leq \t N\leq N$, as
\beqa
	\big|\p_x^j \tau - \p_x^j \tau'\big|
	&=&\Big|\p_x^j \sum_{r\not\subseteq s}
	\tau_r\Big|\n
	&\leq&
	\sum_{r\not\subseteq s}
	\Big|\sum_{i\in r}2\chi_i\sgn(d_i)\Big|^j\,
	\tau_r\n
	&\leq&
	(2C\b N^{\beta+1})^j\sum_{r\not\subseteq s}
	\tau_r\n
	&\leq&
	(2C)^j\b N^{j\beta+j+1}\re ^{E\b N^{\alpha+\sigma+\mu+\nu}-\chi_*L}
	\label{boundertau}
\eeqa
where we used the last bound in \eqref{c1} as well as \eqref{resbound}.  By a similar calculation,
\beq
	\big|\p_x^j\tau\big|
	\leq
	\Big|\sum_{r\in \dbra1,N\dket}\p_x^j\tau_r\Big|
	\leq
	(2C\b N^{\beta+1})^j\sum_{r\in \dbra1,N\dket}\tau_r
\eeq
therefore, with a similar calculation for $\tau'$,
\beq\label{dtt}
	\Big|\frc{\p_x^j\tau}{\tau}\Big|,\ 
	\Big|\frc{\p_x^j\tau'}{\tau'}\Big|
	\leq
	(2C\b N^{\beta+1})^j.
\eeq
Noting that $\tau,\tau'\geq 1$, using \eqref{boundertau} and \eqref{dtt} we now bound
\beq
	\Big|\frc{\p_x^j \tau}{\tau} - \frc{\p_x^j\tau'}{\tau'}\Big|
	\leq
	\Big|\frc{\p_x^j \tau}{\tau}\frc{\tau'-\tau}{\tau'}\Big|
	+
	\Big|\frc{\p_x^j \tau-\p_x^j\tau'}{\tau'}\Big|
	\leq
	2(2C)^j\b N^{j\beta+j+1}
	\re^{E\b N^{\alpha+\mu+\nu}-\chi_*L}.
\eeq
This bound, along with \eqref{dtt} again, gives us
\beqa
	\Big|\frc{(\p_x^j \tau)^m}{\tau^m} - \frc{(\p_x^j\tau')^m}{(\tau')^m}\Big|
	&=& \Big|\frc{\p_x^j \tau}{\tau} - \frc{\p_x^j\tau'}{\tau'}\Big|
	\,\Big|\sum_{a=0}^m \frc{(\p_x^j \tau)^a}{\tau^a} \frc{(\p_x^j \tau')^{m-a}}{(\tau')^{m-a}}\Big|\n
	&\leq&
	(m+1)(2C\b N^{\beta+1})^{mj}\Big|\frc{\p_x^j \tau}{\tau} - \frc{\p_x^j\tau'}{\tau'}\Big|\n
	&\leq &
	2(m+1)(2C)^{mj+j}\b N^{mj\beta + j\beta + mj + j + 1}
	\re^{E\b N^{\alpha+\mu+\nu}-\chi_*L}.
	\label{hdag}
\eeqa
We have, with \eqref{dtt}, by the telescopic argument
\beqa
	\prod_{j=1}^{n+2} a_j - \prod_{j=1}^{n+2} a_j'
	&=&
	(a_1-a_1') \prod_{j=2}^{n+2} a_j 
	+ a_1' (\prod_{j=2}^{n+2} a_j -\prod_{j=2}^{n+2} a_j') \n
	&=&
	(a_1-a_1') \prod_{j=2}^{n+2} a_j 
	+
	a_1' (a_2-a_2') \prod_{j=3}^{n+2} a_j
	+
	a_1'a_2' (\prod_{j=3}^{n+2} a_j -\prod_{j=3}^{n+2} a_j') \n
	&=&\cdots
\eeqa
 and using $(2C\b N^{\beta+1})^j\geq 1$ for all $j=1,\ldots,n+2$ and for all $\b N\in\N,\,\b N\geq (2C)^{-1/(\beta+1)}$, we find
\beq
	\Bigg|\prod_{j=1}^{n+2} \frc{(\p_x^j \tau)^{m_j}}{\tau^{m_j}}
	-
	\prod_{j=1}^{n+2} \frc{(\p_x^j \tau')^{m_j}}{(\tau')^{m_j}}\Bigg|
	\leq
	(2C\b N^{\beta+1})^{\sum_{j=1}^{n+2}jm_j}
	\sum_{j=1}^{n+2}
	\Big|\frc{(\p_x^j \tau)^{m_j}}{\tau^{m_j}} - \frc{(\p_x^j\tau')^{m_j}}{(\tau')^{m_j}}\Big|.
\eeq
With the condition $m_1,\ldots,m_{n+2}\geq 0$, $\sum_j jm_j= n+2$ we have both $m_j\leq n+2$ and $j\leq n+2$. Then the bound \eqref{hdag} gives
\beq
	\Bigg|\prod_{j=1}^{n+2} \frc{(\p_x^j \tau)^{m_j}}{\tau^{m_j}}
	-
	\prod_{j=1}^{n+2} \frc{(\p_x^j \tau')^{m_j}}{(\tau')^{m_j}}\Bigg|
	\leq
	2(n+2)(n+3)(2C)^{(n+2)(n+4)}
	\b N^{(n+2)(n+4)(\beta+1)+1}
	\re^{E\b N^{\alpha+\mu+\nu}-\chi_*L}.
\eeq
Hence with \eqref{dsupr},
\beq
	|\p_x^n u-\p_x^n u'| \leq D_n'
	2(n+2)^{n+3}
	(n+3)(2C)^{(n+2)(n+4)}
	\b N^{(n+2)(n+4)(\beta+1)+1}
	\re^{E\b N^{\alpha+\mu+\nu}-\chi_*L}.
\eeq
This shows \eqref{proj} with $D_n=D_n'
	2(n+2)^{n+3}
	(n+3)(2C)^{(n+2)(n+4)}$ and $\kappa_n = (n+2)(n+4)(\beta+1)+1$. A similar calculation can be done for time derivatives (giving powers of $\b N^{3\beta}$ instead of $\b N^\beta$), and this completes the proof of Theorem \ref{theo}.
\eproof

\subsection{Fluid-cell mean as fluid-cell projection}\label{ssectprooffc}

In this section we prove Theorem \ref{theofc}.

By Assumption \ref{assimp},
\beq
	\Delta x\leq VN^\omega\Delta X=VRN^{\alpha+2\beta+\mu+2\omega}=\Delta I
\eeq
hence
\beq
	\lim_{N\to\infty} \Delta I L^{-1}=0.
\eeq
So the interval
\beq
	I^\circ := [I_-+\Delta I,I_+-\Delta I]
\eeq
is non-empty and more than a single point for all $N$ large enough.

The following is a re-statement of Corollaries \ref{corolfc}, \ref{corolvanishfc} and \ref{coroboundedfc} which we express here for convenience and will refer to in the proof below.
\begin{lemma}\label{lemproof}
We have
\beq\label{convuprojfcpr}
	\sup_{x\in  I^\circ} \Big|\p_x^n \big(u(x) -u'(x)\big)\Big| = \mathcal O(N^{-\infty}) \quad (N\to\infty)\qquad \forall n\in\Z_+.
\eeq
Further, there exist $a\geq 1,\,\ep>0$ such that
\beq\label{uvanishfcpr}
	\sup_{x\in \R\setminus J} |\re^{\ep \dist(x,I)}\p_x^n u'(x)| = \mathcal O(N^{-\infty})\quad (N\to\infty)\qquad \forall\, n\in\Z_+
\eeq
where
\beq\label{Jfclast}
	J:= [I_--a\Delta I,I_++a\Delta I].
\eeq
Finally, there exists $N_*\in\N$ and $T>0$ such that for every integers $n\geq 0,\,N\geq N_*$,
\beq\label{bound2}
	\sup_{x\in \R}|\p_x^nu(x)|,\,\sup_{x\in \R}|\p_x^nu'(x)|\leq 
	2(n+2)!\,T^{n+2}\,\Biggr(\sum_{m_1,\ldots,m_{n+2}\geq 0\atop \sum_j jm_j= n+2}
	\frc{\Big(\sum_j m_j-1\Big)!}{\prod_{j=1}^{n+2} j!^{m_j}m_j!}\Biggr)\,
	N^{(\alpha+5\beta+\mu+2\omega)(n+2)}.
\eeq
\end{lemma}

\medskip

\noindent{\em Proof of Theorem \ref{theofc}.}
We use $J$ defined in \eqref{Jfclast}. We denote by $K$ the Lipschitz coefficient,
\beq\label{diffF}
	|F(\bs u)-F(\bs u')| \leq K\sum_{k=1}^n |u_k-u_k'|.
\eeq
We set $F_{\rm max} = \sup_{\bs u\in\R^{n+1}}(|F(\bs u)|)<\infty$. By Lemma \ref{lemproof}, and by the definition of the degree Eq.~\eqref{degree}, there exists $P_{\rm max}>0$ such that
\beq\label{supP}
	\sup_{x\in\R} |P(u(x),\p_x u(x),\ldots,\p_x^nu(x))|,\ 
	\sup_{x\in\R} |P(u'(x),\p_x u'(x),\ldots,\p_x^nu'(x))|\leq P_{\rm max}
	N^{(\alpha+5\beta+\mu+2\omega)r}.
\eeq

{\em Case I.} We have
\beqa
	\lefteqn{\Bigg|\int_{I} \dd x\,F(u,\p_x u,\ldots,\p_x^nu)
	-
	\int_{-\infty}^\infty \dd x\,F(u',\p_x u',\ldots,\p_x^nu')
	\Bigg|} &&\n
	&=&\Bigg|\Big(\int_{I^\circ} \dd x
	+
	\int_{I\setminus I^\circ} \dd x
	\Big)\,F(u,\p_x u,\ldots,\p_x^nu)
	-
	\int_{-\infty}^\infty \dd x\,F(u',\p_x u',\ldots,\p_x^nu')
	\Bigg|
	\n
	&=&\Bigg|-\int_{\R\setminus I^\circ} \dd x\,F(u',\p_x u',\ldots,\p_x^nu')
	+
	\int_{I\setminus I^\circ} \dd x
	\,F(u,\p_x u,\ldots,\p_x^nu) + \mathcal O(N^{-\infty})\Bigg|\n
	&\leq & \Bigg|\int_{\R\setminus I^\circ} \dd x\,F(u',\p_x u',\ldots,\p_x^nu')\Bigg|
	+
	2F_{\rm max}\Delta I + \mathcal O(N^{-\infty})\n
	&= & \Bigg|\int_{J\setminus I^\circ} \dd x\,F(u',\p_x u',\ldots,\p_x^nu') + \mathcal O(N^{-\infty})\Bigg|
	+
	2F_{\rm max}\Delta I + \mathcal O(N^{-\infty})
	\n
	&\leq &
	(4+2a)F_{\rm max}\Delta I + \mathcal O(N^{-\infty}).
	\label{proo1}
\eeqa
In the third line (the second equality) we used, by Lemma \ref{lemproof} and with \eqref{diffF},
\beqa
	\lefteqn{\Bigg|\int_{I^\circ} \dd x
	\,(F(u,\p_x u,\ldots,\p_x^nu)
	-
	F(u',\p_x u',\ldots,\p_x^nu')\Bigg|} &&\n
	&\leq&
	\int_{I^\circ} \dd x\,
	|F(u,\p_x u,\ldots,\p_x^nu)
	-
	F(u',\p_x u',\ldots,\p_x^nu')| \n
	&\leq&
	K\int_{I^\circ} \dd x\,
	\sum_{k=0}^n|\p_x^k (u(x)-u'(x))|\n
	&\leq&
	KL
	\sum_{k=0}^n|\sup_{x\in I^\circ} \p_x^k (u(x)-u'(x))|
	= \mathcal O(N^{-\infty})
\eeqa
using \eqref{convuprojfcpr} and \eqref{Lgamma}. In fourth line of \eqref{proo1} we used the fact that $F$ is bounded and that the size of $I\setminus I^\circ$ is $2\Delta I$. In the fifth line, we used Lemma \ref{lemproof} again,
\beqa
	\Bigg|\int_{\R\setminus J} \dd x\,F(u',\p_x u',\ldots,\p_x^nu')\Bigg|
	&\leq&
	\int_{\R\setminus J} \dd x\,|F(u',\p_x u',\ldots,\p_x^nu')|\n
	&\leq&
	K\int_{\R\setminus J} \dd x\,\sum_{k=0}^n
	|\p_x^ku'(x)|\n
	&\leq&
	K\sum_{k=0}^n\sup_{x\in\R\setminus J} |e^{\ep \dist(x,I)}\p_x^ku'(x)|\int_{\R\setminus J} \dd x\,
	e^{-\ep \dist(x,I)}\n
	&=&\mathcal O(N^{-\infty})
\eeqa
as $\int_{\R\setminus J} \dd x\, e^{-\ep \dist(x,I)} = 2\int_0^\infty \dd x\,\re^{-\ep (x+a\Delta I)}\to 0$ as $N\to\infty$. Lastly, in the sixth line of \eqref{proo1} we again used the bound on $F$, and the fact that the size of $J\setminus I^\circ$ is $2(1+a)\Delta I$.
Finally, using conditions \eqref{case1}, the result of \eqref{proo1} implies \eqref{res1}.

\medskip

{\em Case II.} We use the algebraic fact that for a polynomial $P(u_0,\ldots,u_n)$ with $\deg(P(u_0,\ldots,u_n)) = r$, we have
\beq\label{diffP}
	P(u_0,\ldots,u_n) - P(u_0',\ldots,u_n')
	=\sum_{k=0}^n P_k(u_0,\ldots,u_n,u_0',\ldots,u_n')(u_k-u_k')
\eeq
for some polynomials $P_k(u_0,\ldots,u_n,u_0',\ldots,u_n')$ with $\deg(P_k(u_0,\ldots,u_n,u_0',\ldots,u_n'))\leq r-k-2$. By Lemma \ref{lemproof}, there exist $P_{k,{\rm max}}>0$ such that
\beq\label{Pkmax}
	|P_k(u,\p_x u,\ldots,\p_x^n u,u',\p_x u',\ldots,\p_x^n u')|
	\leq P_{k,{\rm max}} N^{(\alpha+5\beta+\mu+2\omega)(r-k-2)}\quad \forall\ k=0,\ldots,n.
\eeq

We redo the derivation \eqref{proo1}, but now with $P$ in place of $F$:
We have
\beqa
	\lefteqn{\Bigg|\int_{I} \dd x\,P(u,\p_x u,\ldots,\p_x^nu)
	-
	\int_{-\infty}^\infty \dd x\,P(u',\p_x u',\ldots,\p_x^nu')
	\Bigg|} &&\n
	&=&\Bigg|\Big(\int_{I^\circ} \dd x
	+
	\int_{I\setminus I^\circ} \dd x
	\Big)\,P(u,\p_x u,\ldots,\p_x^nu)
	-
	\int_{-\infty}^\infty \dd x\,P(u',\p_x u',\ldots,\p_x^nu')
	\Bigg|
	\n
	&=&\Bigg|-\int_{\R\setminus I^\circ} \dd x\,P(u',\p_x u',\ldots,\p_x^nu')
	+
	\int_{I\setminus I^\circ} \dd x
	\,P(u,\p_x u,\ldots,\p_x^nu)\Bigg| + \mathcal O(N^{-\infty})\n
	&\leq & \Bigg|\int_{\R\setminus I^\circ} \dd x\,P(u',\p_x u',\ldots,\p_x^nu')\Bigg|
	+
	2P_{\rm max}N^{(\alpha+5\beta+\mu+2\omega)r}\Delta I + \mathcal O(N^{-\infty})\n
	&= & \Bigg|\int_{J\setminus I^\circ} \dd x\,P(u',\p_x u',\ldots,\p_x^nu') + \mathcal O(N^{-\infty})\Bigg|
	+
	2P_{\rm max}N^{(\alpha+5\beta+\mu+2\omega)r}\Delta x + \mathcal O(N^{-\infty})
	\n
	&\leq &
	(4+2a)P_{\rm max}N^{(\alpha+5\beta+\mu+2\omega)r}\Delta I + \mathcal O(N^{-\infty}).\n
	\label{proo2}
\eeqa
In the third line (the second equality) we used, by Lemma \ref{lemproof} along with \eqref{diffP} and \eqref{Pkmax},
\beqa
	\lefteqn{\Bigg|\int_{I^\circ} \dd x
	\,(P(u,\p_x u,\ldots,\p_x^nu)
	-
	P(u',\p_x u',\ldots,\p_x^nu')\Bigg|} &&\n
	&\leq&
	\int_{I^\circ} \dd x\,
	|P(u,\p_x u,\ldots,\p_x^nu)
	-
	P(u',\p_x u',\ldots,\p_x^nu')| \n
	&\leq&
	\int_{I^\circ} \dd x\,
	\sum_{k=0}^n
	|P_k(u,\ldots,\p_x^n u,u',\ldots,\p_x^nu')|\,|\p_x^k (u(x)-u'(x))|\n
	&\leq&
	L\sum_{k=0}^n
	P_{k,\rm max} N^{(\alpha+5\beta+\mu+2\omega)(r-k-2)}\,
	|\sup_{x\in I^\circ} \p_x^k (u(x)-u'(x))|
	= \mathcal O(N^{-\infty})
\eeqa
using \eqref{convuprojfcpr} and \eqref{Lgamma}. In the fourth line of  \eqref{proo2}, we used \eqref{supP}:
\beq
	\Bigg|\int_{I\setminus I^\circ} \dd x
	\,P(u,\p_x u,\ldots,\p_x^nu)\Bigg|\leq
	2P_{\rm max}N^{(\alpha+5\beta+\mu+2\omega)r}\Delta I.
\eeq
In the fifth line of \eqref{proo2}, we used Lemma \ref{lemproof} again, specifically \eqref{uvanishfcpr}, along with \eqref{diffP} and \eqref{Pkmax}:
\beqa
	\lefteqn{\Bigg|\int_{\R\setminus J} \dd x\,P(u',\p_x u',\ldots,\p_x^nu')\Bigg|}&&\n
	&\leq&
	\int_{\R\setminus J} \dd x\,|P(u',\p_x u',\ldots,\p_x^nu')|\n
	&\leq&
	\sum_{k=0}^n \int_{\R\setminus J} \dd x\,
	|P_k(u',\ldots,\p_x^n u',0,\ldots,0)|\,
	|\p_x^ku'(x)|\n
	&\leq&
	\sum_{k=0}^nP_{k,{\rm max}} N^{(\alpha+5\beta+\mu+2\omega)(r-k-2)}\sup_{x\in\R\setminus J} |e^{\ep \dist(x,I)}\p_x^ku'(x)|\int_{\R\setminus J} \dd x\,
	e^{-\ep \dist(x,I)}\n
	&=&\mathcal O(N^{-\infty}).
\eeqa
Lastly, in the sixth line of \eqref{proo2} we again used the bound \eqref{supP} and the fact that the size of $J\setminus I^\circ$ is $2(1+a)\Delta I$.
Finally, using conditions \eqref{case2}, the result of \eqref{proo2} implies \eqref{res2}. The result \eqref{res2conserved} follows from \eqref{integralconserved}. \eproof

\subsection{Weak limit of conserved densities}\label{ssectweak}

In this section we prove Theorem \ref{theoweak}.
\proof Let
\beq\label{lambdacondweak}
	\lambda > (1+r)\alpha+(2+5r)\beta+r\mu+(2+2r)\omega
\eeq
and
\beq\label{condLambda}
	\Lambda>\lambda+(\alpha+5\beta+\mu+2\omega)r.
\eeq
By \eqref{condLambdacombined}, this is possible. We use the result of Theorem \ref{theofc}, Case II, for every fluid cell $N\mapsto I_{x_*} = [x_*-N^{\lambda}/2,x_*+N^\lambda/2]$, uniformly on $x_*\in\R$. By \eqref{supP} we have
\beq\label{supPfrak}
	\sup_{x\in\R}|\mathtt P_k[u](x)|\leq
	P_{\rm max}N^{(\alpha+5\beta+\mu+2\omega)r}.
\eeq
For every $N$ large enough, we write
\beqa
	\int \dd x\,f(x)\,\mathtt P_k[u](N^\Lambda x)
	&=&
	\frc1{N^\Lambda}\int \dd x\,f(x/N^\Lambda)\,\mathtt P_k[u](x)\n
	&=&
	\frc1{N^{\Lambda-\lambda}}\sum_{\ell\in \Z}
	\frc1{N^\lambda}\int_{-N^\lambda/2}^{N^\lambda/2} \dd x\,f(\ell /N^{\Lambda-\lambda}+x/N^\Lambda)
	\,\mathtt P_k[u](\ell N^\lambda+x).\n
\eeqa
As $f$ is a Schwartz function, then there exists a Schwartz function  $g>0$ such that, for every $N\in\N$ and $z\in\R$,
\beq\label{aoqjf}
	\sup_{x\in [-N^\lambda/2,N^\lambda/2]}
	\big|f(z+x/N^\Lambda)
	-
	f(z)\big|\leq
	\frc1{N^{\Lambda-\lambda}} g(z).
\eeq
In particular
\beq
	\sum_{\ell\in\Z} \frc1{N^{\Lambda-\lambda}}g(\ell/ N^{\Lambda-\lambda})\leq G<\infty \quad \forall\ N\in\N.
\eeq
Therefore, using \eqref{supPfrak},
\beqa
	\lefteqn{\Bigg|\sum_{\ell\in\Z}\frc1{N^\lambda}\int_{-N^\lambda/2}^{N^\lambda/2} \dd x\,\big(f(\ell /N^{\Lambda-\lambda}+x/N^\Lambda)
	-
	f(\ell /N^{\Lambda-\lambda})\big)
	\,\mathtt P_k[u](\ell N^\lambda+x)\Bigg|}&&\n
	&\leq&
	\sum_{\ell\in\Z}\frc1{N^\lambda}\int_{-N^\lambda/2}^{N^\lambda/2} \dd x\,\big|f(\ell /N^{\Lambda-\lambda}+x/N^\Lambda)
	-
	f(\ell /N^{\Lambda-\lambda})\big|
	\,|\mathtt P_k[u](\ell N^\lambda+x)|
	\n
	&\leq&
	\sum_{\ell\in\Z} \frc1{N^{\Lambda-\lambda}} g(\ell/ N^{\Lambda-\lambda})
	\,P_{\rm max}N^{(\alpha+5\beta+\mu+2\omega)r}
	\n
	&\leq &
	G P_{\rm max}N^{(\alpha+5\beta+\mu+2\omega)r}.
\eeqa
Hence with \eqref{condLambda} we have
\beq\label{limaa1}
	\lim_{N\to\infty} \Bigg(\int \dd x\,f(x)\,\mathtt P_k[u](N^\Lambda x)
	-
	\frc1{N^{\Lambda-\lambda}}\sum_{\ell\in \Z}
	f(\ell /N^{\Lambda-\lambda})\frc1{N^\lambda}\int_{-N^\lambda/2}^{N^\lambda/2} \dd x\,
	\,\mathtt P_k[u](\ell N^\lambda+x)
	\Bigg) = 0.
\eeq
By Theorem \ref{theofc}, we have
\beq\label{fhsgq}
	\lim_{N\to\infty} \Bigg(\frc1{N^\lambda}\int_{I_{\ell N^\lambda}}\dd x\,
	\,\mathtt P_k[u](x)
	-
	\frc1{N^\lambda}\sum_{i:x_i\in I_{\ell N^\lambda}}\chi_i^{2k+1}\Bigg)
	=
	0
\eeq
uniformly for $\ell\in\Z$, and therefore, because $f$ is a Schwartz function,
\beq\label{limaa2}
	\lim_{N\to\infty}
	\frc1{N^{\Lambda-\lambda}}\sum_{\ell\in \Z}
	f(\ell /N^{\Lambda-\lambda})
	\Bigg(\frc1{N^\lambda}\int_{-N^\lambda/2}^{N^\lambda/2} \dd x\,
	\,\mathtt P_k[u](\ell N^\lambda+x)
	-
	\frc1{N^\lambda}\sum_{i:x_i\in I_{\ell N^\lambda}}\chi_i^{2k+1}\Bigg)
	=0.
\eeq
Finally, for every $x_i\in I_{\ell N^\lambda}$, using \eqref{aoqjf}
\beq
	|f(\ell N^\lambda/ N^\Lambda) - f(x_i/N^\Lambda)|
	\leq \frc1{N^{\Lambda-\lambda}}g(\ell N^\lambda)
\eeq
and with \eqref{fhsgq} (uniformly on $\ell$) and \eqref{supPfrak},
\beq
	\Bigg|\frc1{N^\lambda}\sum_{i:x_i\in I_{\ell N^\lambda}}\chi_i^{2k+1}\Bigg|
	\leq
	P_{\rm max}N^{(\alpha+5\beta+\mu+2\omega)r} +
	W(N)\quad \forall\ell\in\Z
\eeq
for some function $W(N)$, independent of $\ell$, with $\lim_{N\to\infty} W(N)=0$. Therefore
\beq\label{limaa3}
	\lim_{N\to\infty}
	\frc1{N^{\Lambda-\lambda}}\sum_{\ell\in \Z}\frc1{N^\lambda}
	\sum_{i:x_i\in I_{\ell N^\lambda}}
	(f(\ell /N^{\Lambda-\lambda})
	-
	f(x_i/N^\Lambda)
	)
	\chi_i^{2k+1}
	=0.
\eeq
Combining \eqref{limaa1}, \eqref{limaa2}, \eqref{limaa3}, and using
\beq
	\sum_{\ell\in\Z}
	\sum_{i:x_i\in I_{\ell N^\lambda}}
	=
	\sum_{i=1}^N
\eeq
we obtain \eqref{weaklimitmain}.
\eproof

\ssect{Proofs of the simplified theorems}\label{ssectproofoverviewtheorems}

Here we prove Theorems \ref{theoextentintro}, \ref{theoboundintro}, \ref{theointro}, \ref{theofcintro} and \ref{theoweakintro}.

We note that the regularity of spectral parameter, Eq.~\eqref{c1intro}, amounts to making Assumption \ref{asssp} with $\beta=0$, and for all all $\alpha>0$, for the first and last condition being uniform (with the same $\chi_*,\,C$ for all $\alpha$), and with $\eta=\alpha/2$. Further, the statement that a soliton gas $N\mapsto \bs\chi,\bs y$ is away from the condensate regime, Eq.~\eqref{c4intro}, is Assumption \ref{asscont} for all exponents $\eta>0$, with uniform $\omega=\mu=0$, $B=1$, and with
\beq
	V = \frc2{1-W}\geq 2.
\eeq
As we need the case $\eta=\alpha/2$ (Eq.~\eqref{etachoice}), and Assumption \ref{asssp} holds for all $\alpha>0$, then this is a consistent choice. By Lemma \ref{lemconcentration}, this implies Assumption \ref{assimp} with $\omega=0,\,B=1,\,\mu=0$ for all $\eta>0$. Finally, we note that, by Lemma \ref{lemrealcontracted}, the statement that a soliton gas is away from the condensate regime also implies Assumption \ref{assde} uniformly for $x_*\in\R$, with \eqref{DVintro} and $\nu=0,\,B=1,\,\mu=0$ for all $\eta> 0$.

\medskip
\noindent {\em Proof of Theorem \ref{theoextentintro}.} The condition \eqref{c2intro} is sufficient to imply Assumption \ref{assde} with $\nu=0,\,\mu=0$ and $B=1$, and every $\eta\geq 0$, with uniform $D$, at the boundaries $x_*=x^\pm(\bs\chi,\bs y)$. Therefore the assumptions of Corollary \ref{corolextent}, with $w=0$, are satisfied for all $\alpha$, and the result \eqref{projvanishintro} holds.
\eproof

\medskip
\noindent {\em Proof of Theorems \ref{theoboundintro} and \ref{theointro}.} The discussion above shows that the assumptions of Corollary \ref{corolbound2} hold, which shows \eqref{bound2intro}.
For \eqref{projintro}, by the discussion above the assumptions of
Theorem \ref{theo} are satisfied uniformly for $x_*\in\R$. It is then a matter of choosing appropriate local projections. We use Lemma \ref{lemfclocal} with $\Delta X=L/2$ and condition \eqref{c4intro}, giving $\Delta x(\Delta X;\bs\chi,\bs y)\leq VL/2$ for all $L\geq 2$. Thus we can apply Theorem \ref{theo} to all $x_*\in[I_-+VL/2,I_+-VL/2]$, and this gives \eqref{projintro}.
\eproof

\medskip
\noindent {\em Proof of Theorems \ref{theofcintro} and \ref{theoweakintro}.} These are immediate from the discussion above and Theorems \ref{theofc} and \ref{theoweak}.
\eproof

\section{Conclusion}\label{secconclu}

In this paper we have proposed a precise meaning to the {\em quasiparticle problem} of the KdV soliton gas: the problem of identifying meaningful solitons' positions in a $N$-soliton field, for a sequence of spectral and impact parameters $\bs\chi\in\Rpu^N,\,\bs y\in\R^N$ in the limit where $N$ is large. We have introduced the quasiparticle criterium, which asks for good properties of fluid-cell projections defined with respect to quasiparticle positions: it only keeps solitons within the fluid cell, such that the projected field gives an approximation of the original field on $I$ that becomes exact as $N\to\infty$, and is supported on the fluid cell. We have proposed the explicit, universal form of fluid-cell projections, Eqs.~\eqref{fcintro}, \eqref{fcintro2}, in terms of solitons' spectral and impact parameters. This is based on explicit quasiparticle positions $x_i(\bs\chi,\bs y)$, which solve the semiclassical Bethe equations (SBE) \eqref{bethe} exactly, generalising the contraction map of the hard rod gas \cite{spohn2012large}. We have proposed simple, finite recursive algorithms giving exact solutions to the SBE for negative scattering shifts (such as in the KdV case), and use this as our quasiparticle positions. These satisfy a minimal-separation condition \eqref{separation}, generalising the minimal separation of hard rods and naturally related to the KdV condensate condition \cite{el2021soliton,suret2024soliton}. Using this, we have obtained a full and explicit characterisation of multi-soliton fields in terms of spectral parameters and solitons' positions. We have then shown how the quasiparticle criterium is indeed satisfied, relating these KdV rods with the KdV multi-soliton fields, under weak assumptions on spectral parameters, Eq.~\eqref{c1}, and natural assumptions concerning densities of solitons, Eqs.~\eqref{c2} and \eqref{c4}. Solving the quasiparticle problem is an essential step in deriving the kinetic equation / GHD equation for soliton gases: without it, the kinetic equation for the density of states -- the spectral-spatial density of solitons -- is not predictive for local observables. Our conditions on densities of solitons are physically natural, and we have argued that it is satisfied as soon as the soliton gas is away from its condensate regime \cite{el2021soliton,suret2024soliton}. We have shown that all conditions are satisfied in dilute soliton gases. We leave for future works further analysis of these conditions, including how they are stable under time evolution.

The SBE are known to give rise to the correct Euler-scale GHD equations, see e.g.~the derivations or proofs in \cite{croydon2021generalized,doyon2024new,doyon2023ab}. Therefore our result could form the basis for a proof of the GHD equation (soliton gas kinetic equation) for the KdV model. Specifically the GHD equation is a continuity equation for $\rho(\chi,x,t)$, and in order to show GHD from the KdV multi-soliton solutions, one needs to:
\begin{enumerate}
\item define positions of quasiparticles and the associated empirical density $\rho_{\rm emp}(\chi,x,t)$, Eq.~\eqref{empirical};
\item show \eqref{euler}, or more generally establish the quasi-particle criterium \eqref{loose}, for almost all times in macroscopic intervals, $t\in[0,NT]$;
\item show that the weak limit, $\rho_{\rm emp}(\cdot,N\cdot,N\cdot)\to\rho(\cdot,\cdot,\cdot)$, satisfies the GHD equation.
\end{enumerate}
Our results solve point 1 and partially point 2. In order to show the emergence of GHD from our result, we need to: (i) establish how our conditions are preserved over time, at least for a large family of initial spectral and impact parameters (part of point 2); and (ii) show that the empirical density based on our quasiparticle positions gives in the weak limit a solution to the GHD equation (point 3). We have chosen a particular algorithm solving the SBE (the ``rightward'' version), but other choices are possible (the ``leftward'' version, or other algorithms yet to be written).

The SBE Eqs.~\eqref{bethe} have been conjectured to give a good description of soliton positions in previous works \cite{doyon2024new}. This was given credence recently by showing that they are a good approximation for centres of Lax eigenvectors in thermal states  of the Toda model, which are shown to be localised by Anderson localisation \cite{aggarwal2025effective,aggarwal2026asymptoticscatteringrelationtoda} and represent quasiparticle positions in this model. Our results further confirm this, in a different way, via our formulation of the quasiparticle problem.

The quasiparticle positions, Eqs.~\eqref{bethe}, can be used to propose a general local-equilibrium grand canonical measure that admits all conserved quantities, following \cite{doyon2026generalised}:
\beq\label{measure}
	\bigoplus_{N=0}^\infty \frc1{N!}\prod_{i=1}^N\dd \chi_i\dd y_i\, \prod_{i=1}^N \Theta(
	\chi_i)\,\re^{-\sum_i \beta(\chi_i,x_i(\bs\chi,\bs y)/\ell)}
\eeq
for a macroscopic length $\ell$. The {\em a priori} measure $\dd \chi_i\dd y_i$ is invariant, and the Boltzmann weight guarantees that the full measure is locally invariant. The grand canonical ensemble on a finite interval of length $\ell$ is obtained with $\beta(\chi,x) = \beta(\chi)\Theta(|x|-1/2)$ where $\Theta$ is Heaviside's step function. One is interested in taking the thermodynamic limit $\ell\to\infty$. For the KdV rods, we have a bijection (Theorem \ref{theodesc}) to a space $\mathsf Q$, Eq.~\eqref{Qexact}, on which the map \eqref{bethe} is simply the identity, so we have
\beq\label{measureKdVrods}
	\bigoplus_{N=0}^\infty \frc1{N!}\prod_{i=1}^N\dd \chi_i\dd x_i\,
	\prod_{i=1}^N \Theta(\chi_i)
	\prod_{i,k:x_k<x_i}\Theta\Big(
	x_i-x_k+\sum_{j:x_j\in[x_k,x_i)} \varphi_{ij}\Big)
	\re^{-\sum_i \beta(\chi_i,x_i/\ell)}.
\eeq
It is expected that the above measure has strong clustering properties in space for local observables of $x_i(\bs\chi,\bs y)$, hence by the quasiparticle criterium, it is clustering for local observables of the KdV field. Thus, this is a good family of local-equilibrium measures for the KdV soliton gas, and we believe this is the most general one. In the case where $\varphi_{ij}\geq 0$, and using a mollification of the sign function, it is shown in \cite{doyon2026generalised} using the matrix-tree theorem that the free energy corresponding to \eqref{measure} satisfies the classical thermodynamic Bethe ansatz; this was also proposed for the KdV soliton gas in \cite{bonnemain2022generalized}.

In fact, Eqs.~\eqref{bethe} have been conjectured to give rise to the correct hydrodynamic theory for many-body integrable systems to all orders in the hydrodynamic scale \cite{doyon2024new,urilyon2025simulating,kethepalli2026ballistic,urilyon2026anomalousdiffusionsuperdiffusionintegrable}, in agreement with the vanishing of hydrodynamic noise in integrable systems \cite{PhysRevLett.134.187101,doyon2026hydrodynamic}, including fluctuations and correlations \cite{doyon2018exact,doyon2020fluctuations,myers2020transport,doyon2023ballistic,doyon2023emergence,kundu2025macroscopic,PhysRevLett.134.187101}. In the Toda case, the Lax centres satisfy the SBE up to a power of $\log(N)$, which should not affect either the Euler large-scale limit or higher-order hydrodynamics. In order to establish the hydrodynamic theory to all order, however, it is not sufficient to study the SBE: one must also control the corrections, as $N\to\infty$, on how local observables -- here the KdV field $u(x)$ and its derivatives $\p_x^n u(x)$ -- are fixed in terms of solitons lying near $x$, for all $x\in\R$ (the quasiparticle problem). This can lead to additional correction terms for local observables. Both the localisation techniques of \cite{aggarwal2025effective,aggarwal2026asymptoticscatteringrelationtoda} and our techniques should provide enough control, something which we hope to come back in future works. This would then open the way for rigorous proofs of the above conjecture. 

Our main technical theorem, Theorem \ref{theo}, establishes how the KdV field around an observation point $x_*$ is controlled by solitons in a neighbourhood of $x_*$. This requires contracted distances around $x_*$, where the expansion effect of interactions is taken away, should have a density that does not increase too much with $N$. Corrections terms are exponentially decaying, and controlled by the slowest, widest soliton (the minimal impact parameter). This exponential is similar to the fundamental result (see e.g.~\cite{gesztesy1992limits}) saying how a $N$-soliton field decays at large distances.


We have also established a new result about the support of a $N$-soliton field as $N\to\infty$, showing that it must lie, under density conditions, within an interval controlled by the left-most and right-most soliton positions, for which we give explicit formulas in terms of spectral and impact parameters, Eq.~\eqref{coreintro}.

The methods are based on the tau-function representation of the multi-soliton solution, and in particular on the construction of new tau functions with a clear physical interpretation. The techniques we present are novel, where a representation of the tau function is taken which guarantees that far-solitons have negligible contributions. This suggests that in soliton theory, the phenomenon of Anderson localisation of Lax eigenvectors is not necessary in order to be able to locate solitons -- although a more in-depth analysis of our assumptions would be required. For instance, it is a simple matter to construct explicit, non-random spectral and impact parameters satisfying the assumptions of Theorem \ref{theo} for any given observation point $x_*\in\R$: for $\bs\chi$ this is done in \eqref{exchi}, which are regularly distributed so as to satisfy \eqref{c2}; then set $D_i$'s to be all positive, of finite density and regularly distributed, and by Lemma \ref{lemnat}, $X_i=D_i\sgn(x_i-x_*)$ satisfy the magnifying-glass equation \eqref{eqxymainXstr}, thus we can construct $\bs y = \mathcal C_{\bs \chi,x_*}(\bs X)$, Eq.~\eqref{eqxymainXstr}. An example that satisfy the assumptions of all our main results is given in Section \ref{ssectdilute}. Our methods also go beyond the application to thermal states or generalised Gibbs ensembles: the fluid-cell projection can be applied to any multi-soliton KdV field that satisfy certain conditions, and randomness, ``genericity'', or long-wavelength, hydrodynamic-type conditions are not necessary.

An interesting observation is that the proof is based on the specific logarithmic form of the KdV two-body scattering shift: the fact that $\log|\chi-\chi'|$ is related to the 2-dimensional Laplacian operator, or has a strictly negative regular part of Fourier transform. This is a fact that was also used in \cite{kuijlaars2021minimal} in analysing the nonlinear dispersion relations of the soliton gases related to KdV. It appears as though in classical soliton theory, the two-body scattering shift has special properties that guarantee that the concepts of quasiparticle and soliton gas make sense. We believe our proofs could be generalised to a large family of scattering shifts with similar property.

We believe this work opens many directions. Tau functions representations of solitonic solutions are available in many other models, and therefore our techniques should be readily applicable. In particular, they are available in higher-dimensional integrable field equations such as the Kadomtsev–Petviashvili hierarchy, and it would be interesting to use the techniques developed here to find locations of line solitons in this context, in view of understanding higher-dimensional soliton gases \cite{bonnemain2025two}.

We also believe that we can weaken some of our assumptions; for instance, the minimal separation of spectral parameters may be bounded by a decaying power law, instead of a stretched exponential, in Eq.~\eqref{c1}.

Of course it would be extremely interesting to make the arguments that we have presented to derive the kinetic equation mathematically accurate: the rigorous derivation of such equations remains a long-standing problem in soliton theory. It would also be interesting to have a construction of quasiparticles directly in the limit $N\to\infty$, and to extend our analysis to the inclusion of radiative modes, which might be necessary in order to describe more general measures on the KdV field that are not restricted to the set of multi-soliton fields.

\medskip
{\bf Acknowledgment.} I thank Amol Aggarwal, Gino Biondino, Thibault Bonnemain, Gennady El, Tamara Grava, Stefano Olla, St\'ephane Randoux, Giacomo Roberti, Makiko Sasada, Pierre Suret, and Alex Tovbis for many discussions on the subject of soliton gases. I am grateful in particular to Amol for his comments on a draft of the manuscript, to Tamara for pointing out to me some important literature, and to Makiko for her private communication showing the SBE in the box-ball system. I am also grateful to St\'ephane for showing me his observation that the support of the Lax eigenvectors in the nonlinear Schr\"odinger equation represented soliton positions, during the programme ``Modulation theory and dispersive shock waves'', Isaac Newton Institute, 2022. Although I did not develop this line of attack (this is developed in \cite{aggarwal2025effective,aggarwal2026asymptoticscatteringrelationtoda,bonnemain2026exact}), this inspired me, after much time, to understand the problem of soliton positions more deeply. I am grateful to UK Research and Innovation for support under the Horizon Europe Guarantee, Advanced Grant Scheme, grant number EP/Z534304/1. No artificial intelligence tools were used in the production of this research or writing of this paper.

\appendix

\section{Magnifying-glass vectors}\label{appmag}

In Section \ref{ssectsbe}, we introduced the notion of regular magnifying-glass vector. These are special cases of magnifying-glass vectors, from which a different representation of the tau function can be obtained, with an associated local projection theorem. We develop this here\footnote{This was the basis for fluid-cell projection theorems in a first preprint version of this paper \cite{firstpreprint}, however the new methods developed in the main text are more powerful. In particular, the notion of effective position was defined in terms of magnifying-glass position, and shown to satisfy an approximate for of the SBE; this has been replaced by the explicit solutions to the SBE that we have obtained in the main text.}.

We fix a regularisation of the sign function: for every $\varep>0$, this is a continuous function $\sgn_\varep:\R\to\R$ such that
\beq\label{sgnepgen}
	\sgn_{\varep}(d) = \sgn(d)\quad \forall\;|d|\geq \varep,
\eeq
and
\beq\label{sgnepgen1}
	|\sgn_{\varep}(d)|\leq 1 \quad (d\in\R).
\eeq
An example is
\beq\label{sgnep}
	\sgn_\varep(d) = \lt\{\ba{ll}
	-1 & (d<-\varep)\\
	d/\varep & (-\varep\leq d\leq \varep)\\
	1 & (d>\varep),
	\ea\rt.
\eeq
but any other continuous function satisfying \eqref{sgnepgen} can be used in the following. Throughout, we fix $\varep>0$.

\ssect{Magnifying-glass vectors and the associated centred form of the tau function}

We now extend Definition \ref{defimagstr} to the case where components of the magnifying-glass vector can be 0 (here the point 0 is extended to the interval $[-\varep,\varep]$).
\begin{defi}[Magnifying-glass vectors for observation point $x_*\in\R$]\label{defimag}
Let $N\in\N$, $\bs \chi\in\Rpu^N,\,\bs y\in\R^N$. A {\em magnifying-glass vector at $x_*$} is a solution $\bs X\in\R^N$ to the following system of equations:
\beq\label{eqxymainX}
	y_i = X_{i} - \frc12\sum_{j=1\atop j\neq i}^N\sgn_{\varep}(X_j-x_*)\varphi_{ij},\quad i=1,\ldots,N.
\eeq
The corresponding {\em displacement vector} is
\beq\label{defX}
	d_i = X_i - x_*
\eeq
which solves
\beq\label{eqxymaind}
	y_i  - x_* = d_{i} - \frc12\sum_{j=1\atop j\neq i}^N\sgn_{\varep}(d_{j})\varphi_{ij},
	\quad i=1,\ldots,N.
\eeq
We modify the partial contraction map away from $x_*$ from Definition \ref{defimagstr} to $\mathcal C_{\bs \chi,x_*}: \R^N\to\R^N$ (which now involves $\varep$) as
\beq\label{mapC}
	\bs y = C_{\bs\chi,x_*}(\bs X)\quad\forall\;\bs\chi\in\Rpu^N,\,x_*\in\R.
\eeq
Note that the explicit $x_*$ dependence is simple: $\mathcal C_{\bs \chi,x_*}(\bs X) =x_* \1_N + \t{\mathcal C}_{\bs \chi}(\bs X - x_* \1_N)$ where $\t{\mathcal C}_{\bs\chi}(\bs d)$ is given by the right-hand side of \eqref{eqxymaind}.
\end{defi}

We emphasise that magnifying-glass vector $\bs X$ explicitly depends on the observation point $x_*$. The absolute value of displacements, $|d_i|$, are the counterpart, in this construction, to the contracted distances $D_i$ of Definition \ref{definatmag}.  Here $|d_i|$ is, of course, never negative, however magnifying-glass positions can accumulate on $[-\varep,\varep]$, and such positions  correspond phenomenologically to semiclassical Bethe positions with negative contracted distances -- this is the ``dark space'' discussed in Section \ref{ssectsbe}.

A local projection will be defined in Section \ref{appprojections}, an operation by which we keep only solitons with small enough displacements. Theorem \ref{theomg} expressed in Sec.~\ref{appprojections}, gives the justification for the interpretation of $d_i$'s as soliton displacements from $x_*$, as its says that a $N$-soliton solution is closed to its local projection, under certain conditions.

By techniques similar to those of \cite{doyon2026generalised}, it is possible to show that Eqs.~\eqref{eqxymaind} have a unique solution (fixing the regularised sign functions) if $\varphi_{ij}>0\;\forall\;i,j$; but for the KdV case, we have $\varphi_{ij}<0$.  Our main lemma for the system of equations \eqref{eqxymaind} in the case $\varphi_{ij}<0$ is that solutions indeed exist, and have a nice structure for impact parameters that are far enough from $x_*$ or for large enough $|x_*|$ (this lemma holds for any matrix $(\varphi_{ij})_{i,j=1}^N$ with $\varphi_{ij}<0\;\forall\;i\neq j$, in place of the matrix \eqref{varphi}):
\begin{lemma}[Existence and properties of soliton displacements] \label{lemd}
Let $N\in\N$, $\bs \chi\in\Rpu^N$. The map $\t{\mathcal C}_{\bs \chi}:\R^N\to\R^N$ is surjective. If there is a set $s\subseteq\dbra 1,N\dket$ such that
\beq\label{condz}
	z_i \leq \frc12\sum_{j=1\atop j\neq i}^N \varphi_{ij}-\varep\ \forall\ i\in s,\quad
	z_i \geq - \frc12\sum_{j=1\atop j\neq i}^N \varphi_{ij}+\varep\ \forall\ i\in
	\b s:=\dbra 1,N\dket\setminus s,
\eeq
then there is a unique $\bs d\in\R^N$ satisfying $\bs z = \t{\mathcal C}_{\bs \chi}(\bs d)$, and it is such that $d_i\leq -\varep\; (i\in s)$ and $d_i\geq \varep\; (i\in \b s)$ and is given by
\beq\label{valuesdisplacements}
	d_i = z_i- \frc12\sum_{j=1\atop j\neq i}^N\sgn_s(j)\varphi_{ij}\quad\forall\ i.
\eeq
Finally, let
\beq\label{xmpdef}\begin{aligned}
	x_-(\bs\chi,\bs y) &:= 
	\sup\{x_*\in\R: \exists\, \bs d\ \mbox{solving}\ \eqref{eqxymaind}\,|\, d_i\geq \varep\,\forall\,i\}+\varep\\
	x_+(\bs\chi,\bs y) &:= 
	\inf\{x_*\in\R: \exists\, \bs d\ \mbox{solving}\ \eqref{eqxymaind}\,|\, d_i\leq -\varep\,\forall\,i\}-\varep.
	\end{aligned}
\eeq
These exist and give the values in \eqref{condxstar}.
\end{lemma}
\proof Note that for the first two statements it is sufficient to consider the case $x_*=0$ of \eqref{eqxymaind}, in which case $\bs z= \bs y$.

Let $N\in\N$, $\bs y\in\R^N$, $\t\varphi_{ij}\in\R$ for all $i,j\in\dbra1,N\dket$. Consider first the case $\varep>0$. We show that there exists at least one solution $\bs d\in\R^N$ to
\beq\label{eqxymainapprox}
	y_i = d_i - \frc12\sum_{j\neq i}\sgn_{\varep}(d_j)\t\varphi_{ij}
\eeq
with \eqref{sgnepgen}. Indeed, we construct the continuous function $H:[-1,1]^N\to \R^N$ as
\beq
	H_i(\bs u) = d(u_i) - \frc12\sum_{j\neq i}\sgn_{\varep}(d(u_j))\t\varphi_{ij}
	- y_i
\eeq
with the monotonically increasing continuous function
\beq
	d(u) = \frc{u}{1+\eta-u^2}:[-1,1]\to \R
\eeq
and $1/\eta = \sum_i|y_i| + \frc12 \sum_{i,j:j\neq i}|\t\varphi_{ij}|+1$. This has the property that $H_i>0$ when $u_i=1$ and $H_i<0$ when $u_i=-1$, for all $i\in\dbra 1,N\dket$. Therefore by the Poincar\'e-Miranda theorem, there exists $\bs u\in[-1,1]^n$ such that $\bs H(\bs u) = \bs 0$. For these $u_i$, we therefore have solutions $d_i = d(u_i)$ to the system \eqref{eqxymainapprox}. Setting $\t\varphi_{ij} = \varphi_{ij}$, this shows \eqref{eqxymaind}.

Under the first  condition in \eqref{condz}, as $\varphi_{ij}<0$, we have $z_i \leq - \frc12\sum_{j=1\atop j\neq i}^N |\varphi_{ij}|-\varep$ and therefore by \eqref{eqxymaind} and  \eqref{sgnepgen1}, it must be that $d_i\leq -\varep$ for $i\in s$. Similarly the second condition gives $d_i\geq \varep$ for $i\in \b s$. By  \eqref{sgnepgen} we then obtain \eqref{valuesdisplacements}.

Finally, consider \eqref{xmpdef}. For any $x_*\in\R$ and solution $\bs d$ to \eqref{eqxymaind} such that $d_i\geq \varep\,\forall\,i$, then $y_i-x_*\geq \varep - \frc12\sum_{j=1\atop j\neq i}^N\sgn_{\varep}(d_{j})\varphi_{ij}$ and hence $x_*\leq\min_i\Big(y_i + \frc12\sum_{j\neq i} \varphi_{ij}\Big)-\varep$. Further, for all $x_*\leq\min_i\Big(y_i + \frc12\sum_{j\neq i} \varphi_{ij}\Big)-\varep$, the solution is unique and gives $d_i\geq \varep\,\forall\,i$ by our previous result. Hence the supremum in \eqref{xmpdef} exists and gives \eqref{condxstar}. A similar argument holds for the infimum in \eqref{xmpdef}.
\eproof
\medskip

We note that the case $\varep=0$ can also be done, however there we need to replace $\sgn_\varep(x)$ by $\sgn_j(x) = \sgn(x)\;(x\neq\0),\,\varep_j\in[-1,1]\;(x=0)$. We must show that
\beq
	\exists\,\bs\varep\in[-1,1]^N: y_i  - x_* = d_{i} - \frc12\sum_{j=1\atop j\neq i}^N\lt\{\ba{ll} \sgn(d_{j}) & (d_j\neq 0)\\ 
	\varep_j & (d_j=0)
	\ea\rt\}\;
	\varphi_{ij},
	\quad i=1,\ldots,N
\eeq
always possesses a solution. For this purpose, we adapt the above proof by using, instead of \eqref{eqxymainapprox}, the starting point
\beq
	y_i = k_1(\t d_i) - \frc12\sum_{j\neq i}k_2(\t d_j)\t\varphi_{ij}
\eeq
with
\beq
	k_1(d) = \lt\{\ba{ll}
	d+\ep & (d\leq -\ep)\\
	0 & (-\ep<d<\ep)\\
	d-\ep & (d\geq \ep)
	\ea\rt.,\quad
	k_2(d) = \lt\{\ba{ll}
	-1 & (d\leq -\ep)\\
	d/\ep & (-\ep<d<\ep)\\
	1 & (d\geq \ep)
	\ea\rt.
\eeq
for some $\ep>0$. As the right-hand side is continuous in $\t d_i$ for all $i$, then the proof can be adapted. The solution is $d_i = k_1(\t d_i)$ with $\varep_i = \sgn_{i}(d_i) = k_2(\t d_i)$.

Lemma \ref{lemd} immediately implies the following. 
\begin{corol}[Properties of magnifying-glass vectors] \label{corold}
First, there is a unique solution $\bs X$ to \eqref{eqxymainX} whenever $x_*\leq x_-(\bs\chi,\bs y)-\varep$, resp.~$x_*\geq x_+(\bs\chi,\bs y)+\varep$, and it satisfies
\beq
	X_i = X_i^\mp(\bs\chi,\bs y)
\eeq
with \eqref{condxstar}. Second, if $x_*> x_-(\bs\chi,\bs y)-\varep$, then for any solution to \eqref{eqxymainX} there exists $i\in\dbra 1,N\dket$ such that $X_i<x_*+\varep$. Similarly, if $x_*< x_+(\bs\chi,\bs y)+\varep$, then for any solution to \eqref{eqxymainX} there exists $i\in\dbra 1,N\dket$ such that $X_i>x_*+\varep$.
\end{corol}

Much like \eqref{formvectors}, a simple observation is that any magnifying-glass vector is of the form
\beq
	X_i = y_i + \frc12\sum_{j=1\atop j\neq i}^N\sigma_j\varphi_{ij},\quad
	\sigma_i\in\{\pm\}.
\eeq
In particular,
\beq\label{magboundinv}
	[\min_i X_i,\max_i X_i]\subseteq\core(\bs\chi,\bs y) .
\eeq

Again the non-uniqueness of solutions to  \eqref{eqxymaind} is related to the lack of a unique definition of soliton positions, and as we see in Lemma \ref{lemtaumainmg} below, it is irrelevant to the $N$-soliton KdV field itself.

Note that, because $\varphi_{ij}<0$, magnifying-glass positions $X_i$ are to the right of $y_i$ if the observation point $x_*$ is large and positive, while they are to the left if it is large and negative. Thus interactions have the effect of expanding the positions in space, as with interactions, positions appear closer to the observation point when this observation point is far in space (either to the right or to the left). This justifies the name ``contraction map'' for the map $\mathcal C_{\bs\chi,x_*}$ that takes away the interaction to give freely-propagating impact parameters, and ``expansion map'' for $\mathcal E_{\bs\chi,x_*}$ that partially puts the interaction into physical soliton displacements (it is only partial, again as discussed in Sec.~\ref{sssectdisp}).

Time evolution is naturally induced on soliton displacements, and on magnifying-glass positions, by using \eqref{yt}.

Our main justification for introducing the magnifying-glass positions is their relation with tau functions. Combining the notion of displacement vectors with a judicious choice of the tau-function representation \eqref{tauresult}, \eqref{tauprimet}, we obtain the following; this gives a {\em different representation} from that of Lemma \ref{lemtaumain}, but of similar style (we use the same notation for the tau function for simplicity):
\begin{lemma}[Centred form - magnifying-glass setup] \label{lemtaumainmg}
Let $N\in\N$, $\bs\chi\in\Rpu^N,\,\bs y,\,\Delta\bs y\in\R^N$, and  $x_*\in\R$. For every $\bs d\in\R^N$ solving \eqref{eqxymaind}, and under the equivalence \eqref{equiv}, we have
\beq\label{taueqcentredmg}
	\tau_{\bs\chi,\bs y+\Delta\bs y}  \equiv \tau_{\bs\chi,\bs d,\Delta\bs y,x_*}^\#
\eeq
where
\beq
	\tau_{\bs\chi,\bs d,\Delta\bs y,x_*}^\#(x) := \sum_{r\subseteq \dbra1,N\dket} 
	\exp\Big[-\sum_{i\in r} 2\chi_i\big(
	|d_i|
	+\sgn(d_i)(x_*-x+\Delta y_i+e_i)\big)+ R_r\Big]
	\label{taumainmg}
\eeq
where
\beq\label{eRmg}
	e_i := \frc12\sum_{j=1\atop j\neq i}^N \big(\sgn(d_j) - \sgn_{\varep}(d_j)\big)\varphi_{ij},\quad
	R_r := \sum_{i,j\in r,\,i\neq j}
	\sgn(d_i)
	\sgn(d_j) \log|S_{i,j}|
\eeq
(recall \eqref{varphi} and \eqref{Sp}).
\end{lemma}
\proof We use the representation \eqref{taudecomp} with the first line of \eqref{tauresult}. We make the replacement $\bs y\to\bs y+\Delta\bs y$, and keep explicitly $\Delta\bs y$ out side of $x_i^s$:
\beq\label{tauresultproof}
	\tau_{\bs \chi,\bs y+\Delta\bs y}^s(x)=
	\sum_{r\subseteq \dbra 1,N\dket}\ 
	\exp \sum_{i\in r} 2\chi_i\sgn_s(i)(x-x_i^{s}-\Delta y_i)\ 
	\exp \sum_{i,j,\in r,\,i\neq j} \sgn_s(i)\sgn_s(j)
	\log|S_{i,j}|
\eeq
where
\beq\label{xistproof}
	x_i^{s} := y_i + \frc12\sum_{j=1\atop j\neq i}^N \sgn_s(j)\varphi_{ij}
	\quad(i\in\dbra 1,N\dket).
\eeq
Choose
\beq
	s = \{i:d_i\geq 0\}.
\eeq
Then
\beq\label{dsajk}
	\sgn_{s}(i) = \sgn(d_i)
\eeq
recalling the convention \eqref{sgn} for the sign function. Denote $x_i = x_i^{s}$ for lightness. Then we have
\beq
	y_i = x_i - \frc12\sum_{j\neq i}\sgn(d_j)\varphi_{ij}.
\eeq
Also, using the definition \eqref{defX} for $X_i$,
\beqa
	\sgn_{s}(i)(x-x_i-\Delta y_i)
	&=& \sgn(d_i)(-d_i + x-x_* + X_i-x_i-\Delta y_i)\n
	&=& -(|d_i| + \sgn(d_i)(x_*-x + \Delta y_i + e_i))\label{akak}
\eeqa
where, using of \eqref{eqxymainX},
\beq
	e_i = x_i - X_i = \frc12 \sum_{j\neq i}
	(\sgn(d_j) - \sgn_{\varep}(d_j))\varphi_{ij}.
\eeq
Combining \eqref{tauresultproof} and \eqref{akak} this shows the lemma.
\eproof
\medskip


Note how the discussion around Eq.~\eqref{discussioncentredform} applies here, with $|d_i|$ replacing the contracted distance $D_i$. Therefore, we have a different form of the tau function, based on a different notion of solitons' positions where only contracted metrics around observation points $x_*$ are considered, which still allow us to obtain a centred form where terms with solitons that are far -- in such metrics -- can be shown to be small (Theorem \ref{theomg}).

%
%

\begin{rema}\label{remaX}
From \eqref{eqxymainX}, we see that as we scan $x_*$ from the left to the right, whenever $x_*$ crosses the position of a soliton, all solitons, except that one, tend to ``jump'' (in a continuous way) towards the right. But if the density is high enough so that other crossings happen during this jump, then a rearrangement occurs over which we do not have a good control yet. Nevertheless, in general we expect that over long displacements of $x_*$ towards the right, the magnifying-glass position $X_i$ is affected by an overall displacement towards the right, and generically the overall displacement of $X_i$ is approximately proportional to the displacement of $x_*$.
\end{rema}

\begin{rema} In \eqref{eqxymaind} and \eqref{mapC}, we cannot in general replace $\sgn_\varep$ by a fixed sign function, such as $\sgn$: our preliminary analysis suggests that resulting map $\mathcal C_{\bs \chi}$ would not be surjective for generic $\varphi_{ij}<0$ that are not symmetric, $\varphi_{ij}\neq \varphi_{ji}$, as is the case for the KdV scattering shifts \eqref{varphi}.
\end{rema}

\ssect{Local projection theorem}\label{appprojections}

We now define a local projection map that keeps only solitons whose magnifying-glass positions are near enough to the observation point $x_*$. We start with a lemma that mirrors Lemma \ref{lemextractsbe}, but applied to magnifying-glass positions. The following has less constraints on the sets $s_+,\,s_-$, because of the simpler form of Eq.~\eqref{eqxymainX} compared to Eq.~\eqref{bethemain}. It turns out that the same modification of impact parameters $\bs y^{(s_+,s_-)}$, Eq.~\eqref{yispsm}, is involved here.
\begin{lemma}[How to extract solitons: magnifying glass positions]
\label{lemseparated} Let $N\in\N$, and let $\bs X\in\R^N$ be a magnifying-glass vector at $x_*\in\R$ for $\bs \chi\in\Rpu^N,\,\bs y\in\R^N$ (Definition \ref{defimag}). Let $s_-\subseteq\{i:X_i\leq x_*-\varep\}$, $s_+\subseteq\{i:X_i\geq x_*+\varep\}$, and
\beq\label{sspsmmg}
	s = \dbra 1,N\dket \setminus (s_+\cup s_-).
\eeq
Then $\bs X_s$ is a magnifying-glass vector at $x_*$ for $\bs\chi_s,\,\bs y^{(s_+,s_-)}$.
\end{lemma}
\proof Because $X_j\leq x_*-\varep$ for all $j\in s_-$, then in \eqref{eqxymainX} with all $i\in s$, we can bring the terms with such $j$ to the left-hand side. This gives the second sum in \eqref{yispsm}. The first sum is obtained similarly, and the case $\varep=0$ is obtained similarly.
\eproof
\medskip

For every $N\in\N,\,\bs\chi\in\Rpu^N,\,\bs y\in\R^N,\,x_*\in\R$, we choose a solution $\bs X\in\R^N$ to \eqref{eqxymainX}, and their associated displacement vectors:
\beq
	\bs X(x_*;\bs\chi,\bs y),\quad \bs d(x_*;\bs\chi,\bs y).
\eeq
These exist by Lemma \ref{lemd}. With this choice, we now define {\em local projections with respect to $I$}, where $I\subset\R$ is an interval. Such a local projection is a map that keeps at least all solitons that lie within $I$, in the magnifying-glass metric at the centre of $I$, with respect to the chosen set of solutions. We set
\beq\label{intervalapp}
	I=[x_*-L/2,x_*+L/2]\subset \R.
\eeq
This definition mirrors Definition \ref{defimap}, except it does not need to notion of a fluid-cell projection, because of the lesser constraints on the sets $s_\pm$ in Lemma \ref{lemseparated}.
\begin{defi}[Local projection] \label{defimapmg}
Let $x_*\in\R$, $L\geq2\varep$ and $I\subset\R$ be the interval \eqref{intervalapp}. A local projection with respect to the $I$ is a map
\beq
	\mathcal L: \mathfrak{S}\to \mathfrak{S}
\eeq 
such that for every $(\bs \chi,\bs y)\in \mathsf S$ and with $\bs d =\bs d(x_*;\bs\chi,\bs y)$, there exists $s\supseteq \{i:|d_i|\leq L/2\}$ such that
\beq\label{localspsm}
	\mathcal L u_{\bs \chi,\bs y} = u_{\bs \chi_s,\bs y^{(s_+,s_-)}}
\eeq
where $s_+,\,s_-\subseteq\dbra 1,N\dket$ are defined such that $d_i>L/2\,\forall\,i\in s_+$, $d_i<-L/2\,\forall\,i\in s_-$ and $s_+\cup s_- = \dbra 1,N\dket\setminus s$. Equivalently, we have
\beq\label{locallimit}
	\mathcal L u_{\bs \chi,\bs y}(x)
	= 
	\lim_{\bs z:s_+,s_-}
	u_{\bs \chi,\bs y+\bs z}(x)\quad\forall\, x\in\R.
\eeq
\end{defi}
The statement \eqref{locallimit} come from Lemmas \ref{lemproj}. A local projection takes
\beq
	u_{\bs \chi,\bs y}\mapsto u_{\bs \chi',\bs y'}
\eeq
with $\bs\chi\in\Rpu^N,\,\bs y\in\R^N$ and $\bs\chi'\in\Rpu^M,\,\bs y'\in\R^M$ for some $M\leq N$. This extracts from the full set of solitons, the solitons in the set $s$, which contain at least all those whose magnifying-glass positions that lie in $I$. By Lemma \ref{lemseparated}, $u_{\bs \chi',\bs y'}=\mathcal L u_{\bs \chi,\bs y}$ is a multi-soliton field for which $\bs X(x_*;\bs\chi,\bs y)_s$ are magnifying-glass positions at $x_*$ (although they may be different from $\bs X(x_*;\bs\chi',\bs y')$).
The spectral parameters $\bs\chi'=\bs\chi_s$ of the projected, $M$-soliton solution are a subset of those of the original $N$-soliton one, $\bs\chi$, but the  impact parameters $\bs y'=\bs y^{(s_+,s_-)}$ are not a subset of $\bs y$.

Because $\mathcal L u_{\bs \chi,\bs y}$ is a multi-soliton field for which $\bs X(x_*;\bs\chi,\bs y)_s$ are magnifying-glass positions at $x_*$, then $\bs d(x_*,\bs\chi,\bs y)_s\in\R^{|s|}$ solves \eqref{eqxymaind} for $\bs\chi_s,\,\bs y^{(s_+,s_-)}$. Hence we immediately have the equivalent of Corollary \ref{coroltauproj} in the magnifying-glass setup, as a corollary of Lemma \ref{lemtaumainmg}.
\begin{corol}[Projected form of the centred tau function, magnifying-glass setup]\label{coroltauprojmg} In the context of Definition \ref{defimapmg},  let $(\bs\chi,\bs y)\in {\mathtt S}$ and $\Delta\bs y \in \R^{|s|}$. The tau function of the projected KdV field $\tau_{\bs\chi_s,\bs y^{(s_+,s_-)}+\Delta\bs y}$ is equivalent, in the sense of \eqref{equiv}, to the restriction to $s$ of the centred tau function $\tau_{\bs\chi,\bs x,\Delta\bs y,x_*}^\#$: with the notation of Definition \ref{lemtaumainmg},
\beq\label{tauprojmg}
	\tau_{\bs\chi_s,\bs y^{(s_+,s_-)}+\Delta\bs y}
	\equiv\tau_s,\quad
	\tau_s(x) :=
	\sum_{r\subseteq s} 
	\exp\Big[-\sum_{i\in r} 2\chi_i\big(
	|d_i|
	+\sgn(d_i)(x_*-x+\Delta y_i+e_i)\big)+ R_r\Big]
\eeq
where $\bs d = \bs d(x_*;\bs\chi,\bs y)$ and
\beq\label{eRmgproj}
	e_i := \frc12\sum_{j=1}^N \big(\sgn(d_j) - \sgn_{\varep}(d_j)\big)\varphi_{ij},\quad
	R_r := \sum_{i,j\in r,\,i\neq j}
	\sgn(d_i)
	\sgn(d_j) \log|S_{i,j}|.
\eeq
\end{corol}
\proof The result is immediate from what we said above and Lemma \ref{lemtaumainmg}, if $e_i$ in \eqref{eRmgproj} is replaced by $e_i := \frc12\sum_{j\in s} \big(\sgn(d_j) - \sgn_{\varep}(d_j)\big)\varphi_{ij}$. But as $|d_j|> L\geq \varep$ for all $j\in s_+\cup s_- = \dbra 1,N\dket\setminus s$, and by \eqref{sgnepgen}, we can extend the sum to that written in \eqref{eRmgproj}.
\eproof
\medskip

We are now ready to combine the above with the centred form of the tau function from Lemma \ref{lemtaumainmg}, in order to show the main theorem of this Appendix, Theorem \ref{theomg}. This shows how, under certain conditions, any local projection gives a good local approximation of the KdV field locally. However, because of the use of magnifying-glass positions, where precision is lost the further we are from $x_*$, we obtain an approximation of the field on an interval that is smaller then the ``real-space'' interval corresponding to $I$; we only know which solitons {\em may} contribute to the field there, but not which {\em actually do} contribute. In particular, the support of the local projection of the field goes beyond the region where it agrees with the original field. This is not sufficient to obtain the quasiparticle criterium \eqref{loose} -- here we will not show how to solve this criterium in the magnifying-glass setup, as it was solved already in the main text using the SBE.

The conditions we need are Assumption \ref{asssp}, as well as a no-accumulation assumption and an assumption that adapts Assumption \ref{assde} to soliton displacements instead of contracted distances. We consider, as in the main text, sequences \eqref{seqgen} for our soliton gas. The following two assumptions are for  a sequence $\b N\mapsto N,\,\bs\chi,\,\bs y,\,x_*$ and we take $\bs d =  \bs d(x_*;\bs\chi,\bs y)$.

\begin{assum}[Bound on accumulation of small soliton displacements]\label{assac}
There exists $U>0$ and $\sigma\geq 0$ such that for every $\b N$:
\beq\label{c3}
	|\{i:|d_{i}|\leq \varep\}| \leq U\b N^{\frc\sigma2}.
\eeq
\end{assum}
The data of this assumption is $U,\,\sigma$.

\begin{assum}[Bound on soliton displacements' density for exponent $\eta\geq 0$]\label{assdemg} There exist $B>0,\,D>0$ and $\mu\geq  0,\,\nu\geq 0$ such that for every $\b N$:
\beq
	\label{c2app}
	\rho_{\bs d}(d)\leq D N^{\nu}\quad \forall\  d\geq B\b N^{\eta +\mu}
\eeq
where
\beq
	\rho_{\bs d}(d) := \frc{\big|\{i:|d_i|\leq d\}\big|}{2d}.
\eeq
\end{assum}
The data of this assumption is $B,\,D,\,\mu,\,\nu$, while $\eta$ is a parameter to be chosen appropriately in terms of other assumptions' data.

For our main local projection theorem of the magnifying-glass context, we consider a sequence as follows:
{\em 
\beq\label{seqmg}
	\N\ni \b N\mapsto N\in \dbra 1,\b N\dket,\,\bs\chi\in\Rpu^N,\,\bs y\in\R^N,\,x_*\in \R,\,L\geq \varep\quad \mbox{such that:}
\eeq
the spectral regularity Assumption \ref{asssp} holds; and the bound on accumulation Assumption \ref{assac} and the bound on density Assumption \ref{assdemg} for exponent
\beq
	\eta = \frc{\alpha+\sigma}2.
\eeq
}
\begin{theorem}[Local form of the multi-soliton field] \label{theomg}
There exists $D_{n,m},\kappa_{n,m}>0\ (n,m=0,1,2,\ldots),\,E>0$ and $N_*>0$ depending only on the assumptions' data, with $E$ given in \eqref{Edef}, such that for every sequence \eqref{seq}, for every integer $n,m\geq 0$, and for every sequence of local projections $\b N\mapsto \mathcal L$ with respect to $I$ (Eq.~\eqref{intervalapp}, Definition \ref{defimapmg}), we have for all $\b N\geq N_*$,
\beq\label{projapp}
	\sup_{(x,t)\in J\times K}\Big|\p_x^n\p_t^m \big(u_{\bs \chi,\bs y+\bs v t}(x) - \tau_t \mathcal L u_{\bs \chi,\bs y}(x)\big)\Big| \leq D_{n,m}
	\b N^{\kappa_{n,m}} \exp\Big(E \b N^{\alpha+\sigma+\mu+\nu}-\chi_*L\Big)
\eeq
where
\beq\label{xregionapp}
	J:=[x_*- C^{-1}\b N^{\eta-\beta}\,,\,
	x_*+C^{-1}\b N^{\eta-\beta}],
\eeq
\beq\label{tregionapp}
	K:=[- C^{-3}\b N^{\eta-3\beta}\,,\,
	C^{-3}\b N^{\eta-3\beta}].
\eeq
In particular, if we can take $L=R N^\gamma$ with either
\beq\label{gammaapp}
	R>\frc{E}{\chi_*},\ \gamma = \alpha+\sigma+\mu+\nu,\quad\mbox{or}\quad R>0,\ \gamma>\alpha+\sigma+\mu+\nu,
\eeq
then we have, uniformly on sequences \eqref{seq} with fixed assumptions' data, and sequences $N\mapsto \mathcal L$ as above,
\beq\label{convuprojapp}
	\sup_{(x,t)\in J\times K} \Big|\p_x^n \p_t^m\big(u_{\bs \chi,\bs y+\bs v t}(x) - \tau_t\mathcal L u_{\bs \chi,\bs y}(x)\big)\Big| = \mathcal O(N^{-\infty}) \quad (\b N\to\infty)\qquad \forall \,n,m\in\Z_+.
\eeq
\end{theorem}
The proof is provided in Sec.~\ref{proofmain}, simultaneously with the proof of Theorem \ref{theo}, as it has exactly the same structure.

\subsection{Magnifying glass effect}\label{sssectdisp}

%

Consider \eqref{eqxymaind},  \eqref{eqxymainX}. Their meaning is as follows: from the position $X_i$ of a tracer soliton, we subtract half of its scattering shifts with all other solitons $j\neq i$ which stand on the left of the observation point $x_*$ (at least a distance $\varep$ away), and add those for solitons $j$ on the right, in order to obtain a quantity, the impact parameter for soliton $i$, that evolves linearly in time. The result is that whenever the position $X_j$ of another soliton crosses the observation point, the position $X_i$ jumps by a scattering shift with that soliton. This jump is exactly what happens over a full interaction event when soliton $i$ interact with soliton $j$ in a two-body collision. Over long times and distances, the accumulation of these shifts represent the global effect of factorised soliton scattering. This is very similar to the system \eqref{bethe} encoding the collision rate ansatz. However, there is an important difference: here we account for solitons that cross the observation point $x_*$, instead of those that actually cross the tracer soliton's position itself as in \eqref{bethe}. For positions $X_i$ that are near to $x_*$, this does not change very much from the viewpoint of large-scale observations. But for solitons that are far, we introduce an ``observation error''. It is as if we were looking at solitons ``through a lens'', with deformations away from the centre of the lens. In the metric language of \cite{doyon2018geometric}, the soliton displacements account for the metric change from the impact-parameter space to real space around $x_*$, but not for that at positions far from it, and therefore, soliton displacements that are large do not live in the real-space metric of the position they represent. This is why $X_i$'s must depend on $x_*$, the point where we look.

Assuming that \eqref{bethe} give soliton positions that are accurate on large scales -- that is, within the correct metric --, we may heuristically describe the metric change induced by \eqref{eqxymainX} away from the observation point $x_*$. For, say, $X_i<x_*$, the difference is in accounting for positions that lie within the interval $(X_i,x_*)$ (disregarding the $\varep$ regularisation of the sign function): on the right-hand side of \eqref{eqxymainX}, these give a sign function that is negative, while within \eqref{bethe} it is positive. As $\varphi_{ij}<0$, for fixed $y_i$, the value of $X_i$ is therefore ``too much to the right''. Thus, the magnifying-glass effect, the deformation of the metric away from the observation point $x_*$, is expected on large scales to bring accurate soliton positions, overall, towards $x_*$. This is, of course, in a way that is not smooth on small scales, because of interaction effects not accounted for by this simple arguments.

Theorem \ref{theomg} applies to these deformed positions. The ``error'' can be estimated by noticing that for a displacement $d$, the number of crossing missed is proportional to the ``true'' soliton density times $d$. For a finite, bounded density, the true displacements are therefore proportional to the soliton displacements with respect to the observation point.

Let us denote $X_i(x_*) = X_i(x_*;\bs\chi,\bs y)$. It is simple to argue  that the fixed points $x_i$ defined by $X_i(x_i) = x_i$ should satisfy  \eqref{bethe}, if for every $i$, the function $X_i(x_*)$ is continuous and has a single fixed point:
\beq
	x_{*,i}\ :\  X_i(x_{*,i}) = x_{*,i}.
\eeq
The fixed point is the position of soliton $i$ when it is in the centre of the lens, hence it is a ``precise'' soliton position. As, by Lemma \ref{lemd}, for all $x_*$ small enough (towards $-\infty$) we have $X_i(x_*)>x_*$, then for every $x_*$ such that $X_i(x_*)>x_*$ (resp.~$X_i(x_*)<x_*$), we have $x_{*,i}>x_*$ (resp.~$x_{*,i}<x_*$). Hence, for all $\varep$ small enough (towards 0) and assuming nonzero separations of soliton positions,
\beq
	\sgn_\varep(X_j(x_{*,i}) - x_{*,i}) = \sgn(x_{*,j}-x_{*,i}) \quad (j\neq i).
\eeq
Eq.~\eqref{eqxymainX} then implies that
\beq\label{xuniversal}
	x_i = x_{*,i}
\eeq
satisfies \eqref{bethe}. The argument fails because in general $X_i(x_*)$ are not necessarily continuous functions, and do not necessarily have a single fixed point.


\end{document}